%% file: main.tex
\documentclass{article}
\usepackage{graphicx} 
\usepackage[dvipsnames]{xcolor}

\usepackage{float}

\usepackage{fullpage}
\usepackage{amsthm}
\usepackage{amsmath,amssymb,amsfonts,nicefrac}
\usepackage{mathtools}
\usepackage{xspace}
\usepackage{color}
\usepackage{url}
\usepackage{hyperref}
\usepackage{stackengine}
\usepackage{bm}
\usepackage{bbm}
\usepackage{times}
\usepackage{thmtools,thm-restate}
\usepackage{enumitem}
\usepackage[capitalise,nameinlink,noabbrev]{cleveref}
\usepackage{booktabs}

\newcommand\skipi{{\vskip 10pt}}

\newcommand{\E}{\mathop{\mathbb{E}}}

\newcommand{\ind}{\mathbbm{1}}

\newcommand{\eps}{\varepsilon}

\newcommand{\mc}{\mathcal}
\newcommand{\poly}{{\sf poly}}
\newcommand{\Dec}{{\sf Dec}}
\renewcommand\leq{\leqslant}
\renewcommand\geq{\geqslant}
\renewcommand\le{\leqslant}
\renewcommand\ge{\geqslant}

\DeclarePairedDelimiter{\ceil}{\lceil}{\rceil}
\DeclarePairedDelimiter{\floor}{\lfloor}{\rfloor}

\theoremstyle{plain} 
\newtheorem{theorem}{Theorem}[section]

   \newtheorem{fact}[theorem]{Fact}
   \newtheorem{lemma}[theorem]{Lemma}
   \newtheorem{corollary}[theorem]{Corollary}
   \newtheorem{claim}[theorem]{Claim}
   \newtheorem{remark}[theorem]{Remark}
   \newtheorem{definition}[theorem]{Definition}

\DeclareMathOperator\supp{supp}

\title{
$3$-Query RLDCs are Strictly Stronger than $3$-Query LDCs}
\author{
Tom Gur\thanks{Department of Computer Science and Technology, University of Cambridge, UK. Supported by ERC Starting Grant 101163189 and UKRI Future Leaders Fellowship MR/X023583/1.}
\and
Dor Minzer
\thanks{Department of Mathematics, Massachusetts Institute of Technology, Cambridge, USA. Supported by NSF CCF award 2227876 and 
 NSF CAREER award 2239160.}
\and
Guy Weissenberg
\thanks{EPFL, Lausanne, Switzerland. Supported in part by the Ethereum Foundation.}
\and
Kai Zhe Zheng\thanks{Department of Mathematics, Massachusetts Institute of Technology, Cambridge, USA. Supported by the NSF GRFP DGE-2141064.}
}
\date{\vspace{-5ex}}

\newcommand{\List}{{\sf List}}

\newcommand{\comp}{{\sf comp}}

\newcommand{\dec}{{\sf dec}}
\usepackage{xcolor}

\newcommand{\sat}{{\sf SAT}}
\newcommand{\polylogn}{\poly(\log(N))}
\DeclareMathOperator{\ZZ}{ZZ}
\DeclareMathOperator{\Rep}{R}
\newcommand{\maj}{{\sf Maj}}
\DeclareMathOperator{\inmap}{IN}
\DeclareMathOperator{\outmap}{OUT}
\newcommand{\hdx}{{\sf hdx}}
\newcommand{\dip}{{\sf dp}}
\newcommand{\dist}{{\sf dist}}
\newcommand{\out}{{\sf out}}
\newcommand{\inner}{{\sf in}}

\newcommand{\Out}{\mathsf{out}}
\newcommand{\OutPCPDecoder}{\PCPDecoder_{\Out}}
\newcommand{\OutListSize}{\ListSize_{\Out}}
\newcommand{\OutPCPError}{\delta_{\Out}}
\newcommand{\OutLeftSide}{\LeftSide_{\Out}}
\newcommand{\OutRightSide}{\RightSide_{\Out}}
\newcommand{\OutEdgeSet}{\EdgeSet_{\Out}}
\newcommand{\OutConstraint}[1]{\Constraint{\Out,#1}}
\newcommand{\OutEdgeDistributionName}{\EdgeDistributionName^{\Out}}
\newcommand{\OutEdgeDistribution}[1]{\OutEdgeDistributionName_{#1}}
\newcommand{\OutLeftAlphabet}{\Alphabet{\OutLeftSide}}
\newcommand{\OutRightAlphabet}{\Alphabet{\OutRightSide}}

\newcommand{\OutIndexInDecodingSet}{t}

\newcommand{\OutDecoder}[1]{\DecoderSymbol^{\Out}_{#1}}

\newcommand{\OutLeftVertex}{\textbf{a}}
\newcommand{\OutRightVertex}{\textbf{b}}
\newcommand{\OutLeftProof}{\TableProof{\OutLeftSide}}
\newcommand{\OutRightProof}{\TableProof{\OutRightSide}}

\newcommand{\In}{\mathsf{in}}
\newcommand{\InPCPDecoderSymbol}{\PCPDecoder_{\In}}
\newcommand{\InPCPDecoder}{\PCPDecoder_{\In,\OutLeftVertex}}
\newcommand{\InListSize}{\ListSize_{\In}}
\newcommand{\InPCPError}{\delta_{\In}}
\newcommand{\InLeftSide}{\LeftSide_{\In,\OutLeftVertex}}
\newcommand{\InRightSide}{\RightSide_{\In,\OutLeftVertex}}
\newcommand{\InEdgeSet}{\EdgeSet_{\In,\OutLeftVertex}}
\newcommand{\InConstraint}[1]{\Constraint{\In,#1}}
\newcommand{\InLeftAlphabet}{\Alphabet{\InLeftSide}}
\newcommand{\InRightAlphabet}{\Alphabet{\InRightSide}}

\newcommand{\InDecoder}[2]{\Decoder{\In,\OutLeftVertex}{#1}}
\newcommand{\InRightVertex}{(\aaa; b)}
\newcommand{\InLeftProof}{\TableProof{\InLeftSide}}
\newcommand{\InRightProof}{\TableProof{\InRightSide}}

\newcommand{\InSpecificRandomness}{r}

\begin{document}

\maketitle
\begin{abstract}
We construct $3$-query relaxed locally decodable codes (RLDCs) with constant alphabet size and length $\tilde{O}(k^2)$ for $k$-bit messages.
Combined with the lower bound of $\tilde{\Omega}(k^3)$ of [Alrabiah, Guruswami, Kothari, Manohar, STOC 2023] on the length of locally decodable codes (LDCs) with the same parameters, we obtain a separation between RLDCs and LDCs, resolving an open problem of [Ben-Sasson, Goldreich, Harsha, Sudan and Vadhan, SICOMP 2006].

Our RLDC construction relies on two components. First, we give a new construction of probabilistically checkable proofs of proximity (PCPPs) with $3$ queries, quasi-linear size, constant alphabet size, perfect completeness, and small soundness error. This improves upon all previous PCPP constructions, which either had a much higher query complexity or soundness close to $1$. 
Second, we give a query-preserving transformation from PCPPs to RLDCs.

At the heart of our PCPP construction is a $2$-query decodable PCP (dPCP) with matching parameters, and our construction builds on the HDX-based PCP of [Bafna, Minzer, Vyas, Yun, STOC 2025] and on the efficient composition framework of [Moshkovitz, Raz, JACM 2010] and [Dinur, Harsha, SICOMP 2013].
More specifically, we first show how to use the HDX-based construction to get a dPCP with matching parameters but a large alphabet size, and then prove 
an appropriate composition theorem (and related transformations) to reduce the alphabet size in dPCPs.
\end{abstract}

\section{Introduction}
Probabilistically checkable proofs (PCPs, for short) are a central object of study in theoretical computer science, with connections to hardness of approximation, cryptography, error-correcting codes, expansion in graphs and hypergraphs, property testing, analysis of Boolean functions and more. The connection to error-correcting codes, which is the focus of this paper, runs in both directions. In one direction, arguably all PCP constructions rely on error-correcting codes (such as the Reed-Solomon, Reed-Muller and Hadamard codes) as building blocks. In the other direction, as we shall see later on, strong PCPs are often building blocks in error-correcting codes with additional desirable properties.
To discuss PCPs and their parameters, we shift to the perspective of constraint satisfaction problems, defined as follows.
\begin{definition}\label{def:PCP_const}
    An instance $\Psi$ of a constraint satisfaction problem (CSP) consists of 
    a weighted uniform hypergraph $G = (V, E, \mc{P})$, where $\mc{P}\colon E\to (0,\infty)$ is a distribution over the edges, as well as an alphabet $\Sigma$ and a collection of constraints $\{\Phi_e\colon \Sigma^q\to\{0,1\}\}_{e\in E}$, one for each edge. 
    
    An assignment 
    $T\colon V\to\Sigma$ is said to satisfy 
    the constraint on $e=(v_1,\ldots,v_q)\in E$ if 
    $\Phi_e(T(v_1),\ldots,T(v_q)) = 1$, and the value of $T$ is defined as the total weight of edges $e\in E$ that are satisfied. The value of the instance $\Psi$ is the maximum value obtained by any assignment $T$.
\end{definition}
The PCP theorem~\cite{FGLSS,AroraSafra,ALMSS} asserts that 
there exists $\eps>0$ and a polynomial time reduction $f$ from $3$-SAT to CSP instances 
on $O(1)$-uniform hypergraphs with constant alphabet size, such that:
\begin{enumerate}
    \item {\bf Completeness:} if $\varphi$ is satisfiable, then $f(\varphi)$ has value $1$.
    \item {\bf Soundness:} if $\varphi$ is unsatisfiable, then $f(\varphi)$ has value at most $1-\eps$.
\end{enumerate}
Since the parameters of PCPs turn out to be of great importance in applications, subsequent research has put significant effort into optimizing each one of them. These parameters include the blow-up of the reduction $f$ (sometimes called the length/size of the PCP), the number of queries $q$ (the uniformity of the hypergraph), the structure of the constraints $\Phi_e$, the soundness error, and the alphabet size. Indeed, in applications, one ideally wants a nearly-linear size PCP (as opposed to a polynomial blow-up as above) with as few queries as possible (optimally $q=2$), and with soundness close to $0$ (as opposed to just bounded away from $1$ as above).

\subsection{PCPs of Proximity, Decodable PCPs and Their Combinatorial Analogs}

 It turns out that to optimize parameters of PCPs, one needs, as building blocks, PCPs 
 that have more specialized properties. In this section, we discuss two key such variants. 

\subsubsection{PCPs of Proximity} 
PCPs of proximity (PCPPs), also known as assignment testers, were initially introduced by the works ~\cite{bghsv,DinurReingold} to get PCPs of small size (say, nearly linear). They are variants of PCPs with a more subtle notion of soundness. Whereas a PCP encoding of an assignment to $3$-SAT can be used to check whether it is satisfying or not, a PCPP enables checking if the encoded assignment is \emph{close} to a satisfying assignment. Here and throughout, the distance between a string and a language is the normalized Hamming distance, measuring the fraction of coordinates that need to be changed to get an input in the language. Below is the formal definition of PCPPs.
\begin{definition} \label{def:pcpp}
Let $q\in\mathbb{N}$ and $\eps,\eta>0$.
A $q$-query probabilistically checkable proof of proximity (PCPP) with proximity parameter $\eta$ and soundness $1-\eps$ for a language $\mc{L} \subseteq \Sigma_0^n$ consists of a CSP instance $\Psi = (X\cup V, E, \mc{P}, \Sigma, \{\Phi_e\}_{e \in E})$ over a $q$-uniform hypergraph, where $X$ has size $n$ and is identified with the coordinates of $\Sigma_0^n$ and $\Sigma_0\subseteq \Sigma$. On an input $w: X \to \Sigma_0$ it has the following properties:
\begin{itemize}
    \item \textbf{Completeness:} if $w \in \mc{L}$, then there exists an assignment $A: X\cup V \to \Sigma$ extending $w$ such that 
    \[
    \Pr_{e \sim \mc{P}}[\Phi_e(A(e)) = 1] = 1.
    \]
    \item \textbf{$(\eta, \eps)$-Soundness:} if $w$ is $\eta$-far from all members of $\mc{L}$, then for any $A: X\cup V \to \Sigma$ extending $w$ we have
    \[
    \Pr_{e\sim \mc{P}}[\Phi_e(A(e)) = 1] \leq 1 - \eps. 
    \]
\end{itemize}
\end{definition}

We note that in contrast to a standard PCP, where one has full access to the input $w$, in a PCPP one only has \emph{oracle access} to $w$. In particular, the number of locations read in $w$ is also counted towards the query complexity. Thus, if the input string is very close to the language, one cannot expect constant soundness to hold in a PCPP with constantly many queries, and hence PCPPs must give soundness guarantees which are parameterized by the distance of the input from the language.

The definition of PCPPs was originally motivated by considerations arising from composition of PCPs. In a sense, a PCPP can be used to ensure that if an assignment $A$ passes almost all of the constraints of a CSP, then it is associated with some (unique) actual word from $\mc{L}$. This assertion is helpful when composing PCPs with soundness close to $1$. In fact, using PCPPs, one can state generic composition theorems in the close-to-$1$ soundness regime; see~\cref{sec:techniques_PCP} for a more thorough discussion. This feature plays a crucial role in Dinur's combinatorial proof of the PCP theorem~\cite{BS,DinurReingold,Dinur}, which, among other things, is the first PCP construction of quasi-linear length and soundness bounded away from $1$.  

\paragraph{Constructing PCPPs.} The works~\cite{bghsv,DinurReingold} show composition theorems for PCPPs, which, alongside Dinur's gap amplification technique~\cite{Dinur}, yield a $2$-query PCPP with small (but constant) $\eta,\eps>0$, see~\cite{Mie}. As we will 
see, however, some applications require PCPPs with small soundness. That is, the $1-\eps$ above is replaced by a small positive $\eps$, and intuitively this may be achievable if $\eta$ is close to $1$. The composition techniques of~\cite{bghsv,DinurReingold} cannot achieve soundness better than $1/2$, and we are not aware of any other previous attempts at establishing a small soundness PCPP with concretely small values of $q$, such as $q=2,3$ (one can always improve the soundness by increasing the query complexity, of course). While PCPPs seem only slightly stronger than standard PCPs, we are not aware of any black-box transformation from PCPs to PCPPs. For example, we are not aware of a PCPP analogous to the construction of~\cite{MoshkovitzRaz}.

\subsubsection{Decodable PCPs} 
The parallel repetition theorem of Raz~\cite{Raz} 
is a generic method for improving the soundness of
a $2$-query PCP. It is, however, incapable of giving soundness which is vanishing with the instance size. Additionally, it cannot give small (but constant) soundness PCPs with nearly-linear length. Thus, to achieve such a size-efficient 
PCPs, one must resort to other techniques, 
which all involve the composition of PCPs in one way or another. Unfortunately, composition is considerably more difficult in the low-soundness regime. The works~\cite{RazSafra,AroraSudan,DFKRS,DHK} show how
to get low-soundness PCPs from a suitable composition theorem, but the number of queries their PCPs make grows as the soundness decreases. 
In particular, their approach falls short of giving $2$-query PCPs with small soundness.

The work of~\cite{MoshkovitzRaz}, later abstracted in~\cite{dh}, gives a novel composition technique in the low-soundness regime that maintains the number of queries being $2$. A
key concept of that approach is decodable PCPs (dPCPs, for short).\footnote{The work~\cite{MoshkovitzRaz} uses concrete instantiations of this object and refers to them as locally decode/reject codes. The terminology, which by now is more standard in the area, is from~\cite{dh}.} Below is a somewhat simplified definition sufficient for this introductory section, and we refer the reader to~\cref{def:dpcp} for the more precise definition that we actually use.
\begin{definition} \label{def:dPCP_intro}
Let $q,L\in\mathbb{N}$ and $\eps>0$.
A $q$-query decodable PCP (dPCP) for a language $\mc{L} \subseteq \Sigma_0^n$ consists of a CSP instance $\Psi = (V, E, \mc{P}, \Sigma, \{\Phi_e\}_{e \in E})$, a collection of decoding maps $\{D_t\colon \Sigma^q\to\Sigma_0\cup\{\bot\}\}_{t\in [n]}$, and a collection of decoding distributions $\{\mc{P}_t\}_{t\in [n]}$ over $E$, one for each coordinate of $\mc{L}$. It has the following properties:
\begin{itemize}
    \item \textbf{Completeness:} for any $w\in \mc{L}$ there exists an assignment 
    $A\colon V\to \Sigma$ satisfying that
    \[
    \Pr_{e \sim {\mc P}}[\Phi_e(A(e)) = 1] = 1,
    \qquad
    \Pr_{t\in [n],e\sim\mc{P}_t}[D_t(A(e)) = w_t]=1.
    \]
    In words, there is an assignment to $\Psi$ with value $1$ that decodes to $w$.
    \item \textbf{$(L, \eps)$-List-decoding soundness:} for any $A: V \to \Sigma$ we may find a list of words of size at most $L$, say $\{y_1,\ldots,y_L\}\subseteq\mc{L}$, such that
    \[
    \Pr_{t\in [n], e \sim {\mc P}_t}[\Phi_e(A(e)) = 1\land D_t(A(e)) \not\in\{(y_1)_t,\ldots,(y_L)_t\}] \leq \eps. 
    \] 
    In words, the probability that a randomly chosen edge $e$ is satisfied by $A$ but decodes to something inconsistent with the list, is at most $\eps$.
\end{itemize}
The size of $\Psi$ is defined to be $|V|+|E|$.
\end{definition}
Intuitively, in the low-soundness regime, there is no hope of associating a single word from $\mc{L}$ with an assignment $A$ that satisfies a non-negligible fraction of the constraints, so one must resort to a list-decoding notion of soundness. It turns out that requiring the local decoding feature is the crucial additional feature that enables generic composition theorems, and we elaborate on this in~\cref{sec:techniques_PCP}.

\paragraph{Constructing dPCPs.} Most known dPCP constructions are based on some variant of the manifold-versus-point test (arising from Reed-Muller codes and low-degree testing). The only other known dPCP construction is by Dinur and Meir~\cite[Theorem 1.6]{DinurMeir}, which gives a $2$-query dPCP construction with small soundness. Their construction, though, only achieves polynomial size (as opposed to nearly-linear size), and has super-constant alphabet size. While the fact that their dPCP is not of nearly-linear size seems inherent to their approach, their alphabet size being super-constant is due to the lack of composition results for dPCPs. 

\subsubsection{Relaxed Locally Decodable Codes}
The above variants of PCPs have close relations to locally testable codes, locally decodable codes, and relaxed locally decodable codes, which can all be thought of as combinatorial variants of PCPs. Indeed, composition techniques in the high-soundness regime often require constructions coming from locally testable codes, and in the low-soundness regime, they often require constructions coming from locally decodable codes. 

In the other direction, Dinur~\cite{Dinur} shows how to obtain locally testable codes  matching the parameters of her PCP. The work~\cite{bghsv} shows (and actually, defines) a relaxed variant of locally decodable codes, called \emph{relaxed locally decodable codes} (or RLDC in short), which is built upon their PCPP construction. RLDCs will be a central object in this paper, and roughly speaking, they are codes in which individual message bits can be locally recovered from a corrupted codeword, with the relaxation that the decoder may output a rejection symbol $\bot$ upon detecting corruption. Below is the formal definition.

\begin{definition} 
A code $C:  \{0,1\}^k \to \Sigma_0^n$ is called $q$-query RLDC with decoding radius $\delta$ if there is a randomized decoding algorithm $\Dec$, which satisfies the following:
    \begin{itemize}
        \item \textbf{Completeness:} for every $m \in \{0,1\}^k$ and $i \in [k]$, $\Pr[\Dec(\mc{C}(m), i) = m_i] = 1$.
        \item \textbf{Soundness:} if $w \in \Sigma_0^n$ is $\delta$-close to $\mc{C}(m)$ for some $m\in\{0,1\}^k$, then for every $i \in [k]$ we have that $\Pr[\Dec(w, i)\in \{m_i, \perp \}] \geq \frac{2}{3}$.
        \item \textbf{Success-rate:} if $w \in \Sigma_0^n$ is $\delta$-close to $\mc{C}(m)$ for some $m\in\{0,1\}^k$, then for at least $(1-16\delta-o(1))$-fraction of $i \in [k]$ we have that $\Pr[\Dec(w, i) = m_i] \geq \frac{2}{3}$.
    \end{itemize}
\end{definition}

The above relaxation of LDCs allows for constructions with dramatically better parameters compared to their non-relaxed counterparts, where the decoder is not allowed to output $\bot$. In particular, the construction in \cite{bghsv} is an $O(1)$-query RLDC with nearly-linear length, while the best known construction of $O(1)$-query (non-relaxed) LDCs has superpolynomial length \cite{yekhanin08,efr12}. These improved parameters, in turn, have led to many applications, including in property testing lifting lemmas \cite{GR18,CG18}, data structures for code membership and polynomial evaluation \cite{CGW09}, distributions that defy the law of large numbers \cite{briet2016outlaw}, constructions of universal locally testable codes \cite{GG18}, interactive proofs of proximity \cite{goldreich21}, quantum complexity separations \cite{DGMT22}, and beyond.

However, despite the great attention that RLDCs have received (cf. \cite{GGK19, GRR20, GL19, AS21,block2022relaxed, CGS22, CY22, Goldreich24, KM24, CY24}), prior to this work, no separation result was known between RLDCs and LDCs. This is since, other than the recent lower bound of $\tilde{\Omega}(k^3)$ of  \cite{AlrabiahGKM23} for $3$-query LDCs, the known lower bounds for RLDCs \cite{GL19,DGL23,GGS25} nearly match the state-of-the-art lower bounds for non-relaxed LDCs \cite{KT00,WW05}.

\paragraph{Constructing RLDCs.} There are two main paradigms for constructing RLDCs. In the low-query regime, constructions start with a base code and augment each bit of the message with a PCPP that certifies consistency of the (relevant local view of the) codeword with that bit. In the high-rate regime, constructions are closer to those for high-rate locally testable codes. Namely, they are obtained by iteratively alternating between transformations that reduce query complexity at the expense of distance, and transformations that increase distance while only mildly worsening locality (in the spirit of Dinur’s gap-amplification proof of the PCP theorem).

\subsection{Main Results}\label{sec:main_results}
In this section, we state our main results, which will be convenient to express in terms of the Circuit-SAT problem.
\begin{definition}
For a circuit $\varphi:\Sigma_0^n\to \{0,1\}$ we define the language ${\sf SAT}(\varphi) :={x\in \Sigma_0^n\ | \ \varphi(x)=1}$.
\end{definition}
We will also discuss the size of a circuit $\varphi:\Sigma_0^n\to \{0,1\}$, which is defined in the standard way. A circuit is a directed acyclic graph with a single vertex of out-degree $0$ called the output. Each node has fan-in $0$ or $2$, and is labeled either by an input variable if its fan-in is $0$, or else by an arbitrary function from $\Sigma_0^2$ to $\Sigma_0$. The size of a circuit is the sum of the number of nodes and the number of edges in the graph, and we say it computes $\varphi$ if for any input $x\in\Sigma_0^n$, the value of the output node is $\varphi(x)$.
\subsubsection{Quasi-linear, Low-soundness $2$-Query dPCPs and $3$-Query PCPPs}
Our first result is a new dPCP construction for ${\sf SAT}(\varphi)$ for any circuit $\varphi$, improving upon all previous constructions.
\begin{theorem}\label{thm:dPCP_intro}
For all $\eps>0$ there exists $L\in\mathbb{N}$, such that for all $\Sigma_0$ there is $C\in\mathbb{N}$ such that the following holds.
If $\varphi\colon \Sigma_0^n\to\{0,1\}$ is a circuit of size $N$, then the language ${\sf SAT}(\varphi)$ has a $2$-query dPCP with size
$N\log^C N$ that has $(L,\eps)$-list-decoding soundness and $O_{\eps,|\Sigma_0|}(1)$ alphabet size.
\end{theorem}
It is clear that any dPCP must make at least $2$ queries, so~\cref{thm:dPCP_intro} achieves the best parameters possible, up to the logarithmic factor in the size and the implicit dependency between $L$ and $\eps$ (which for us is $L = (1/\eps)^{O(1)}$).

Using a straightforward transformation from dPCPs to PCPPs, we use~\cref{thm:dPCP_intro} to conclude a 3-query PCPP construction with small soundness.

\begin{theorem}\label{thm:PCPP_intro}
For all $\eps>0$ there is $\eta>0$, such that for all $\Sigma_0$ there is $C>0$, such that the following holds. If $\varphi\colon \Sigma_0^n\to\{0,1\}$ is a circuit of size $N$, then the language ${\sf SAT}(\varphi)$ has a $3$-query PCPP with size $N\log^C N$ that has $(1-\eta,1-\eps)$-soundness and $O_{\eps,|\Sigma_0|}(1)$ alphabet size.
\end{theorem}
Intuitively, the construction in~\cref{thm:PCPP_intro} is the same as the construction in~\cref{thm:dPCP_intro}, except that we query a random index $i$ of the supposed word $w$ from $\mc{L}$, decode it using the dPCP, and compare to $w_i$. In other words, we ``spend’’ one query to check if the assignment to the CSP decodes to $w$. As long as one sticks to such a recipe, it seems impossible to achieve query complexity smaller than $3$. As far as we know, though, it may be possible to achieve a small soundness, $2$-query PCPP for certain languages of interest.

At the cost of paying one additional query, we can get a tight relationship between $\eps$ and $\eta$, which is important in some applications. The statement below is meaningful when $\eps>0$ is thought of as small, and $\eta$ can be taken to be any value between $0$ and $1$.
\begin{theorem}\label{thm:4PCPP_intro}
For all $\eps>0$, $0<\eta<1$ and $\Sigma_0$, there is $C>0$ such that the following holds. If $\varphi\colon \Sigma_0^n\to\{0,1\}$ is a circuit of size $N$, then the language ${\sf SAT}(\varphi)$ has a $4$-query PCPP with size $N\log^CN$ that has $(\eta,\eta+\eps)$-soundness and $O_{\eps,\eta,|\Sigma_0|}(1)$ alphabet size.
\end{theorem}

\subsubsection{$3$-Query RLDCs are Strictly Stronger than $3$-Query LDCs}

We provide a new black-box transformation of PCPPs into RLDCs, which preserves query complexity. All prior similar transformations led to an increase in the query complexity, and in particular, they cannot be used to get a $3$-query RLDC.
Combining our transformation with~\cref{thm:PCPP_intro}, we construct RLDCs with $3$ queries. More precisely, we show the following.
\begin{theorem}\label{thm:RLDC_intro}
There are $\eps,\delta>0$, constant $C$, and alphabet $\Sigma_0$  such that for sufficiently large $k$, there is a $3$-query relaxed locally decodable code $\mathcal{C}\colon \{0,1\}^k\to \Sigma_0^{n}$ with decoding radius $\delta$, distance $\eps$ and
$n\leq k^2\log^C k$.
\end{theorem}

In words,~\cref{thm:RLDC_intro} establishes the existence of a $3$-query RLDC with constant decoding radius and distance, and roughly quadratic blow-up in the message length. All previous $O(1)$-query RLDC constructions (such as in~\cite{bghsv,GGK19,GRR20,AS21,CGS22}) require more than $3$ queries.
The significance of~\cref{thm:RLDC_intro} achieving $3$ queries is more apparent once it is compared to the lower bounds for LDCs of~\cite{AlrabiahGKM23}:
\begin{theorem}[Theorem A.2 in \cite{AlrabiahGKM23}]\label{thm:ldc_lb}
Suppose $C : \{0,1\}^k \to \Sigma_0^n$ is a $3$-query locally decodable code with decoding radius $\Omega(1)$, distance $\Omega(1)$ and alphabet $|\Sigma_0|=O(1)$. Then $n \ge \Omega(k^3/\log^6 k)$.
\end{theorem}
In particular, combining~\cref{thm:RLDC_intro,thm:ldc_lb} gives the following corollary, which asserts that $3$-query RLDCs achieve strictly stronger parameters than $3$-query LDCs. This provides a separation between RLDCs and LDCs, resolving an open problem due to Ben-Sasson, Goldreich, Harsha, Sudan, and Vadhan \cite{bghsv} (see also Goldreich’s textbook~\cite[Chapter 13.4.4]{GolBook} and recent surveys \cite{Goldreich24,Gas24}).

\begin{corollary}
There is a family of $3$-query RLDCs with constant decoding radius, distance, and alphabet size that achieves polynomially better length than any family of $3$-query LDCs with constant decoding radius, distance, and alphabet size.
\end{corollary}

\begin{remark}
The recent work \cite{grigorescu2025relaxed} showed that linear 3-query RLDCs with a \emph{binary} alphabet imply non-relaxed LDCs with similar parameters. While our construction is not linear, this barrier suggests that a non-binary alphabet is likely a prerequisite for achieving such polynomial separation.
\end{remark}

\paragraph{Other Implications.}
Our results have a few direct implications using existing results in the literature, and we mention one of them.

 Using the connection between gaps of PSPACE-hardness for the $q$-CSP reconfiguration problem and parameters of PCPPs established by~\cite{guruswami_reconfiguration}, we get the following two corollaries:  
\begin{corollary}\label{cor:recon1}
For any constant $\eta>0$ there is an alphabet
$\Sigma$ of size $O_{\eta}(1)$, 
such that $\mathrm{Gap}_{1,\frac{1}{2}+\eta}\,5$-\textrm{CSP}$_\Sigma$
Reconfiguration is $\mathsf{PSPACE}$-hard.
\end{corollary}
\begin{corollary}\label{cor:recon2}
For any constant $\eta>0$ there is an alphabet
$\Sigma$ of size $O_{\eta}(1)$, 
such that $\mathrm{Gap}_{1,0.9+\eta}\,2$-\textrm{CSP}$_\Sigma$
Reconfiguration is $\mathsf{PSPACE}$-hard.
\end{corollary}
These results improve upon \cite{HiraharaOhsaka2024,guruswami_reconfiguration}, who proved that $\mathrm{Gap}_{1,0.9942+\eta}\,2$-\textrm{CSP}$_\Sigma$ is PSPACE-hard for some $\Sigma$ of size $O_{\eta}(1)$. We remark that for both of them, it is important that we use~\cref{thm:4PCPP_intro}, which has a near-tight relationship between the distance and rejection parameters.

\subsection{Techniques in our PCP Constructions}\label{sec:techniques_PCP}
In this section, we discuss the techniques in the proofs of~\cref{thm:dPCP_intro,thm:PCPP_intro}. 
We begin with a general discussion about how PCPs are constructed.
\subsubsection{The General PCP Recipe}\label{sec:recipe}
Essentially all known PCP constructions can be viewed as using the following high-level recipe. Starting with a language $\mc{L}\subseteq\Sigma_0^n$, they consist of:
\begin{enumerate}
    \item {\bf Outer PCP:} is a fairly good PCP construction for the language $\mc{L}$, except that its alphabet $\Sigma$ is too large (being either polynomial in $n$ or sometimes even quasi-polynomial in $n$, provided it has small ``decision complexity'', a notion we will define shortly).
    \item {\bf Inner PCP:} is a PCP construction for a language $\mc{L}'\subseteq (\Sigma_0')^{n'}$ that arises from the constraint structure or the alphabet structure of the outer PCP. The key additional features of the inner PCP construction are that its alphabet size is much smaller than that of the outer PCP, and that it satisfies a stronger notion of soundness as discussed above. The inner PCP could be fairly ``lossy'' with respect to other parameters, such as its length relative to $n'$. This is because the $n'$ we take is typically much smaller than the above $n$, as the intention here is to encode an alphabet symbol of the outer PCP.
    \item {\bf A composition procedure:} this is a technique that combines the outer PCP and the inner PCP constructions to yield a composed PCP. Roughly speaking, the soundness and completeness of the composed PCP are governed by the soundness and completeness of the two components, the alphabet size of the composed PCP is inherited from the inner PCP construction, and the size is the product of the sizes.
\end{enumerate}
To illustrate this recipe, we define a special class of CSPs, which goes by the name Label-Cover.
\begin{definition}
    An instance $\Psi$ of Label-Cover consists of a weighted bipartite graph $G=(A\cup B, E, \mc{P})$, where $\mc{P}\colon E\to (0,\infty)$ is a distribution over the edges, alphabets $\Sigma_A$, $\Sigma_B$ and constraints $\{\Phi_e\colon \Sigma_A\times\Sigma_B\to\{0,1\}\}_{e\in E}$. 
    \begin{enumerate}
        \item The value of a pair of assignments $T_A\colon A\to\Sigma_A$, 
    $T_B\colon B\to\Sigma_B$ is the sum of weights of constraints they satisfy, and the value of $\Psi$ is the maximum weight of constraints satisfied by any $T_A,T_B$.
    \item We say $\Psi$ is a projection Label-Cover instance if for each edge $e\in E$, the constraint $\Phi_e$ is projection from $\Sigma_A$ to $\Sigma_B$. Namely, there exists a map $\phi_e\colon \Sigma_A\to\Sigma_B$ such that $\Phi_e(\sigma,\tau) = 1$ if and only if $\tau = \phi_e(\sigma)$.\footnote{The notion of projection Label-Cover is essentially equivalent to the notion of robust PCPs, which is sometimes used in the literature.}
    \end{enumerate}
\end{definition}
Suppose we have an outer PCP construction $\Psi = (G=(V,E,\mc{P}),\Sigma,\{\Phi_e\}_{e\in E})$ as in~\cref{def:PCP_const} which is good (in the sense that $\eps>0$ and $q$ are absolute constants), but its alphabet $\Sigma$ is too large. To discuss composition, it is first convenient to convert $\Psi$ into a Label-Cover instance $\Psi'$. 

The instance $\Psi'$ is defined over the graph $H = (A\cup B, E', \mathcal{P}')$ whose sides $A$ and $B$ are 
the edge set $E$ and the vertex set $V$ of $\Psi$ respectively, its weight function $\mc{P}'$ is naturally derived from $\mc{P}$ by: sample $e\sim \mc{P}$, sample a vertex $v$ from $e$ uniformly and output $(e,v)$. For $e = (v_1,\ldots,v_q)\in A$, we interpret the alphabet $\Sigma_A$ as an assignment of $\Sigma$ symbols to the variables $v_1,\ldots,v_q$ that satisfies the constraint $\Phi_e$; for $v\in B$, we interpret the alphabet $\Sigma_B$ simply as $\Sigma$. Finally, the constraint on $(e,v)\in E'$ is that the symbol assignment to $e$ and the symbol assigned to $v$ are consistent with respect to the value given to $v$. Using error-correcting codes, the instance $\Psi'$ can be transformed into an instance $\Psi''$
over a similar graph $(A\cup C, E'',\mc{P}'')$ with similar constraints and alphabets, except that the alphabet $\Sigma_{C}$ now has constant size, and for the sake of this introduction, we just take $\Sigma_C = \{0,1\}$.

Suppose for the sake of simplicity that the constraint $\Phi_e$ is the same for all edges $e\in E$. In this case, to reduce the alphabet size of $\Psi'$, we would like to be able to encode assignments to $H$, so that we can check: (1) the hardcoded condition in $\Sigma_A$, which is morally how the constraint $\Phi_e$ is manifested, 
(2) the consistency between the assignment to $A$ and $B$, which ultimately ensures that the local assignments to the $A$ side come from a legitimate assignment to $V$. Towards this end, we may think of symbols in $\Sigma_A$ as strings of length $n'$, say $\Sigma_A\subseteq \{0,1\}^{n'}$, and take an inner PCP construction suitable for the language
\[
\mathcal{L'} = \{x\in \{0,1\}^{n'}~|~x\in \Sigma_A\}.
\]

At this point, it becomes clear that we care about the complexity of the language $\mc{L}'$, otherwise there is no reason we would have a decent PCP construction for it. This is the motivation for the notion of \emph{decision complexity}, which captures the circuit complexity of the constraints (in the context of CSP instances) or alphabet membership checks (in the context of Label-Cover), and we define it formally later. In our applications, $n'$ and the circuit complexity of $\mc{L}'$ will both be poly-logarithmic in $n$, and in particular much smaller than $n$.

In the high-soundness regime, the composition can now be performed by replacing a vertex in the $A$-side of $\Psi''$ with a copy of the inner PCP of $\mc{L}'$. Morally, this automatically takes care of issue (1) above, namely the hardcoded condition in the alphabet $\Sigma_A$. To address issue (2), namely the consistency between $A$ and $C$, we require the inner PCP to be a PCP of proximity for $\mc{L}'$. In that case, we could actually reliably read off coordinates of the PCPP corresponding to the coordinates in $\mc{L}'$ and compare them to what they are supposed to be from the assignment to the $C$ side.
This composed PCP can be either thought of as having more than $2$ queries, or equivalently, as having $2$ queries but soundness close to $1$ (by simply running each one of the above-mentioned tests with probability $1/2$).

\skipi
Composing in the low-soundness regime is more challenging, since the two tests above need to be 
combined into a single test in a way that preserves soundness. This is
precisely where the notion of decodable PCPs comes in handy. By definition, by reading off the locations corresponding to a constraint in a dPCP, we can -- at the same time -- check a constraint and decode a coordinate from the supposed word of the language.
With this description in mind, the main challenges in proving~\cref{thm:dPCP_intro} are as follows.
\begin{enumerate}
    \item {\bf Starting point:} no prior dPCP constructions achieve the parameters of~\cref{thm:dPCP_intro}, even if we relax the requirement on the alphabet size and allow it to be somewhat large (so as to serve as the outer PCP in the above recipe). The closest result in the literature is due to Dinur and Meir~\cite{DinurMeir}, but its length is significantly bigger than quasi-linear.
    \item {\bf Composition theorem:} we are not aware of any prior composition results that yield a composed dPCP. 
    As a matter of fact, in the low-soundness regime, we are only aware of a single query-efficient composition result for any variant of PCPs, namely the one in~\cite{MoshkovitzRaz,dh}. 
\end{enumerate}
In the rest of this section we discuss the way we overcome both of these challenges.
\subsubsection{Our Starting Point}
Our starting point is the recent PCP construction of~\cite{bmv}, which is based on high-dimensional expanders. Roughly speaking, the bulk of their work
goes into showing the existence of a quasi-linear, low-soundness PCP for NP, with a quasi-polynomial alphabet and poly-logarithmic decision complexity. They then reduce its alphabet using the techniques from~\cite{MoshkovitzRaz,dh}.

We show that, after appropriate modifications, applying their construction on the PCPP of Mie~\cite{Mie} gives a decodable $2$-query PCP (as opposed to just a normal PCP) with similar parameters. Since we are not aware of any black-box transformation, our argument has to carefully go through the steps of their construction and show that they all preserve the stronger notion of soundness. These steps include:
\begin{enumerate}
    \item {\bf Embedding a given PCPP into an HDX}: we use the same link-to-link routing protocol  to transform a given $2$-query dPCP to a $2$-query dPCP over the base graph of a high-dimensional expander $X$, provided the latter satisfies sufficiently good expansion properties.
    \item {\bf Amplification via direct product testing:} we show that the sparse direct product testing results in the low-soundness regime of~\cite{BM,BLM,DD1,DD2,DDL} as used in~\cite{bmv}, can also be used in the context of decodable PCPs to decrease the soundness of our HDX-based $2$-query dPCPs.
\end{enumerate}
Combining these two components gives a version of the $2$-query dPCP of Dinur and Meir~\cite{DinurMeir} with only quasi-linear size.
The next step in our argument is to show a way of reducing the alphabet size.

\subsubsection{Our Composition Theorem} 
To reduce the alphabet size of our HDX-based dPCP, we also use a composition theorem, except that we have to prove versions that yield dPCPs. To discuss our composition theorems we 
have to define a few additional parameters of interest of dPCPs, and 
to do so we first give the appropriate definition of label-cover dPCPs.
\begin{definition}\label{def:dpcp_intro2}
Let $\mc{L}\subseteq\Sigma_0^n$ be a language. A label-cover, decodable PCP for $\mc{L}$ is a tuple 
$$\mc{D} = ((A\cup B,E), 
\{\mc{P}_{t}\}_{t\in [n]},\Sigma_A,\Sigma_B,\{\Phi_e\}_{e\in E},\{D_{t}\}_{t\in [n]}),$$
    where $(A\cup B,E)$ is graph, 
    $\{\Phi_e\colon\Sigma_A\times\Sigma_B\to\{0,1\}\}_{e\in E}$ is a collection of constraints, $\mathcal{P}_t$ is a distribution over $E$ for each $t\in[n]$ and $D_t\colon A\times \Sigma_A\to \Sigma_0$ is a decoding map for each $t\in [n]$. 
    We say that $\mathcal{D}$ is a projection dPCP if for each $e\in E$, $\Phi_e$ is a projection constraint.
\end{definition}
There are several differences between~\cref{def:dpcp_intro2} and~\cref{def:dPCP_intro}. First, the former uses bipartite graphs, whereas the latter uses hypergraphs. Second, in the former,  
the decoding maps only read the left vertex and an assignment to it, whereas in the latter, they read an assignment to an edge. These two differences are relatively minor and are very natural in light of the transformation discussed in~\cref{sec:recipe}. 

Besides the obvious PCP-related parameters (such as completeness and soundness, which are defined analogously), we will also care about graph-theoretic parameters, such as left and right degrees of vertices. A somewhat less standard parameter that will be important for us is the decoding degree: 
\begin{definition}
    Let $\mc{D}$ be a label-cover decodable PCP as in~\cref{def:dpcp_intro2}. The decoding-degree of a vertex $a\in A$, denoted by ${\sf ddeg}(a)$, is the number of indices it is in charge of decoding, namely
    $|\{t\in [n]~|~\exists b\in B, (a,b)\in \supp(\mc{P}_t)\}|$. The decoding-degree of $\mc{D}$ is defined as $\max_{a\in A}{\sf ddeg}(a)$.
\end{definition} 
With this notion in place, we now discuss our composition theorem.\vspace{-1ex}
\paragraph{Composition when the decoding-degree is constant:} our composition theorem applies in the case we have a projection label-cover dPCP $\mc{D}$ with constant decoding-degree. In this case, we show how to combine two copies of it, one of the scale of $n$ as an outer PCP and one at the scale of a much smaller $n'$ as an inner PCP, to get a projection label-cover dPCP with similar parameters, except that its alphabet size is much smaller. 

This construction is a natural extension of the abstract composition theorem of Dinur and Harsha~\cite{dh}, and morally proceeds as follows. First, we use the standard right-alphabet reduction technique from~\cite{MoshkovitzRaz,dh} to ensure that $\Sigma_B$ is constant. Next, in the context of standard PCPs, the alphabet $\Sigma_A$ is replaced by $\Sigma_A'\subseteq\Sigma_B^{{\sf deg}(a)}$ reflecting the supposed assignments of all neighbors $b\in B$ of $A$; only strings that are fully consistent with some symbol in $\Sigma_A$ are allowed. 
This change in the alphabet of $A$ is not quite valid in our setting of dPCPs since it is not clear how to use the decoders of $\mc{D}$ to construct decoders for the new construction. To remedy that, we actually replace the alphabet $\Sigma_A$ with $\Sigma_A'\subseteq \Sigma_B^{{\sf deg}(a)}\times\Sigma_0^{{\sf ddeg}(a)}$, which is thought of as a string reflecting the assignments of all of the neighbors $b\in B$ of $a$, as well as of all of the $t\in [n]$ that $a$ is in charge of. Once again, in $\Sigma_A'$ we only allow strings that are fully consistent with some label $\sigma$ to $a$ from the original alphabet $\Sigma_A$, which additionally decodes correctly (under the decoding map $D_t(a,\cdot)$) on all $t$ that $a$ is in charge of.

In the simplified setting, all constraints $\Phi_e$ are the same, and a similar analysis as in~\cite{dh} can be carried out. The case that $\Phi_e$ may be different introduces several technical difficulties, stemming from the following technical issue. In the composed PCP, a left vertex corresponds to a right vertex $b$ in the outer PCP, and to the invocation of all of the inner PCPs of neighbors ${\bf a}$ of $b$ with a certain seed of randomness. Thus,
we need the inner PCPs neighbors ${\bf a}$ of $b$ to be ``aligned''; for example, it would be nice if they use the same seed of randomness. We show how to accomplish this, and this requires a few fairly technical transformations of dPCP that manipulate its size while retaining its parameters. 
\skipi
Unfortunately, this composition result alone is insufficient for our purposes. As is clear from the above description, the alphabet size of the composed PCP is at least exponential in the decoding degree of the outer PCP. With that in mind, our HDX-based $2$-query dPCP construction has poly-logarithmic decoding degree, so if we plainly apply the aforementioned composition theorem, the alphabet size of the composed PCP will be (at the very least) $2^{{\sf poly}(\log n)}$. In other words, nothing was gained from the composition.
Having a poly-logarithmic decoding degree seems inherent in the HDX construction (due to the use of routing protocols). In particular, it is not clear how to use the composition theorem we have just described to get a better dPCP. 

\vspace{-1ex}
\paragraph{Decoding-degree reduction:} in light of the above, it makes sense to ask if there is a transformation that reduces the decoding degree of a given dPCP without hurting the other parameters too much. Such transformations do exist for more combinatorial notions, such as right degree and right alphabet size.
While we do not have a generic result of this form, we give a transformation that works provided that the dPCP satisfies an additional property called \emph{agnosticism}, which our HDX-based construction luckily possesses. To present it, we first introduce the complete decoding distribution.
\begin{definition}[Complete Decoding Distribution]\label{def:complete-distribution} Let $\mc{D}$ be a decodable PCP as in~\cref{def:dpcp_intro2}. The complete decoding distribution of $\mc{D}$ is the distribution $\mc{Q}$ over 
$[n]\times A\times B$, in which a sample is drawn by taking $t\in[n]$ uniformly,  $(a,b)\sim\mc{P}_t$, and outputting $(t,a,b)$.
\end{definition}
Letting $\mc{Q}$ be the complete decoding distribution of a dPCP $\mc{D}$, we denote by $\mc{Q}(t,\cdot,b)$ the marginal distribution of $a$ conditioned on $t$ and $b$, and by $\mc{Q}(\cdot,\cdot,b)$ the marginal distribution of $(t,a)$ conditioned on $b$.
\begin{definition}
    We say a dPCP $\mc{D}$ has agnostic decoding distributions if
    for all $t$ and $b$, the marginal distribution of $\mc{Q}(\cdot,\cdot,b)$ on $a$ is identical to $\mc{Q}(t,\cdot,b)$.
\end{definition}
With this notion in mind, we can now describe our decoding-degree reduction transformation.
Given a dPCP $\mc{D}$ with an agnostic complete decoding distribution, we construct a new dPCP $\mc{D}_2$ as follows. We first replace a right-vertex $b\in B$ 
with copies $(b,t)$ of it for $t\in[n]$ such that there is a neighbour $a$ of $b$ such that $(a,b)\in\supp(\mc{P}_t)$, and modify the alphabets and constraints appropriately. Namely, denoting the right degree of $\mc{D}$ by $d$, the alphabet of the right is now a subset of $\Sigma_A^d\times \Sigma_0$. Thinking of a right vertex $b\in B$ and its neighbors $a_1,\ldots,a_d$, a symbol $(\tau_1,\ldots,\tau_d,s)$ for $(b,t)$ is allowed if all the projections $\phi_{(a_i,b)}(\sigma_i)$ are the same, and for any $i$ such that $t$ is in the decoding neighbourhood of $a_i$ we have that $D_t(a_i,\tau_i) = s$. The constraint on the edge $(a_i,(b,t))$ is that the assignments $\sigma$ to $a$ and $(\tau_1,\ldots,\tau_d, s)$ to $(b,t)$ satisfy that $\sigma =\tau_i$. The complete decoding distribution of $\mc{D}_2$ is defined by choosing $t\in [n]$ uniformly, then $(a,b)\sim\mc{P}_t$ and output $(t, (b,t), a)$. In particular, the side $A$ of $\mathcal{D}$ is the right side of $\mc{D}_2$, and the pairs $(b,r)$ are the left side of $\mc{D}_2$.
As per the new decoder, when $D_t$ is run on $((b,t),(\tau_1,\ldots,\tau_d,s))$, the output is simply $s$.

In a sense, in this transformation we have created copies of $b$ according to which $t$ we want it to be able to decode, and also we ``flipped'' the sides of the projections (by making the alphabet of $(b,t)$ include an opinion for the label of each neighbour of $b$). The latter step is similar to the flipping-of-sides transformation in~\cite{MoshkovitzRaz,dh}. 

Most of the features of the new dPCP construction, such as decoding degree (which is $1$ in $\mc{D}_2$), alphabet size, and so on, are rather immediate to analyze, and the key challenge is the soundness. If $\mc{D}$ were an arbitrary PCP, because we have split a vertex $b$ into several copies depending on $t$, the extra knowledge of both $b$ and $t$ immediately reveals from which decoding distribution $\mc{P}_t$ the constraint was sampled. When designing an assignment for $(b,t)$, this gives extra information about which vertex $a$ is chosen to test the constraint, which breaks the soundness guarantee of $\mc{D}$. This is precisely where the fact that $\mc{D}$ has agnostic decoding distribution comes in handy, since in that case no extra information about $a$ is gained, and the soundness guarantee of $\mc{D}$ persists.

\begin{remark}
    Our actual decoding-degree procedure in~\cref{sec:dec_deg_red} has one minor difference, which is that we may create several copies of each $(b,t)$. This is related to the fact that we want to ``align'' the size of the inner PCPs as per our earlier discussion.
\end{remark}

\subsection{Techniques in our RLDC Construction}\label{sec:techniques_RLDC}
We obtain our $3$-query RLDC via a new query-preserving black box transformation from PCPPs to RLDCs. 

While many RLDC constructions invoke PCPPs as a subroutine, prior approaches either employ multiple invocations of PCPP verifiers or introduce additional queries beyond those of the underlying PCPP.

Our starting point is the quadratic-length, $O(1)$-query construction in~\cite{bghsv}, which uses a PCP of proximity to construct an RLDC as follows. Let $\mc{C}_0\colon \{0,1\}^k\to\{0,1\}^{n}$ be a linear-time encodable code with constant distance rate and distance, and let $\pi_i(\cdot)$ be a PCP of proximity encoding for the language $\mc{L}_i = \{(z,y)~|~\exists x\in \{0,1\}^k,~y=\mc{C}_0(x), z=(x_i)^{n}\}$. The code $\mc{C}$ constructed by~\cite{bghsv} is given as:
\begin{equation}\label{eq:bghsv}
\mc{C}(x) = (\mc{C}_0(x))^{t} \circ (x)^{t'} \circ \pi_1(x_i^n,\mc{C}(x))\circ\cdots\circ \pi_k(x_i^n,\mc{C}(x)),
\end{equation}
namely, the first part is $t$ repetitions of the encoding of $x$ under $\mc{C}_0$, the second part is $t'$ repetitions of $x$ itself, and the last part consists of a PCP of proximity for each bit in $x$. Here, the purpose of $t,t'$ is to balance the length of the different parts, so that the code $\mc{C}$ has a constant distance. On input word $w$ and index $i\in [k]$, the basic decoder of~\cite{bghsv} proceeds by running the PCP of proximity on the oracle $((z,y),\pi)$, where 
(1) $z$ is taken to be $z = \sigma^n$ where $\sigma$ is the symbol appearing in a random place in the received word $w$ that is supposed to be $x_i$, (2) $y$ is taken to be a supposed random copy of $\mc{C}_0(x)$, (3) $\pi$ is taken to be the supposed PCPP encoding of $\mc{L}_i$. The query complexity of this basic decoder is the same as that of the PCPP, and~\cite{bghsv} show that it works provided that the PCPP for $\mc{L}_i$ has sufficiently small decoding radius $\eta$. This already implies that the PCPP must make $\Omega(1/\eta)$-queries, which means the RLDC construction above must make some unspecified, likely not-too-small, number of queries.

Our construction is syntactically similar, but with a few important distinctions. Since we wish to use~\cref{thm:PCPP_intro}, whenever we apply a PCPP decode, we want it to be in the case that the word is very far from the relevant language. In particular, we must take the code $\mc{C}_0\colon\{0,1\}^k\to\Sigma_0^n$ to have distance close to $1$, and (more importantly) we cannot afford to have the ``$x$-part'' as in the encoding in~\eqref{eq:bghsv}. We thus use the following variant:
\begin{equation}\label{eq:our_rldc}
\mc{C}(x) = (\mc{C}_0(x))^{t} \circ \pi_{1,x_1}'(\mc{C}(x))\circ\cdots\circ \pi_{k,x_k}'(\mc{C}(x)),
\end{equation}
where again $t$ is a suitable repetition parameter, and the PCPP part of the encoding is a bit different. To define $\pi_{i,b}'$ we first take $\pi_{i,b}$ to be the PCPP encoding for the language 
\[
\mc{L}_{i,b} = \{w \in\Sigma_0^n~|~\exists x\in\{0,1\}^k, \mc{C}_0(x) = w, x_i=b\}.
\]
Denoting the alphabet of $\pi_{i,b}$ by $\Sigma$, the alphabet of $\pi_{i,b}'$ is now
$\Sigma\times\{0,1\}$, and the $j$th symbol in 
$\pi_{i,b}'(w)$ is the $j$th symbol in $\pi_{i,b}(w)$ with $b$ appended to it.
 
The reason we need to use $\pi_{i,b}'$ in place of $\pi_{i,b}$ is that, for a fixed $i$, upon reading a single symbol from $\pi_{i,b}(w)$, we 
may not know whether we need to run the verifier of $\mc{L}_{i,0}$ or of $\mc{L}_{i,1}$. Appending $b$ to each symbol resolves this issue.

With this construction in mind, when run on a word $w$ and an index $i\in [k]$, the RLDC decoder looks at $i$th part in the PCPP part of~\eqref{eq:our_rldc} and invokes it on a randomly chosen copy of $\mc{C}_0(x)$ from the first part of~\eqref{eq:our_rldc}. Indeed, we show that this decoder works (provided that the PCPP decoder has a few additional technical features, which hold in our case).

\subsection{Discussion and Open Problems}
While~\cref{thm:dPCP_intro} gives dPCPs with essentially the best possible query complexity, size, and soundness, the same cannot be said about~\cref{thm:PCPP_intro}. Existing lower bounds for PCPPs~\cite{SHLM} do not rule out quasi-linear $2$-query PCPPs with constant (but not Boolean) alphabet size that have small soundness, and we leave the problems of giving better $2$-query PCPP constructions or lower bounds to future research. We remark that in some applications, there may be more context-dependent reductions that use our dPCP (instead of our PCPP) directly and thus do not suffer the loss incurred from the additional query.

Another interesting question is whether one can establish a PCPP with a tight relationship between the proximity parameter and the soundness, as in~\cref{thm:4PCPP_intro}, but with a smaller number of queries. This will yield direct improvements in some applications, such as better hardness gap in~\cref{cor:recon2}.
Lastly, it would be interesting to come up with a variant of our RLDC construction that is additionally locally testable (we suspect our code is locally testable, but investigating it seems to require uniqueness properties from our PCP constructions, which go beyond the scope of this work). 

Lastly, our codes from~\Cref{thm:RLDC_intro} have $3$ queries and nearly quadratic length, and in a companion paper we show how to generalize these codes for any odd number of queries $q\geq 3$. We do not know if better $3$-query RLDCs exist, and it would be interesting to either prove lower bounds that match the performance of our codes, or come up with sub-quadratic $3$-query RLDCs.

\section{Preliminaries}
In this section we introduce several basic notions and tools that will be used in our arguments.
\vspace{-2ex}
\paragraph{Notations.} Throughout the paper, for $n\in\mathbb{N}$ we denote $[n] = \{1, \ldots, n\}$ and use $t\in[n]$ to denote a uniformly chosen element in $[n]$. Given a distribution $\mc{Q}$ over a finite set $S$, we write $s \sim \mc{Q}$ to denote a random sample distributed according to $\mc{Q}$.  
For a subset $U \subseteq S$, $\mc{Q}(U) = \sum_{s\in U}\mc{Q}(s)$ is the measure of $U$ under $\mc{Q}$. The total variation distance between two distributions, $\mc{Q}, \mc{Q}'$ is defined as
\[
{\sf TV}(\mc{Q}, \mc{Q}') = \frac{1}{2}\sum_{s \in S} |\mc{Q}(s) - \mc{Q}'(s)| = \max_{U \subseteq S}  |\mc{Q}(U) - \mc{Q}'(U)|.
\]
We say that $\mc{Q}$ is $M$-discrete if for every $s \in S$, there is $w(s) \in \mathbb{N}$ such that $\mc{Q}(s) = \frac{w(s)}{M}$.

When $\mc{Q}$ is a multivariate distribution, say $m$-variate, we denote by $\mc{Q}(a_1,\ldots,a_i,\cdot,\ldots\cdot)$ the distribution of $a_{i+1},\ldots,a_m$ conditioned on $a_{1},\ldots,a_i$. We use a similar notation when the coordinates we condition on are not consecutive.
We use the notation $\mc{Q}(a_1,\ldots,a_i,\circ,\ldots,\circ)$ to denote  $\sum_{a_{i+1},\ldots,a_m}\mc{Q}(a_1,\ldots,a_i,a_{i+1},\ldots,a_m)$, i.e., the marginal distribution of $\mc{Q}$ on the first $i$ coordinates.

Given a length $n$ vector $v$, we write $v_i$ to denote the $i$th coordinate of $v$. Given a function $F$ over a domain $X$ and a subset $Y\subseteq X$, we denote by $F|_Y$ the restriction of $F$ to $Y$.

We use standard big-$O$ notation: for functions $f,g\colon\mathbb{N}\to[0,\infty)$, the notation $f = O(g)$ means that there is an absolute constant $C$ such that $f(n)\leq Cg(n)$ for all $n\in\mathbb{N}$, and the notation  $f = \Omega(g)$ means that there is an absolute constant $C>0$ such that $f(n)\geq Cg(n)$ for all $n\in\mathbb{N}$. Often we will have dependency on auxiliary parameters, and we write $f = O_{\eps}(g)$ or $f=\Omega_{\eps}(g)$ if the constant $C$ depends on $\eps$. When we write 
${\sf poly}_{\eps}(N)$ we mean a polynomial in $N$ whose degree and coefficients may be increasing in $1/\eps$.

\subsection{Error-Correcting Codes}

An error-correcting code over an alphabet $\Sigma$ is a map $\mc{C}: \{0,1\}^k \to \Sigma^n$. We refer to $k$ as the message length and to $n$ as the blocklength of the code. Abusing notation, we will often use $\mc{C}$ to refer to both the map and its image, and it will always be clear from context which we are referring to. We refer to the strings in $\mc{C}$ as codewords. The distance between $w, w' \in \Sigma^n$ is the relative Hamming distance between them, namely
\[
\dist(w, w') = \Pr_{i \in [n]}[w_i \neq w'_i].
\]
The relative distance of a set $\mc{C}$ is defined as
\[
\dist(\mc{C}) = \min_{w, w' \in \mc{C}, w \neq w'} \dist(w,w').
\]
Informally, when discussing error-correcting codes 
we will deal with sets $\mc{C}$ that have $\Omega(1)$ relative distance.

We will require an error-correcting code with message length $k$, quasi-linear time encoding, quasi-linear time decoding, constant alphabet size, relative distance arbitrarily close to $1$ and blocklength $n = O(k)$. By quasi-linear time encoding, we mean that there is a circuit of size  $\tilde{O}(k)$ that on input $x\in\{0,1\}^k$ computes its encoding, and similarly for decoding.
There are multiple examples of such codes throughout the literature~\cite{spielman1996linear,GI,roth2006improved}, and we use the following formulation from~\cite[Theorem 3]{GI}.
\begin{lemma}\label{lm: code GI}
    For every $\eps > 0$, $k \in \mathbb{N}$, sufficiently large alphabet $\Sigma$ relative to $1/\eps$, and blocklength $n = O_{\eps}(k)$ there exists a quasi-linear time encodable/decodable code $\mc{C}: \{0,1\}^k \to \Sigma^n$ with distance $1-\eps$. 
    Furthermore, given $k$ one can construct in time $k^{O_{\eps}(1)}$ a circuit of size $\tilde{O}_{\eps}(k)$ computing $\mc{C}$.
\end{lemma}
We will also require the well-known Johnson Bound, which bounds the number of codewords in $\mc{C} \subseteq \Sigma^n$ that lie within proximity of a given string $w \in \Sigma^n$. We refer to \cite[Fact 5.3]{dh} for a proof.
\begin{fact} \label{fact: johnson}
    Let $\mc{C} \subseteq \Sigma^n$ be a code with minimum distance $1-\delta$. Then, for any $w \in \Sigma^n$ and $\eta > 2\sqrt{\delta}$, the number of $w' \in \mc{C}$ such that $\dist(w, w') \leq 1 -  \eta$ is at most $\frac{2}{\eta}$.
\end{fact}

\subsection{Expanders and Bipartite Expanders}
Our arguments will use bipartite expander graphs and the bipartite expander mixing lemma, and we introduce these notions in this section. 
Let $G$ be a weighted bipartite graph with sides $U, V$, and let $w: E \to \mathbb{R}_{\geq 0}$ be the weight function. 
The weighted adjacency matrix of $G$ is $M\in \mathbb{R}^{(U\cup V)\times (U\cup V)}$ defined as $M=D^{-1/2} W D^{-1/2}$, where $W(u,v) = w(u,v)$ and the degree matrix $D$ is a diagonal matrix with the entries $D(u,u) = \sum_{v}w(u,v)$. The transition matrix of $G$ is defined as $P = D^{-1} M_{U,V}$, where $M_{U,V}\in \mathbb{R}^{U\times V}$ is the matrix with $M_{U,V}(u,v) = w(u,v)$. 
It is well known that the largest singular of $P$ is $1$, and is given by the vectors $x\in \mathbb{R}^{U}$ where $x(u) = D(u,u)$. Therefore, the stationary distribution of $G$ on each side is proportional to the weighted-degrees given in $D$. We denote samples according to the stationary on $U$ and $V$ by $u\sim U$ and $v\sim V$ respectively.
It is well known that the second largest eigenvalue of $M$ is the same as the second largest singular value of $P$, and we often refer to it as the second singular value of $G$.

\begin{lemma} \label{lm: expander mixing}
    Let $G = (U, V, E)$ be a weighted bipartite graph whose second singular value is at most $\lambda$. Then for any functions $F: U \to [0,1], G: V \to [0,1]$ we have
    \[
   \left| \E_{(u,v) \sim E}[F(u) \cdot G(v)] - \mu_U(F) \mu_V(G)\right| \leq \lambda \sqrt{\mu_U(F)\mu_V(G)},
    \]
    where $\mu_U := \E_{u \sim U}[F(u)], \mu_V(G) := \E_{v \sim V}[G(v)]$.
    
    In the special case where $F$ and $G$ are indicator functions of sets, say of $A \subseteq U$ and $B \subseteq V$, we get
    \[
    \left| \Pr_{(u,v) \sim E}[u \in A, v \in B] - \mu_U(A) \mu_V(B) \right| \leq \lambda \sqrt{\mu_U(A)\mu_V(B)}.
    \]
\end{lemma}

\subsubsection{Constructions of Expanders and Bipartite Expanders}
We will require the fact that expanders of arbitrary degree and size can be constructed efficiently.
\begin{lemma} \label{lm: poly time expander}
    For integers $n \geq d$ there is a polynomial time algorithm which constructs an (unweighted) $d$-regular graph with $n$ vertices and second singular value $O(d^{-1/2})$.
\end{lemma}
\begin{proof}
    The algorithm is by \cite{alon_expanders} if $n$ is sufficiently large and by brute force otherwise.
\end{proof}

We also require a version of the above lemma for bipartite, bi-regular expanders.

\begin{lemma} \label{lm: poly time bip expander}
    For integers $n_1,n_2,d_1,d_2$ such that $n_1\cdot d_1=n_2\cdot d_2$ and either $n_1$ divides $n_2$ or $n_2$ divides $n_1$, there is a ${\sf poly}_{d_1,d_2}(n_1,n_2)$-time algorithm which constructs an (unweighted) $(d_1,d_2)$-bi-regular bipartite graph with $n_1$ vertices on the left and $n_2$ vertices on the right and second singular value $O\left((\min\{d_1,d_2\}\right)^{-1/2})$.
\end{lemma}
\begin{proof}
    The algorithm is by \cite{GribinskiMarcus2021} if $n_1,n_2$ are sufficiently large and by brute force otherwise.
\end{proof}

\subsection{High-Dimensional Expanders}
In this section, we give some necessary background on high-dimensional expanders (HDXs).

\subsubsection{Simplicial Complexes}
A $d$-dimensional high-dimensional expander (HDX) on a ground set $\mc{U}$ is a downward closed collection of subsets of $\mc{U}$, typically denoted by $X = (X(1), \ldots, X(d))$, 
where $X(1) = \mc{U}$ and $X(i)$ consists of size $i$ subsets of $\mc{U}$. By downwards closed we mean that for 
$1\leq i\leq j\leq d$, if $I\subseteq J$ have sizes $i$ and $j$ respectively and
$J \in X(j)$, then $I \in X(i)$. We refer to $X(1)$ as the vertices of the complex, and to 
$X(k)$ as the $k$-faces of $X$.

For any $I \subseteq \mc{U}$, one may consider the sets in $X$ containing $I$. These form a smaller, $(d- |I|)$-dimensional complex, which goes by the name the link of $I$.
\begin{definition}
Given a set $I \in X(i)$ for $0 \leq i \leq d-2$, the link of $I$ is the $(d-i)$-dimensional complex $X_I$ whose $k$-faces are 
\[
X_I(k) = \{J \setminus I \; | \; J \in X(k+i), J \supseteq I \}.
\]
\end{definition}

\subsubsection{Measures and Walks on Complexes}

A complex $X = (X(1), \ldots, X(d))$ is equipped with  a natural measure arising as the ``push down'' measure of the uniform measure over $X(d)$. More precisely, for each $i \in [d]$, let $\mu_{i}$ denote the measure over $X(i)$ obtained by choosing $U \in X(d)$ uniformly and then outputting a uniformly chosen $A \subseteq U$ of size $i$. At times, when $i$ is clear from context, we drop $i$ from the subscript and simply write $\mu$ to be the appropriate measure over $X(i)$. We will write $U \sim X(i)$ to denote a random $U$ chosen according to the measure $\mu_i$.

We also consider the above measure when restricted to a link of $X$. Specifically, for a set $I \in X(i)$ where $0 \leq i \leq d-2$, we define the measure $\mu_{I, j-i}$ over $X_I(j-i)$ to be the measure obtained by choosing $J \sim X(j)$ conditioned on $J \supseteq I$, and then outputting $J \setminus I$. At times, if the parameter $j-i$ is clear from context, we drop it from the subscript and simply write $\mu_I$.

These measures also induce weighted graphs between layers of the complex. Throughout the paper we will consider two graphs derived from a complex $X = (X(1), \ldots, X(d))$.
\vspace{-1ex}
\paragraph{The Base Graph.} This is the graph with vertex set $X(1)$ and edge set $X(2)$. The edges are weighted by the measure $\mu_2$ over $X(2)$ and the stationary distribution of this graph is given the measure $\mu_1$ over $X(1)$.

\vspace{-1ex}
\paragraph{Inclusion Graphs.} For two levels $j < k$, the inclusion graph between $X(j)$ and $X(k)$ is a bipartite graph with parts $X(j)$ and $X(k)$. The edges of this graph are all pairs $(V,U) \in X(j) \times X(k)$ such that $V \subset U$ and are weighted according to the following natural distribution derived from $X$: choose $U \sim X(k)$, $V \subset U$ uniformly of size $j$, and output $(V,U)$.

\subsubsection{Direct Product Testing on HDX}
We will need direct product testers over HDXs. Fix a complex $X = (X(1), \ldots, X(d))$ and suppose that for $k \in [d]$ (thought of as a large constant but much smaller than $d$) there is a table $T: X(k) \to \Sigma^k$ assigning local functions on each $U \in X(k)$. We wish to determine if these individual functions assigned to each set of the complex actually correspond to some global function, $F: X(1) \to \Sigma$, meaning $T[U] = F|_U$. We will consider a natural $2$-query tester, given by the following definition.
\begin{definition}
    Let $X$ be a $d$-dimensional complex and let $k < d$. Given a  table $T: X(k) \to \Sigma^k$, the $(k, \sqrt{k})$-direct product test can be described as follows.
    \begin{itemize}
        \item Sample $D \in X(d)$ uniformly.
        \item Sample $B \subseteq D$ of size $\sqrt{k}$ uniformly.
        \item Sample $A, A'$ independently and uniformly at random conditioned on $B \subseteq A, A' \subseteq D$.
        \item The test passes if and only if $T[A]|_B = T[A']|_B$.
    \end{itemize}
\end{definition}

\begin{definition}
    We say that the $(k, \sqrt{k})$-direct product test on $X$ has soundness $\delta$ if the following holds. Let $T$ be any table that passes the $(k, \sqrt{k})$-direct product test with probability at least $\delta$. Then there exists a function $F: X(1) \to \Sigma$ such that
\[
\Pr_{A \sim X(k)}[\dist(T[A], F|_A) \leq \delta] \geq \poly(\delta).
\]
In words, if a table passes the $(k, \sqrt{k})$-direct product test on $X$ with probability at least $\delta$, then it almost agrees with a global direct product function on a sizable weight of the $k$-faces.
\end{definition}

\subsubsection{An Algebraic HDX Construction}
We will require an HDX construction $X$ as in~\cite{bmv}, which is a variant of the constructions in~\cite{DDL,BLM}. Some key features of this construction that are important for us include the fact they have poly-logarithmic degree, good expansion properties, and that the $(k,\sqrt{k})$-direct product test on $X$ has arbitrarily small soundness.

\begin{theorem}\cite[Theorem 2.13]{bmv} \label{thm: hdx construction}
    For all $\delta \in (0,1)$, $C \in \mathbb{N}$, for sufficiently large $k \in \mathbb{N}$ and sufficiently large $d\in\mathbb{N}$ the following holds. For any sufficiently large $n \in \mathbb{N}$ and $q =  \Theta(\log^C(n))$, there is an algorithm which takes as input $\delta, C, d, k, n, q$, and outputs in $\poly(n)$-time, a $d$-dimensional complex $X = (X(1), \ldots, X(d))$ which satisfies the properties in~\cite[Theorem 2.13]{bmv}, as well as the following properties.
    \begin{enumerate}
        \item $n \leq |X(1)| \leq O_{d}(n)$.
        \item Every vertex is contained in at most $q^{O(d^2)}$ $d$-faces.
        \item For any $1 \leq i \leq d$, the second singular value of the weighted bipartite inclusion graph between $(X(1), X(i))$ is  $\frac{1}{i}+\frac{i^{O(1)}}{d}$.
        \item For any vertex $v \in X(1)$, the second eigenvalue of the graph $(X_{v}(1), X_{v}(2))$ is at most $O(1/\sqrt{q})$.
        \item The $(k , \sqrt{k})$-direct product test on $X$ has soundness $\delta$.
    \end{enumerate}
    We think of $\delta, C, d, k, n, q$ as inputs to this algorithm and refer to $\delta$ as the direct-product soundness parameter, $C$ as the exponent, $d$ as the dimension parameter, $k$ as the direct-product dimension parameter, $n$ as the target number of vertices, and $q$ as the prime.  
\end{theorem}
\begin{proof}
    All items, except for the third and fourth,  follow from the statement of \cite[Theorem 2.13]{bmv}. The fourth item follows from the proof in \cite[Theorem 2.13]{bmv}, and more precisely from the paragraph ``Local spectral expansion of $X$'' in the proof.  
    The third item follows from \cite[Lemma 2.8]{bmv}.
\end{proof}
Since our arguments use several results from~\cite{bmv} in a black-box way, we make an implicit use of the properties of $X$ outlined in~\cite[Theorem 2.13]{bmv}. We thus omit a detailed description of them (as well as the set-up they require). The properties we mention in~\cref{thm: hdx construction} will be used explicitly in our arguments.

We need the following (rather standard) sampling bound.
\begin{lemma} \cite[Claim 4.7]{bmv} \label{lm: bmv hdx expander sampling}
    Let $X$ be a $d$-dimensional complex satisfying \cref{thm: hdx construction} with prime $q$, and let $\mc{E} \subseteq X(2)$ be a set of edges with measure $\mu_2(\mc{E}) \geq \eps$ for some $\eps > 0$. Then,
    \[
    \Pr_{u \in X_1(1)}[\mu_{u, 2}(\mc{E}) \geq \sqrt{\eps} ] \leq \frac{\poly(d)}{q \cdot \eps}.
    \]
\end{lemma}

\input{initial_udPCP}
\input{hdx_route}

\input{direct_product}
\input{decoding_degree_reduction}

\input{composition}

\input{applications}
\bibliographystyle{alpha}
\bibliography{references}
\appendix
\input{hadamard}
\end{document}

%% file: initial_udPCP.tex
\section{Decodable PCPs and Uniquely-Decodable PCPs}
In this section we formally define dPCPs, their unique-decoding analog called uniquely-decodable PCPs, and give a more precise version of~\cref{thm:dPCP_intro}, which is~\cref{thm:dPCP main} below.

\newcommand{\Language}{\mathcal{L}}
\newcommand{\PCPDecoder}{\mathcal{D}}
\newcommand{\Word}{w}
\newcommand{\Bits}{\{0,1\}}
\newcommand{\ListSize}{\mathsf{L}}
\newcommand{\PCPError}{\delta}
\newcommand{\LeftSide}{A}
\newcommand{\RightSide}{B}
\newcommand{\EdgeSet}{E}
\newcommand{\Constraint}[1]{\Phi_{#1}}
\newcommand{\EdgeDistributionName}{\mathcal{P}}
\newcommand{\EdgeDistribution}[1]{\EdgeDistributionName_{#1}}
\newcommand{\Alphabet}[1]{\Sigma_{#1}}
\newcommand{\LeftAlphabet}{\Alphabet{\LeftSide}}
\newcommand{\RightAlphabet}{\Alphabet{\RightSide}}
\newcommand{\WordLength}{n}
\newcommand{\IndexInWord}{j}
\newcommand{\LanguageAlphabet}{\Alphabet{0}}
\newcommand{\Decoder}[2]{D^{#1}_{#2}}
\newcommand{\Restricted}[2]{{#1}_{#2}}
\newcommand{\IndexInDecodingSet}{t}
\newcommand{\Degree}[1]{\mathsf{deg}\!\left(#1\right)}
\newcommand{\Neighbourhood}[1]{\Gamma\!\left(#1\right)}
\newcommand{\DecNeighbourhood}{\Gamma_{\dec}}
\newcommand{\LeftAlphabetSymbol}{\sigma}
\newcommand{\LeftVertex}{a}
\newcommand{\RightVertex}{b}
\newcommand{\TableProof}[1]{T_{#1}}
\newcommand{\LeftProof}{\TableProof{\LeftSide}}
\newcommand{\RightProof}{\TableProof{\RightSide}}
\newcommand{\DecisionComplexity}{\mathsf{DC}}
\newcommand{\Circuit}{C}
\newcommand{\ddeg}{\mathsf{ddeg}}

\newcommand{\AllDistribution}{\mathcal{Q}}

\newcommand{\SideDecodings}{[n]}
\newcommand{\DecoderSymbol}{D}

\newcommand{\Assignment}{T}

\subsection{Decodable PCPs} 
The definition of dPCPs below expands~\cref{def:dPCP_intro} as it includes the definitions of additional parameters of interest.

\begin{definition} [Decodable PCP]\label{def:dpcp}
Let $\Alphabet{0}$ be an alphabet and let $\Language \subseteq \Alphabet{0}^n$ be a language.
A decodable PCP (dPCP) for $\Language$ is a tuple
    \[
    \PCPDecoder = \left(\LeftSide\cup \RightSide,\EdgeSet,\LeftAlphabet,\RightAlphabet,\{\Constraint{e} \}_{e\in \EdgeSet},\{\EdgeDistribution{t}\}_{t \in [n]}, \{\DecoderSymbol_{t} \}_{t \in [n]}\right)
    \]
    where:
\begin{itemize}
    \item \textbf{Constraint Graph.} $(\LeftSide \cup \RightSide, \EdgeSet)$ is a bipartite graph with parts $\LeftSide, \RightSide$ and edge set $\EdgeSet$.
    \item \textbf{Left Alphabet.} $\Alphabet{\LeftSide}$ is the alphabet of the vertices in $\LeftSide$. We allow each individual vertex $a\in A$ to have a constrained alphabet, which is a subset of $\Sigma_A$.
    \item \textbf{Right Alphabet.} $\Alphabet{\RightSide}$ is the alphabet of the vertices in $\RightSide$. We allow each individual vertex $b\in B$ to have a constrained alphabet, which is a subset of $\Sigma_B$.
    \item \textbf{Constraints.} For each $e = (\LeftVertex,\RightVertex)\in \EdgeSet$, $\Constraint{e}: \Alphabet{\LeftSide} \times \Alphabet{\RightSide} \to \{0, 1\}$ is a constraint on the edge $e$. We say that an input satisfies it if it evaluates to $1$ and unsatisfied if it evaluates to $0$. If either input to the constraint is an invalid alphabet symbol for $a$ or $b$, then the constraint outputs $0$.

    \item \textbf{Projection Constraints.} We say a constraint $\Phi_e$ is a projection constraint if for every $\sigma \in \Sigma_A$, there is a unique $\sigma' \in \Sigma_B$ such that $\Phi_{e}(\sigma, \sigma') = 1$. We say $\mc{D}$ is a projection dPCP if all of its constraints are projections. 
    \item \textbf{Decoding Distribution.} For each $t \in [n]$, $\mc{P}_t$ is a distribution over $E$.
     \item \textbf{Decoder.} 
    For each $t \in [n]$, the decoder $\DecoderSymbol_t\colon A\times\Sigma_A\to\Sigma_0$ receives the left vertex of a constraint along with a symbol for it, and outputs a decoding for $t$. 
    \end{itemize}
We also define the following parameters related to dPCPs:
\begin{itemize}
    \item \textbf{Length.} The number of vertices in the constraint graph.
    \item \textbf{Alphabet Size.} $\max(|\Sigma_A|, |\Sigma_B|)$, and typically this will be $|\Sigma_A|$.
     \item \textbf{Decision Complexity.} Viewing $\Sigma_A = \{0,1\}^{\log |\Sigma_A|}, \Sigma_B = \{0,1\}^{\log |\Sigma_B|}$, the decision complexity is the maximum over all constraints $(a,b)$ of the size of the circuit with inputs $ \{0,1\}^{\log |\Sigma_A| + \log |\Sigma_B|}$  that outputs $1$ if the constraint on $(a,b)$ is satisfied and both alphabet symbols are valid and $0$ otherwise.  
     \item \textbf{Projection Decision Complexity.} In the case that the dPCP is a projection dPCP,  its projection decision complexity is the maximum over all constraints $(a, b)$ of the size of the circuit with input $\sigma \in \Sigma_A$ thought of as an element in $ \{0,1\}^{\log |\Sigma_A|}$, outputs the unique $\Sigma_B$ symbol which satisfies the constraint on $(a,b)$, and outputs $\bot$ if $\sigma$ is not a valid $a$ alphabet symbol. 
     
     We note that the (standard) decision complexity is at most $O(\log|\Sigma_B|)$ more than the projection decision complexity.
\item \textbf{Decoding Complexity.} Viewing $\Sigma_A = \{0,1\}^{\log |\Sigma_A|}$, the decoding complexity is maximum over all $a \in A$ of the size of the circuit which computes the decoding function $D_t(a, \cdot): \Sigma_A \to \Sigma_0$.
    
     \item \textbf{Neighborhoods and Degrees.} The neighborhood of $u \in A \cup B$ is $\Gamma(u) = \{v \in A \cup B \; | \; (u,v) \in E \}$, and the degree of $u$ is $\deg(u) = |\Gamma(u)|$. The left degree of $\mc{D}$ is $\max_{a \in A}\deg(a)$ and the right degree of $\mc{D}$ is  $\max_{b \in B}\deg(b)$.
    \item \textbf{Decoding-Neighborhoods and Decoding-Degrees.} The decoding neighborhood of $a \in A$ is $\Gamma_{\dec}(a) = \{t \in [n] \; | \; a \in \supp\left(\mc{P}_{t}\right) \}$ and the decoding-degree of $a$ is $\ddeg(a) = |\Gamma_{\dec}(a)|$. The decoding degree of $\mc{D}$ is $\max_{a \in A} \ddeg(a)$.
\end{itemize}
\end{definition}

We next move on to define the completeness and soundness of a dPCP. We will always use the notion of perfect completeness (all of the constructions in this paper have it). 

\begin{definition}[Perfect Completeness]\label{def:dpcp-completeness}Let $\PCPDecoder$ be a decodable PCP for $\Language$. We say that $\mc{D}$ has perfect completeness if for all $\Word \in \Language$ there are assignments $\Assignment_{\LeftSide}: \LeftSide \to \Alphabet{\LeftSide}$ and $\Assignment_{\RightSide}: \RightSide \to \Alphabet{\RightSide}$ such that 
    \[
    \Pr_{\substack{\IndexInDecodingSet \in \SideDecodings,\\ (\LeftVertex,\RightVertex)\sim \EdgeDistribution{\IndexInDecodingSet}}}[\Phi_{(\LeftVertex,\RightVertex)}(\Assignment_{\LeftSide}[\LeftVertex], \Assignment_{\RightSide}[\RightVertex])=1 \land \DecoderSymbol_{\IndexInDecodingSet}(\LeftVertex,\Assignment_{\LeftSide}[\LeftVertex]) = \Word_t] = 1.
    \]

\end{definition}

As per the soundness notion, we 
use a list-decoding soundness in the style of \cite{dh}.\footnote{The more common definition in the literature asserts that the list $w_1,\ldots,w_L$ can be constructed from an assignment $T_B$ to the right side of the given dPCP $\mc{D}$. We chose the notion below as it is more convenient for us to work with, and in~\cref{sec:switch_side} we show a relation with the more usual one.}

\begin{definition}[List-Decoding Soundness]\label{def:list-decoding-soundness}
Let $\PCPDecoder$ be a decodable PCP for $\Language$. Given $\eps \in [0,1]$, $L \in \mathbb{N}$ we say that $\mc{D}$ satisfies $(L, \eps)$-list-decoding soundness if the following holds. For any left assignment $\Assignment_{\LeftSide}: \LeftSide \to \Alphabet{\LeftSide}$ there exists a list $\{\Word_1, \ldots, \Word_L \} \subseteq \Language$ such that for any $\Assignment_{\RightSide}: \RightSide \to \Alphabet{\RightSide}$, 
\[
\Pr_{\substack{\IndexInDecodingSet \in [n],\\ (\LeftVertex,\RightVertex)\sim \EdgeDistribution{\IndexInDecodingSet}}}[\Phi_{(\LeftVertex,\RightVertex)}(\Assignment_{\LeftSide}[\LeftVertex], \Assignment_{\RightSide}[\RightVertex])=1 \land \DecoderSymbol_{\IndexInDecodingSet}(\LeftVertex,\Assignment_{\LeftSide}[\LeftVertex]) \notin \{\Restricted{(\Word_i)}{\IndexInDecodingSet} \}_{i \in [L]}] \leq \eps.
\]
\end{definition}

\paragraph{Agnostic Decoding Distributions.} 
We recall the notions of complete decoding distribution and agnosticism. 
\begin{definition}[Complete Decoding Distribution]\label{def:complete-distribution}Let $\PCPDecoder$ be a decodable PCP for $\Language$. The complete decoding distribution of $\PCPDecoder$, denoted by  $\AllDistribution$, is the distribution over 
$\SideDecodings \times \LeftSide\times\RightSide$ generated by choosing $t \in [n]$ uniformly, sampling $(a,b) \sim \mc{P}_t$ and outputting $(t,a,b)$.
\end{definition}

\begin{definition}[Agnostic Decoding Distribution]\label{def:agnostic-distributions} For a decodable PCP, $\PCPDecoder$, we say that its complete decoding distribution $\AllDistribution$ is agnostic if for all  $\RightVertex\in\RightSide,\IndexInDecodingSet\in\SideDecodings$, the distributions  
$ \AllDistribution(\IndexInDecodingSet, \cdot, \RightVertex)$ and  $\AllDistribution(\circ, \cdot, \RightVertex)$ are the same. 
\end{definition}

\subsection{Uniquely-Decodable PCPs}
We will also consider a variant of decodable PCPs called uniquely-decodable PCPs. Uniquely-decodable PCPs can be thought of as high-soundness versions of decodable-PCPs and were first introduced by Dinur and Meir \cite{DinurMeir}. Uniquely-decodable PCPs differ from decodable PCPs in three ways: 
\begin{enumerate}
    \item The constraint graph is not necessarily bipartite. We remark that in some cases it will be bipartite, and this will be important for us.
    \item The decoding distributions are over vertices, as opposed to being over edges.
    \item The constraint graph is a weighted graph and in the soundness cases edges are sampled according to this weighting. This should be thought of as the analog of the marginal distribution over edges of the complete decoding distribution $\mc{Q}$ of 
    a dPCP as in~\cref{def:dpcp}.
\end{enumerate}

\begin{definition} [Uniquely-Decodable PCP]\label{def:udpcp}
Let $\Alphabet{0}$ be an alphabet and let $\Language \subseteq \Alphabet{0}^n$ be a language.
A uniquely-decodable PCP (udPCP) for $\Language$ is a tuple
    \[
    \PCPDecoder = \left(V,\EdgeSet,\Sigma, \{\Constraint{e} \}_{e\in \EdgeSet}, \{\EdgeDistribution{t}\}_{t \in [n]}, \{\DecoderSymbol_{t} \}_{t \in [n]}\right)
    \]
    where:
\begin{itemize}
    \item \textbf{Constraint Graph:} The constraint graph $(V, E)$, which may have multi-edges, and we think of it as equipped with the uniform measure.
    \item \textbf{Alphabet.} $\Sigma$ is the alphabet and we allow each individual vertex $v \in V$ to have a constrained alphabet $\Sigma_v\subseteq \Sigma$.
    \item \textbf{Constraints.} For each $e \in E$, $\Phi_{e}: \Sigma \times \Sigma \to \{0,1\}$ is a constraint on the edge $e$.
    \item \textbf{Decoding Distribution.} For each $t \in [n]$, $\mc{P}_t$ is a distribution over $V$.
    \item \textbf{Decoder.} For each $t \in [n]$, the decoder $\DecoderSymbol_t\colon V\times\Sigma\to\Sigma_0$ receives a vertex and a symbol for it, and outputs a decoding for $t$. 
\end{itemize}
We define the parameters length, alphabet size, decision complexity, decoding complexity, neighborhoods, decoding-neighborhoods, degrees, and decoding degrees in the same way as in \cref{def:dpcp}.
\end{definition}
For technical purposes, we will also want to distinguish when a uniquely-decodable PCP has bipartite constraint graph and projection constraints (as defined in \cref{def:dpcp}). In the case that it does, we will also use a notion called left-canonical, defined below.
\begin{definition} [Bipartite and Projection Uniquely-Decodable PCP]
    We say that a uniquely-decodable PCP is \emph{bipartite} if its constraint graph is bipartite, and in this case we will denote its left and right alphabets separately. We say that a bipartite uniquely-decodable PCP has projection constraints if all of its constraints are projections from the left side to the right side.
\end{definition}

\begin{definition}[Left-Canonical Uniquely-Decodable PCP]
We say that a projection bipartite uniquely-decodable PCP is left-canonical if for every left vertex, and any labeling of its right neighborhood, there is at most one assignment to the left vertex that satisfies all of the constraints involving that vertex.
\end{definition}

In order for the decoding distributions to be meaningful they should be aligned with the uniform distribution over $E$, in the following sense. Choosing $t \in [n]$, $u \sim \mc{P}_t$, and $v$ a random neighbor of $u$ (weighted by multiplicities), $(u,v)$ is uniformly random in $E$. Indeed, this will always be the case for our uniquely-decodable PCP constructions. \footnote{Morally speaking, a uniquely-decodable PCP must satisfy this property, at least approximately, to have any useful notion of soundness.}

Perfect completeness and unique-decoding soundness are defined similarly to \cref{def:dpcp-completeness} and \cref{def:list-decoding-soundness}.

\begin{definition}[Perfect Completeness]\label{def:udpcp-completeness}Let $\PCPDecoder$ be a uniquely-decodable PCP for $\Language$. We say that $\mc{D}$ has perfect completeness if for all $\Word \in \Language$ there is an assignment $\Assignment: V \to \Alphabet{}$ such that 
    \[
    \Pr_{(a,b) \sim E}[\Phi_{(\LeftVertex,\RightVertex)}(\Assignment_{\LeftSide}[\LeftVertex], \Assignment_{\RightSide}[\RightVertex])] = 1,
    \qquad 
    \Pr_{t \in [n], a \sim \mc{P}_t}[D_t(a, T[a]) = w_t] = 1.
    \]
In words, the first condition says that the assignment is fully satisfying and the second condition says that the assignment always decodes according to $w$.
\end{definition}

\begin{definition}[Unique-Decoding Soundness]\label{def:unique-decoding-soundness}
 Let $\PCPDecoder$ be a uniquely-decodable PCP for $\Language$. Given $\eps, \eta \in [0,1]$, we say that $\mc{D}$ has $(\eta, \eps)$-unique decoding soundness 
 if the following holds. For any assignment $T: V \to \Sigma$ satisfying
 \[
 \Pr_{(a,b) \sim E}[\Phi_{(a,b)}(T[a], T[b]) = 1] \geq 1-\eps,
 \]
 there exists $w \in \mc{L}$ such that 
 \[
 \Pr_{\substack{t \in [n],\\ a \sim \mc{P}_t}}[D_t(a, T_A[a]) \neq w_t] \leq \eta.
 \] 
\end{definition}

Finally, we define the complete decoding distribution for uniquely-decodable PCPs. 
\begin{definition}[Complete Decoding Distribution]\label{def:complete-distribution}Let $\PCPDecoder$ be a uniquely-decodable PCP for $\Language$. The complete decoding distribution of $\PCPDecoder$, denoted by  $\AllDistribution$, is the distribution over 
$\SideDecodings \times V$ generated by choosing $t \in [n]$ uniformly, $a \sim \mc{P}_t$, and outputting $(t,a)$.
\end{definition}

\subsection{Main Decodable PCP Theorem}
With the definitions above in mind, below we state a more detailed version of our main dPCP construction from~\cref{thm:dPCP_intro}, specifying additional features of our new dPCP construction.
\begin{theorem}\label{thm:dPCP main}
    For all $\eps>0$ there exist $L, C\in\mathbb{N}$ such that $L = \poly(1/\eps)$ and the following holds for all alphabets $\Sigma_0$. If $\varphi\colon \Sigma_0^n\to\{0,1\}$ is a circuit of size $N$, then the language ${\sf SAT}(\varphi)$ has a $2$-query, poly-time constructible dPCP with the following properties:
    \begin{itemize}
        \item \textbf{Length.} $O_{|\Sigma_0|}\left(N \log^C N\right)$,
        \item \textbf{Alphabet Size.} $O_{\eps, |\Sigma_0|}(1)$,
        \item \textbf{Complete Decoding Distribution.} The marginal of the complete decoding distribution over the right vertices is uniform, and the marginal of the complete decoding distribution over the left vertices is also uniform.
        \item \textbf{Completeness.} The dPCP has perfect completeness,
        \item \textbf{Soundness.}  $(L,\eps)$-list-decoding soundness.
    \end{itemize}
\end{theorem}

\paragraph{Proof Outline.} 
The proof of~\cref{thm:dPCP main} spans~\cref{sec:udpcp} to~\cref{sec: final dpcp}, and proceeds as follows:
\begin{enumerate}
    \item In~\cref{sec:udpcp} we give a uniquely-decodable PCP construction with constant alphabet size and quasi-linear length.
    \item In~\cref{sec:hdx_routing} we show how to transform the previous result to a uniquely-decodable PCP construction over the base graph of the high-dimensional expander $X$ from~\cref{thm: hdx construction}. This increases the alphabet size to be super-constant.
    \item In~\cref{sec:amp} we perform the soundness amplification, transforming the HDX-based uniquely-decodable PCP to a $2$-query dPCP with small soundness, but large alphabet size. The goal in the rest of the argument is to reduce the alphabet size of this construction.
    \item In~\cref{sec:dPCP_transfor} we give dPCP variants of a few standard transformations on PCPs.
    \item In~\cref{sec:dec_deg_red} we show our decoding degree reduction transformation. 
    \item In~\cref{sec:comp_thms} we prove our composition theorem. 
    \item In~\cref{sec: final dpcp} we show how to combine the above tools to prove~\cref{thm:dPCP main}.
\end{enumerate}

\section{A Uniquely Decodable PCP Construction}\label{sec:udpcp}
Our starting point is a quasi-linear size uniquely-decodable PCP construction that is implicit in prior works~\cite{BS,Dinur,Mie}. Since we are 
unaware of an explicit reference we give it in full detail here. We start with~\cite[Theorem 1]{Mie}, stated below.

\begin{theorem} \label{thm: pcpp starting}
There exists $\rho>0$ such that for all $\eta > 0$ the following holds. For every circuit
$\varphi: \Sigma_0^n \to \{0,1\}$ of size $N$ there is a poly-time constructible PCPP $\Psi_0 = (V_0, E_0, \mc{P}_0,\Sigma, \{\Phi_{0,e}\}_{e\in E_0})$ for ${\sf SAT}(\varphi)$, such that:
\begin{itemize}
    \item \textbf{Length.} The constraint hypergraph has $O_{|\Sigma_0|}(N \cdot \poly(\log (N)))$ vertices and $O_{|\Sigma_0|}(N \cdot \poly(\log (N)))$ hyperedges, each of size $k = O(1)$. The distribution $\mc{P}_0$ is uniform over the hyperedges. 
    \item \textbf{Alphabet size.} $|\Sigma|=O_{|\Sigma_0|}(1)$ .
    \item \textbf{Completeness.} The PCPP has perfect completeness.
    \item \textbf{Soundness.} If $w \in \Sigma_0^n$ is $\eta$-far from the language $\sat(\varphi)$, then any assignment $A$ extending $w$ satisfies at most $(1-\rho \cdot \eta)$-fraction of the constraints.
\end{itemize}
\end{theorem}
The rest of this section is devoted to deriving from~\cref{thm: pcpp starting} our initial uniquely-decodable PCP construction, and the argument is organized as follows:
\begin{enumerate}
    \item In~\cref{sec:transform_pcpp_to_udpcp} we show how to transform a PCPP into a uniquely-decodable PCP.
    \item In~\cref{sec:play_with_size} we show how to change the number of vertices in a uniquely-decodable PCP without affecting its parameters too much. This transformation will be needed in~\cref{sec:hdx_routing} for technical reasons. 
    \item In~\cref{sec:udpcp_reg} we show how to make a uniquely-decodable PCP right-regular.
    \item In~\cref{sec:udPCP_final} we combine the transformations in prior sections and prove~\cref{lm: initial udPCP regular}, which is the uniquely-decodable PCP construction we use in subsequent sections.
\end{enumerate}

\subsection{Converting a PCPP into a udPCP}\label{sec:transform_pcpp_to_udpcp}
We now describe how to obtain a udPCP from the PCPP in ~\cref{thm: pcpp starting}. Our transformation is similar to the one in~\cite[Proposition 6.12]{DinurMeir}, but it has quasi-linear length rather than quadratic length at the expense of achieving a weaker soundness guarantee.

\begin{lemma}\label{lm: initial udPCP}
There exists $k = O(1)$ such that for all $\eta > 0$ and any circuit
$\varphi: \Sigma_0^n \to \{0,1\}$ of size $N$, there is a poly-time constructible $(k+1)$-left regular, left-canonical, projection bipartite udPCP 
\[
\mc{D} =  (A \cup B, E, \Sigma_A, \Sigma_B, \{\Phi_e\}_{e \in E}, \{D_t \}_{t \in [n]},  \{\mc{P}_t \}_{t \in [n]} ),
\]
satisfying:
\begin{itemize}
    \item \textbf{Length.} $O_{|\Sigma_0|}(N \cdot \poly(\log (N)))$.
    \item \textbf{Alphabet Size.} Both left and right alphabets are of size $O_{|\Sigma_0|}(1)$.
    \item \textbf{Complete Decoding Distribution.} Letting $\mc{Q}$ be the complete decoding distribution, the distribution $\mc{Q}(\circ,\cdot)$ is uniform over $A$.
    \item \textbf{Completeness.} The udPCP has perfect completeness.
    \item \textbf{Soundness.} $(\eta, \frac{\rho\cdot \eta}{2(k+1)})$-unique decoding soundness, where $\rho$ is the constant from \cref{thm: pcpp starting}.
\end{itemize}
\end{lemma}
\begin{proof}
Take $\Psi = (V_0, E_0, \Sigma', \{\Phi_{0,e}\}_{e\in E_0})$ to be the PCPP from \cref{thm: pcpp starting}. We assume without loss of generality that $|E_0| / n$ is an integer, otherwise we may add trivial edges to $\Psi_0$ at the cost of at most halving the constant $\rho$ in the soundness guarantee. We construct the udPCP $\mc{D}$ as follows.
\vspace{-1ex}
\paragraph{Vertices.} 
The vertex set, which we denote by $V$, consists of two parts $A, B$, where $B := V_0$; we recall that $[n]\subseteq V_0$ by~\cref{def:pcpp}. To construct $A$, we first partition the edges $E_0$ evenly into $n$ 
parts of equal size, which we denote by $\{E_{0,t}\}_{t \in [n]}$. For each $t \in [n]$ and $e \in E_{0, t}$ we add the vertex $(e, t)$ to $A$, so in total $|A| = |E_0|$. Writing an edge $e \in E_0$ as a $k$-tuple $e = (v_1,\ldots, v_k) \in B^k$, we also think of $(e, t)$ as $(v_1,\ldots, v_k, t)$. Throughout the proof, we will refer to each $v_i$ as a vertex in $e$.

\vspace{-1ex}
\paragraph{Edges.} For each $(e, t) \in A$ where $e = (v_1, \ldots, v_k)$, we add an edge from $(e, t)$ to $t$ and to each one of $\{v_1,\ldots,v_k\}$. Thus in total, the number of neighbors of $(e,t)$ is $k+1$.
\vspace{-1ex}
\paragraph{Alphabets.} We describe a constrained alphabet for each individual vertex. For each  $v \in B$ its alphabet is $\Sigma'$. For each $(e, t) = (v_1,\ldots,v_k,t) \in A$, the alphabet $\Sigma_{e,t}$ consists of all maps $F_{e,t}: \{v_1,\ldots,v_k\} \cup \{t\} \to \Sigma'$ such that the assignment to $e$ satisfies the constraint $\Phi_{0,e}$ from $\Psi_0$. If $v = t$ or $v$ is a vertex in $e$, we will write $F_{e,t}(v)$ to denote the value of the symbol corresponding to $v$.
\vspace{-1ex}
\paragraph{Constraints.} For each $(e,t) \in A$ and $v \in B$ such that $v$ is a vertex in $e$ or $v= t$, the constraint $\Phi_e$ checks for agreement.  That is, for $F_{e,t} \in \Sigma_{e,t}$  and $\sigma \in \Sigma_{v}$, the constraint is given by
\[
\Phi_{((e,t), v)}(F_{e,t}, \sigma) = 1 \iff F_{e,t}(v) = \sigma.
\]
\vspace{-5ex}
\paragraph{Decoding Distributions.} For each $t \in [n]$, the decoding distribution 
$\mc{P}_t$ chooses a uniform  
$(e,t) \in \{(e,t)~|~e\in E_{0,t}\}$.  One can check that choosing $t \in [n]$ uniformly, and $(e,t)$ according to the marginal of $\mc{P}_t$, we have $(e,t)$ is uniform in $A$. This is because we split $|E_0|$ into $n$ equal parts before appending the edges of each part with a unique $t \in [n]$. Hence the condition on the complete decoding distribution is satisfied. 

\vspace{-2ex}
\paragraph{Decoder.} The decoder is defined as $D_t((e,t), F_{e,t}) = F_{e,t}(t)$.

\skipi
We now move on to the analysis of of $\mc{D}$.  The length of our construction is $|A| + |B| = |V_0| + |E_0|$ and this is $N \cdot \polylogn$. It is also clear that the alphabet size is $O(1)$, that the constraint graph is $(k+1)$-left regular and bipartite, 
that the constraints are projections, and that the udPCP is left-canonical.

\vspace{-2ex}
\paragraph{Perfect Completeness.} Fix $w \in \mc{L}$; then by the completeness of $\Psi$, there is an assignment $A_0: V_0 \to \Sigma'$ which satisfies all of the constraints of $\Psi_0$ and is equal to $w$ on the vertices of $V_0$ associated with $[n]$. Define the assignment $F: A \to B$ by setting $F(v) = A_0(v)$  for each $v \in B$, and setting $F((e,t))$ to be the restriction of $A_0$ to the vertices in $(e,t)$ for each $(e, t) \in A$. It is easy to see that $F$ assigns a valid alphabet symbol to each vertex, that it satisfies all of the constraints in $\mc{D}$, and that for any $t$ and $(u,v) \in \supp(\mc{P}_t)$, we have $D_t(u, F(u)) = w_t$. 

\paragraph{Unique Decoding Soundness.} Fix $\eta > 0$, set $\eps = \frac{\rho \cdot \eta}{2k}$ and suppose $A: V \to \Sigma$ is an assignment to $\mc{D}$ that satisfies at least $(1-\eps)$-fraction of the constraints. We will show that there exists $w \in \mc{L}$ such that the decoding error with respect to this $w$ is at most $\eta$.

For each $(e, t) \in A$ define $p(e,t)= \Pr_{v \in (e,t)}[A(v) \neq A(e,t)(v)]$, and note that the assumption that $A$ satisfies at least  $(1-\eps)$-fraction can be written as
$\E_{(e,t) \in A}[p(e,t)] \leq \eps$. Note that if $p(e,t) > 0$ for some $(e,t)$, then it must be the case that $p(e,t) \geq \frac{1}{k+1}$. Applying Markov's inequality we conclude that
\begin{equation} \label{eq: left side consistency} 
\Pr_{(e,t) \in A}[p(e,t) = 0] \geq 1- (k+1) \eps.
\end{equation}

In words, for at least $(1- (k+1)\eps)$-fraction of $(e,t) \in A$, under the assignment $A$, there is perfect agreement with all $v$ appearing in $e$ and with $v=t$. Consider the restriction $A_0 = A|_{B}$ and view it as an assignment to the PCPP $\Psi_0$. Then
\begin{equation} \label{eq: sat old pcpp}
\Pr_{e \in E_0}[\Phi_{0,e}(A_0(e)) = 1] \geq \Pr_{(e,t) \in A}[p(e,t) = 0] \geq 1- (k+1)\eps \geq 1 - \frac{\rho \cdot \eta}{2},
\end{equation}
so $A_0$ satisfies at least $(1-\rho\eta/2)$-fraction of the constraints of $\Psi_0$.
Here, in the first transition we used the fact that for a uniformly random $(e,t) \in E$, the distribution of $e$ is uniform over $E_0$, and also that if $p(e,t) = 0$, then $A_0$ satisfies the constraint $\Phi_{0,e}$. 
Combining~\eqref{eq: sat old pcpp} with the soundness guarantee of~\cref{{thm: pcpp starting}} we get that the restriction of $A_0|_X$ is $\frac{\eta}{2}$-close to $\mc{L}$. Namely, there exists $w \in \mc{L}$ such that 
\begin{equation} \label{eq: consistency with w}
    \Pr_{t \in [n]}[A_0(t) \neq w_t] \leq \frac{\eta}{2}.
\end{equation}
The decoding error of $\Psi_0$ can now be bounded by:
\begin{align*}  
\Pr_{t \in [n], (e,t) \sim \mc{P}_t}[D_t((e,t), A(e,t)) \neq w_t] 
&= \Pr_{(e,t) \in E}[ A(e,t)(t) \neq w_t] \\
&\leq  \Pr_{(e,t) \in E}[ A(e,t)(t) \neq w_t \land A(e,t) = A_0|_{(e,t)}] \\
&+ \Pr_{(e,t) \in E}[A(e,t) \neq A_0|_{(e,t)}] \\
&\leq \frac{\eta}{2} + (k+1)\eps \\
&\leq \eta.
\end{align*}
By $A_0|_{(e,t)}$, we mean the restriction of $A_0$ to the vertices in $e$ and $t$. In the last transition we bound the first term using \eqref{eq: consistency with w} and use the fact that when choosing $(e,t) \in E$ the element $t$ from the tuple is uniformly random in $[n]$. We bound the second term by the probability that $p(e,t) > 0$, which is at most $(k+1)\eps$ by \eqref{eq: left side consistency}.
\end{proof}

\subsection{Enlarging the Constraint Graph}\label{sec:play_with_size}
Here, we describe how to enlarge the constraint graph of a bipartite, projection, left-canonical udPCP. When working with standard PCPs, one can simply create multiple disjoint copies of the constraint graph and be done. It is not hard to see that essentially all properties, including the soundness, of the PCP are preserved when copying the constraint graph. This transformation does not work for udPCPs, since it does not necessarily preserve unique decoding soundness. Indeed, there could be two completely different, fully satisfying assignments on each copy that decode to different members of the language. We circumvent this issue by using the left-canonical property.  

\begin{lemma} \label{lm: repeat A side}
    Suppose that 
    \[
    \mc{D} = (A \cup B , E,\Sigma_A, \Sigma_B, \{\Phi_{e}\}_{e \in E}, \{D_t \}_{t \in [n]}, \{\mc{P}_t \}_{t\in [n]})
    \]
    is a bipartite, $k$-left-regular, projection, left-canonical udPCP for the language $\mc{L} \subseteq \Sigma_0^n$ with $(\eta_1, \eta_2)$-unique decoding soundness. Suppose additionally that the marginal of the complete decoding distribution over the vertices is uniform in $A$. Then for all $C\in \mathbb{N}$, there is a $\poly(C, |A| + |B|)$-time algorithm which takes as input $\mc{D}$ and outputs a udPCP for $\mc{L}$,
    \[
    \mc{D}' = (A \times [C] \cup B , E', \Sigma_A, \Sigma_B, \{\Phi'_{e}\}_{e \in E'}, \{\mc{P}'_t \}, \{D_t \}_{t \in [n]})
    \]
      with the following properties:
    \begin{itemize}
        \item \textbf{Constraint graph.} The constraint graph has $C|A| + |B|$ vertices, and each $(a,i) \in A \times [C]$ is adjacent to the same vertices that $a$ is adjacent to in the original constraint graph.
        \item \textbf{Degrees.} The left degree and the left decoding degree of each vertex is preserved, while the right degree of each vertex is multiplied by $C$.

        \item \textbf{Soundness.} $\mc{D}'$ has $(\eta_1 + \eta_2, \frac{\eta_2}{2k})$-unique decoding soundness.
    \end{itemize}
    Furthermore, the completeness, decision complexity, decoding complexity, and alphabet sizes of $\mc{D}$ and $\mc{D}'$ are the same, and the marginal of the complete decoding distribution over the vertices is uniform over $A \times [C]$. If $\mc{D}$ has projection constraints then $\mc{D}'$ does as well and if $\mc{D}$ is left-canonical then $\mc{D}'$ is as well.
\end{lemma}
\begin{proof}
      Construct the new udPCP
    \[
    \mc{D}' = (A \times [C] \cup B , E', \{\mc{P}'_t \}_{t\in [n]}, \Sigma_A, \Sigma_B, \{\Phi'_{e}\}_{e \in E'}, \{D'_t \}_{t \in [n]})
    \]
    as follows. The edges are $E' = \{((a,i),b) \in (A \times [C]) \times B \; | \; (a,b) \in E \}$. 
    The constraint on $((a,i), b)$ is the same as that on $(a,b)$ in the original udPCP, $\mc{D}$. The new decoding distribution $\mc{P}'_t$ is generated by choosing $a \sim \mc{P}_t$, choosing $i \in [C]$ uniformly at random, and outputting $(a,i)$. 
    The new decoder $D'_t$ on vertex $(a,i)$ is the same as the original decoder $D_t$ on vertex $a$. 
    
    It is clear that completeness is preserved, the alphabets are unchanged, the distribution of $(a,i) \in A \times [C]$ is uniform in the complete decoding distribution of $\mc{D}'$. It is also straightforward to see that $\mc{D}'$ has projection constraints if $\mc{D}$ has projection constraints and $\mc{D}'$ is left-canonical if $\mc{D}$ is left-canonical. Furthermore, the conditions regarding the constraint graph, and degrees are straightforward to check.
    
    We complete the proof by showing $\mc{D}$ has $(\eta_1 + \eta_2, \frac{\eta_2}{2k})$-unique decoding soundness. 
    Suppose $T_1: A \times [C] \to \Sigma_A$ and $T_2: B \to \Sigma_B$ are a $(1 - \eta')$-satisfying assignments to $\mc{D}'$ for $\eta':= \frac{\eta_2}{2k}$. For each $(a,i) \in A \times [C]$, let $p(a,i)$ denote the fraction of constraints involving $(a,i)$ that are not satisfied. Since each $(a,i)$ has degree at most $k$, we have that if $p(a,i) \neq 0$ then $p(a,i) \geq \frac{1}{k}$, and $\E_{(a,i) \in A \times [C]}[p(a,i)] \leq \eta'$. By Markov's inequality, we have that 
    \begin{equation} \label{eq: pai is 0}   
    \Pr_{(a,i) \in A \times [C]}[p(a,i) = 0] = \E_{a \in A}\left[\Pr_{i \in [C]}[p(a,i) = 0]\right]\geq 1 - k\eta'.
    \end{equation}
    By Markov's inequality again, we have that for at least $(1-k\eta')$-fraction of $a \in A$, $\Pr_{i \in [C]}[p(a,i) = 0] > 0$, and hence for $(1-k\eta')$-fraction of $a$, there exists an assignment in $\mc{D}$ to $a$ which satisfies all of the constraints involving $a$ conditioned on the right side being assigned $T_2$. We call such $a$ good. Let $T^*: A \to \Sigma_A$ be the table such that $T^*[a] = T_1[(a,i)]$ if there exists $i \in [C]$ such that $p(a,i) = 0$, and $T^*[a]$ is arbitrary otherwise. Note that since $\mc{D}$ is left-canonical, if $a$ is good then $T_1[(a,i)]$ is the same for all $i \in [C]$ such that $p(a,i) = 0$, so the assignment $T^*$ is well defined. 
    
    It follows that for the $T^*, T_2$ satisfy $(1-k\eta')$-fraction of the constraints in $\mc{D}$. Since $k \eta' \leq \eta_2$ and $\mc{D}$ satisfies $(\eta_1, \eta_2)$-unique decoding, we conclude that there exists $w \in \mc{L}$ such that
    \begin{equation} \label{eq: udpcp_size_manip}   
    \Pr_{t \in [n], a \sim \mc{P}_t}[D_t(a, T^*[a]) \neq w_t] \leq \eta_1.
     \end{equation}
    We now bound the decoding error of $T_1, T_2$ in $\mc{D}'$ as
    \begin{align*}
     &\Pr_{t \in [n], a \sim \mc{P}_t, i \in [C]}[D_t((a,i), T_2[(a,i)]) \neq w_t] \\
     &\leq \Pr_{t \in [n], a \sim \mc{P}_t, i \in [C]}[ T_2[(a,i)] \neq T^*[a]] +  \Pr_{t \in [n], a \sim \mc{P}_t, i \in [C]}[D_t(a, T^*[a]) \neq w_t] \\
     &\leq \Pr_{a \in A, i \in [C]}[\text{$a$ is not good} \lor p(a,i) \neq 0] + \eta_1 \\
     &\leq 2k\eta' + \eta_1 \\
     &= \eta_1 + \eta_2
    \end{align*}
In the second transition we are using the assumption that $a \in A$ is uniform in the complete decoding distribution of $\mc{D}$ and that if $a$ is good and $p(a,i) = 0$, then we must have $T^*[a] = T_2[(a,i)]$ by definition of $T^*$. In the third transition we are using~\eqref{eq: pai is 0},~\eqref{eq: udpcp_size_manip} and the fact that at least $(1-k\eta')$-fraction of $a \in A$ are good.
\end{proof}

\subsection{Making the udPCP Right-Regular}\label{sec:udpcp_reg}
The next step is to convert the construction in~\cref{lm: initial udPCP} into a bi-regular one. Towards this end we use the following transformation (similar transformations have been used in the context of standard PCPs in~\cite{MoshkovitzRaz,dh}). 
\begin{lemma} \label{lm: right degree reduction udpcp}
    Suppose $\mc{L}$ has a $k$-left-regular projection udPCP $\mc{D} = \left( A \cup B, E, \Sigma_A, \Sigma_B ,\{D_t \}_{t \in [n]}, \{\mc{P}_t \}_{t \in [n]} \right)$ with $(\eta, \eps)$-unique decoding soundness and minimum right degree $d_0$. Then for every integer $d' \leq d_0$, there is a polynomial time transformation which takes $\mc{D}$ and outputs a udPCP, $\mc{D}' = \left(A \cup B', E', \Sigma_A, \Sigma_B, \{D_t\}_{t \in [n]} , \{\mc{P}'_t \}_{t \in [n]} \right)$, for $\mc{L}$ that is $(d'k, d')$-regular and has $(\eta, \eps - O(d'^{-1/2}))$-unique decoding soundness. Furthermore $|B'| = |E| = k \cdot |A|$ and if $\mc{D}$ is left-canonical, then $\mc{D}'$ is as well.
\end{lemma}
\begin{proof}
Fix $d'$ as in the statement. By \cref{lm: poly time bip expander} we get that for every $N' \geq d'$ there is a ${\poly}(N')$-time algorithm that constructs a $d'$-regular bipartite expander graph with $N'$ vertices on each side and second singular value $O(d'^{-1/2})$, and we will use these graphs in our construction.
 
 The udPCP $\mc{D}'$ is constructed as follows. Recall that $(A \cup B, E)$ is the constraint graph of $\mc{D}$ and that for each vertex $b \in B$, $\Gamma(b)$ is the neighborhood of $b$ in this graph. For each $b \in B$ we construct a $d'$-regular bipartite expander graph $H_b = (A_b, B_b, E_b)$ with second singular value $O(d'^{-1/2})$, and we identify $A_b$ with $\Gamma(b)$.
 With these expander graphs in mind, the left side of $\mc{D}'$ is $A$, while the right side of $\mc{D}'$ is $\bigcup_{b \in B} B_b$. The edges, $E'$, are obtained as follows. For each $b \in B$, add the edges $E_b$ from the graph $H_b$, where the left side, $A_b$, is identified with $\Gamma(b)$ and the right side is $B_b$. Overall, the edges are $E' = \cup_{b \in B} E_b$. Note that choosing $(a,b) \in E$ uniformly at random and $b' \in B_b$ adjacent to $a$ uniformly at random, the edge $(a,b') \in E'$ is uniformly random. 
 
 The left and right alphabets of $\mc{D}'$ are still $\Sigma_A, \Sigma_B$ and the decoders of $\mc{D}'$ are still $\{D_t \}_{t \in [n]}$. For each $(a, b') \in E'$, if $b' \in B_b$, the constraint on $(a,b')$ is the same as the constraint on $(a,b)$ in $\mc{D}$. Finally the decoding distributions for $\mc{D}'$ are the same as in $\mc{D}$.
 
This completes the description of $\mc{D}'$, and we now move on to the analysis.
It is clear from the construction that $|B'|=\sum_{b\in B}|\Gamma(b)| = |E|=k|A|$. Also, $\mc{D'}$ is $(d'k, d')$-regular as each neighbor $b$ of $a\in A$ in $\mc{D}$ contributes $d'$ neighbors for $a$ in $\mc{D}'$. We move on to showing that the left-canonical property is preserved (if $\mc{D}$ satisfies it) as well as verifying completeness and soundness. Throughout, we write $(a', b') \sim H_b$ to denote a uniformly random edge from $H_b$.

\vspace{-1ex}
\paragraph{The left-canonical property is preserved:} suppose $\mc{D}$ is left-canonical, and consider any vertex $a \in A$, let $b_1, \ldots, b_k$ be neighbors of $a$ in $\mc{D}$, and let $b'_1, \ldots, b'_k$ be neighbors of $a$ in $\mc{D}'$, where $b'_i \in B_{b_i}$ for each $i \in [k]$. Fix any assignment to $b'_1, \ldots, b'_k$, and consider the same assignment to $b_1, \ldots, b_k$. Since $\mc{D}$ is left-canonical, there is at most one assignment to $a$ satisfying all of the constraints in $\mc{D}$, and since the constraints between $(a, b_i)$ in $\mc{D}$ and $(a, b'_i)$ in $\mc{D}'$ are the same, there is also at most one assignment to $a$ satisfying all of the constraints in $\mc{D}'$. Hence, $\mc{D}'$ is left-canonical as well.

\vspace{-1ex}
\paragraph{Perfect Completeness:} fix $w \in \mc{L}$; by the perfect completeness of $\mc{D}$, there are assignments $T_1:A \to \Sigma_A$ and $T_2: B \to \Sigma_B$ that satisfy all of the constraints in $\mc{D}$, and additionally for each $t\in [n]$, the decoder $D_t$ outputs $w_t$ with probability $1$. 
We take the assignments $T_1$ and $T_2'$ to $\mc{D}'$, where $T'_2: B' \to \Sigma_B$ is defined by $T'_2[b'] = T_2[b]$ if $b'\in B_b$. It is clear that $T_1$ and $T_2'$ satisfy all of the constraints in $\mc{D}'$ and that for all $t\in [n]$, the decoder $D_t'$ outputs $w_t$ with probability $1$.

\vspace{-1ex}
\paragraph{Unique Decoding Soundness:} fix assignments $T_1: A \to \Sigma_A$ and $T_2': B' \to \Sigma_B$ to $\mc{D}'$, and suppose that $T_1, T_2'$ satisfy at least $\left(1 - \eps'\right)$-fraction of the constraints, where $\eps':= \eps - O(d^{-1/2})$. For each $b \in B$ and $\sigma \in \Sigma_B$ define 
\[
X_{b,\sigma} = \{b' \in B_b \; | \; T'_2[b'] = \sigma \} \quad \text{and} \quad Y_{b, \sigma} = \{a' \in A_b \; | \; \Phi_{(a', b)}(T_1[a'], \sigma) = 1\}.
\]
In words, $X_{b, \sigma}$ is the set of vertices in $B_b$ that are assigned $\sigma$ under $T'_2$, and $Y_{b,\sigma}$ is the set of vertices in $a' \in \Gamma(b)$ such that the assignments $a \to T_1[a']$ and $b \to \sigma$ satisfy the constraint on $(a', b)$ in $\mc{D}$. It is clear that for each $b\in b$, the sets $\{X_{b,\sigma}\}_{\sigma \in \Sigma_B}$ form a partition of $B_b$. Additionally, since the constraints in $\mc{D}$ are all projections, for all $b\in B$ the sets $\{Y_{b,\sigma}\}_{\sigma \in \Sigma_B}$ form a partition of $A_b$. Define $F_{b,\sigma}: A_b \to [0,1]$, $G_{b,\sigma}: B_b \to \{ 0,1\}$ by
\[
G_{b, \sigma}(b') := \ind_{b' \in X_{b,\sigma}} \quad \text{and} \quad F_{b, \sigma}(a') := \ind_{a' \in Y_{b,\sigma}},
\]
and denote  $\mu(F_{b, \sigma}) = \E_{a \in A_b}[F_{b, \sigma}(a')]$ and $\mu(G_{b, \sigma})= \E_{b' \in B_b}[G_{b, \sigma}(b')]$. 
Consider the randomized assignment $T^*_2: B \to \Sigma_B$ obtained by setting $T^*_2(b) = \sigma$ with probability $\mu(G_{b, \sigma})$. The expected fraction of constraints in $\mc{D}$ satisfied by $T_1$ and $T^*_2$ is
\begin{equation} \label{eq: prior expected ud error} 
\E_{b \sim B}\left[\sum_{\sigma \in \Sigma_B} \mu(G_{b, \sigma}) \cdot \mu(F_{b, \sigma})]\right], 
\end{equation}
where $b \sim B$ is sampled according to the stationary distribution over $b \in B$ in the constraint graph of $\mc{D}$.
Next, the fraction of constraints satisfied by $T_1$ and $T_2$ in $\mc{D}'$ is
\begin{equation} \label{eq: sum over sigma_b in ud}  
\E_{b \sim B}\left[ \sum_{\sigma \in \Sigma_B} \E_{(a',b') \sim H_b}[F_{b,\sigma}(a') \cdot G_{b,\sigma}(b') ] \right] \geq 1 - \eps',
\end{equation}
where the inequality is because we assumed that $T_1, T_2$ satisfy at least $(1-\eps')$-fraction of constraints in $\mc{D}'$.

Now fix a $b \in B$ and a $\sigma \in \Sigma_B$. We will attempt to relate the inner expectation from \eqref{eq: sum over sigma_b in ud} to a product in the summation of \eqref{eq: prior expected ud error}. By~\cref{lm: expander mixing} we get
\[
\E_{(a',b') \sim H_b}[F_{b,\sigma}(a') \cdot G_{b,\sigma}(b') ] \leq \mu(F_{b,\sigma}) \cdot \mu(G_{b,\sigma}) + O(d'^{-1/2}) \cdot \sqrt{ \mu(F_{b,\sigma}) \cdot \mu(G_{b,\sigma})}.
\]
Keeping $b \in B$ fixed and summing this inequality over all $\sigma \in \Sigma_B$, we get
\begin{align*}  
 \sum_{\sigma \in \Sigma_B} \E_{(a',b') \sim H_b}[F_{b,\sigma}(a') \cdot G_{b,\sigma}(b') ] &\leq \sum_{\sigma \in \Sigma_B} \mu(F_{b,\sigma}) \cdot \mu(G_{b,\sigma}) + O(d'^{-1/2}) \cdot \sqrt{ \mu(F_{b,\sigma}) \cdot \mu(G_{b,\sigma})} \\
 &\leq \sum_{\sigma \in \Sigma_B} \mu(F_{b,\sigma}) \cdot \mu(G_{b,\sigma}) + O(d'^{-1/2}),
\end{align*}
where the second inequality uses Cauchy-Schwarz and the fact that $\sum_{\sigma \in \Sigma_B} \mu(G_{b,\sigma}) \leq 1$ and $\sum_{\sigma \in \Sigma_B} \mu(F_{b,\sigma}) = 1$. Thus, we have
\begin{align*}
 \E_{b \sim B}\left[\sum_{\sigma \in \Sigma_B} \mu(F_{b, \sigma}) \cdot \mu(G_{b, \sigma})]\right] \geq  \E_{b \sim B}\left[ \sum_{\sigma \in \Sigma_B} \E_{(a',b') \sim H_b}[F_{b,\sigma}(a') \cdot G_{b,\sigma}(b') ] \right] - O(d'^{-1/2}) \geq 1- \eps' - O(d'^{-1/2}),
\end{align*}
which is at least $1- \eps$ by the choice of $\eps'$.
It follows that there exists a choice for $T_2^{*}: B \to \Sigma_B$ such that $T_1$ and $T_2^{*}$ satisfy at least $(1-\eps)$-fraction of constraints in $\mc{D}$. By the $(\eta,\eps)$-unique decoding of $\mc{D}$, there exists $w \in \mc{L}$ such that
\[
\Pr_{t \in [n], a \sim \mc{P}_t}[D_t(a, T_1[a]) \neq w_t] \leq \eta.
\]
Notice that this is exactly the same as the decoding error of $\mc{D'}$ with assignments $T_1, T_2$. This is because for all $t$,  $\mc{D}_t$ and $\mc{D}'_t$ are the same function and the marginal distribution over $A$ in $\mc{P}_t$ and $\mc{P}'_t$ is the same. Thus, we have
\[
\Pr_{t \in [n], a' \sim \mc{P}_t'}[D'_t(a', T_1[a']) \neq w_t]  = \Pr_{t \in [n], a  \sim \mc{P}_t}[D_t(a, T_1[a]) \neq w_t] \leq \eta.
\]

It follows that $\mc{D}'$ satisfies $(\eta, \eps-O(d^{-1/2}))$-unique decoding soundness.
\end{proof}

\subsection{A Bi-regular Uniquely-decodable PCP of Arbitrary Size}\label{sec:udPCP_final}

In this section we put all of the previous pieces together to obtain a final udPCP that will go into the routing transformation of the next round. We highlight the following important properties:
\begin{itemize}
    \item The constraint graph is $d'$-regular and bipartite.
    \item The number of vertices in the constraint graph can be multiplied by any arbitrary $C \in \mathbb{N}$.
    \item The marginal of the complete decoding distribution over the vertices is uniform.
\end{itemize}

We remark that the last item above is in contrast to our original udPCP, from \cref{lm: initial udPCP}, whose complete decoding distribution has a marginal which is uniform over only a subset of the vertices.

\begin{lemma} \label{lm: initial udPCP regular}
   There is $\rho'>0$ such that for every $\eta > 0$ there is a sufficiently large $d' \in \mathbb{N}$ such that the following holds. For any alphabet $\Sigma_0$, there exists a function $s(N) = O_{|\Sigma_0|,\eta}\left(N \cdot \poly(\log(N))\right)$ such that for any $C \in \mathbb{N}$ and any size-$N$ circuit $\varphi: \Sigma_0^n \to \{0,1\}$, the language $\sat(\varphi)$ has a
   udPCP which satisfies the following properties:
    \begin{itemize}
        \item \textbf{Constraint Graph.} The constraint graph is a $d'$-regular bipartite graph.
        \item \textbf{Length.} $C \cdot s(N)$.
        \item \textbf{Alphabet Size.} $O_{|\Sigma_0|}(1)$.
        \item \textbf{Degrees.} The left and right degree are both $d'$ and the decoding degree is $1$.
        \item \textbf{Complete Decoding Distribution.} The marginal of the complete decoding distribution over the vertices is uniform.
        \item \textbf{Completeness.} The udPCP has perfect completeness.
        \item \textbf{Soundness.} The udPCP has $\left(\eta, \rho' \cdot \eta^5  \right)$-unique decoding soundness.
    \end{itemize}
\end{lemma}
\begin{proof}
Fix $\eta > 0$ and take the bipartite, $(k+1)$-left regular, projection, left-canonical udPCP $\mc{D}$, from \cref{lm: initial udPCP} with $\left(\eta, \frac{\rho \cdot \eta}{2(k+1)}\right)$-unique decoding soundness. Let $(A \cup B, E)$ be the constraint graph of this  udPCP, set $d' = \Theta(\eta^{-2})$, and fix an arbitrary $C \in \mathbb{N}$. Also set $s(N) = 2d'(k+1)|A| = 
O_{|\Sigma_0|,\eta}\left(N \cdot \poly(\log(N))\right)$. Let $\{D_t\}_{t \in [n]}$ and $\{ \mc{P}_t\}_{t \in [n]}$ be the decoders and decoding distributions respectively of $\mc{D}$.

We start by repeating every edge in $\mc{D}$ $d'$-times so that the minimum degree on the right hand side is $d'$. Now apply \cref{lm: repeat A side} with repetition parameter $C$. At this point, we get a $d'(k+1)$-left regular udPCP which has $C|A|$ vertices on the left, $|B|$ vertices on the right, and $(\eta + \frac{\rho \eta}{2(k+1)}, \frac{\rho \eta}{4(k+1)^2})$-unique decoding soundness.

Next, we apply \cref{lm: right degree reduction udpcp} with degree parameter $d' = O(\eta^{-2})$ chosen sufficiently large so that $O(d'^{-1/2})$ is sufficiently small relative to $\frac{\rho \eta}{4(k+1)^2}$. Let $R = d'^2(k+1)$. The resulting udPCP's constraint graph has $C|A|$ vertices on the left, $C|A|d'(k+1)$ vertices on the right, is $(R, d')$-bi-regular, and has  $(\eta + \frac{\rho \eta}{2(k+1)}, \frac{\rho \eta}{5(k+1)^2})$-unique decoding soundness.

We then apply \cref{lm: repeat A side} one more time with repetition parameter $d'(k+1)$ to obtain an $R$-regular udPCP. Since the constraint graph is fully regular, the number of vertices is just twice the number on either side, so the length of the resulting udPCP is $2Cd'(k+1)|A| = C \cdot s(N)$. Let $\mc{D}'$ be the resulting udPCP and $G' = (A' \cup B', E')$ be its constraint graph. Note that $G'$ is $R$-regular and $|A'| = |B'|$.
Since $C$ was arbitrary, we can indeed obtain length $C \cdot s(N)$ for any $C\in\mathbb{N}$ for $\mc{D}'$. Since $\rho$ and $k$ are constants and $d'$ is large enough, one can check through \cref{lm: repeat A side} that $\mc{D}'$ has $(1.1\eta, \frac{\rho \eta}{5(k+1)^2 R})$-unique decoding soundness. Here we use the fact that
\[
\eta + \frac{\rho \eta}{2(k+1)} + \frac{\rho \eta}{5(k+1)^2} \leq 1.1\eta.
\]
It is also straightforward to verify that perfect completeness is preserved from $\mc{D}$, the alphabets are the same as $\mc{D}$, that the left decoding degree is still $1$ as in $\mc{D}$, and that the marginal of the complete decoding distribution over the vertices is
uniform over $A$, the left side of the constraint graph.

Now we will further modify $\mc{D}'$, to obtain a new udPCP, $\mc{D}''$ whose parts are as follows. In the below, recall that each $a\in A$ has decoding degree $1$, so we denote by $t(a) \in [n]$ its only decoding neighbor.

\begin{itemize}
    \item \textbf{Constraint Graph.}  Starting with $G'$, the constraint graph of $\mc{D}'$, choose an arbitrary matching $\pi: A' \to B'$, meaning the edges of the matching are $\{(a, \pi(a)) \}_{a \in A'}$. Add $R$ copies of this matching to the graph (meaning between each $a$ and $\pi(a)$, we add $R$ parallel edges). The resulting $2R$-regular graph, $G'' = (A' \cup B', E'')$, is the constraint graph of $\mc{D}''$.
    \item \textbf{Alphabet.} The alphabet for each vertex in $A'$ is the same as in $\mc{D}'$.For each right vertex $b \in B'$, we require its alphabet to additionally include one $\Sigma_0$ symbol. It is intended that this symbol holds the $t(\pi^{-1}(b))$-th index of some satisfying assignment.

\item \textbf{Decoding Distribution.} For each $t \in [n]$, the  decoding distribution $\mc{P}''_t$ now samples $a \sim \mc{P}_t$, and then outputs one of $a$ or $\pi(a)$ uniformly at random. 
    \item \textbf{Decoder.} For each $a \in A'$ and symbol $\sigma$ for $a$, the decoder $D''_t(a, \sigma)$ is unchanged and outputs $D_t(a, \sigma)$. For $b = \pi(a) \in B'$, the decoder outputs the $\Sigma_0$ symbol held in $b$'s symbol (we note that it supposedly corresponds to what the decoder would have output if given the vertex $t(\pi^{-1}(b))=t(a)$ and its label).
 \item \textbf{Constraints.} For each $(a,b) \in E''$, if $b \neq \pi(a)$, then 
 the constraint on $(a,b)$ checks the constraint from $\mc{D}'$. Specifically, given labels $\sigma_a$ for $a$ and $\sigma_b$ for $b$ is as follows. Recall that $\sigma_b$ contains a label for $b$ in the udPCP $\mc{D}'$. Let $\sigma'_b$ be this symbol. The constraint checks if $\sigma_a, \sigma'_b$ satisfy the constraint on $(a,b)$ in $\mc{D}'$. If $b = \pi(a)$, then the constraint checks the previous condition and in addition checks that the symbols held by $a$ and $\pi(a)$ lead to the same decoding under $D'_t$. That is, given symbols $\sigma, \sigma'$, the constraint additionally checks if $D'_t(a, \sigma) = D'_t(b, \sigma')$.
\end{itemize}
We now verify that the guarantees of the lemma are still satisfied after the above modification. It is clear that the length is unchanged, the alphabet size only increases by a constant factor, the constraint graph is  $2R$-regular, every vertex has decoding degree $1$, and perfect completeness is preserved. Regarding the new decoding distribution, it is clear that choosing $t \in [n]$ and sampling $v \sim \mc{P}''_t$, we have that $v$ is uniformly random in $A'' \cup B''$. This is because, choosing $t \in [n]$ and $a \sim \mc{P}'_t$ we have that $a$ is uniformly random in the left side. Since $\pi$ is a matching, if we then output $a$ or $\pi(a)$ uniformly at random, the resulting vertex is uniform in the whole graph.

Finally, we show that $\mc{D}''$ satisfies $(1.11 \eta, \frac{\rho \eta}{  10(k+1)^2 R})$ unique decoding soundness. Suppose $T$ is an assignment to $A' \cup B'$ satisfying at least $(1 - \frac{\rho \eta}{10(k+1)^2 R})$-fraction of the constraints in $\mc{D}'$. Abusing notation, we will also think of $T$ as an assignment to $\mc{D}$ by removing the $\Sigma_0$ symbol attached to each $T(b)$.

Let us split the constraints into $E_0$, consisting of the constraints in $\mc{D}$, and $E_1$ consisting of the new constraints added from the matching. Since each of $E_0, E_1$ consists of half of the overall constraints, we get that the fraction of satisfied constraints in both $E_0, E_1$ is at least $\left(1 -  \frac{\rho \eta}{5(k+1)^2 R}\right)$-fraction of the constraints there, and hence applying the unique-decoding soundness of $\mc{D}'$, we get that there exists a $w\in\sat(\varphi)$ such that
\begin{equation} \label{eq: non bip 1}   
\Pr_{t \in [n], a \sim \mc{P}'_t}[D'_t(a, T[a]) \neq w_t] \leq 1.1\eta.
\end{equation}
On the other hand the decoding error of $T$ in $\mc{D}''$ relative to $w$ can be bounded by
\begin{equation}
 \begin{split}   
&\Pr_{t \in [n], a \sim \mc{P}'_t}[D''_t(a, T[a]) \neq w_t \lor D''_t(\pi(a), T[\pi(a)]) \neq w_t] \\
&\leq \Pr_{t \in [n], a \sim \mc{P}'_t}[D'_t(a, T[a]) \neq w_t]  + \Pr_{t \in [n], a \sim \mc{P}'_t}[D'_t(\pi(a), T[\pi(a)]) \neq D'_t(a, T[a])]  \\
&\leq  1.1\eta + \frac{\rho \eta}{5(k+1)^2 R} \\
&\leq 1.11 \eta,
\end{split}
\end{equation}
where in the third transition we bound the first term using \eqref{eq: non bip 1} and the second term using the fact that it is at most the fraction of unsatisfied constraints in $E_1$. This shows that $\mc{D}''$ has $(1.11 \eta, \frac{\rho \eta}{  10(k+1)^2 R})$-unique decoding soundness which is as desired because $R = O(\eta^{-2})$.
\end{proof}

%% file: hdx_route.tex
\section{Routing Onto an HDX}\label{sec:hdx_routing}
In this section, we show how to transform 
the udPCP construction from~\cref{lm: initial udPCP regular} into a udPCP whose constraint graph is the base layer of an HDX from~\cref{thm: hdx construction}, and for that we use the fault-tolerant routing protocol of~\cite{bmv}. 
The main difference is that we analyze unique-decoding soundness, whereas~\cite{bmv} analyze the standard notion of soundness.

\begin{theorem}\label{thm: udpcp graph to hdx}
For every sufficiently small $\eta > 0$, any $\delta \in (0,1)$, sufficiently large  $C'$, sufficiently large $k\in\mathbb{N}$ and sufficiently large $d\in \mathbb{N}$ the following holds. Any circuit $\varphi: \Sigma_0^n \to \{0,1\}$ of size $N$ has a udPCP $\mc{D}$ which satisfies the following properties:
\begin{itemize}
    \item \textbf{Constraint graph.} Let $X$ be the complex constructed from the algorithm in \cref{thm: hdx construction} where the direct product soundness, exponent, direct-product dimension, target number of vertices, and prime parameters are set to $\delta$, $C'$, $k$, $d$, $N \cdot \poly(\log N)$, $\Theta(\log^{C'}(N))$ respectively. The constraint graph of $\mc{D}$ is $(X(1), X(2))$, where the multiplicity of each edge $(u,v)$ is proportional to the weight of $(u,v)$ in stationary distribution over $(X(1), X(2))$. The udPCP has length $|X(1)| = O_{|\Sigma_0|, \eta, \delta}(N  \poly(\log(N))$.
    \item \textbf{Degrees.} Each vertex in the constraint graph has degree $\poly_{\delta}(\log N)$ and decoding degree $\poly_{\delta}(\log N)$.
    \item \textbf{Alphabet Size.} $O_{|\Sigma_0|}\left(2^{\poly_{\eta, \delta}(\log(N ))}\right)$. 
    \item \textbf{Decision Complexity.} $O_{ |\Sigma_0|}(\poly_{\eta,\delta}(\log(N )))$.
    \item \textbf{Decoding Distribution.} The marginal of the complete decoding distribution over the vertices is the stationary distribution of $(X(1), X(2))$ over $X(1)$.
    \item \textbf{Decoding Complexity.} $O_{|\Sigma_0|}(\poly_{\eta, \delta}(\log N ))$. 
    \item \textbf{Unique-decoding Soundness.} The udPCP has $(\eta, \eta^2/6)$-unique decoding soundness.
\end{itemize}
\end{theorem}

\subsection{The Zig-Zag Product}
In order to convert from a udPCP over an arbitrary constraint graph $G = (V, E)$ into a udPCP over the base layer $(X(1), X(2))$ of an HDX, we need a way to associate the vertices of $G$ with the vertices $X(1)$. As in~\cite{bmv}, to do so we rely on the zig-zag product of~\cite{RVWzigzag}, which we present next.

The zig-zag product takes a graph $G = (V, E)$, with minimum degree $m_0$, and a family of $k$-regular expander graphs $\mc{G} = \{\mc{G}_m\}_{m \geq m_0}$ where $\mc{G}_m$ has $m$ vertices, and outputs a $k^2$-regular graph on $2|E|$ vertices. We denote the output of the zig-zag product as $\ZZ(G, \mc{G})$. Typically, the family of graphs $\mc{G}$ is chosen to have second singular value at most $\sigma$, for some chosen parameter $\sigma$. 
Before describing the zig-zag product, it is convenient to describe the replacement product of $G'$ and $\mc{G}$. We denote this graph by $\Rep(G, \mc{G})$, defined in the following way.
\begin{itemize}
    \item For each vertex $v$, arbitrarily order its neighbors from $1$ to $\deg(v)$. We will refer to the $c$th neighbor of $v$ as the vertex labeled $c$ under this ordering.
    \item Replace each vertex $v \in V'$ with a copy of $\mc{G}_{\deg(v)}$, where $\deg(v)$ is the degree of $v$ in $G'$. Hence $v$ is replaced by a cloud of $\deg(v)$ vertices, and we also label these vertices by $\{v\}\times[\deg(v)]$. 
     \item Let $(u, v) \in E'$ be an edge where $v$ is the $c_1$th neighbor of $u$ and $u$ is the $c_2$th neighbor of $v$, then we add the edge $((u, c_1), (v, c_2))$ to $\Rep(G, \mc{G})$. In addition, we add the edge $((v,c_1), (v,c_2))$ to $\Rep(G, \mc{G})$ for all $v$ and $(c_1,c_2)$ such that $((v,c_1),(v,c_2))$ is an edge in $\mc{G}_{\deg(v)}$.
\end{itemize}
Note that $\Rep(G, \mc{G})$ has $2|E|$ vertices and is $(k+1)$-regular. We define the zig-zag product, $\ZZ(G', \mc{G})$ as follows.
\begin{itemize}
    \item The vertex set is $V(\ZZ(G, \mc{G})) = V(\Rep(G, \mc{G}))$.
    \item We add an edge $((u, c_1), (v, c_4))$ to $E(\ZZ(G, \mc{G}))$ if there exist $c_2$ and $c_3$ such that all of $((u, c_1), (u,c_2))$, $((u, c_2), (v,c_3))$, $((v, c_3), (v,c_4))$ are edges in $E(\Rep(G, \mc{G}))$. In words, the vertex $(u, c_1)$ is adjacent to $(v, c_4)$ if the latter vertex can be reached, in $\Rep(G, \mc{G})$, by first taking a step in the cloud of $u$, then taking a step between the clouds $u$ and $v$, and finally taking a step in the cloud of $v$. 
\end{itemize}
It is straightforward to verify that $\ZZ(G, \mc{G})$ has $2|E|$ vertices and is $k^2$ regular.
\subsection{Routing Protocols}
We now discuss routing protocols and associated useful notations, and we start by explaining what a $T$-round routing protocol on a graph $G = (V,E)$ is. It has three stages: input, communication, and output, described as follows.
\vspace{0.1cm}

\noindent \textit{Input.} Each vertex $v \in V$ holds an initial message, which we think of as a symbol from some global alphabet $\Sigma$.

\vspace{0.1cm}

\noindent \textit{Communication.}  During each round $r \in [T]$, each vertex $v\in V$ performs some computation and sends a message, which is a symbol from $\Sigma$, to each of its neighbors. The message to each neighbor may be different. Moreover, the computation $v$ makes depends on its initial message and any messages received from its neighbors in prior rounds. The history of messages $v$ sent/received before step $r$ is called the \emph{transcript} of $v$ in round $r$, so that $v$'s outgoing messages during round $r$ are functions of this transcript. For $r = 1$, $v$'s message depends only on its initial message.

\vspace{0.1cm}

\noindent \textit{Output.} After round $T$, each vertex calculates an output (as some function of its transcript in round $T$). 

\begin{definition}
Given an arbitrary graph $G = (V, E)$, a $T$-round routing protocol on $G$, which we denote by $\mc{R}$, consists of the following items.

\begin{itemize}
    \item \textbf{Initial Message.} An initial message $A: V \to \Sigma$ where $\Sigma$ is some alphabet.
    \item \textbf{Out Maps.} For each round $r \in [T]$, a map $\outmap_{r}: 2E \times \Sigma^{*} \to \Sigma$, where $\outmap_r(u,v, \sigma)$ denotes the message that $u$ sends to $v$ in round $r$ given that $u$'s transcript at round $r$ is $\sigma$. Here $2E$ consists of all ordered pairs of adjacent vertices, meaning each edge $(u,v) \in E$ appears as $(u,v)$ and $(v,u)$ in $2E$, and the map is defined only for $\sigma$ which represent a valid transcript of the protocol up to step $r$.
    \item \textbf{In Maps.} For each round $r \in [T]$, a map $\inmap_r: 2E \to \Sigma$ where for each $r \in [T]$, $\inmap_r(u,v)$ should be thought of as the message received by $v$ from vertex $u$ in round $r$ of the protocol.
    \item \textbf{Output}. A function $A_T: V \times \Sigma^{*} \to \Sigma$, where $A_T(v, \sigma)$ is the symbol output by $v$ at the end of the routing protocol given that its transcript at the end of the routing protocol is $\sigma$. We refer to $A_T$ as the function computed by the protocol. 
\end{itemize}
\end{definition}
We will also have to consider a quantity known as the \emph{work complexity} of a routing protocol.
\begin{definition}
    The work complexity of a $T$-round routing protocol over $G = (V,E)$ is the maximum over $r\in[T]$ and $(u,v) \in 2E$, of the circuit complexity of the function $\outmap_{r}(u,v, \cdot): \Sigma^* \to \Sigma$.
\end{definition}

Although we are able to use the routing protocol of \cite{bmv} in a black-box way, it will still be helpful to give a brief description of why routing protocols are useful for transforming the constraint graph of udPCPs. We refer the reader to \cite{bmv} for a more detailed discussion.

In the setting of udPCPs, routing protocols are used as follows. Let $\mc{D}$ be an initial udPCP with bipartite, $k$-regular constraint graph $G = (A \cup B, E)$, and suppose that $G' = (V, E')$ is the target graph which we want to use as the constraint graph of a new udPCP, $\mc{D}'$. Suppose for this exposition that $G$ is a matching, that $|V| = |A| = |B|$, and let $\pi: A \to B$ denote this matching. Identifying both $A$ and $B$ with $V$, we can view $\pi$ as a permutation on $V$, i.e.\ as a mapping $\pi: V \to V$. Now the initial message $A: V \to \Sigma$ to the routing protocol should be thought of as the assignment of $A$ and the goal of the communication of the routing protocol is to pass the value $A(v)$, which is initially held at vertex $v$, to the vertex $\pi(v)$. If this is done correctly, then each vertex $\pi(v)$ can check if the assignment $A$ actually satisfied the constraint on $(v, \pi(v))$ in the original udPCP. Hence, the desired output function is $A_T(\pi(v), \sigma) = A(v)$, where here $\sigma$ is the transcript of the honest routing protocol received by vertex $\pi(v)$. Assuming $\pi(v)$ correctly computes $A_T(\pi(v), \sigma) = A(v)$, it can check itself if $A(v)$ and its own initial value $A(\pi(v))$ should satisfy the constraint on $(v, \pi(v))$.

The above gives a high level overview of the connection between routing protocols and transforming constraint graphs of udPCPs, and roughly speaking it describes what happens in the perfect completeness case. In order to show soundness, one has to consider what happens to the routing protocol under adversarial edge corruptions. 

\begin{definition}
    We say that an edge $(u,v)$ is adversarially corrupted if, for every round, the outgoing message from $u$ may be an adversarially selected symbol in $\Sigma$. 
\end{definition}

In \cite{bmv}, routing protocols are considered in a slightly different setting. There, the target graph  $G = (X(1), X(2))$ is the base layer of an arbitrary HDX, $X$, and one considers the zig-zag product $Z = \ZZ(G, \mc{G})$, where $\mc{G}$ is an appropriately chosen family of expanders. Instead of considering a function directly over the vertices of $G$, one considers a function over the zig-zag product, $Z$, which is then associated with $G$. Strictly speaking, the zig-zag product is not necessary to describe the \cite{bmv} routing protocol, but it is helpful to have it in mind as it virtually appears in the definition of the initial messages and the function computed. We will similarly use $Z$ and below we describe how routing protocols over $Z$ translate to protocols over $G$.

Let us label the vertices of $Z$ as $V(Z) = \{(u,v) \; | \; u \in X(1), (u,v) \in X(2) \}$ 
and let the initial messages over $V(Z)$ be given by $A'_0: V(Z) \to \Sigma$. Fix a permutation $\pi: V(Z) \to V(Z)$. The goal of the routing protocol is to transmit $A'_0(v)$ to $A'_0(\pi(v))$ for all $v \in V(Z)$. The actual routing protocol should take place over $G$ though, so we would like a way to simulate the initial messages and final function (which are over $V(Z)$), as functions over $X(1)$. To this end, define $\mc{H} = \{(j, u) \; | \; j \in V(Z), u \in X_{j_1}(1) \}$, and define the function $A_0: \mc{H} \to \Sigma$, where $A_0(j, u) = A'_0(j)$. Now, we think of the initial message of a vertex $u \in X(1)$ as holding all of the symbols $(A_0(j,u))_{j \in V(Z)\text{ such that }u \in X_{j_1}(1)}$. One can think of the function $A_0$ as a function over the links of $X$: for each $j\in V(Z)$, each vertex $u\in X_{j_1}(1)$ contains the symbol of $j$ (among other things). Then, the routing protocol over $Z$ naturally translates over to a link-to-link transfer protocol over $G$.

Below we describe the formal guarantee of the routing protocol from \cite[Lemma 4.6]{bmv}. In the soundness case, we have the following lemma.
\begin{lemma}\label{lm: routing}
    There exist $\eps_0, \alpha > 0$ and $d_0, C' \in \mathbb{N}$ such that the following holds for sufficiently large $n \in \mathbb{N}$. Let $\mc{G}$ be a family of expanders, each with second singular value at most $\frac{\alpha}{2}$, and let $X$ be any complex constructed as in \cref{thm: hdx construction} with dimension parameter, exponent, target number of vertices, and prime set to $d \geq d_0$, $C'$, $n$, and $q = \Theta\left(\log^{C'}(n)\right)$ respectively. 
    Let $Z = \ZZ(G, \mc{G})$, and let $\pi: V(Z) \to V(Z)$ be any permutation over $V(Z)$. 
    
    Then there is a routing protocol on $G = (X(1), X(2))$ with round complexity $T = O(\log(n))$ and work complexity $q^{O(d^2)}\log |\Sigma|$  such that for any initial function $A_0 : \mc{G} \to \Sigma$ which satisfies
    \[
    \Pr_{j \in V(Z)}[\maj_{0.99}(A_0(j,u)) \; | \; u \sim X_{j_1}(1) )\neq \perp] \geq 1 - \xi,
    \]
    and for all possible adversarial corruptions of at most $\eps$-fraction of edges with $\eps \leq \eps_0$, the protocol computes a function $A_T: \mc{G} \to \Sigma$ which satisfies
    \[
    \Pr_{j \in V(Z)}[\maj_{0.99}(A_T(\pi(j), w) \; | \; w \sim X_{\pi(j)_1}(1)) = \maj_{0.99}(A_T(j, u) \; | \; u \sim X_{\pi(j)_1}(1))] \geq 1 - \xi - \frac{1}{\log n}.
    \]
\end{lemma}

In words,~\cref{lm: routing} gives, for any permutation $\pi$ over $V(Z)$, a routing protocol that is tolerant against up to $\eps$ errors. If an overwhelming majority value exists in at least $1-\xi$ of the links in the input message, then the same overwhelming majority value exists in at least $1-\xi - \frac{1}{\log n}$ of the links in the output message. 

The completeness of the protocol from \cref{lm: routing} is stated in the following lemma. 

\begin{lemma} \label{lm: routing completeness}
Let $G = (X(1), X(2))$ where $X = (X(1), \ldots, X(d))$ is a $d$-dimensional complex, let $\mc{G}$ be an appropriate family of expanders as above and let $A_0: \mc{G} \to \Sigma'$ be a function such that $A_0(j,v)$ is the same for all $v \in X_{j_1}(1)$. Then for each permutation $\pi\colon V(Z)\to V(Z)$, an honest run of the protocol from~\cref{lm: routing} satisfies
\[
    A_{0}(j, v) = A_{T}(\pi(j), u)
\]
    for every $(j, v), (\pi(j), u) \in \mc{G}$.
\end{lemma}

We will require the following fact from \cite{bmv} which translates the uniform distributions over vertices and edges of $Z$ to the stationary distributions over vertices and edges of $X$.
\begin{fact} \label{fact: hdx sampling}
    Choosing $j \sim V(Z)$ uniformly at random and outputting $j_1$, we have that $j_1 \in X(1)$ is distributed according to the underlying distribution of $X(1)$. Additionally, choosing an edge $(j,j')$ of $Z$ uniformly at random and outputting $(j_1, j'_1) \in X(2)$, we have that $(j_1, j'_1)$ is distributed according to the underlying distribution of $X(2)$.
\end{fact}

\subsection{Proof of \cref{thm: udpcp graph to hdx}}
In this section we prove \cref{thm: udpcp graph to hdx}. 

\subsubsection{The Construction}
Fix $\eta > 0$, 
$\delta \in (0,1)$, a size $N$ circuit $\varphi: \Sigma_0^n \to \{0,1\}$ and set $\mc{L} = \sat(\varphi) \subseteq \Sigma_0^n$. Let
\[
 \mc{D}' = \left(A' \cup B', E', \Sigma', \{\Phi_{e} \}_{e \in E'}, \{\mc{P}'_t \}_{t \in [n]}, \{D'_t \}_{t \in [n]}\right)
\]
to be the udPCP for $\mc{L}$ guaranteed by \cref{lm: initial udPCP regular} with $\left(\frac{\eta}{100}, \rho' \cdot \left(\frac{\eta}{100}\right)^5\right)$-unique decoding soundness. 
Note that by \cref{lm: initial udPCP regular}, we may take the length of $\mc{D}'$ to be $C\cdot s(N)$ for any $C \in \mathbb{N}$, where $s(N) = O_{|\Sigma_0|,\eta}(N \cdot \log^{c_0}(N))$ for some constant $c_0$. Here we are slightly abusing notation as technically, there is a separate udPCP for each $C$, but let us just refer to the udPCP as $\mc{D}'$ for now and fix $C$ later. The udPCP $\mc{D}'$ is bipartite and $d'$-regular. We denote its constraint graph by $G' = (A' \cup B', E')$, write $V' = A' \cup B'$, and set $N' = |V'|$.

Take $\eps_0, \alpha, d_0, C'$ to be the constants from \cref{lm: routing} and let $\mc{G}$ be a family of $O(\alpha^{-2})$-regular expander graphs such that each graph has second singular value at most $\frac{\alpha}{2}$. In particular, we use \cref{lm: poly time expander} where $d = O(\alpha^{-2})$, and obtain a family $\mc{G}$ with one graph of each size $m$ for all $m \geq \Omega(\alpha^{-2})$. Note that by \cref{lm: poly time expander},
each graph in the family $\mc{G}$ can be constructed in time polynomial in the number of vertices.
We take the complex $X$ from \cref{thm: hdx construction} so that the following items hold:
\begin{itemize}
    \item the direct-product soundness parameter is set to the $\delta$ fixed above, 
    \item the dimension parameter $d \geq d_0$,
    \item the number of vertices satisfies $N \cdot \log^{c_0+10}(N) \leq |X(1)| \leq N \cdot \log^{10c_0 + 10}(N)$, and is sufficiently large in accordance with \cref{lm: routing},
    \item the prime $q$ is set to $q = \log^{c'}(N)$ for some constant $c'$,
    \item $N \cdot \log^{c_0+10}(N) \leq |X(2)| \leq N \cdot \log^{c_1}(N)$, where $c_1 \geq c_0 + 10$ is some constant. This item is implied by the fact that $d$ is constant along with the prior two items and item 3 of \cref{thm: hdx construction}.
\end{itemize}
Let us briefly argue why such a complex $X$ can be constructed from \cref{thm: hdx construction}. There, one can fix $d$ arbitrarily and then choose a desired number of vertices, say $n'$, arbitrarily. In our case, we set $n'$ to be in the range $[N \cdot \log^{c_0+10}(N), N \cdot \log^{10c_0+10}(N)]$, and note that we can assume $n'$ is sufficiently large in accordance with \cref{lm: routing} by assuming that the original circuit is sufficiently large. After $n'$ is fixed, we set $q = \Theta(\log^{C'}(n'))$ and we run the algorithm of \cref{thm: hdx construction} with direct-product soundness, exponent, dimension, target number of vertices, and prime parameters set to $\delta, C', d, n', q$ and direct-product dimension set arbitrarily. \cref{thm: hdx construction} guarantees that this $X$ is constructed in  $\poly(n') = \poly(N)$ time.
\vspace{-2ex}
\paragraph{Viewing $(X(1), X(2))$ as a multigraph.} We will want to consider $(X(1), X(2))$ as a weighted graph according to its stationary distribution and this is accomplished by having each edge $(u,v) \in X(2)$ have multiplicity proportional to its weight under the stationary distribution. Specifically, recall that the stationary distribution over $(u,v)$ samples a $d$-face $U \sim X(d)$ and then 
$\{u, v\} \in \binom{U}{2}$ uniformly, so the weight of $(u,v)$ under this distribution is, $\frac{m(u,v)}{\binom{d}{2}}$, where $m(u,v)$ is the number of $d$-faces containing $u$ and $v$, and is hence $\polylogn$. For each $(u,v) \in X(2)$, we think of it as having multiplicity $m(u,v)$ in the multigraph $(X(1), X(2))$.

\vspace{-2ex}
\paragraph{Identifying $V'$ with $V(Z)$ by adding a negligible number of vertices.}
The first step will be to modify $G'$ so that it has size equal to $|V(Z)|$. Recall $|V(Z)| = 2|X(2)|$, so 
\begin{equation}
    N \cdot \log^{c_0 + 10}(N) \leq |X(2)| \leq N \cdot \log^{c_1}(N),
\end{equation}
for some $c_1 \geq c_0$. We can thus write $|V(Z)| = s(N) \cdot Q + R$ for integers $Q$ and $R$ satisfying $\log^{10}(N) \leq Q \leq \log^{c_1 - c_0}(N)$ and $0\leq R < s(N)$. We set the repetition parameter $C = Q$, so that we have $|V'| = Q \cdot s(N)$ and add $R$ vertices to the constraint graph of $\mc{D}'$. On these $R$ additional vertices we add an arbitrary bipartite $d'$-regular graph with equality constraints. Note that these $d'R$ extra equality constraints are at most $\frac{d'R}{d' \cdot s(N) \cdot Q} \leq \frac{1}{\log^{10}(N)}$ fraction of all constraints, so they can essentially be ignored. The decoders for $\mc{D}'$ are unchanged after adding $R$ additional vertices (we can define the decoders arbitrarily when one of these $R$ additional vertices is given as input). We define the new decoding distributions as follows. Let $\xi := R/|V'|$-so that the $R$-additional vertices account for $\xi$-fraction of the vertices overall in $G'$ and divide the $R$ additional vertices into equal parts of size $n$ and label them $1$ through $n$.
For each $t \in [n]$, the decoding distribution $\mc{P}'_t$ is modified as follows
\begin{itemize}
    \item With probability $1-\xi$ choose $u$ according to the old decoding distribution for $t$ without the $R$ extra vertices.
    \item With probability $\xi$ choose one of the extra $R$ vertices in the part labeled by $t$.
\end{itemize}
One can check that choosing $t \in [n]$ and then $u \sim \mc{P}'_t$ as above, we still get that $u$ is uniform in $V'$ (with the $R$ additional vertices). Since $\xi \leq \frac{1}{\log^{10}(N)}$, it is easy to see that after adding the $R$-additional vertices, the udPCP still has $\left(\frac{\eta}{100}, \frac{\rho'}{2}\left(\frac{\eta}{100}\right)^5\right)$-unique decoding soundness.
It is also straightforward to check that the resulting udPCP has length $|V(Z)|$  and perfect completeness. Furthermore, alphabet size, $d'$-regularity, and left-decoding degree are maintained. Henceforth, we slightly abuse notation and still refer to the modified constraint graph as $G' = (A' \cup B', E')$ and assume that it has the same number of vertices as $Z$ for the remainder of the section. 

\vspace{-2ex}
\paragraph{Breaking $G'$ into matchings and setting the routing protocols.}
Since $G'$ is a $d'$-regular bipartite graph we can decompose $E'$ into $d'$ perfect matchings $\pi_1, \ldots, \pi_{d'}$. Since $|V(Z)| = |V'|$ we can associate $V'$ with $V(Z)$ in an arbitrary manner, and we fix such an identification henceforth. Abusing notation, we view each $\pi_i$ as a function from $V(Z) \to V(Z)$ satisfying that $\pi_i^2$ is the identity. In particular, for $j \in V(Z)$ corresponding to some $a \in A'$, $\pi_i(a)$ is the vertex in $V(Z)$ corresponding to the vertex in $B'$ matched with $a$ under $\pi_i$, and vice-versa if $j \in V(Z)$ corresponds to some $b \in B'$.

In order to describe the construction, it will be helpful to define $\mc{H} = \{(j,u) \; | \; j \in V(Z) , u \in X_{j_1}(1) \}$ and consider an initial function $A_0: \mc{H} \to \Sigma'$ where recall that $\Sigma'$ is the alphabet of $\mc{D}'$. One should think of $A_0$ as coming from a satisfying assignment of $\mc{D}'$, so that $A_0(j,u)$ is the $\Sigma'$ symbol assigned to the vertex in $V'$ corresponding to $j$ under this satisfying assignment. Moreover, one should think of every $A_0(j,u)$ as being ``held'' at the vertex $u \in X(1)$, so that $u \in X(1)$ holds many $\Sigma'$ symbols (and in particular, one for each link it is in). That is, an assignment to the new udPCP is supposed to provide such values for each $u \in X(1)$.

For each $i \in [d']$ let $\mc{R}_i$ be the routing protocol from \cref{lm: routing} for the permutation $\pi_i$. The protocols $\mc{R}_i$ all have round complexity $T$ for some $T = O(\log N)$. Here, one should think of the routing protocols as being run with the initial maps $A_0$ above. Then in words, for every $j \in V(Z)$, the routing protocol $\mc{R}_i$ is meant to route an assignment of the $V'$ vertex corresponding to $j$ from every vertex in the link $X_{j_1}$ to every vertex in the link $X_{\pi_i(j_1)}(1)$. For each $i \in [d'], r \in [T]$, let $\inmap_{i, r}$ and $\outmap_{i,r}$ denote the in map and out map respectively during round $r$ of protocol $\mc{R}_i$. Also, let $\mc{A}_{i,T}$ denote the output function of protocol $\mc{R}_i$.
\vspace{-2ex}
\paragraph{The description of $\mc{D}$.}
We denote the new dPCP that we construct as
\[
\mc{D} = \left(X(1), X(2), \Sigma, \{\Phi_{e} \}_{e \in X(2)}, \{D_t\}_{t \in [n]}, \{\mc{P}_t \}_{t \in [n]}\right),
\]
and define its parts as follows.
\begin{itemize}
    \item \textbf{Constraint Graph.} The constraint graph is the base layer of $X$, i.e.\ $(X(1), X(2))$ (with the multiplicities proportional to the weights of the stationary distribution of $X(2)$ over $X(1)$).
    \item \textbf{Alphabets:} For each $u \in X(1)$, an alphabet symbol for $u$ consists of the following.
    \begin{itemize}
        \item \textbf{Initial Messages:} one $\Sigma'$ symbol for each $j \in V(Z)$ such that $u \in X_{j_1}(1)$. This symbol is intended to be the value $A_0(j, u)$ for each $j \in V(Z)$ which is supposed to come from a satisfying assignment to $\mc{D}'$ as described above.
        \item \textbf{Routing Transcript:} the maps $\inmap_{i, r}: 2E(G) \to \Sigma$ for each $i \in [d']$ and $r \in [T]$. For each $i \in [d']$ and $r \in [T]$, this is the in map of round $r$ for the protocol $\mc{R}_i$. Recall that $\inmap_{i,r}(u,v)$ should be thought of as the message received by $v$ from $u$ during round $r$ in the routing protocol for the matching $\pi_i$. 
    \end{itemize}
    In addition, we constrain the alphabet of $u$ in the following ways.
    \begin{itemize}
        \item \textbf{Hardcoding $\mc{D}'$ Constraints:} for each $i \in [d']$ and $u\in V'$, we only allow symbols $\sigma$ for $u$ such that the values in $\sigma$ corresponding to $A_0(j, u)$ and $A_{i, T}(\pi_i(j), u)$ satisfy the constraint $\Phi'_{(j, \pi_i(j))}$ in $\mc{D}'$, for every $j$ such that $u \in X_{j_1}(1)$.
        \item \textbf{Hardcoding Decoding Consistency:} for all $t\in [n]$ and $u\in V'$, we only allow symbols $\sigma$ for $u$ which satisfy the following. Denoting by $\sigma_j$ the symbol corresponding to $A_0(j, u)$ in $\sigma$,  
        we require that $\mc{D}'_t(j, \sigma_j)$ is the same for all $j$ such that $u \in X_{j_1}(1)$ and $j \in \supp(\mc{P}_t')$. 
    \end{itemize}
    We denote the constrained alphabet for $u$ as $\Sigma_u$. 
        \item \textbf{Constraints:} For each $(u,v) \in X(2)$, the constraint between $u$ and $v$ checks the following. 
        \begin{itemize}
            \item For each $j \in V(Z)$ such that $u,v \in X_{j_1}(1)$, check that
            \[
            A_0(j, u) = A_0(j,v) \quad \text{and} \quad A_{i,T}(j,u) = A_{i,T}(j,v), \; \forall i \in [d'].
            \]
            \item For each $i \in [d']$ and $r \in [T-1]$, check that 
            \[
            \inmap_{i, r+1}(u, v) = \outmap_{i, r}(v, u) \quad \text{and} \quad \inmap_{i, r+1}(v,u) = \outmap_{i,r}(v,u).
            \]
    \end{itemize}
    In words, the first item checks that the initial message parts of the symbols assigned to $u$ and $v$ are consistent. The second item checks that the routing transcripts of $u$ and $v$ are consistent, that is, for each protocol $\mc{R}_i$ and every round $r \in [T-1]$, the constraint checks that the message that $v$ receives from $u$ (according to $v$'s alphabet symbol), is equal to the message that $u$ sends to $v$ (calculated based on the prior round in maps held at $u$'s alphabet symbol), and vice-versa.
    \item \textbf{Decoding Distributions:} For each $t \in [n]$, the distribution $\mc{P}_t$ is as follows
    \begin{itemize}
        \item Choose $j \sim \mc{P}'_t$ and think of $j \in V(Z)$ as a vertex in identifying $V'$ (by the identification we have between $V'$ and $V(Z)$).
        \item Output $u \sim X_{j_1}(1)$.
    \end{itemize}
   
    \item  \textbf{Decoders:} For each $t \in [n]$, $u \in \supp(\mc{P}_t)$  
    and label $\sigma \in \Sigma_u$ for $u$, the decoder outputs 
    \[
    \mc{D}_t(u, \sigma) = \mc{D}'_t(j, \sigma_j),
    \]
    where $j$ is an arbitrary vertex in $V(Z)$ such that $u \in X_{j_1}(1)$, and  $\sigma_{j}$ refers to the part of $\sigma$ which contains a label for $A_0(j, u)$. Note that the specific identity of $j$ above does not matter due to the way the alphabet of $\Sigma_u$ is constrained (specifically, the second item) and hence the decoder is well defined.
\end{itemize}

\subsubsection{Proof of Everything Except Unique-Decoding Soundness}

We start by proving all of the properties except for unique decoding soundness, which is deferred to the next section. 
\vspace{-2ex}
\paragraph{Constraint Graph and Length.} It is clear that the constraint graph is $(X(1), X(2))$ and hence the length is $|X(1)| = O_{|\Sigma_0|,\eta}(N \cdot \poly(\log(N)))$ by construction of the complex $X$. 
\vspace{-2ex}
\paragraph{Degrees.} Let $\Gamma', \Gamma'_{\dec}$ denote neighborhoods and decoding neighborhoods in $\mc{D}'$ and let $\Gamma, \Gamma_{\dec}$ be defined similarly for $\mc{D}$.  The degree of each vertex in the constraint graph is $\poly_{\delta}(\log n)$ by the properties of the HDX. In particular, it follows from the fact that each vertex is contained in $q^{O(d^2)}$ many $d$-faces by item 3 of \cref{thm: hdx construction}. As for the decoding degree, we have that $t \in \Gamma_{\dec}(v)$ if and only if there exists $j \in V(Z)$ such that $v \in X_{j_1}(1)$ and $t \in  \Gamma'_{\dec}(j)$. Since $|\Gamma'_{\dec}(j)| = O(1)$ for all $j \in V(Z)$, and the number of $j \in V(Z)$ such that $v \in X_{j_1}(1)$ is $\poly_{\delta}(\log n)$, the size of $v$'s decoding neighborhood is $\poly_{\delta}(\log n)$.
\vspace{-2ex}
\paragraph{Alphabet Size.} An alphabet for each vertex consists of the initial messages which the vertex holds, and the messages it received throughout all of the rounds of the routing protocol. For the initial messages, each vertex $v$ contains one symbol from a $O_{|\Sigma_0|}(1)$-sized alphabet for each $j \in V(Z)$ such that $v \in X_{j_1}(1)$. The number of such $j$ is $\poly_{\delta}(\log(N))$ since $v$ has $\poly_{\delta}(\log(N))$ neighbors in $(X(1), X(2))$ and for each $u$, there are a constant number of $j \in V(Z)$ such that $j_1 = u$. Thus, the initial message part of the alphabet consists of $\poly_{\delta}(\log(N))$ symbols from a $O_{|\Sigma_0|}(1)$-sized alphabet. For the routing transcript part, the alphabet of a vertex $v$ consists of one message per round of the protocol. By the work complexity of the protocol, the message contains at most $\poly_{\delta}(\log(N))$ bits, and there are $T = O(\log N)$ rounds. Altogether, this shows that each alphabet symbol for $v$ can be written using $\poly_{\delta}(\log(N))$ symbols from a $O_{|\Sigma_0|}(1)$-sized alphabet and hence the alphabet size is $2^{O_{|\Sigma_0|}(\poly_{\delta}(\log(N)))}$ as desired.

\vspace{-2ex}
\paragraph{Decision Complexity.} Fix an edge $(u,v)\in X(2)$. The decision complexity of $(u,v)$ is the sum of the following 2 parts: 
first, circuit sizes for checking that the alphabet symbols for $u$ and $v$ are valid symbols (i.e.\ that they satisfy the hardcoded constraints), and second, checking that the constraint is satisfied. For the first part, let us describe the circuit complexity for checking validity of $u$'s symbol, as the circuit complexity for checking $v$'s symbol is the same. The circuit needs to check that for each of $O_{\eta}(1)$ matchings $\pi$, the parts of the alphabet corresponding to $A_0(j, \pi(u))$ and $A_{i, T}(j,u)$ satisfy some constraint in $\mc{D}'$ for every $j$ such that $u \in X_{j_1}(1)$. Since the constraint in $\mc{D}'$ is over $O_{|\Sigma_0|}(1)$ sized alphabet symbols, its circuit complexity is $O_{|\Sigma_0|}(1)$. Moreover, there are at most $\polylogn$ many such $j$'s. Altogether, this results in circuit complexity $O_{|\Sigma_0|}(\poly_{\eta}(\log(N)))$. For the second part, the circuit needs to calculate the outgoing messages of $u$ and $v$ in each of the routing protocols $\mc{R}_1, \ldots, \mc{R}_{d'}$, and then check an equality on the computed message. As the work complexity is $O_{|\Sigma_0|}(\polylogn)$ and there are $O_{\eta}(1)$ routing protocols, this results in $O_{|\Sigma_0|}(\poly_{\eta}(\log(N)))$ circuit complexity as well. Overall, the decision complexity is $O_{|\Sigma_0|}(\poly_{\eta}(\log(N)))$.

\vspace{-2ex}
\paragraph{Decoding Distribution.} Choosing $t \in [n]$ uniformly and $j \sim \mc{P}'_t$ uniformly, we have that $j \in V(Z)$ is distributed uniformly at random by the guarantee of starting udPCP $\mc{D}'$ (specifically the fourth item of \cref{lm: initial udPCP regular}). Then the decoding distribution of $\mc{D}$ chooses $u \sim X_{j_1}(1)$, and we have $u \sim X(1)$ by \cref{fact: hdx sampling}.
\vspace{-2ex}
\paragraph{Decoding Complexity.} It is clear from the description of the decoder, that given a vertex $u$ and a symbol $\sigma$ for it, the decoder can restrict $\sigma$ to a symbol from the alphabet of $\mc{D}'$ and run the decoder of $\mc{D}'$ (which is another restriction of the symbol to a symbol from $\Sigma_0$). The circuit complexity for both of these tasks is on the order of logarithmic in the left alphabet size (which is also the length of the circuit's input), and hence it is $O_{|\Sigma_0|}(\poly_{\eta, \delta}(\log(N )))$. 
\vspace{-2ex}
\paragraph{Completeness.} Fix $w \in \mc{L}$ and let $A': V' \to \Sigma'$ be the assignment to $\mc{D}'$ which satisfies all of the constraints and always decodes according to $w$. Now, define the initial messages according to $A'$ and run the routing protocols $\mc{R}_1, \ldots, \mc{R}_{d'}$ with these initial messages. Let $A$ be the resulting assignment for $\mc{D}$. We will show that $A$ is a valid assignment satisfying all of the constraints to $\mc{D}$ and that it always decodes according to $w$. 

For the first part, it is clear that all of the constraints of $\mc{D}$ are satisfied, and it remains to check that $A$ is a valid assignment, i.e.\ it satisfies the harded constraints of the alphabets. Indeed, for any $v$, any of the matchings $\pi_i$, and any $j$ such that $v \in X_{j_1}(1) \cap X_{\pi_i(j)_1}(1)$, the alphabet symbols corresponding to $A_0(j, v), A_{i,T}(\pi(j),v)$ are precisely $A'(j)$ and $A'(\pi(j))$ by \cref{lm: routing completeness}, so they satisfy the constraint in $\mc{D}'$. For the decoding hardcoded constraints fix any vertex $v \in X(1)$ and $j \in V(Z), t \in [n]$ such that $v \in X_{j_1}(1)$ and $t$ is in the decoding neighborhood of $j$ in $\mc{D}'$. Then by the completeness of $\mc{D'}$, we have that $D'_t(j, A'(j)) = w_t$. Since the symbol corresponding to the initial value of $j$ under assignment $A(v)$ is precisely $A'(j)$, this shows that the folding for decoding consistency is satisfied.
\vspace{-1ex}
\subsubsection{Proof of Unique-Decoding Soundness}
We will show that $\mc{D}$ satisfies $(\eta, \eta^2/6)$-unique decoding soundness for a fixed $\eta$ sufficiently small relative to the constant $\eps_0$ from~\cref{lm: routing}.
Suppose there is an assignment to $\mc{D}$ that satisfies at least $(1-\eps)$-weight of the constraints, for $\eps := \eta^2 / 6$. Let $A: X(1) \to \Sigma$ be the assignment to $\mc{D}$ achieving this value. From $A$, we can derive the function $A_0: \mc{H} \to \Sigma'$. Specifically, for each $(j, v) \in \mc{H}$, we set $A_0(j,v)$ to be the part of the symbol $A(v)$ corresponding to an assignment for $j$ from the alphabet of $\mc{D}'$, $\Sigma'$. Now, define the following assignment, $B\colon V'\to\Sigma'$:
\[
B(j) = {\sf plurality}_{u \sim X_{j_1}(1)} \{A_0(j, u)\},
\]
where the function ${\sf plurality}$ refers to the most common value $A_0(j, u)$ when choosing $u \sim X_{j_1}(1)$. It is shown in \cite{bmv} that this assignment $B$ is a nearly satisfying assignment for $\mc{D}'$.

\begin{lemma} \cite[Lemma 5.3]{bmv}
Suppose that $A$ satisfies $(1-\eps)$-weight of the constraints in $\mc{D}$. Then $B$ defined as above satisfies at least $\left(1 - O\left(\frac{1}{\log N}\right)\right) \geq   \left(1- \frac{\rho'}{2}\left(\frac{\eta}{100}\right)^5\right)$-fraction of the constraints of $\mc{D}'$.
\end{lemma}
\begin{proof}
    The lemma is implicitly shown in the proof of the soundness claim in Lemma 5.3 of \cite{bmv}. Specifically, they argue that if $A$ is a $(1-\eps)$-satisfying assignment to $\mc{D}$, then for each of the routing protocols $\mc{R}_i$ the routing transcript of $\mc{R}_i$ corresponding to the assignments of $A$ corresponds to a routing on $(X(1), X(2))$ with at most $\eps$-fraction of adversarial edge corruptions. 
    
     Additionally, since $A$ is $(1-\eps)$-satisfying, one can show that for $\left(1-\frac{1}{\log N}\right)$-fraction of $j \in V(Z)$, we have that the probability a random vertex in the link $X_{j_1}(1)$ holds a value for $j$ that is consistent with $B$, i.e.\ the probability $\Pr_{u \in X_{j_1}(1)}[A_0(j,u) = B(j)]$, is very close to $1$ (in fact this statement is required for our proof of unique-decoding soundness too, so the argument is reproduced, essentially in \cref{lm: j bad prob}, below). As a result, the initial messages $A_0$ satisfy the condition in \cref{lm: routing} and applying \cref{lm: routing}, we get that for $\left(1-\frac{1}{\log N}\right)$-fraction of $j \in V(Z)$, the values $B(j)$ and $B(\pi_i(j))$ are consistent with those held at some vertex $u \in X_{j_1}(1)$. By the hardcoding of the alphabet of $u$, this shows that $B$ satisfies the constraint on $(j, \pi_i(j))$, and applying this over all routing protocols shows that $B$ satisfies almost all of the constraints in the original udPCP.
\end{proof}

Because $\mc{D}'$ satisfies $\left(\frac{\eta}{100}, \frac{\rho'}{2} \cdot \left(\frac{\eta}{100}\right)^5\right)$-unique decoding soundness with respect to $\mc{L}$ for some constant $\rho$, there exists $w \in \mc{L}$ such that 
\begin{equation} \label{eq: ref B decoding error}
\Pr_{t \in [n], j \sim \mc{P}'_t}[D'_t(j, B(j)) \neq w_t] \leq  \frac{\eta}{100}.
\end{equation}
We can bound the decoding error of the assignment $A$ for the udPCP $\mc{D}$ as follows
\begin{equation} \label{eq: bound decoding hdx}
\begin{aligned}
    \Pr_{t \in [n], u \sim \mc{P}_t}[D_t\left(u, A(u) \right) \neq w_t]  &= \Pr_{t \in [n], j \sim \mc{P}'_t, u \sim X_{j_1}(1)}[D'_t(j, A_0(j,u)) \neq w_t]\\
    &\leq \Pr_{t \in [n], j \sim \mc{P}'_t}[D'_t(j, B(j)) \neq w_t] \\
    &+ \Pr_{t \in [n], j \sim \mc{P}'_t, u \sim X_{j_1}(1)}[A_0(j, u) \neq B(j)] \\
    &\leq \frac{\eta}{100} +\Pr_{t \in [n], j \sim \mc{P}'_t, u \sim X_{j_1}(1)}[A_0(j, u) \neq B(j)]. 
\end{aligned}
\end{equation}

In the first transition we used the fact that $D_t\left(u, A(u) \right) = D'_t(j, A_0(j,u))$, in the second transition we used the union bound, and in the third transition we used~\eqref{eq: ref B decoding error}. Note that for any fixed $j$, we have 
\[
\Pr_{u \sim X_{j_1}(1)}[A_0(j,u) \neq B(j)] \leq \Pr_{u,v \sim X_{j_1}(1)}[A_0(j,u) \neq A_0(j,v)].
\]
This is because $B(j)$ is defined to be the most popular value of $A_0(j,u)$ when sampling $u \sim X_{j_1}(1)$. Plugging this into~\eqref{eq: bound decoding hdx} gives 
\begin{equation} \label{eq: hdx route 1}
\begin{split}
 \Pr_{t \in [n], u \in \mc{P}_t}[D_t\left(u, A(u)  \right) \neq w(t)] &\leq  \frac{\eta}{100} + \E_{t \in [n], j \sim \mc{P}'_t}\left[\Pr_{u,v \in X_{j_1}(1)}[A_0(j,u) \neq A_0(j,v)]\right].
 \end{split}
\end{equation}
 We say $j \in V(Z)$ is bad if 
\begin{equation} \label{eq: j bad assumption}
   \Pr_{u,v \in X_{j_1}(1)}[A_0(j, u) \neq A_0(j,v)] \geq 2\sqrt{\eps}. 
\end{equation}
The expectation from \eqref{eq: hdx route 1} can be bounded by

\begin{equation} \label{eq: hdx route dec error bound}
     \E_{t \in [n], j \sim \mc{P}'_t}\left[\Pr_{u,v \in X_{j_1}(1)}[A_0(j,u) \neq A_0(j,v)]\right] \leq 2\sqrt{\eps} + \Pr_{t \in [n], j \sim \mc{P}'_t}[\text{$j$ is bad}].
\end{equation}

To bound the probability that $j$ is bad, we will show that if $j$ is bad, then the link $X_{j_1}(1)$ must contain many violated constraints. To this end, let $\mc{E} \subseteq X(2)$ denote the set of constraints that are not satisfied by the assignment $A$. By assumption, $A$ satisfies at least $(1-\eps)$-fraction of constraints, so $\mu_{2}(\mc{E}) \leq \eps$.  On the other hand, if $A_0(j,u) \neq A_0(j,v)$ for any $j \in V(Z)$, and $(u,v) \in X_{j_1}(2)$, then the constraint on $(u,v)$ is violated and $(u,v) \in \mc{E}$. We can thus show the following.

\begin{lemma} \label{lm: j bad prob}
    If $j$ is bad then $\mu_{j_1,2}(\mc{E}) \geq \sqrt{\eps}$, and consequently  $\Pr_{t \in [n], j \sim \mc{P}'_t}[j \text{ is bad}] \leq \frac{\poly(d)}{q}$.
\end{lemma}
\begin{proof}
    Fix $j$ that is bad and for each symbol for $j$ in the original udPCP $\mc{D}'$,  $\sigma \in \Sigma'_{j}$, let 
    \[
    U_{\sigma} = \{u \in X_{j_1}(1) \; | \; A_0(j, u) = \sigma \}.
    \]
Then
\begin{equation} \label{eq: routing proof aux 1}
 \mu_{j_1,2}(\mc{E}) \geq \Pr_{(u,v) \sim X_{j_1}(2)}[A_0(j, u) \neq A_0(j, v)]
 = 1 - \sum_{\sigma \in \Sigma'_j} \Pr_{(u,v) \sim X_{j_1}(2)} [u \in U_{\sigma} \land v \in U_{\sigma}]. 
\end{equation}

Now, we use the fact that the graph $(X_{j_1}(1), X_{j_1}(2))$ has second eigenvalue at most $O\left(\frac{1}{\sqrt{q}}\right)$ by item 4 of \cref{thm: hdx construction}, so we can bound
\begin{equation} \label{eq: routing proof aux 2}
\begin{split}
    \sum_{\sigma \in \Sigma'_j} \Pr_{(u,v) \sim X_{j_1}(2)} [u \in U_{\sigma} \land v \in U_{\sigma}] &\leq 
  \sum_{\sigma \in \Sigma'_j} \mu_{j_1}(U_{\sigma})^2 + O(q^{-1/2})  \mu_{j_1}(U_{\sigma}) \\
  &\leq O(q^{-1/2}) + \sum_{\sigma \in \Sigma'_j} \mu_{j_1}(U_{\sigma})^2.
\end{split}
\end{equation}
Combining \eqref{eq: routing proof aux 1}, \eqref{eq: routing proof aux 2}, and the observation that $1 - \sum_{\sigma \in \Sigma'_j} \mu_{j_1}(U_{\sigma})^2 = \Pr_{u,v \sim X_{j_1}(1)}[A_0(j, u) \neq A_0(j,v)]$ we get that 
\begin{equation*}
\mu_{j_1,2}(\mc{E}) \geq 1 -  \sum_{\sigma \in \Sigma'_j} \mu_{j_1}(U_{\sigma})^2 -  O(q^{-1/2}) =  \Pr_{u,v \sim X_{j_1}(1)}[A_0(j, u) \neq A_0(j,v)] -  O\left(\frac{1}{\sqrt{q}}\right) \geq 2\sqrt{\eps} -  O(q^{-1/2}) \geq \sqrt{\eps},
\end{equation*}
where the penultimate transition uses the assumption that $j$ is bad from \eqref{eq: j bad assumption}. This shows the first part of the lemma.

To obtain the ``consequently'' part of the lemma, recall that by \cref{fact: hdx sampling}, when choosing $j \in V' = A' \cup B'$ uniformly at random and then outputting $j_1$, the resulting $j_1 \in X(1)$ is distributed according to the stationary distribution over $X(1)$. On the other hand, choosing $t \in [n]$ and $j \sim \mc{P}'_t$, we have $j$ uniformly distributed in $V'$. Thus,
\[
 \Pr_{t \in [n], j \sim \mc{P}'_t}[j \text{ is bad}] =  \Pr_{t \in [n], j \sim \mc{P}'_t}[\mu_{j_1}(\mc{E})  \geq \sqrt{\eps}] = \Pr_{j \in V'}[\mu_{j_1}(\mc{E})  \geq \sqrt{\eps}] \leq \frac{\poly(d)}{q}.
\]
In the second transition we used the above observation on the distribution of $j_1$. 
In the third transition we used \cref{lm: bmv hdx expander sampling}, $\mu_2(\mc{E}) \leq \eps$, 
and the fact that a uniformly random $j \in V'$ is on the side $A'$ with probability $1/2$.
\end{proof}

We can now conclude the proof of unique decoding soundness.
\begin{proof}[Proof of Unique Decoding Soundness]
Fix $\eta > 0$ sufficiently small, and suppose there is an assignment, $A$, to $\mc{D}$ that satisfies at least $(1-\eps)$-fraction of the constraints, for $\eps = \eta^2/6$. Then, combining \eqref{eq: hdx route 1}, \eqref{eq: hdx route dec error bound}, and \cref{lm: j bad prob}, we get that there is $w \in \mc{L}$ such that the decoding error of $A$ with respect to $w$ is at most
\[
\Pr_{t \in [n], u \sim \mc{P}_t}[D_t(u, A(u)) \neq w_t]\leq \frac{\eta}{100} +  2\sqrt{\eps} + \frac{\poly(d)}{q} \leq \eta.
\qedhere
\]
\end{proof}

%% file: direct_product.tex
\newcommand{\disc}{\mathsf{disc}}
\newcommand{\tv}{\mathsf{TV}}
\newcommand{\wt}{\mathsf{wt}}
\section{Gap Amplification}\label{sec:amp}

Starting from the udPCP construction from \cref{thm: udpcp graph to hdx}, we now obtain a dPCP with small list-decoding soundness error by using a direct product based gap amplification.
\begin{theorem} \label{thm: dp dpcp}
   For all $\eps > 0$ and a circuit $\varphi: \Sigma_0^n \to \{0,1\}$ of size $N$, the language $\sat(\varphi)$ has a dPCP satisfying the following:
        \begin{itemize}
        \item \textbf{Length.} $O_{|\Sigma_0|}\left(N \cdot \poly_{\eps}(\log N)\right)$.

            \item \textbf{Left Alphabet.} For each left vertex $a$, the left alphabet of $a$ is contained in $\Sigma^k$ for some $\Sigma \supseteq \Sigma_0$ of size $|\Sigma_2| = 2^{O_{|\Sigma_0|}(\poly_{\eps}(\log N))}$ and $k = O_{\eps}(1)$. Hence the left alphabet size is $2^{O_{|\Sigma_0|}\left(\poly_{\eps}(\log N)\right)}$.
         \item \textbf{Right Alphabet.} The right alphabet is $\Sigma^{\sqrt{k}}$ where $\Sigma$ and $k$ are the same as in the left alphabet condition.
        \item \textbf{Degrees.} The left degree is $O_{\eps}(1)$, the decoding degree is $\poly_{\eps}(\log N)$, and the right degree is $\poly_{\eps} (\log N)$.
         \item \textbf{Decision Complexity.} $O_{|\Sigma_0|}(\poly_{\eps} (\log N))$.
         \item \textbf{Decoding Complexity.} $O_{|\Sigma_0|}(\poly_{\eps}(\log N))$.
         \item \textbf{Decoding Distribution.} The complete decoding distribution is agnostic. Furthermore, denoting the complete decoding distribution by $\mc{Q}$, for every $\OutRightVertex \in X(\sqrt{k})$ we have that
         \[
         \mc{Q}(\circ, \circ, \OutRightVertex) \leq \frac{\poly_{\eps}(\log N)}{N}.
         \]
         \item \textbf{Completeness.} The dPCP has perfect completeness.
         \item \textbf{Soundness.} The dPCP has $(L, \eps)$-list-decoding soundness for $L:= \eps^{-O(1)}$.
         \end{itemize}
    \end{theorem}
We now dsecribe the construction for~\cref{thm: dp dpcp}. Fix $\eps > 0$ and a size $N$ circuit $\varphi: \Sigma_0^n \to \{0,1\}$ and set $\mc{L} = \sat{(\varphi)}$. Set $\delta := \delta(\eps)>0$ sufficiently small relative to $\eps$ and take $\mc{D}_{\hdx}$ to be the udPCP for $\sat{(\varphi)}$ obtained through \cref{thm: udpcp graph to hdx} with the direct product soundness parameter there set to be $\delta/100$ and $\eta>0$ there to be a small absolute constant so that the result holds. After setting $\delta$, the other parameters are all chosen sufficiently large as described in \cref{thm: udpcp graph to hdx}, and in particular, we have that the parameters of $X$ are as follows. The direct product parameter $k = O_{\delta}(1)$ is sufficiently large, the dimension parameter is $d = O_{\delta}(1)$ is sufficiently large, the number of vertices is $O_{|\Sigma_0|,\delta}(N \cdot \polylogn)$, the prime is $q = \poly_{\delta}(\log N)$. Let us denote the parts of $\mc{D}_{\hdx}$ by

\[
\mc{D}_{\hdx} = \left(X(1), X(2),\Sigma, \{\Phi^{\hdx}_e\}_{e \in X(2)}, \{\mc{P}^{\hdx}_{t}\}_{t \in [n]},  \{D^{\hdx}_{t}\}_{t \in [n]} \right).
\]
The amplified dPCP takes the form \[
\mc{D}_{\dip} = \left( X(k) \cup X(\sqrt{k}), E ,\Sigma^k, \Sigma^{\sqrt{k}}, \{\Phi^{\dip}_e\}_{e \in E},\{\mc{P}^{\dip}_{t}\}_{t \in [n]}, \{ D^{\dip}_{t}\}_{t \in [n]}\right).
\] 
Throughout the proof we will use boldface letters to denote vertices of $\mc{D}_{\dip}$. The components of $\mc{D}_{\dip}$ are as follows:
\begin{itemize}
    \item \textbf{Constraint Graph.} The bipartite inclusion graph with parts $X(k)$ and $X(\sqrt{k})$. Throughout the section we will use $\OutLeftVertex$ to denote a generic vertex in $X(k)$ and $\OutRightVertex$ to denote a generic vertex in $X(\sqrt{k})$. We remark that $\OutLeftVertex$ and $\OutRightVertex$ should also be thought of, respectively, as size $k$ and $\sqrt{k}$ subsets of $X(1)$. For $\OutLeftVertex \in X(k), \OutRightVertex \in X(\sqrt{k})$ we include an edge between them in $E$ if $\OutLeftVertex \supseteq \OutRightVertex$.
    
    \item \textbf{Left Alphabet:} For each $\OutLeftVertex \in X(k)$, its alphabet is a subset $\Sigma^k$ and each symbol $\sigma$ should be thought of as a map from assigning a $\Sigma$ symbol to each $v \in \OutLeftVertex$. We constrain the left alphabet of each $\OutLeftVertex \in X(k)$ to only contain maps  $\sigma: \OutLeftVertex \to \Sigma$ that satisfy the following properties:   
    \begin{itemize}
        \item \textbf{Hardcoding  $\mc{D}_{\hdx}$ constraints.} For any $ (u,v) \in X(2)$ such that both vertices are contained in $\OutLeftVertex$, $\sigma(u)$ and $\sigma(v)$ satisfy the constraint on $(u,v)$ from $\mc{D}_{\hdx}$.
        \item \textbf{Hardcoding for decoding consistency.}  For any $t \in [n]$, and any two $u, v \in \OutLeftVertex$ such that $t$ is in the decoding neighborhood of both $u$ and $v$ in $\mc{D}_{\hdx}$, we have
        \[
        D^{\hdx}_{t}(u, \sigma(u)) =  D^{\hdx}_{t}(v, \sigma(v)).
        \]
    \end{itemize}
    
    \item \textbf{Right Alphabet.} For each $\OutRightVertex \in X(\sqrt{k})$, its alphabet is $\Sigma^{\sqrt{k}}$, and we similarly think of each symbol as a map from $R \to \Sigma$ assigning a $\Sigma$ symbol to each $v \in \OutRightVertex$.
    \item \textbf{Constraints.} For each edge of the constraint graph, $e = (\OutLeftVertex, \OutRightVertex)$, the constraint $\Phi_e$ checks that the assignments to $\OutLeftVertex$ and $\OutRightVertex$ agree. That is, if the left and right vertices are assigned symbols $\sigma \in \Sigma^k,  \sigma' \in \Sigma^{\sqrt{k}}$ respectively,
    \[
    \Phi_{e}(\sigma, \sigma') = 1 \quad \text{if and only if} \quad \sigma|_{\OutRightVertex} =\sigma',
    \]
    where $\sigma|_\OutRightVertex \in \Sigma^{\sqrt{k}}$ is the restriction of $\sigma$ to the vertices in $\OutRightVertex$.
    \item \textbf{Decoding Distributions:} Let $\mc{Q}_{\hdx}$ be the complete decoding distribution (over $[n] \times X(1)$) of $\mc{D}_{\hdx}$ wherein one chooses $t \in [n]$ uniformly, $u \sim \mc{P}_{\hdx}(t)$, and outputs $(t,u)$. The total decoding distribution of $\mc{D}_{\dip}$, denoted by $\mc{Q}_{\dip}$, is a distribution over triples $(t,\OutLeftVertex, \OutRightVertex) \in [n]\times X(k) \times X(\sqrt{k})$ is as follows:
    \begin{itemize}
        \item Choose $\OutLeftVertex \sim X(k)$, and then $\OutRightVertex \subseteq \OutLeftVertex$ uniformly.
        \item Choose  $u \in \OutRightVertex$ uniformly and choose $t\sim \mc{Q}_{\hdx}(\cdot, u)$.
    \end{itemize}
    Then the decoding distribution $\mc{P}^{\dip}_{t}$ is defined as $\mc{Q}_{\dip}(t,\cdot, \cdot)$.
        
    \item \textbf{Decoders:} For $(t, \OutLeftVertex, \OutRightVertex) \in \supp(\mc{Q}_{\dip})$ and left alphabet symbol $\sigma \in \Sigma^k$ for $\OutLeftVertex$, the decoder outputs
    \[
    D^{\dip}_{t}(\OutLeftVertex, \sigma_\OutLeftVertex) = D^{\hdx}_{t}(u, \sigma_\OutLeftVertex(u)).
    \]
    where $u \in X(1)$ is an arbitrary vertex contained in both $\OutLeftVertex$ and $\supp(\mc{P}^{\hdx}_{t})$. Due to the hardcoding for decoding consistency, the above value is the same for any $u$ and hence the decoder is well defined.
    \end{itemize}
    The above completes the description of the new dPCP, and the remainder of the section is dedicated to its analysis.
\subsection{Proof of \cref{thm: dp dpcp}: Everything Except the Soundness}
In this section we analyze the length, alphabet size, decoding degree, decision complexity and the completeness of the construction. 
\vspace{-2ex}
\paragraph{Length.} The length is $|X(k)| + |X(\sqrt{k})| = O_{|\Sigma_0|}(N \cdot \poly_{\eps}(\log N))$ by the properties of $X$ from \cref{thm: hdx construction}.
\vspace{-2ex}
\paragraph{Alphabet Size.} The left and right alphabet sizes are at most $|\Sigma|^k$ and $|\Sigma|^{\sqrt{k}}$ respectively. Since $k = O_{\eps}(1)$ and $|\Sigma| = 2^{O_{|\Sigma_0|}(\poly_{\eps}(\log N))}$, both the left and right alphabet sizes are $2^{O_{|\Sigma_0|}(\poly_{\eps}(\log N))}$.
\vspace{-2ex}
\paragraph{Degrees.} For $\OutLeftVertex \in X(k)$, its degree is $\binom{k}{\sqrt{k}} = O_{\eps}(1)$ and its decoding degree is at most the sum of the decoding degrees, in $\mc{D}_{\hdx}$, of the vertices $u\in \OutLeftVertex$. Hence, the decoding degree is at most $k\poly_{\eps}(\log N)\leq \poly_{\eps}(\log N)$. For $\OutRightVertex \in X(\sqrt{k})$, its degree is at most $d(\OutRightVertex) \cdot d^k$, where $d(\OutRightVertex)$ is the number of $d$-faces of the complex $X$ containing $\OutRightVertex$. Since $d(\OutRightVertex) \leq \poly_{\eps}(\log N)$ by \cref{thm: hdx construction}, the right degree is $\poly_{\eps}(\log N)$.
\vspace{-2ex}
\paragraph{Projection Decision Complexity.} Given a left alphabet symbol, it is clear that the unique right alphabet symbol which could satisfy the constraint is simply a restriction, so it can be computed with a circuit of size $O_{|\Sigma_0|}(\poly_{\eps}( \log N))$. Additionally, the projection circuit needs to checking that the left alphabet symbol $\sigma \in \Sigma^k$ satisfies (1) the hardcoding  of $\mc{D}_{\hdx}$ constraints, and (2) checking that the left alphabet symbol $ \sigma \in \Sigma^{k}$ satisfies hardcoding for decoding consistency.

The circuit complexity of (1) is $O_{|\Sigma_0|}(\poly_{\eps}( \log N))$ as there are $O_{\eps}(1)$ edges $e \in X(2)$ contained in $\OutLeftVertex$ and each $\mc{D}_{\hdx}$ constraint has circuit complexity $O_{|\Sigma_0|}(\poly_{\eps}(\log N))$. The circuit complexity of (2) is also $O_{|\Sigma_0|}(\poly_{\eps}( \log N))$. This is because, by the decoding degree of $\mc{D}$, there are $\poly_{\eps}(\log N)$ pairs of $u \in \OutLeftVertex$ and $t \in [n]$ such that $t$ is in the decoding neighborhood of $u$ in $\mc{D}_{\hdx}$ and the decoders $\mc{D}_{\hdx, t}$ all have circuit complexity $O_{|\Sigma_0|}(\poly_{\eps}( \log N))$ by the decoding complexity of \cref{thm: udpcp graph to hdx}. Finally, computing the restriction and checking that is is equal to a given string amounts to equality between two strings of $O_{|\Sigma_0|}(\poly_{\eps}( \log N))$ bits, which can be done by a circuit of size $O_{|\Sigma_0|}(\poly_{\eps}( \log N))$. Altogether, this shows that the projection decision complexity is $O_{|\Sigma_0|}(\poly_{\eps}( \log N))$. 

\vspace{-2ex}
\paragraph{Decoding Distribution.}  To see that $\mc{Q}_{\dip}$ is agnostic, note that for any $(t,\OutRightVertex)\in\supp(\mc{Q}_{\dip}(\cdot,\circ,\cdot))$, both of the distributions 
$\mc{Q}_{\dip}(\circ,\cdot,\OutRightVertex)$ and 
$\mc{Q}_{\dip}(t,\cdot,\OutRightVertex)$ are of $\OutLeftVertex\sim X(k)$ 
conditioned on $\OutLeftVertex\supseteq \OutRightVertex$, so they are the same.

For the second property, we describe the following way of choosing $(\OutLeftVertex, \OutRightVertex) \sim \mc{Q}(\circ, \cdot, \cdot)$:
\begin{itemize}
        \item Choose $U \in X(d)$ uniformly at random,
        \item Choose $\OutLeftVertex \subseteq U$ of size $k$ uniformly and $\OutRightVertex \subseteq \OutLeftVertex$ of size $\sqrt{k}$ uniformly,
        \item Output $(\OutLeftVertex, \OutRightVertex)$.
    \end{itemize}
It is clear by the above that for every  
$\OutRightVertex \in X(\sqrt{k})$, we have 
\begin{equation} \label{eq: hdx disc} 
\mc{Q}(\circ,\circ, \OutRightVertex)  
=
\sum\limits_{\OutLeftVertex\supseteq\OutRightVertex}\mc{Q}(\circ,\OutLeftVertex, \OutRightVertex)
=
\sum_{\substack{\OutLeftVertex\supseteq\OutRightVertex\\ U\in X(d), U \supseteq \OutLeftVertex}} \frac{1}{|X(d)|} \cdot \frac{1}{\binom{d}{{k}}\binom{k}{\sqrt{k}}} \leq \frac{\poly_{\eps}(\log N)}{N}.
\end{equation}
In the last transition, we use the fact that for a fixed $\OutRightVertex$, 
the number of $d$-faces $U$ and $k$-faces $\OutLeftVertex$ that contain $\OutRightVertex$ is at most $\poly_{\eps}(\log N)^d\leq \poly_{\eps}(\log N)$ (as the degree of $X$ is $\poly_{\eps}(\log N)$), and also that $|X(d)|\geq N$.

\vspace{-2ex}
\paragraph{Decoding Complexity.} It is clear that given a left vertex $\OutLeftVertex$ and a symbol for it, $\sigma$, the decoder runs the decoder of $\mc{D}_{\hdx}$ on a restriction of $\sigma$, and hence the decoding complexity is $O_{|\Sigma_0|}(\poly_{\eps}( \log N))$ by~\cref{thm: udpcp graph to hdx}.

\vspace{-2ex}
\paragraph{Perfect Completeness.} Fix $w \in \mc{L}$. By the perfect completeness of $\mc{D}_{\hdx}$, there is an assignment $A': X(1) \to \Sigma$ which satisfies all of the constraints in $\mc{D}_{\hdx}$ and decodes to $w_t$ with probability $1$ for all $t \in [n]$. Now consider the assignments $T_1: X(k) \to \Sigma^k$ and $T_2: X(\sqrt{k}) \to \Sigma^{\sqrt{k}}$ such that $T_1[\OutLeftVertex] = A'|_\OutLeftVertex$ for each $\OutLeftVertex \in X(k)$ and $T_2[\OutRightVertex] = A'|_\OutRightVertex$ for each $\OutRightVertex \in X(\sqrt{k})$. It is clear that the alphabet hardcoding is satisfied by $T_1[\OutLeftVertex]$ for every $\OutLeftVertex \in X(k)$. Indeed, for any $e \in X(2)$ contained in $\OutLeftVertex$, the restriction of $T_1[\OutLeftVertex]$ to the vertices in $e$ is consistent with $A'$ and hence the constraint $e$ is satisfied by the corresponding values of $T_1[\OutLeftVertex]$. As for the hardcoding for decoding consistency, for any $u \in \OutLeftVertex$ and $t \in [n]$ which is a decoding neighbor of $u$ in $\mc{D}_{\hdx}$, we have that $D'_t(u, T_1[\OutLeftVertex](u)) = w_t$. This latter fact also shows that $T_1, T_2$ decode to $w_t$ with probability $1$ for all $t \in [n]$. Finally, it is straightforward to see that  all of the constraints in $\mc{D}_{\dip}$ are satisfied.

\subsection{Proof of \cref{thm: dp dpcp}: Tools for List-Decoding Soundness}

The proof of list-decoding soundness involves two steps. First, we show a list-decoding variant of the soundness guarantee of the direct product test, informally  saying that there is always a small list of functions $f_i: X(1) \to \Sigma$ which explains almost all of the passing probability of the direct product test in $\mc{D}_{\dip}$. Additionally, each $f_i$, when viewed as an assignment to $\mc{D}_{\hdx}$, is nearly a satisfying assignment of $\mc{D}_{\hdx}$. 
Second, we translate the list-decoding guarantee above into a list-decoding guarantee relative to valid members of the language $\mc{L}$. To accomplish this, we combine the unique decoding soundness of $\mc{D}_{\hdx}$ with fact that the $f_i$'s are nearly satisfying assignments of $\mc{D}_{\hdx}$.

The first step outlined above is achieved via the following lemma.
\begin{lemma} \label{lm: list decoding dp 1}
Let $T_1: X(k) \to \Sigma^k$ be a left assignment to $\mc{D}_{\dip}$ and let $\delta$ be defined as in the construction of $\mc{D}_{\dip}$. Then, there exists an absolute constant $C$ such that, there is a list of at most $L \leq \delta^{-C}$ functions $f_1, \ldots, f_L: X(1) \to \Sigma$, along with disjoint sets $\mc{U}_1, \ldots, \mc{U}_L \subseteq X(k)$ which together satisfy the following:
\begin{itemize}
    \item for each $i\in[L]$, we have $\mu_k(\mc{U}_i) \geq \delta^{C}$ and  $\dist(f_i|_\OutLeftVertex, T_1[\OutLeftVertex]) \leq \delta$ for all $\OutLeftVertex \in \mc{U}_i$; 
    \item for any right assignment $T_2: X(\sqrt{k}) \to \Sigma^{\sqrt{k}}$, we have
    \[
    \Pr_{\OutLeftVertex \sim X(k),  \OutRightVertex \subseteq \OutLeftVertex}[T_1[\OutLeftVertex]|_\OutRightVertex = T_2[\OutRightVertex] \land \OutLeftVertex \notin \cup_{i=1}^{L} \mc{U}_i] \leq \frac{\sqrt{\delta}}{10}.
    \]
\end{itemize}
\end{lemma}

Towards the proof of \cref{lm: list decoding dp 1}, we first show an auxiliary lemma which says that there exists an assignment $X(k) \to \Sigma^k$ which has only trivial agreement with all possible functions $F: X(1) \to \Sigma$.

\begin{lemma} \label{lm: bad table}
    There exists a table $T^*: X(k) \to \Sigma^k$ such that for any function $F: X(1) \to \Sigma$,
    \[
    \Pr_{\OutLeftVertex \sim X(k)}\left[\dist(F|_\OutLeftVertex, T^*[\OutLeftVertex]) \leq \frac{1}{3}\right] \leq \frac{10}{\sqrt{d}},
    \]
    where recall $d$ is the dimension of the complex $X$.
\end{lemma}
\begin{proof}
We use the fact that the direct-product dimension parameter, $k$, in constructing $\mc{D}_{\dip}$ is large and in particular $k \geq 8$. First let us show that there exists a table $T^*$ satisfying such that every pair for every pair $\OutLeftVertex, \OutLeftVertex' \in X(k)$ with nonempty intersection, we have that $T^*[\OutLeftVertex]$ and $T^*[\OutLeftVertex']$ disagree at every $x \in \OutLeftVertex \cap \OutLeftVertex'$. Indeed, $|\Sigma|$ is super-polynomial in $n$ whereas $k \cdot |X(k)| \leq n^{O_{\eps}(1)}$, and hence we can arrange that the $T^*[\OutLeftVertex]: \OutLeftVertex \to \Sigma$ have disjoint image over all $\OutLeftVertex \in X(k)$.

    Suppose for the sake of contradiction that this is not the case, and that there exists $F: X(1) \to \Sigma$, such that
    \begin{equation} \label{eq: f assumption}
    \Pr_{\OutLeftVertex \sim X(k)}[\dist(F|_\OutLeftVertex, T^*[\OutLeftVertex]) \leq \frac{1}{3}] > \frac{10}{\sqrt{d}}
    \end{equation}

    Now consider the probability that, when choosing $U \sim X(d)$,  $\OutRightVertex \subseteq U$ of size $k-1$, and finally $\OutLeftVertex, \OutLeftVertex'$ of size $k$ such that $\OutRightVertex \subseteq \OutLeftVertex, \OutLeftVertex' \subseteq U$, we have both $\dist(F|_\OutLeftVertex, T^*[\OutLeftVertex]) \leq \frac{1}{3}$ and $\dist(F|_{\OutLeftVertex'}, T^*[\OutLeftVertex']) \leq \frac{1}{3}$. For simplicity, let us refer to these two events as ${\sf E}(\OutLeftVertex)$ and ${\sf E}(\OutLeftVertex')$ respectively. Then,
    \[
    \E_{U \sim X(d), \OutRightVertex \subseteq U}\left[\Pr_{\OutRightVertex \subseteq  \OutLeftVertex, \OutLeftVertex' \subseteq U}[{\sf E}(\OutLeftVertex) \land {\sf E}(\OutLeftVertex')]  \right] = 
    \E_{U \sim X(d), \OutRightVertex \subseteq U}\left[\Pr_\OutLeftVertex[{\sf E}(\OutLeftVertex)~|~U,\OutRightVertex]^2\right] \geq \E_{U \sim X(d), \OutRightVertex \subseteq U}\left[\Pr_\OutLeftVertex[{\sf E}(\OutLeftVertex)~|~U,\OutRightVertex]\right]^2  > \frac{100}{d}.
    \]
    In the first transition we used the fact that $\OutLeftVertex, \OutLeftVertex'$ are independent, in the second transition we used the Cauchy-Schwarz inequality, and in the third transition we used~\eqref{eq: f assumption}. Now using the fact that the probability of choosing $\OutLeftVertex = \OutLeftVertex'$ is  at most $\frac{1}{d-k+1} \leq \frac{100}{d}$, we get there exist $\OutLeftVertex$ and $\OutLeftVertex'$ distinct such that both ${\sf E}(\OutLeftVertex)$ and ${\sf E}(\OutLeftVertex')$ occur. Then 
    $T^*[\OutLeftVertex]$ and 
    $T^*[\OutLeftVertex']$ disagree on at most 
    $\frac{2}{3}k+1 < k$ coordinates, which is a contradiction.
\end{proof}
\begin{proof}[Proof of \cref{lm: list decoding dp 1}]
Initialize $\mc{U}_1 = \mc{U}_1^{\star}= \emptyset$ and $i := 1$. Also set $T^{(1)} := T_1$. While there exists $F: X(1) \to \Sigma$ such that 
    \[
    \Pr_{\OutLeftVertex \sim X(k)}[\dist(F|_\OutLeftVertex, T^{(i)}[\OutLeftVertex]) \leq \delta \land \OutLeftVertex \notin \mc{U}^\star_i] \geq \delta^{C}
    \]
pick such $F$, set $f_i = F$, set $\mc{U}_i = \{\OutLeftVertex \in X(k) \; | \; \dist(F|_\OutLeftVertex, T^{(i)}[\OutLeftVertex]) \leq \delta \land \OutLeftVertex \notin \mc{U}^\star_i \}$, and update $\mc{U}^\star_{i+1} = \mc{U}^\star_i \cup \mc{U}_i$. Then increment $i$ by $1$, and repeat. In the above we fix $C$ to be some absolute constant which we can set arbitrarily large.

Note that this process must terminate in $J \leq \delta^{-C}$ steps.  Indeed, after each increment of $i$, the measure $\mu_k\left(\mc{U}^\star_i\right)$ increases by at least $\delta^{C}$. At the end of the process, let $f_1, \ldots, f_J$ and $\mc{U}_1, \ldots, \mc{U}_J$ be the functions and sets obtained respectively, and let $\mc{U}^\star:= \bigcup_{i \in [J]} \mc{U}_i$. The following conditions hold:
\begin{itemize}
    \item For each $i \in [J]$ and $\OutLeftVertex \in \mc{U}_i$ we have $\dist(f_i|_{\OutLeftVertex}, T_1[\OutLeftVertex]) \leq \delta$.
    \item $\mu_k(\mc{U}_i) \geq \delta^C$,
    \item The sets $\mc{U}_i$ are disjoint.
    \item For any $\OutLeftVertex \notin \mc{U}^\star$, and any $i \in [J]$, we have $\dist(f_i|_\OutLeftVertex, T[\OutLeftVertex]) > \delta$. 
    \item For any $F: X(1) \to \Sigma$, we have 
    \begin{equation} \label{eq: disagree outside L}
        \Pr_{\OutLeftVertex \sim X(k)}[\dist(T[\OutLeftVertex], F|_\OutLeftVertex) \leq \delta \land \OutLeftVertex \notin \mc{U}^\star] \leq \delta^{C}.
    \end{equation}
\end{itemize}
We will now show that $f_1, \ldots, f_{J}$ and $\mc{U}_1, \ldots, \mc{U}_J$ satisfy \cref{lm: list decoding dp 1}. First, note that the number of functions is indeed $J \leq \delta^{-C}$ for some absolute constant $C$ and the sets $\mc{U}_i$ are indeed disjoint as required by \cref{lm: list decoding dp 1}. It is also clear that the first item of \cref{lm: list decoding dp 1} holds by construction of the $\mc{U}_i$'s.

It remains to verify the second item. Let $T^*$ be the table from \cref{lm: bad table} --- which, recall, does not have significant agreement with any function over $X(1)$ --- and define the table $T'$ as follows. For $\OutLeftVertex \in \mc{U}^\star$, set $T'[\OutLeftVertex] = T^*[\OutLeftVertex]$ and for $\OutLeftVertex \notin \mc{U}^\star$, set $T'[\OutLeftVertex] = T[\OutLeftVertex]$. Suppose for the sake of contradiction that the second item fails, so that
\begin{equation} \label{eq: dp list decode temp}  
\Pr_{\OutLeftVertex \sim X(k), \OutRightVertex \subseteq \OutLeftVertex}[T'[\OutLeftVertex]|_\OutRightVertex = T_2[\OutRightVertex]] > \frac{\sqrt{\delta}}{10}.
\end{equation}

Choosing $U \in X(d)$, $\OutRightVertex \subseteq U$ of size $\sqrt{k}$, and $\OutLeftVertex, \OutLeftVertex'$ of size $k$ such that $\OutRightVertex \subseteq \OutLeftVertex, \OutLeftVertex' \subseteq U$ according to the direct product testing distribution, we have
\begin{equation}
    \begin{split}
        \Pr_{U, \OutLeftVertex, \OutLeftVertex', \OutRightVertex}[T'[\OutLeftVertex]|_\OutRightVertex = T'[\OutLeftVertex']|_\OutRightVertex] &\geq \E_{U} \left[\Pr_{\OutLeftVertex,\OutLeftVertex', \OutRightVertex}[T'[\OutLeftVertex]|_\OutRightVertex = T_2[\OutRightVertex] \land T'[\OutLeftVertex']|_\OutRightVertex = T_2[\OutRightVertex]]  \right] \\
        &= \E_{U} \left[\Pr_{\OutLeftVertex, \OutRightVertex}[T'[\OutLeftVertex]|_\OutRightVertex= T_2[\OutRightVertex]]^2 \right]  \\
        &\geq \left(\E_{U} \left[\Pr_{\OutLeftVertex, \OutRightVertex}[T'[\OutLeftVertex]|_\OutRightVertex = T_2[\OutRightVertex]] \right]  \right)^2 \\
        &\geq \frac{\delta}{100}.
    \end{split}
\end{equation}
In the third transition we used by Cauchy-Schwarz, and in the last transition we used~\eqref{eq: dp list decode temp}.
Since the complex $X$ supports a direct product test with soundness $\delta/100$, we get that there exists $F: X(1) \to \Sigma$ such that 
\[
\Pr_{\OutLeftVertex \sim X(k)}\left[\dist(T'[\OutLeftVertex], F|_\OutLeftVertex) \leq \frac{\delta}{100}\right] \geq \delta^{C_1},
\]
for some absolute constant $C_1$ coming from the soundness of the direct product test as described in \cref{thm: hdx construction}. However, this is a contradiction because
\begin{align*}  
\Pr_{\OutLeftVertex \sim X(k)}\left[\dist(T'[\OutLeftVertex], F|_\OutLeftVertex) \leq \frac{\delta}{100}\right] &= \Pr_{\OutLeftVertex \sim X(k)}\left[\dist(T'[\OutLeftVertex], F|_\OutLeftVertex) \leq \frac{\delta}{100} \land \OutLeftVertex \in \mc{U}^\star\right] \\
&+\Pr_{\OutLeftVertex \sim X(k)}\left[\dist(T'[\OutLeftVertex], F|_\OutLeftVertex) \leq \frac{\delta}{100} \land \OutLeftVertex \notin \mc{U}^\star\right] \\
&\leq \Pr_{\OutLeftVertex \sim X(k)}\left[\dist(T^*[\OutLeftVertex], F|_\OutLeftVertex) \leq \frac{\delta}{100}\right] + \delta^{C}  \\
&\leq \frac{10}{\sqrt{d}} + \delta^{C} \\
& < \delta^{C_1}.
\end{align*}
where in the second transition we used \eqref{eq: disagree outside L} along with the construction of $T'$ and in the third transition we used $\delta/100 \leq 1/3$ and \cref{lm: bad table}. To obtain the contradiction in the last transition, we use the fact that the dimension parameter of the HDX $X$ is sufficiently large relative to $1/\delta$ and the fact that the absolute constant $C$ fixed at the start of the proof can be chosen sufficiently large relative to the constant $C_1$ coming from the direct product test soundness.
\end{proof}

\subsection{Proof of \cref{thm: dp dpcp}: Proof of List-Decoding Soundness} 
We now accomplish the second step in the proof of the list-decoding soundness.
 Fix a left assignment, $T_1: X(k) \to \Sigma^k$, to $\mc{D}_{\dip}$. Let $f_1, \ldots, f_L$ and $\mc{U}_1, \ldots, \mc{U}_L$ be the functions and subsets of $X(k)$ obtained by applying \cref{lm: list decoding dp 1} with respect to $T_1$, and let $T_2: X(\sqrt{k}) \to \Sigma^{\sqrt{k}}$ be any right assignment to $\mc{D}_{\dip}$. We start by observing that each $f_i: X(1) \to \Sigma$ is a nearly-satisfying assignment to $\mc{D}_{\hdx}$.

\begin{lemma} \label{lm: f_i pass}
    For each $i \in [L]$, we have 
    \[
    \Pr_{(u,v) \sim X(2)}[\Phi^{{\sf hdx}}_{(u,v)}(f_i(u), f_i(v)) = 1] \geq 1 - 3\delta.
    \]
\end{lemma}
\begin{proof} 
    Fix $i$ and let $\mc{B} \subseteq X(2)$ be the set of constraints not satisfied by $f_i$, so that $\mu_2(\mc{B})$ is precisely the probability that a randomly chosen constraint, according to the measure of $X$ over $X(2)$, is not satisfied by $f_i$. To prove the lemma, we show $\mu_2(\mc{B}) \leq 3\delta$. 
    
    To start, notice that if there exists $\OutLeftVertex \in X(k)$ such that $T_1[\OutLeftVertex](u) = f_i(u)$ and $T_1[\OutLeftVertex](v) = f_i(v)$, then we automatically have $\Phi^{{\sf hdx}}_{(u,v)}(f_i(u), f_i(v)) = 1$ due to how the alphabet of $\OutLeftVertex$ is hardcoded. This shows that a constraint is highly likely to be satisfied if it is chosen by first choosing $\OutLeftVertex \sim \mc{U}_i$, and then choosing $u,v \in \OutLeftVertex$. In the sequence below, we write $\OutLeftVertex \sim \mc{U}_i$ to denote $\OutLeftVertex$ chosen according to $\sim X(k)$, conditioned on being in $\mc{U}_i$. After $\OutLeftVertex$ is chosen, $u$ and $v$ are chosen as uniformly random vertices inside $\OutLeftVertex$.
    \begin{equation} \label{eq: list dec soundness on DP 1}
    \begin{split}
    \Pr_{\OutLeftVertex \sim \mc{U}_i, u,v \in \OutLeftVertex, u\neq v}[(u,v) \notin \mc{B}] &\geq  \E_{\OutLeftVertex \sim \mc{U}_i}\left[\Pr_{ u,v \in \OutLeftVertex,u\neq v}[T_1[\OutLeftVertex](u) = f_i(u) \land T_1[\OutLeftVertex](v) =f_i(v)]\right] \\
    &\geq \E_{\OutLeftVertex \sim \mc{U}_i}\left[\Pr_{ u \in \OutLeftVertex}[T_1[\OutLeftVertex](u) = f_i(u)]^2 \right] - \frac{1}{k} \\
    &\geq (1-\delta)^2 - \frac{1}{k}\\
    &\geq 1-2\delta-\frac{1}{k}.
    \end{split}
    \end{equation}
In the second transition we are going from a probability over distinct  $u,v\in \OutLeftVertex$ to a probability over $u, v \in \OutLeftVertex$ independent, and the $-1/k$ term accounts for the probability that $u = v$. In the third transition we use the Cauchy-Schwarz inequality and the fact that $f_i$ and $T_1[\OutLeftVertex]$ agree on at least $(1-\delta)$-fraction of vertices in $\OutLeftVertex$ for every $\OutLeftVertex \in \mc{U}_i$.

Next, by item 3 of \cref{thm: hdx construction} and the fact that $d$ is sufficiently large relative to $k$, we get that the second singular value of the bipartite inclusion graph $(X(k), X(2))$ is $O(1/k)$. Applying~\cref{lm: expander mixing} with the left and right subsets are $\mc{U}_i$ and $\mc{B}$ respectively gives:
\[
\left|\Pr_{\OutLeftVertex \sim X(k),  \{u,v\} \in \binom{\OutLeftVertex}{2}}[\OutLeftVertex \in \mc{U}_i, \{u,v\} \in \mc{B}] - \mu_k(\mc{U}_i) \mu_2(\mc{B}) \right| \leq O\left(\frac{1}{k}\right) \cdot \sqrt{ \mu_k(\mc{U}_i) \mu_2(\mc{B}) }.
\] 
Dividing both sides by $\mu_k(\mc{U}_i)$ and rearranging gives 
\begin{align*}  
\mu_2(\mc{B})  &\leq \Pr_{\OutLeftVertex \sim \mc{U}_i,  \{u,v\} \in \binom{\OutLeftVertex}{2}}[(u,v) \in \mc{B}] +  O\left(\frac{1}{k}\right)  \cdot \frac{\sqrt{ \mu_k(\mc{U}_i) \mu_2(\mc{B}) }}{\mu_k(\mc{U}_i)} \\
&\leq 2\delta + \frac{1}{k} + O\left(\frac{1}{k \sqrt{  \mu_k(\mc{U}_i)}}\right) \\
&\leq 3\delta,
\end{align*}
where in the transition we used~\eqref{eq: list dec soundness on DP 1}, and in the last transition we used $\mu_k(\mc{U}_i) \geq \delta^C$ and the fact that $k$ is sufficiently large relative to $1/\delta$.
\end{proof}

Now combining \cref{lm: f_i pass} with the unique decoding soundness assumption of $\mc{D}_{\hdx}$ from \cref{thm: udpcp graph to hdx}, we get that for each $i \in [L]$, there is $w_i \in \mc{L}$ such that 
\begin{equation} \label{eq: f_i decode}  
\Pr_{(t,u) \sim \mc{Q}_{\hdx}}[D^{{\sf hdx}}_t(u, f_i(u)) \neq \left(w_i\right)_t] \leq 5\sqrt{\delta},
\end{equation}
and we recall that the marginal distribution of $u$ in the probability above is the stationary over $X(1)$.
We are now ready to prove list-decoding soundness of $\mc{D}_{\dip}$ by showing that the $w_i$ obtained by \eqref{eq: f_i decode} give the desired list. 

\begin{proof}[Proof of List-Decoding Soundness for \cref{thm: dp dpcp}]
   Take $L = \delta^{-C}$ and $f_i, \mc{U}_i, w_i$ as above for each $i \in [L]$. Set $\mc{U}^\star:= \bigcup_{i \in [L]} \mc{U}_i$. Then
    \begin{align*}
    &\Pr_{(t,\OutLeftVertex, \OutRightVertex) \sim \mc{Q}_{\dip}}\left[T_1[\OutLeftVertex]|_\OutRightVertex = T_2[\OutRightVertex] \land D_t(\OutLeftVertex, T_1[\OutLeftVertex]) \notin \{\left(w_i\right)_t \}_{i \in [L]}\right] \\
    &= \Pr_{(t,\OutLeftVertex, \OutRightVertex) \sim \mc{Q}_{\dip}}\left[T_1[\OutLeftVertex]|_\OutRightVertex = T_2[\OutRightVertex] \land D_t(\OutLeftVertex, T_1[\OutLeftVertex]) \notin \{\left(w_i\right)_t \}_{i \in [L]} \land \OutLeftVertex \in \mc{U}^\star\right] \\
    &+\Pr_{(t,\OutLeftVertex, \OutRightVertex) \sim \mc{Q}_{\dip}}\left[T_1[\OutLeftVertex]|_\OutRightVertex = T_2[\OutRightVertex] \land D_t(\OutLeftVertex, T_1[\OutLeftVertex]) \notin \{\left(w_i\right)_t \}_{i \in [L]} \land \OutLeftVertex \notin \mc{U}^\star\right] \\
    &\leq \Pr_{(t,\OutLeftVertex, \OutRightVertex) \sim \mc{Q}_{\dip}}\left[D_t(\OutLeftVertex, T_1[\OutLeftVertex]) \notin \{\left(w_i\right)_t \}_{i \in [L]} \land \OutLeftVertex \in \mc{U}^\star\right] \\
    &+ \Pr_{(t,\OutLeftVertex, \OutRightVertex) \sim \mc{Q}_{\dip}}\left[T_1[\OutLeftVertex]|_\OutRightVertex = T_2[\OutRightVertex]  \land \OutLeftVertex \notin \mc{U}^\star\right]\\
    &\leq \Pr_{(t,\OutLeftVertex, \OutRightVertex) \sim \mc{Q}_{\dip}}\left[D_t(\OutLeftVertex, T_1[\OutLeftVertex]) \notin \{\left(w_i\right)_t \}_{i \in [L]} \land \OutLeftVertex \in \mc{U}^\star\right] + O\left(\sqrt{\delta}\right),
    \end{align*}
where in the last transition, we are using the second item of \cref{lm: list decoding dp 1} to bound
\[
\Pr_{(t,\OutLeftVertex, \OutRightVertex) \sim \mc{Q}_{\dip}}\left[T_1[\OutLeftVertex]|_\OutRightVertex = T_2[\OutRightVertex]  \land \OutLeftVertex \notin \mc{U}^\star\right] = \Pr_{\OutLeftVertex \sim X(k), \OutRightVertex \subseteq \OutLeftVertex}\left[T_1[\OutLeftVertex]|_\OutRightVertex = T_2[\OutRightVertex]  \land \OutLeftVertex \notin \mc{U}^\star\right] \leq \frac{\sqrt{\delta}}{10}.
\] 
Next, we perform a union bound
    \begin{align*}
        \Pr_{(t,\OutLeftVertex, \OutRightVertex) \sim \mc{Q}_{\dip}}\left[D_t(\OutLeftVertex, T_1[\OutLeftVertex]) \notin \{\left(w_i\right)_t \}_{i \in [L]} \land \OutLeftVertex \in \mc{U}^\star\right] &\leq \sum_{i =1}^{L}\Pr_{(t,\OutLeftVertex, \OutRightVertex) \sim \mc{Q}_{\dip}}\left[D_t(\OutLeftVertex, T_1[\OutLeftVertex]) \neq \left(w_i\right)_t \land \OutLeftVertex \in \mc{U}_i \right] \\
        &\leq \E_{i \sim [L]}\left[\Pr_{(t,\OutLeftVertex, \OutRightVertex) \sim \mc{Q}_{\dip}}[D_t(\OutLeftVertex, T_1[\OutLeftVertex]) \neq \left(w_i\right)_t \; | \; \OutLeftVertex \in \mc{U}_i]\right].
    \end{align*}
In the last line we let $i \sim [L]$ be the distribution where $i$ is chosen with weight proportional to $\mu_k(\mc{U}_i)$. Then, the last line follows from the fact that the $\mc{U}_i$'s are disjoint, so the events $\OutLeftVertex \in \mc{U}_i$ are mutually exclusive. Fix $i$ and examine a probability from the last line more closely. Here, we unpack the joint distribution of $(t, \OutLeftVertex)$ under $\mc{Q}_{\dip}$. By definition, this distribution is obtained by sampling $\OutLeftVertex \sim X(k)$, then a random $u \in \OutLeftVertex$ and finally $t \sim \mc{Q}_{\hdx}(\cdot, u)$. We have
    \begin{equation}\label{eq:list_amp_ana}
        \begin{split}
        \Pr_{(t,\OutLeftVertex, \OutRightVertex) \sim \mc{Q}_{\dip}}[D_t(\OutLeftVertex, T_1[\OutLeftVertex]) \neq \left(w_i\right)_t \; | \; \OutLeftVertex \in \mc{U}_i] 
        &\leq \Pr_{\OutLeftVertex \sim \mc{U}_i, u \in \OutLeftVertex, t \sim \mc{Q}_{\hdx}(\cdot, u)}[D^{{\sf hdx}}_t(u, f_i(u)) \neq \left(w_i\right)_t] \\
        &+  \Pr_{\OutLeftVertex \sim \mc{U}_i,  u \in \OutLeftVertex, t \sim \mc{Q}_{\hdx}(\cdot, u)}[T_1[\OutLeftVertex](u) \neq f_i(u) ].
        \end{split}
    \end{equation}
Now we bound each term individually. Towards the first term, define the function $H_i: X(1) \to [0,1]$ as
\[
H_i(u) = \Pr_{t \sim \mc{Q}_{\hdx}(\cdot, u)}[D^{{\sf hdx}}_t(u, f_i(u)) \neq \left(w_i\right)_t].
\]
In particular, from \eqref{eq: f_i decode} we have
\begin{equation} \label{eq: bound exp in dp sound}  
\mu_1(H_i) := \E_{u \sim X(1)}[ H_i(u)] =   \Pr_{(t,u) \sim \mc{Q}_{\hdx}}[D^{{\sf hdx}}_t(u, f_i(u)) \neq \left(w_i\right)_t] \leq 5\sqrt{\delta}.
\end{equation}

By \cref{lm: expander mixing} and the fact that the bipartite inclusion graph on $(X(k), X(1))$ has second singular value at most $O(1/k)$ by item 3 of \cref{thm: hdx construction} (and the fact that $d$ is sufficiently large relative to $k$), it follows that
\begin{equation} \label{eq: dp soundness end eml}
  \left| \E_{\OutLeftVertex \sim X(k), u \in \OutLeftVertex}[\ind_{\OutLeftVertex \in \mc{U}_i} \cdot H_i(u)] - \mu_k(\mc{U}_i) \mu_1(H_i) \right| \leq O\left(\frac{1}{k}\right) \cdot \sqrt{\mu_k(\mc{U}_i)}.
\end{equation}
Now note that the first term on the right hand side of \eqref{eq:list_amp_ana} is 
\[
\Pr_{\OutLeftVertex \sim \mc{U}_i, u \in \OutLeftVertex, t \sim \mc{Q}_{\hdx}(\cdot, u)}[D^{{\sf hdx}}_t(u, f_i(u)) \neq \left(w_i\right)_t] = \E_{\OutLeftVertex \sim X(k),  u \in \OutLeftVertex}[H_i(u) \; | \; \OutLeftVertex \in \mc{U}_i]=  \frac{\E_{\OutLeftVertex \sim X(k), u \in \OutLeftVertex}[\ind_{\OutLeftVertex \in \mc{U}_i} \cdot H_i(u)]}{\mu_k(\mc{U}_i)}.
\]
Thus, by \eqref{eq: dp soundness end eml},
\[
\Pr_{\OutLeftVertex \sim \mc{U}_i,  u \in \OutLeftVertex, t \sim \mc{Q}_{\hdx}(\cdot, u)}[D^{{\sf hdx}}_t(u, f_i(u)) \neq \left(w_i\right)_t] \leq 5\sqrt{\delta} + O\left(\frac{1}{k \cdot \sqrt{\mu_k(\mc{U}_i)}}\right) \leq 6 \sqrt{\delta}.
\]
The first inequality follows by \eqref{eq: bound exp in dp sound} and the second inequality follows from the fact that $\mu_k(\mc{U}_i) \geq \delta^C$ and $k$ is sufficiently large relative to $1/\delta$. 

For the second term on the right hand side of~\eqref{eq:list_amp_ana}, we have that 
\begin{align*}
  \Pr_{\OutLeftVertex \sim \mc{U}_i, u \in \OutLeftVertex}[T_1[\OutLeftVertex](u) \neq f_i(u) ] 
  =\E_{\OutLeftVertex \sim \mc{U}_i}[\dist(f_i|_{\OutLeftVertex}, T_1[\OutLeftVertex])]
  \leq \delta
\end{align*}
by definition of $f_i$ and $\mc{U}_i$. Altogether, we get that~\eqref{eq:list_amp_ana} is at most $6\sqrt{\delta} + \delta \leq 7 \sqrt{\delta}$, and plugging this up gives that  the list-decoding error is at most $O(\sqrt{\delta})$.
\end{proof}

\section{Transformations on dPCPs}\label{sec:dPCP_transfor}
\cref{thm: dp dpcp} gives us constructions of quasi-linear length dPCPs with arbitrarily small error in the list-decoding soundness, as required by \cref{thm:dPCP main}, but with alphabet size that is super-polynomial, rather than constant. Our goal now is to reduce the alphabet size in~\cref{thm: dp dpcp} down to a constant while maintaining all of the other desired properties; namely, length and small soundness. As discussed in the introduction, we achieve this via a series of transformations, and in this section we present a few of the basic ones here. The transformations in this section are both variants of well known transformations for PCPs.

\subsection{Right Alphabet Reduction}
We start with right alphabet reduction. This procedure is a simple application of error correcting codes along the lines of~\cite[Section 5.2]{dh} except that for our purposes, we also need to keep track of the degree related parameters and analyze the list-decoding soundness (rather than the standard soundness). 

\begin{lemma} \label{lm: right alphabet reduction}
    Suppose a language $\mc{L}$ has a projection dPCP 
    \[
    \mc{D} = \left(A \cup B, E, \Sigma_A, \Sigma_B , \{\mc{P}_t \}_{t \in [n]},\{D_t \}_{t \in [n]} \right)
    \]
    with perfect completeness, $(L, \eps)$-list-decoding soundness, and projection decision complexity ${\sf ProjComp}$. Then for every $0 < \eta < 1$, there is a polynomial-time algorithm which takes $\mc{D}$ and outputs a new projection dPCP for $\mc{L}$,
    \[
    \mc{D}'= \left(A \cup (B \times [m]), E', \Sigma_A, \Omega, \{\mc{P}'_t \}_{t \in [n]},\{D_t \}_{t \in [n]} \right)
    \]
    where $m = O_{\eta}(\log|\Sigma_B|)$  and $\Omega$ is an alphabet of size $O_{\eta}(1)$. Additionally, $\mc{D}'$ satisfies the following:
    \begin{itemize}
        \item \textbf{Degrees.} The decoding degree of each vertex and right degree of each vertex are the same as in $\mc{D}$. The left degree of each vertex is $m$ times its left degree in $\mc{D}$.
        \item \textbf{Projection Decision Complexity.} The decision complexity is ${\sf ProjComp} + O(\log |\Sigma_B| )$.
        \item \textbf{Decoding Complexity.} The decoding complexity is preserved.
        \item \textbf{Preserves Agnosticity.} If the complete decoding distribution of $\mc{D}$ is agnostic, then so is that of $\mc{D}'$.
        \item \textbf{Completeness.} $\mc{D}'$ has perfect completeness.
        \item \textbf{Soundness.} $(L,\eps + 3\eta)$-list-decoding soundness. 
    \end{itemize}
\end{lemma}
\begin{proof}

    Fix $\eta \in (0,1)$ and let $\mc{D}$ denote the original dPCP. Let us arbitrarily identify $\Sigma_B$ with $\{0,1\}^{k}$ for $k = \log(|\Sigma_B|)$ and take $\mc{C}: \{0,1\}^k \to \Omega^m$ to be the code from \cref{lm: code GI} instantiated with relative distance $1 - \eta^3$, alphabet size $|\Omega| = O_{\eta}(1)$, and blocklength $m = O_{\eta}\left(k\right)$.
    
    The new dPCP $\mc{D}'$ is constructed as follows. The constraint graph is bipartite and has vertices $A \cup B \times [m]$. The edges of the constraint graph are all $(a, (b,i)) \subseteq A \times \left( B \times [m]\right)$ such that $(a,b) \in E$. The constraint on $(a,(b,i))$ is as follows. Let $\sigma_a$ be label for $a$ and $\omega$ be a label for $(b,i)$. Since $\mc{D}$ has projection constraints, there is a unique $\sigma_b \in \Sigma_B$ satisfying $\Phi_{(a,b)}(\sigma_a, \sigma_b) = 1$. The constraint on $\mc{D}'$ checks if $\mc{C}(\sigma_b)_i = \omega$. The new decoding distributions $\mc{P}'_t$ are generated by choosing $(a,b) \sim \mc{P}_t$, choosing $i \in [m]$ uniformly, and outputting $(a, (b,i))$.  We proceed to verifying the desired properties. 
    \vspace{-2ex}
\paragraph{Degrees.} It is clear that decoding degree is preserved, as each $a \in A$ has the same decoding neighborhood in both $\mc{D}$ and $\mc{D}'$. The right degree is preserved because each $(b,i) \in B \times [m]$ has the same neighbors in the new constraint graph as neighbors as $b$ in the original constraint graph. For the left degree, note that each $a \in A$ is adjacent to $(b,i)$ for all $b$ that are adjacent to $a$ in the original constraint graph and all $i \in [m]$.
    
\vspace{-2ex}
\paragraph{Projection Decision Complexity.} Fix a constraint $(a, (b,i))$ in $\mc{D}'$ and suppose the assignment to $a$ is $\sigma$. The projection circuit first computes the $\sigma' \in \Sigma_B$ which is the unique $b$ alphabet symbol that would satisfy the constraint in $\mc{D}$. This part has circuit complexity ${\sf ProjComp}$. Then, the projection circuit calculates the encoding of $\sigma'$, under the linear-time encodable code, and outputs its $i$th symbol. This part has circuit complexity $O(\log|\Sigma_B|)$ since $\mc{C}$ is linear-time encodable. Overall, this shows that the projection decision complexity is ${\sf ProjComp}  + O(\log|\Sigma_B|)$.

\vspace{-2ex}
\paragraph{Decoding Complexity.} The decoder is unchanged and hence the decoding complexity is preserved.

\vspace{-2ex}
\paragraph{Preserves Agnosticity} Let $\mc{Q}$ and $\mc{Q}'$ denote the complete decoding distributions of $\mc{D}$ and $\mc{D}'$ respectively. Note that for any $(b,i) \in B \times [m]$ and $t \in [n]$, we have 
$\mc{Q}(\cdot, \cdot, b) = \mc{Q}'(\cdot, \cdot, (b,i))$ 
and 
$\mc{Q}(t, \cdot, b) = \mc{Q}'(t, \cdot, (b,i))$. 
It follows that agnosticity is preserved.

\vspace{-2ex}
\paragraph{Perfect Completeness.} Fix $w \in \mc{L}$ and let $T_1:A \to \Sigma_A$ and $T_2: B\to \Sigma_B$ be the assignment satisfying all of the constraints to $\mc{D}$ and decoding to $w$ with probability $1$ for every $t \in [n]$. Then consider the assignment to $\mc{D}'$ where the left assignment is still $T_1$ and the new right assignment, which we call $T'_2: B \times [m] \to \Omega$ is defined as
   \[
   T'_2[(b,i)] = \mc{C}(T_2[b])_i.
   \] 
   It is straightforward to check that this assignment satisfies all of the constraints of $\mc{D}'$ (and $T_1$ still decodes to $w$ with probability $1$ for every $t \in [n]$ because it is unchanged).
\vspace{-2ex}
    \paragraph{List-Decoding Soundness.} Fix an assignment to $T_1: A \to \Sigma$ and let $w_1, \ldots,w_L \in \mc{L}$ be the list guaranteed by the list-decoding assumption of $\mc{D}$. For each $a \in A$ and $b \in B$ such that $(a,b) \in E$, we let $T_1[a]_{\rightarrow b} \in \Sigma_B$ be the unique symbol $\sigma$ such that $\Phi_{(a,b)}(T_1[a], \sigma) = 1$. Let $\mc{Q}$ be the complete decoding distribution of $\mc{D}$ generated by choosing $t \in [n]$, $(a,b) \sim \mc{P}_t$, and outputting $(t,a,b)$. 
    For each $b \in B$ and $\sigma \in \Sigma_B$, define
    \[
    \delta_b(\sigma) = \Pr_{(t,a) \sim \mc{Q}(\cdot, \cdot, b)}[T_1[a]_{\rightarrow b} = \sigma \land D_t(a, T_1[a]) \notin \{\left(w_i\right)_t \}_{i \in [L]}].
    \]
    For any $T_2: B \to \Sigma_B$, the expectation $\E_{b \sim \mc{Q}}[\delta_b(T_2[b])]$ is precisely the list-decoding error in $\mc{D}$ of $T_1, T_2$ relative to the list $\{w_1,\ldots,w_L\}$. Here and henceforth, $b \sim \mc{Q}$, refers to $b$ sampled according to the marginal of $\mc{Q}$ over $B$. By the list-decoding soundness assumption of $\mc{D}$, we have
    \begin{equation} \label{eq: list decoding assumption in right alph red}  
   \E_{b \sim \mc{Q}}[\delta_b(T_2[b])] \leq \eps.
    \end{equation}
    for any $T_2: B \to \Sigma_B$.
    
    Now, fix an arbitrary right assignment $T'_2: B \times [m] \to \Omega$ for  $\mc{D}'$. We will show that the list-decoding error of $T_1$ and $T'_2$ relative to $\{w_1,\ldots,w_L \}$ is at most $\eps + 3\eta$. We can express this list-decoding error as follows:
    \begin{equation} \label{eq: list decoding error in right alph red}     
    \E_{b \sim \mc{Q}}\left[\sum_{\sigma \in \Sigma_B} \delta_b(\sigma) \cdot \Pr_{i \in [m]}[T'_2[(b,i)] = \mc{C}(\sigma)_i]\right],
    \end{equation}
    so our goal is to now bound the expectation above.
   Fix a $b \in B$ and examine the summation in the expectation. Define ${\sf list}_b$ to consist of all $\sigma \in \Sigma_B$ such that $\mc{C}(\sigma)$ agrees with at least $\eta$-fraction of the entries of the string $(T[(b,1)], \ldots, T[(b, m)]) \in \Omega^m$. By the list-decoding bound in \cref{fact: johnson} and the fact that $\eta > 2\eta^{3/2}$, we have that $|{\sf list}_b| \leq 2/\eta$.  We define a right assignment, $T^*_2: B \to \Sigma_B$,  for $\mc{D}$ as follows. For each $b \in B$, set $T^*_2[b]$ to be the alphabet symbol $\sigma \in \Sigma_B$ which maximizes $\delta_b(\sigma)$. Then, the summation from \eqref{eq: list decoding error in right alph red} can be bounded by
    \begin{align*}
        &\sum_{\sigma \in \Sigma_B} \delta_b(\sigma) \cdot \Pr_{i \in [m]}[T'_2[(b,i)] = \mc{C}(\sigma)_i] \\
        &= \sum_{\sigma \in {\sf list}_b} \delta_b(\sigma) \cdot \Pr_{i \in [m]}[T'_2[(b,i)] = \mc{C}(\sigma)_i] + \sum_{\sigma \notin {\sf list}_b} \delta_b(\sigma) \cdot \Pr_{i \in [m]}[T'_2[(b,i)] = \mc{C}(\sigma)_i] \\
        &\leq \sum_{\sigma \in {\sf list}_b} \delta_b(\sigma) \cdot \Pr_{i \in [m]}[T'_2[(b,i)] = \mc{C}(\sigma)_i] + \eta \\
        &\leq \delta_b(T^*_2[b])\cdot \left(\sum_{\sigma \in {\sf list}_b} \Pr_{i \in [m]}[T'_2[(b,i)] = \mc{C}(\sigma)_i] \right) + \eta \\
        &\leq \delta_b(T^*_2[b]) \cdot \left( \Pr_{i \in [m]}[\exists \sigma \in {\sf list}_b, \; T'_2[(b,i)]=\mc{C}(\sigma)_i]+ \sum_{\sigma \neq \sigma' \in {\sf list}_b} \Pr_{i \in [m]}[T_2(b,i)=\mc{C}(\sigma)_i = \mc{C}(\sigma')_i]\right) + \eta \\
        &\leq \delta_b(T^*_2[b]) \cdot (1 + 2\eta) + \eta \\
        &\leq \delta_b(T^*_2[b]) + 3\eta.
    \end{align*}
    In the second transition we used $\sum_{\sigma \in \Sigma_B} \delta_b(\sigma) \leq 1$ and that  $\Pr_{i \in [m]}[T'_2[(b,i)] = \mc{C}(\sigma)_i] \leq \eta$ for each $\sigma \notin {\sf list}_b$. In the third transition we used the definition of $T^*_2[b]$. In the fourth transition we used the inclusion-exclusion principle. 
    In the fifth transition we used the fact that $|{\sf list}_b|\leq 2/\eta$ and that for $\sigma\neq \sigma'$ we have 
    $\Pr_{i \in [m]}[\mc{C}(\sigma)_i = \mc{C}(\sigma')_i]\leq \eta^3$.
    Plugging this bound back into the expectation from \eqref{eq: list decoding error in right alph red}, we get that the list-decoding error of $\mc{D}'$ is at most 
    \[
        \E_{b \sim \mc{Q}}\left[\sum_{\sigma \in \Sigma_B} \delta_b(\sigma) \cdot \Pr_{i \in [m]}[T'_2[(b,i)] = \mc{C}(\sigma)_i]\right] \leq \E_{b \sim \mc{Q}}\left[\delta_b\left(T^*_2[b]\right)\right] + 3\eta \leq \eps + 3\eta,
    \]
    where in the last transition we used \eqref{eq: list decoding assumption in right alph red} applied to the tables $T_1, T^*_2$.
\end{proof}

\subsection{Right Degree Reduction}
The goal of this subsection is to describe a procedure which transforms an arbitrary dPCP, $\mc{D}$, into a right-regular dPCP $\mc{D}'$ with constant right degree. Additionally, in the complete decoding distribution of $\mc{D}'$, the marginal distribution over the right vertices is uniform. This transformation is similar to the one in~\cref{sec:udpcp_reg}, but some extra care is needed since we deal with list-decoding soundness rather than unique-decoding soundness. 

\begin{lemma} \label{lm: right degree reduction}
    Suppose a language $\mc{L} \subseteq \Sigma_0^n$ has a projection dPCP 
    \[
    \mc{D} = \left( A \cup B, E,\{\mc{P}_t \}_{t \in [n]}, \Sigma_A, \Sigma_B ,\{D_t \}_{t \in [n]} \right)
    \]
    with $(L, \eps)$-list-decoding soundness and suppose there exists $M \in \mathbb{N}$ such that for every $(a,b) \in A \times B$, one can write $\mc{Q}(\circ, a,b) = \frac{w(a,b)}{M}$ for some $w(a,b) \in \mathbb{N}$. Then for any $d' \geq \min_{a,b}w(a,b)$, there is a polynomial-time algorithm which takes $\mc{D}$ and outputs a projection dPCP  
    \[
    \mc{D}' = \left(A \cup B', E',\{\mc{P}'_t \}_{t \in [n]}, \Sigma_A, \Sigma_B, \{D_t\}_{t \in [n]}  \right),
    \]
for $\mc{L}$ such that the following hold:
    \begin{itemize}
        \item \textbf{Length.} $|B'| = M$ and hence the length is $|A| + M$.
        \item \textbf{Degrees.} The left degree of every vertex is multiplied by $d'$, the decoding degree of every vertex is preserved, and the right degree of every vertex is $d'$.
        \item \textbf{Projection decision complexity.} The projection decision complexity is preserved.
        \item \textbf{Decoding complexity.} The decoding complexity is preserved.
        \item \textbf{Complete decoding distribution.} Call the complete decoding distribution $\mc{Q}'$. Then,
        \begin{itemize}
            \item $\mc{Q}'$ is $d'M$-discrete.
            \item For every $b \in B'$, we have $\mc{Q}'(\circ, \circ, b) = \frac{1}{M}$ and hence the marginal of $\mc{Q}'$ over $B'$ is uniform.
            \item For every $a \in A$, we have $\mc{Q}'(\circ, a, \circ) = \mc{Q}(\circ, a, \circ)$

            \item If $\mc{Q}$ is agnostic, then $\mc{Q}'$ is agnostic as well.
        \end{itemize}
        \item \textbf{Completeness.} The dPCP $\mc{D}'$ has perfect completeness.
        \item \textbf{List-decoding soundness. }$(L, \eps + O(d'^{-1/2}))$-list-decoding soundness.
    \end{itemize}
    \end{lemma}
\begin{proof}
Let $G_{\disc}$ be the bipartite graph on $A \cup B$, where there are $w(a,b)$ multi-edges between $a$ and $b$ for each $(a,b) \in A \times B$. 
The dPCP $\mc{D}'$ is constructed as follows. For each vertex $b \in B$, let $\Gamma_{\disc}(b)$ be the neighborhood of $b$ in $G_{\disc}$, where we view $\Gamma_{\disc}(b)$ as a multiset. Specifically, each neighbor $a$ is listed with multiplicity equal to the number of edges going from $b$ to $a$, so that $|\Gamma_{\disc}(b)| = d(b)$. Then, let $H_b = (A_b, B_b, E_b)$ be a $d'$-regular bipartite graph with $d(b)$ vertices on each side and second singular value $O(d'^{-1/2})$, and identify $A_b$ with $\Gamma_{\disc}(b) \subseteq A$. 
Recall that such a graph can be constructed in polynomial-time by \cref{lm: poly time bip expander}. We note that after this identification, there may be multiple vertices in $A_b$ that have been identified with the same vertex $a\in \Gamma_{\disc}(b)$, so there may be multiple edges between two vertices.

The constraint graph of $\mc{D}'$ has left side $A$ and right side $\cup_{b \in B} B_b$. We include an edge $(a,b')$ if for the unique $b$ such that $b' \in B_b$, it holds that $a \in \Gamma(b)$ and $(a, b') \in E_b$ (we recall that each vertex in $A_b$ is identified with a vertex in $A$). The constraint on $(a,b')$ in $\mc{D}'$ is the same as the constraint on $(a,b)$ in $\mc{D}$, where $b' \in B_b$. The left and right alphabets of $\mc{D}'$ remain $\Sigma_A, \Sigma_B$. The decoders for $\mc{D}'$ also remain the same and are still $\{D_t \}_{t \in [n]}$. By this we mean, for a left vertex $a$ and a label $\sigma$ in its alphabet, the decoder for $\mc{D}'$ still outputs $D_t(a, \sigma)$, and hence we still denote it by $D_t$. Finally, the complete decoding distribution of $\mc{D}'$ is $\mc{Q}'$, which is generated by choosing $(t,a,b) \sim \mc{Q}$, sampling a uniform $b' \in B_b$ that is adjacent to $a$ in $H_b$ and outputting $(t,a,b')$. The decoding distributions are thus given by $\mc{P}'_t = \mc{Q}'(t,\cdot, \cdot )$.

It is clear that the algorithm preserves the left alphabet, right alphabet, projection decision complexity, decoding complexity, and decoding degree. Furthermore, one can check that the left degree of each vertex in the constraint graph gets multiplied by a factor of $d'$.

\vspace{-2ex}
\paragraph{Complete Decoding Distribution.} Fix $a \in A$ and $b' \in B'$ and suppose $b'$ appears in $B_b$ for $b \in B$. Let us first check that the complete decoding distribution is $d'M$-discrete. We have
\begin{equation}\label{eq:right_alphabet_red_eq1}
\mc{Q}'(\circ, a, b') = \mc{Q}(\circ, a, b) \cdot \frac{c_b(a,b')}{d'} = \frac{w(a,b) \cdot c_b(a,b')}{d'M}
\end{equation}
where $c_b(a,b')$ is the number of multiedges between 
$a$ and $b'$ in the graph $H_b$, and is an integer in the range $[0,d']$. To see the first equality, note that in order for $(a,b')$ to be sampled from $\mc{Q}'$, we first need $(a,b)$ to be sampled from $\mc{Q}$, where $b$ is such that $b' \in H_b$. Then $\mc{Q}'$ chooses a neighbor of $a$ in $H_b$, so there is a $c_b(a, b')/d'$ probability of choosing $b'$. The second transition is due to the assumption that $ \mc{Q}(\circ, a, b) = \frac{w(a,b)}{M}$ given by the lemma statement. It follows that the complete decoding distribution is $d'M$-discrete 

To see the second item, note that the probability that $b'$ is output by the complete decoding distribution $\mc{Q}'$ is 
\[
\mc{Q}'(\circ, \circ, b') = \mc{Q}(\circ, \circ, b) \cdot \frac{1}{d(b)} = \left(\sum_{a \in A} \frac{w(a,b)}{M}\right) \frac{1}{d(b)} = \frac{1}{M}.
\]
In the first transition we used the fact that conditioned on $b$, the right-vertex we output is uniform in $B_b$ and $|B_b| = d(b)$. 
In the second transition we used $\mc{Q}(\cdot,a,b) = w(a,b)/M$. In the last transition we used the fact that $\sum_{a \in A} w(a,b) = d(b)$. 

For the third item, fix $a \in A$, then 
\[
\mc{Q}'(\circ, a, \circ) = 
\sum\limits_{b\in B}\sum_{b' \in B_b} \mc{Q}'(\circ,a , b') 
=
\sum\limits_{b\in B}\mc{Q}(\circ, a, b)\sum_{b' \in B_b} \frac{c_b(a,b')}{d'}
= \sum\limits_{b\in B}\mc{Q}(\circ, a, b)=\mc{Q}(\circ, a, \circ).
\]
In the second transition we used~\eqref{eq:right_alphabet_red_eq1} and 
in the third transition we used the fact that the degree of $a$ in $H_b$ is $d'$.

For the fourth item, suppose that $\mc{Q}$ is agnostic and take $b' \in B_b$. Note that $\mc{Q}(\circ, \cdot, b) = \mc{Q}'(\circ, \cdot, b')$, and for any $t\in[n]$ 
we have $\mc{Q}(t, \cdot, b) = \mc{Q}'(t, \cdot, b')$.
It follows $\mc{Q}'$ is agnostic.

\vspace{-2ex}
\paragraph{Perfect Completeness.} Suppose $w \in \mc{L}$ and let $T_1: A \to \Sigma_A$, $T_2: B \to \Sigma_B$ be the assignments satisfying all of the constraints in $\mc{D}$ that are decoded to $w$ with probability $1$ over $\mc{Q}$. Define $T_2'$ by setting  $T'_2[b'] = T_2[b]$ for each $b' \in B_b$.  It is straightforward to check that $T_1,T_2'$ satisfy all of the constraints in $\mc{D}'$, and that $T_1$ decodes to $w$ with probability $1$ over $\mc{Q}'$.
\vspace{-2ex}
\paragraph{List-Decoding Soundness.}
Fix an assignment $T_1: A \to \Sigma_A$. 
Then by the $(L, \eps)$-list-decoding soundness of $\mc{D}$ we may find $\{w_1, \ldots,w_L\} \subseteq \mc{L}$  such that for any $T_2: B \to \Sigma_B$, we have
\begin{equation*}
\Pr_{(t,a,b) \sim \mc{Q}}[\Phi_{(a,b)}(T_1[a], T_2[b]) = 1 \land D_t(a, T_1[a]) \notin \{\left(w_i\right)_t \}_{i \in [L]}] \leq \eps.
\end{equation*}

To establish the desired list-decoding soundness for $\mc{D}'$, we will show that for any right-assignment, the list decoding error relative to the same list $\{w_1, \ldots,w_L\}$ is small.

Fix any assignment $T'_2: B' \to \Sigma_B$ and for each $b \in B$ and $\sigma \in \Sigma_B$, define the sets
\[
X_{b,\sigma} = \{b' \in B_b \; | \; T'_2[b'] = \sigma \} \quad \text{and} \quad Y_{b, \sigma} = \{a' \in A_b \; | \; \Phi_{(a', b)}(T_1[a'], \sigma) = 1\}.
\]
In words, $X_{b, \sigma}$ is the set of vertices in $B_b$ that are assigned $\sigma$ under $T'_2$, and $Y_{b,\sigma}$ is the set of vertices in $a' \in A_b$ such that the assignments $a \to T_1[a']$ and $b \to \sigma$ satisfy the constraint on $(a', b)$ in the original dPCP, $\mc{D}$. It is clear that the sets $\{X_{b,\sigma}\}_{\sigma \in \Sigma_b}$ form a partition of $B_b$. Additionally, since the constraints in $\mc{D}$ are all projections, the sets 
$\{Y_{b,\sigma}\}_{\sigma \in \Sigma_B}$ form a partition of $A_b$. Define the functions $F_{b,\sigma}: A_b \to [0,1]$, $G_{b,\sigma}: B_b \to \{ 0,1\}$ where 
\[
G_{b, \sigma}(b') := \ind_{b' \in X_{b,\sigma}} \quad \text{and} \quad F_{b, \sigma}(a) := \ind_{a \in Y_{b,\sigma}} \cdot \Pr_{t \sim \mc{Q}'(\cdot, a,b)}\left[D_t(a, T_1[a]) \notin \{\left(w_i\right)_t \}_{i \in [L]}\right].
\]
Abusing notation, we denote $\mu(F_{b, \sigma}) = \E_{a \in A_b}[F_{b, \sigma}(a')]$ and $\mu(G_{b, \sigma})= \E_{b' \in B_b}[G_{b, \sigma}(b')]$. Consider the randomized assignment $T^*_2: B \to \Sigma_B$ obtained by setting $T^*_2(b) = \sigma$ with probability $\mu(G_{b, \sigma})$. By the list-decoding soundness applied with each fixing of $T^*_2$, we get
\begin{equation} \label{eq: prior expected ld error} 
\begin{split}
  \eps
  &\geq \E_{T^*_2}\left[  \Pr_{(t,a,b) \sim \mc{Q}}\left[\Phi_{(a,b)}(T_1[a], T^*_2[b]) = 1 \land D_t(a, T_1[a]) \notin \{\left(w_i\right)_t \}_{i \in [L]}\right] \right] \\
  &= \E_{(t,a,b) \sim \mc{Q}}\left[\Pr_{b' \in B_{b}}\left[\Phi_{a,b}(T_1[a], T_2[b'])=1 \land D_t(a, T_1[a]) \notin \{\left(w_i\right)_t \}_{i \in [L]}\right] \right]\\
  &= \E_{b \sim \mc{Q}(\circ, \circ, \cdot)}\left[\sum_{\sigma \in \Sigma_B} \mu(G_{b, \sigma}) \cdot \mu(F_{b, \sigma})]\right],
\end{split}
\end{equation}
In the third transition, we are using the fact that conditioned on $b$, the distribution of $a$ under $\mc{Q}$ is uniformly random in $A_b$.
On the other hand, the list-decoding error of the assignment $T_1, T'_2$ in $\mc{D}'$ relative to the list $\{w_1,\ldots,w_L\}$, can be written as:
\begin{equation} \label{eq: sum over sigma_b}  
\E_{b \sim \mc{Q}(\circ, \circ, \cdot)}\left[ \sum_{\sigma \in \Sigma_B} \E_{(a,b') \in E_b}[F_{b,\sigma}(a) \cdot G_{b,\sigma}(b') ] \right].
\end{equation}
Hence, our goal is to now bound the expression above, and we will do so by relating it to \eqref{eq: prior expected ld error} via the expander mixing lemma.
Fix a vertex $b \in B$ and an alphabet symbol $\sigma \in \Sigma_B$. Using~\cref{lm: expander mixing} we get
\[
\E_{(a,b') \in E_b}[F_{b,\sigma}(a) \cdot G_{b,\sigma}(b') ] \leq \mu(F_{b,\sigma}) \cdot \mu(G_{b,\sigma}) + O(d'^{-1/2}) \cdot \sqrt{ \mu(F_{b,\sigma}) \cdot \mu(G_{b,\sigma})}.
\]
Keeping $b \in B$ fixed and summing this inequality over all $\sigma \in \Sigma_B$, we get 
\begin{align*}  
 \sum_{\sigma \in \Sigma_B} \E_{(a,b') \in E_b}[F_{b,\sigma}(a) \cdot G_{b,\sigma}(b') ] &\leq \sum_{\sigma \in \Sigma_B} \mu(F_{b,\sigma}) \cdot \mu(G_{b,\sigma}) + O(d'^{-1/2}) \cdot \sqrt{ \mu(F_{b,\sigma}) \cdot \mu(G_{b,\sigma})} \\
 &\leq \sum_{\sigma \in \Sigma_B} \mu(F_{b,\sigma}) \cdot \mu(G_{b,\sigma}) + O(d'^{-1/2}),
\end{align*}
where the second inequality uses Cauchy-Schwarz and the fact that $\sum_{\sigma \in \Sigma_B} \mu(F_{b,\sigma}) \leq 1$ and $\sum_{\sigma \in \Sigma_B} \mu(G_{b,\sigma}) = 1$. Thus, we have 
\[
    \E_{b \sim \mc{Q}(\circ, \circ, \cdot)}\left[ \sum_{\sigma \in \Sigma_B} \E_{(a,b') \in E_b}[F_{b,\sigma}(a) \cdot G_{b,\sigma}(b') ] \right] \leq \E_{b \sim \mc{Q}(\circ, \circ, \cdot)}\left[\sum_{\sigma \in \Sigma_B} \mu(G_{b, \sigma}) \cdot \mu(F_{b, \sigma})\right] + O(d'^{-1/2})
    \leq \eps+O(d'^{-1/2}),
\]
where in the last transition we used~\eqref{eq: prior expected ld error}.
\end{proof}

\cref{lm: right degree reduction} gives a right degree reduction when the complete decoding distribution is $M$-discrete and results in a dPCP of size at least $M$. For dPCPs with arbitrary complete decoding distribution, which may not be $M$-discrete for any small $M$, we can still obtain a right degree reduction by slightly altering the complete decoding distribution first. Below we state a simple fact about distributions that we will use for this task, which will also be used in subsequent sections.

\begin{lemma} \label{lm: discretize}
    Let $\Delta$ be a distribution over a support of size $n$, say $[n]$. Then for every $\eps > 0$ and every integer $M \geq \frac{n}{2\eps}$, there exist $w(1), \ldots, w(n) \in \mathbb{N}$ such that $\sum_{i\in[n]} w(i) = M$ and the distribution $\Delta_{\disc}$, given by $\Delta_{\disc}(i) = \frac{w(i)}{M}$, satisfies
    \begin{itemize}
    \item  $\tv(\Delta, \Delta_{\disc}) \leq \eps$,
    \item  for every $i \in [n]$, $|\Delta_{\disc}(i) - \Delta(i)| \leq \frac{1}{M}$,
    \item ${\sf supp}(\Delta_{\disc})\subseteq {\sf supp}(\Delta)$.
    \end{itemize}
    As a consequence, every distribution with support size $n$ is $\eps$-close in total-variation distance to an $O_{\eps}(n)$-discrete distribution.
\end{lemma}
\begin{proof}
Suppose $\supp(\Delta)$ has size $n' \leq n$, and without loss of generality suppose $\supp(\Delta) = [n']$. For each $i \in [n']$, let $z_i = \lfloor M \cdot \Delta(i) \rfloor$ and let $r = M - \sum_{i = 1}^{n'} z_i$. Note that $r \leq n'$. Now increment the each of $z_1, \ldots, z_r$ by $1$ and keep the remainder of the $z_i$'s the same. Let the resulting values be $z'_i$ after incrementing. 

We have 
\begin{equation} \label{eq: discrete off}   
\left| z'_i - M \cdot \Delta(i) \right| \leq 1,
\end{equation}
for all $i \in [n']$. Now define the distribution $\Delta_{\disc}(i)$ by
\[
\Delta_{\disc}(i) = 
\begin{cases}
\frac{z'_i}{M} \quad &\text{ if $i \in [n']$,}\\
0 \quad &\text{otherwise.}
\end{cases}
\]
Note that $\Delta_{\disc}(i)$ is indeed a distribution over $[n]$ because $\sum_{i \in [n]} z'_i = M$ and by construction $\supp(\Delta_{\disc}) \subseteq \supp(\Delta)$. Moreover, \eqref{eq: discrete off} implies that for every $i \in [n]$, $|\Delta_{\disc}(i) - \Delta(i)| \leq \frac{1}{M}$, and as a consequence, 
\[
\tv(\Delta, \Delta_{\disc}) = \frac{1}{2} \sum_{i \in [n]} \left| \Delta(i)  - \Delta_{\disc}(i) \right| \leq \sum_{i \in [n']} \frac{1}{2M} \leq \frac{n}{2M} \leq \eps.\qedhere
\]
\end{proof}

The following result is a version of~\cref{lm: right degree reduction} in which the given dPCP $\mc{D}$ does not necessarily have a complete decoding distribution whose marginal over the edges is $M$-discrete . Instead, we require $\mc{D}$ to have bounded decoding degree, and we show how to modify the complete decoding distribution into a discrete one without affecting the rest of the parameters by much. For sake of simplicity and because it suffices for our application, we focus on the case where decoding degree is $1$.
\begin{lemma} \label{lm: gen right degree reduction} 
    Suppose a language $\mc{L} \subseteq \Sigma_0^n$ has a projection dPCP 
    \[
    \mc{D} = \left( A \cup B, E, \Sigma_A, \Sigma_B ,\{D_t \}_{t \in [n]}, \{\mc{P}_t \}_{t \in [n]} \right)
    \]
    with $(L, \eps)$-list-decoding soundness, left degree $\mathsf{ld}$, and decoding degree $1$. Then for every $M \geq \frac{200 \cdot \mathsf{ld} \cdot |A|}{\eps^2 \cdot n}$ and for every $d' \in \mathbb{N}$, there is a polynomial-time algorithm which takes $\mc{D}$ and outputs a dPCP  
    \[
    \mc{D}'' = \left(A'' \cup B'', E'', \Sigma_A, \Sigma_B, \{D_t\}_{t \in [n]} , \{\mc{P}''_t \}_{t \in [n]} \right)
    \]
for $\mc{L}$ such that the following hold:
    \begin{itemize}
        \item \textbf{Length.} $|B''| = d' \cdot n \cdot M$ and $A'' \subseteq A$, so the length is at most $|A| + d' \cdot n \cdot  M$.
        \item \textbf{Degrees.} The decoding degree of every vertex is $1$.
        \item \textbf{Projection decision complexity.} The projection decision complexity is preserved.
        \item \textbf{Decoding complexity.} The decoding complexity is preserved.
        \item \textbf{Complete decoding distribution.} The complete decoding distribution, which we call $\mc{Q}''$, is $d'^2 \cdot n \cdot M$-discrete  and its marginal over $B'$ is uniform. For any $b \in B$, we have
        \[
        \mc{Q}''(\circ, \circ, b) = \frac{d'}{d'^2 \cdot n \cdot M}
        \]
        and for any $a \in A$ we have
        \[
        \mc{Q}''(\circ, a, \circ) \leq \max\left\{ \frac{1}{nM} + \frac{\eps }{10\mathsf{ld} |A|} ,\mc{Q}(\circ, a, \circ) + \frac{\mathsf{ld}}{nM}\right\}.
        \]
        \item \textbf{Completeness.} The dPCP $\mc{D}'$ has perfect completeness.
        \item \textbf{List-decoding soundness. }$\left(L, \frac{9}{8} \cdot \eps + O(d'^{-1/2})\right)$-list-decoding soundness.
    \end{itemize}
\end{lemma}
\begin{proof}
Let $\mc{Q}$ be the complete decoding distribution of $\mc{D}$ and let $\mathsf{ld}, \eps$ and $M$ be as in the lemma statement. For each $a \in A$, we have that the decoding degree of $a$ is $1$, and we denote by $t_a$ its unique decoding neighbor.  Since the decoding degree is $1$, we have
\[
\sum_{t \in [n]} |\supp(\mc{Q}(t, \cdot, \cdot))| \leq \sum_{a \in A} |\supp(\mc{Q}(t_a, a, \cdot))| \leq \mathsf{ld} \cdot |A|.
\]
Hence, taking $S \subseteq [n]$ to be the set of $t \in [n]$ such that $|\supp(\mc{Q}(t, \cdot, \cdot))| \leq \frac{10 \cdot \mathsf{ld} \cdot |A|}{\eps \cdot n}$, we have by Markov's inequality that
\begin{equation} \label{eq: t small neighborhood}  
\Pr_{t \in [n]}[t \in S] \geq 1 - \frac{\eps}{10}.
\end{equation} 
Now, for each $t \in S$, let $\Delta_{t}$ be the discretized distribution over $A \times B$ obtained by applying \cref{lm: discretize} to $\mc{Q}(t, \cdot, \cdot)$ with denominator $M$. For $t \notin S$, let $\Delta_{t}$ be an arbitrary distribution  supported  on $\supp(\mc{Q}(t, \cdot, \cdot))$ such that for each $(a, b) \in \supp(\mc{Q}(t, \cdot, \cdot))$ we have either
\[
\Delta_t(a,b) = \frac{1}{M}\cdot \ceil[\bigg]{\frac{M}{|\supp(\mc{Q}(t, \cdot, \cdot))|}} \quad \text{or} \quad\Delta_t(a,b) = \frac{1}{M} \cdot \floor[\bigg]{\frac{M}{|\supp(\mc{Q}(t, \cdot, \cdot))|}}.
\]
It is straightforward to see that such a distribution indeed exists. 

Now, let the distribution $Q'$ be generated by choosing $t \in [n]$ uniformly at random, choosing $(a,b) \sim \Delta_{t}$, and outputting $(t, a, b)$. We note the following items:
\begin{itemize}
    \item The marginal over $[n]$ of $\mc{Q}'$ is uniform; this is clear by description of $\mc{Q}'$.
    \item $\supp(\mc{Q}') \subseteq \supp(\mc{Q})$; this is clear by construction of $\Delta_t$.
    \item For each $(a, b) \in A \times B$, if $t_a \in S$ then $|\mc{Q}(\circ, a, b) - \mc{Q}'(\circ, a, b)| \leq \frac{1}{nM}$, otherwise $\mc{Q}'(\circ, a, b) \leq \frac{1}{nM} + \frac{\eps }{10\mathsf{ld} |A|}$; we explain this below.
    \item $\tv(\mc{Q}, \mc{Q}') \leq \frac{\eps}{8}$; we explain this below.
\end{itemize}
To see the third item first fix an $(a, b) \in A \times B$. If $t_a \in S$, then 
$\mc{Q}(\circ, a, b) = \mc{Q}(t_a,a,b)$, 
$\mc{Q}'(\circ, a, b) = \Delta_{t_a}(a,b)$ and the item follows from the guarantee of \cref{lm: discretize}. If $t_a\not\in S$ we have,
\[
\mc{Q}'(\circ, a, b) \leq \frac{1}{n} \cdot \Delta_{t_a}(a, b) \leq \frac{1}{n} \cdot \frac{1 + M/|\supp(\mc{Q}(t, \cdot, \cdot))|}{M} \leq \frac{1}{nM} + \frac{\eps }{10\mathsf{ld} |A|}.
\]
For the fourth item, we use the fact that $\mc{Q}$ and $\mc{Q}'$ have the same marginal distribution over $[n]$, hence
\begin{align*}
    \tv(\mc{Q}, \mc{Q}') &=  \sum_{t \in [n]} \mc{Q}(t, \circ, \circ) \cdot \tv(\mc{Q}(t, \cdot, \cdot), \mc{Q}'(t, \cdot, \cdot)) \\
    &\leq \frac{\eps}{10} + \sum_{t \in S} \mc{Q}(t, \circ, \circ) \cdot \tv(\mc{Q}(t, \cdot, \cdot), \mc{Q}'(t, \cdot, \cdot)) \\
    &\leq \frac{\eps}{8}.
\end{align*}
The first transition uses~\eqref{eq: t small neighborhood} and the second transition uses the fact that 
\[
\tv(\mc{Q}(t, \cdot, \cdot), \mc{Q}'(t, \cdot, \cdot)) =  \tv(\mc{Q}(t, \cdot, \cdot), \Delta_t)\leq \frac{\eps}{40}
\]
for all $t \in S$ by \cref{lm: discretize}.

Let $\mc{D}'$ be the dPCP obtained by taking $\mc{D}$ and changing its complete decoding distribution to $\mc{D}'$. Note that the length, projection decision complexity, and decoding complexity are unchanged. For the decoding degrees, note that each $a$ has decoding degree either $0$ or $1$ now. Completeness is preserved by  the fact that $\supp(\mc{Q}') \subseteq \supp(\mc{Q})$, stated above, and $\mc{D}'$ still has $(L, 9\eps/8)$-list decoding soundness by the fact $\tv(\mc{Q}, \mc{Q}') \leq \eps /8$, stated above.

Let us now write $\mc{Q}'(\circ, a, b) = \frac{w(a,b)}{n \cdot M} = \frac{d'w(a,b)}{d' \cdot n \cdot M}$ for $w(a,b) \in \mathbb{N}$. In expressing $\mc{Q}'$ this way, we are using the fact that $a$ has decoding degree $1$. Here, we multiplied both the numerator and denominator by $d'$ to ensure that all of the nonzero numerators are at least $d'$, in accordance with \cref{lm: right degree reduction}.

We now apply \cref{lm: right degree reduction} on $\mc{D}'$ where its complete decoding distribution is written with denominator $d'nM$, and the weight of each edge $(a,b)$ is $d'w(a,b)$. Call the resulting dPCP $\mc{D}''$ and let $\mc{Q}''$ be its complete decoding distribution. The lemma follows by checking the guarantees from  \cref{lm: right degree reduction}. In particular, the length, projection decision complexity, decoding complexity, perfect completeness, and list-decoding soundness are immediate.

For the guarantee on the complete decoding distribution, we first note that  by \cref{lm: right degree reduction}, $\mc{Q}''$ has uniform marginal over the right side and is $d'^2 n M$ discrete. This also implies that 
\[
\mc{Q}''(\circ, \circ, b) = \frac{d'}{d'^2 \cdot n \cdot M}.
\]
For the remaining guarantee on the complete decoding distribution, fix a left vertex $a$. We start by using the third part of the decoding distribution guarantee from \cref{lm: right degree reduction} to get:
\[
\mc{Q}''(\circ, a, \circ)= \mc{Q}'(\circ, a, \circ) =\sum_{b \in B} \mc{Q}'(\circ, a, b).
\]
If $t_a \notin S$, then $\mc{Q}''(\circ, a, \circ)= \mc{Q}'(\circ, a, \circ) \leq \frac{1}{nM} + \frac{\eps }{10\mathsf{ld} |A|}$ by the third item above and we are done. Otherwise,

\[
|\mc{Q}(\circ, a, \circ) -\mc{Q}''(\circ, a, \circ)| = \left|\sum_{b \in B}\left( \mc{Q}(\circ, a, b) -\mc{Q}'(\circ, a, b)\right)\right| \leq \frac{\mathsf{ld}}{nM},
\]
where in the last transition we use the fact that the summation has $\mathsf{ld}$ many nonzero terms and each is at most $\frac{1}{nM}$ by the third item listed above. The above implies, 
\[
\mc{Q}''(\circ, a, \circ) \leq\mc{Q}(\circ, a, \circ) + \frac{\mathsf{ld}}{nM}.
\]

Finally, to get decoding degree $1$, note that each vertex $a \in A$ has decoding degree either $1$ or $0$, and we can delete every vertex with decoding degree $0$ without affecting any of the previous guarantees.
\end{proof}

%% file: decoding_degree_reduction.tex
\section{Decoding Degree Reduction}\label{sec:dec_deg_red}

In this section we show how to reduce the decoding degree of a dPCP assuming that its complete decoding distribution is agnostic. After performing this operation once, we will preserve constant decoding degree throughout the remainder of our construction.
At a high level, after \cref{thm: dp dpcp} the  decoding degree is $\polylogn$, and hence holding this many $\Sigma_0$-symbols itself already causes superpolynomial alphabet size. Therefore any method of reducing alphabet size to a constant must first reduce decoding degree to a constant. This is a non-issue in standard PCPs (as there is no decoding going on), and partly the reason earlier composition techniques did not apply to dPCPs. 

The statement of our decoding degree transformation is as follows.

\begin{lemma} \label{lm: decoding degree reduction}
    Suppose $\Sigma_0$ is an alphabet and that the language $\mc{L} \subseteq \Sigma_0^n$ has a projection dPCP
    \[
        \mc{D} = (A \cup B, E, \{\mc{P}_t \}_{t \in [n]},\Sigma_A, \Sigma_B,  \{ \Phi_{e}\}_{e \in E }, \{D_t \}_{t \in [n]} )
    \]
    with perfect completeness, $(L, \eps)$-list-decoding soundness for $L:= \poly(1/\eps)$, agnostic complete decoding distribution,  projection decision complexity $\mathsf{ProjComp}$, decoding complexity $\mathsf{DecodeComp}$, and left degree, decoding degree, right degree that are at most $\mathsf{ld}, \mathsf{ldd}, d$ respectively.
    Let $\mc{I} = \{(t,b)~|~\exists a\in A, (a,b)\in\supp(\mc{P}_t)\}$, and let $\{C_{t,b} \in \mathbb{N} \; | (t,b)\in\mc{I} \}$ be a collection of integers. Then $\mc{L}$ has a projection dPCP
    \[
    \mc{D}' = (B' \cup A, \Sigma_A^{d} \times \Sigma_0, \Sigma_A, E',\{\mc{P}_t' \}_{t \in [n]}, \{ \Phi_{e}'\}_{e \in E' }, \{D_t' \}_{t \in [n]} ),
    \]
    satisfying the following properties:
    \begin{itemize}
        \item \textbf{Length.} $B' = \bigcup_{(t,b)\in\mc{I}}\{(t,b,i)~|~i\in[C_{t,b}]\}$, and in 
        particular the length is $|A|+\sum_{(t,b) \in \mc{I}} C_{t,b}$.
        \item \textbf{Degrees.} The left degree of every vertex is at most $d$, the right degree of every vertex is at most ${\sf ld} \cdot {\sf ldd} \cdot \max_{(t,b) \in \mc{I}} C_{t,b}$, and the decoding degree of every vertex is $1$. If $\mc{D}$ is $d$-right regular, then $\mc{D}'$ is $d$-left regular.
        \item \textbf{Projection Decision Complexity.} $O(d \cdot(\mathsf{ProjComp} + \mathsf{DecodeComp}))$. 
        \item \textbf{Decoding Complexity.} $O(d\log(|\Sigma_A|) + \log(|\Sigma_0|))$.
        \item \textbf{Decoding Distributions.} For each $(t,b,i) \in B'$ and $a \in A$, $t$ is the unique index in the decoding neighborhood of $(t,b,i)$ and we have
        \[\mc{P}'_{t}((t,b,i), \circ) = \frac{\mc{P}_t(\circ, b)}{C_{t,b}} \quad \text{and} \quad \mc{Q}'(t, (t,b,i), a) = \frac{\mc{Q}(t, b, a)}{C_{t,b}},
        \]
        where $\mc{Q}$ and $\mc{Q}'$ are the complete decoding distributions of $\mc{D}$ and $\mc{D}'$ respectively.
        \item \textbf{Perfect Completeness.} $\mc{D}'$ has perfect completeness.
        \item \textbf{List-Decoding Soundness.} $\mc{D}'$ has $(L \cdot \poly(1/\eps), 2\eps^{1/5})$-list-decoding soundness.
    \end{itemize}
\end{lemma}

\begin{proof}
Let
\[
\mc{D} = \left(A \cup B, \Sigma_A, \Sigma_B, E,  \{ \Phi_{e}\}_{e \in E }, \{D_t \}_{t \in [n]}, \{\mc{P}_t \}_{t \in [n]} \right)
\]
be a dPCP for the language $\mc{L}$ as in the premise of the lemma. Denote its complete decoding distribution by $\mc{Q}$, namely a sample is  obtained by choosing $t \in [n]$ uniformly, $(a,b) \sim \mc{P}_t$, and outputting $(t,a,b)$. 

We denote the new dPCP for $\mc{L}$ by
\[
\mc{D}' = \left(B' \cup A, \Sigma_A^{d} \times \Sigma_0, \Sigma_A, E', \{ \Phi_{e}'\}_{e \in E' }, \{D_t' \}_{t \in [n]}, \{\mc{P}_t' \}_{t \in [n]} \right),
\]
and define its components in the following way. Below, we use $\Gamma$ and $\Gamma_{\dec}$ to denote the neighborhood and decoding neighborhood functions in $\mc{D}$ respectively, and use $\Gamma'$ and $\Gamma_{\dec}'$ to denote the neighborhood and decoding neighborhood functions in $\mc{D}'$ respectively.
\newcommand{\cloud}{\mathsf{cloud}}
\begin{itemize}
    \item \textbf{Constraint Graph.} The constraint graph is as follows:
     \begin{itemize}
        \item The left side $B' = \bigcup_{(t,b)\in\mc{I}}\{(t,b,i)~|~i\in[C_{t,b}]\}$.
        \item The right side is $A$.
        \item The edge set $E'$ consist of all $((t,b,i), a)$ such that $(a,b) \in E$.
    \end{itemize}
    \item \textbf{Alphabets.} The alphabets are as follows.
    \begin{itemize}
        \item The left alphabet is $\Sigma_{B'}\subseteq \Sigma_A^{d} \times \Sigma_0$. For a fixed $(t,b,i)\in B'$, we think of a label to it as assigning a $\Sigma_A$ label for each $a \in \Gamma(b)$ and a $\Sigma_0$ label for $t$. Furthermore, we hardcode constraints on the alphabet of $(t,b,i) \in B'$ as follows. Writing $\Gamma(b) = \{a_1,\ldots, a_d \}$, for a label $\sigma\in \Sigma_A^{d}\times \Sigma_0$ to $(t,b,i)$ we denote by $\sigma_{a_j}$ be the label assigned to $a_j$ in $\sigma$, and by $\sigma_t$ the label for $t$ in $\sigma$. We denote by $\sigma^{(b)}_{a_j}$ the unique value in $\Sigma_B$ satisfying $\Phi_{(a_j, b)}(\sigma_{a_j}, \sigma^{(b)}_{a_j}) = 1$ (recall that there is a unique such element in $\Sigma_B$ since the constraints in $\mc{D}$ are projections). Then, we require that $D_t(a_j, \sigma_{a_j}) = \sigma_t$ for all $i \in [d]$, and that $\sigma^{(b)}_{a_j}$ is the same for all $j \in [d]$.
        \item The right alphabet is $\Sigma_A$.
    \end{itemize}
    \item \textbf{Constraints.} For each edge $((t,b,i), a) \in E'$, the constraint $\Phi'_{((t,b,i), a)}: \Sigma_{B'} \times \Sigma_A \to \{0,1\}$ is defined as follows. Given labels $\sigma \in \Sigma_B$ and $\sigma' \in \Sigma_A$, let $\sigma_a$ be the label for $a$ in $\sigma$. The constraint $\Phi'_{((t,b,i), a)}(\sigma, \sigma')$ is satisfied if and only if $\sigma_a = \sigma'$.
    
    \item \textbf{Decoders.} For each $t \in [n]$, the decoder $D_t': B' \times \Sigma_B \to \Sigma_0$ is defined as follows. For $(t,b,i)\in B'$ and $\sigma\in \Sigma_{B'}$ label for it, the decoder $D_t'((t,b,i),\sigma)$ outputs $\sigma_t$, the label for $t$ in $\sigma$. 
    
    We note that by the constraints on the left alphabet symbol in $\mc{D}$, for any $a \in \Gamma(b)$, letting $\sigma_a$ be the label for $a$ in $\sigma$, we have that $D_t(a, \sigma_a) = \sigma_t = D'_t((t,b,i), \sigma)$.

    \item \textbf{Decoding Distributions.} For each $t \in [n]$, the decoding distribution $\mc{P}_t'$ is defined as follows. To sample from $\mc{P}_t'$, sample $(a,b) \sim \mc{P}_t$, $i \in [C_{t,b}]$ uniformly, and then output $((t,b,i), a)$. Since we think of the side-decoding distribution as uniform over $[n]$, the complete decoding distribution, denoted $\mc{Q}'$ is obtained by sampling $t \in [n]$ uniformly and then $((t,b,i), a) \sim \mc{P}'_t$. We remark that the complete decoding distribution $\mc{Q}'$ can also be obtained by choosing $(t,a,b) \sim \mc{Q}$ and then choosing $i \in [C_{t,b}]$ uniformly and outputting $(t, (t,b,i), a)$.
\end{itemize}

Let us now analyze the parameters of the new dPCP $\mc{D}'$. As usual, the main difficulty is in analyzing the list-decoding soundness.

\vspace{-2ex}
\paragraph{Length.} It is clear that $|B'| = \sum_{(t,b) \in \mc{I}} C_{t,b}$ and hence the length is as described.
\vspace{-2ex}
\paragraph{Degrees.} It is clear that the degree of every vertex in $B'$ is at most $d$, as for each $(t,b,i) \in B'$, $\Gamma'((t,b,i)) = \Gamma(b)$, and in the case that $\mc{D}$ is right-regular, it is clear that $\mc{D}'$ is left-regular. As for the degrees of each $a \in A$, it is clear that its degree is at most ${\sf ld} \cdot {\sf ldd} \cdot \max_{(t,b)\in\mc{I}}C_{t,b}$ as desired. Finally, the decoding degree of every vertex in $B'$ is $1$, as one can see that each $(t,b,i) \in B'$ is only responsible for decoding $t$.

\vspace{-2ex}
\paragraph{Projection Decision Complexity.} Fix a constraint $((t,b,i), a)$ and let $\sigma$ be the symbols assigned to the $(t,b,i)$. The constraint must check that $\sigma$ is a valid alphabet symbol. After that, the unique symbol to $a$ that would satisfy the constraint can be found by taking a restriction in $\sigma$, and this has circuit complexity at most $\mathsf{ProjComp}$. To check that $\sigma$ is valid, the projection circuit must compute one projection from $\mc{D}$ and one decoding for each neighbor of $b$'s $d$ neighbors and perform $d$ equality checks. Overall, this shows the projection decision complexity is $O(d \cdot( \mathsf{ProjComp} + \mathsf{DecodeComp}))$.

\vspace{-2ex}
\paragraph{Decoding Complexity.} The decoder outputs the restriction to a single symbol from an alphabet of size $|\Sigma_A|^d \cdot |\Sigma_0|$ and hence the decoding complexity is $O(d \log |\Sigma_A| + \log|\Sigma_0|)$.

\vspace{-2ex}
\paragraph{Perfect Completeness.} Fix any $w \in \mc{L}$. As $\mc{D}$ has perfect completeness, there are assignments $T_1: A \to \Sigma_A$ and $T_2: B \to \Sigma_B$ which satisfy all of the constraints in $\mc{D}$ and decode to $w_t$ with probability $1$ under $\mc{P}_t$ for every $t \in [n]$. Then we construct assignments $T'_1: B' \to \Sigma_A^{d} \times \Sigma_0$ and $T'_2: A \to \Sigma_A$ in the natural way as follows. Set $T'_2 = T_1$. For each $(t,b,i) \in B'$ write $\Gamma(b) = \{a_1, \ldots, a_d \}$, and take $T'_1[(t,b,i)] = (T_1[a_1],\ldots,T_1[a_d], D_t(a_1,T_1[a_1])$. In words, this is the symbol in $\Sigma_A^{d} \times \Sigma_0$ which contains $T_1[a_i]$ as the label for $a_i$ for all $i \in [d]$, and $D_t(a_i, T_1[a_i])$ as the label for $t$ (we note that by the hardcoded constraints it does not matter which $a_i$ we pick). Note that by the fact that $T_1$ and $T_2$ satisfy all of the constraints in $\mc{D}$ and always decode to $w_t$ under $\mc{P}_t$, this assignment indeed satisfies the constraints on the left alphabet and is a valid assignment. It is straightforward to verify that these assignments satisfy all of the constraints in $\mc{D}'$ and decode to $w_t$ with probability $1$ under $\mc{P}'_t$ for every $t \in [n]$, so $\mc{D}'$ has perfect completeness.
\vspace{-2ex}
\paragraph{List-Decoding Soundness.} Fix any left assignment $T'_1: B' \to \Sigma_A^{d} \times \Sigma_0$. 
Based on $T_1'$, we construct a candidate list of words from $\mc{L}$ and show that it satisfies the list-decoding soundness. Recall that for each $(t,b,i) \in B'$ and $a \in A$ adjacent to it in $\mc{D}'$, the assignment $T'_1[(t,b,i)]$ contains a symbol for $a$, which we denote by $T'_1[(t,b,i)]_a$. For each $a \in A$, let ${\sf cand}(a)$ consist of all symbols $\sigma \in \Sigma_A$ such that 
\[
\Pr_{(t, (t,b,i)) \sim \mc{Q'}(\cdot, a, \cdot)}[T'_1[(t,b,i)]_a = \sigma]  = \Pr_{(t,b) \sim \mc{Q}(\cdot, a, \cdot), i \in [C_{t,b}]}[T'_1[(t,b,i)]_a = \sigma] \geq \eps^{1/5}.
\]
where the first equality is by the observation regarding how to generate $\mc{Q}'$ from $\mc{Q}$ made in the definition of $\mc{D}'$ above.

Clearly, $|{\sf cand}(a)|\leq \eps^{-1/5}:=J$ for every $a\in A$. Therefore, one can find assignments $T_{1,1}, \ldots, T_{1,J}: A \to \Sigma_A$ such that for every $a \in A$ and every $\sigma \in {\sf cand}(a)$ there is some $j \in [J]$ such that $T_{1,j}[a] = \sigma$. Since each $T_{1, j}$ is a left assignment to $\mc{D}$, we can apply the list-decoding soundness of $\mc{D}$ to it. Specifically, for each $j \in [J]$, let $\List_j = \{w_{j,1}, \ldots, w_{j,L} \} \subseteq \mc{L}$ be the list guaranteed by the $(L, \eps)$-list-decoding soundness of $\mc{D}$ for the left assignment $T_{1,j}$. Then by the list-decoding soundness of $\mc{D}$ we conclude that for any right assignment $T_2: B \to \Sigma_B$ and any $j \in [J]$, we have
\begin{equation}\label{eq: flip sides outer list dec soundness}
\Pr_{(t,a,b) \sim \mc{Q}}[\Phi_{(a,b)}(T_{1,j}[a], T_2[b]) = 1 \land D_t(a, T_{1,j}[a]) \notin \{w_{j,1}[t], \ldots, w_{j,L}[t] \}] \leq \eps.
\end{equation}
We show that $\List = \bigcup_{j \in [J]} \List_j$ satisfies the list-decoding soundness condition of $\mc{D}'$ with size $J \cdot L = L \cdot \poly(1/\eps)$ and error $2\eps^{1/5}$. The size is clear, and next we show that for any right assignment $T'_2: A \to \Sigma_A$, the list-decoding error of $T_1'$, $T_2'$ relative to $\List$ is small. We define the following events:
\begin{itemize}
    \item Let $\mc{E}_{(t,b,i),a}$ be the event that the constraint on $((t,b,i), a)$ is satisfied by assignments $T'_1[(t,b,i)]$ and $T'_2[a]$.
    \item Let $\mc{F}_{j, (t,b,i),a}$ be the event that $D'_t((t,b,i), T'_1[(t,b,i)]) \notin \{(w_{j,1})_t, \ldots, (w_{j,L})_t \}$.
\end{itemize} 
It is clear that the list-decoding error can be expressed and bounded as follows
\begin{equation} \label{eq: decoding degree reduction list dec error}
    \begin{split}
      &\Pr_{(t,(t,b,i),a) \sim \mc{Q}'}\left[\mc{E}_{(t,b,i),a} \bigwedge \left(\bigwedge_{j \in [J]}\mc{F}_{j, (t,b,i),a}\right) \right]\\
       &= \Pr_{\substack{(t,a,b) \sim \mc{Q}, i \in [C_{t,b}]}}\left[\mc{E}_{(t,b,i),a} \bigwedge \left(\bigwedge_{j \in [J]}\mc{F}_{j, (t,b,i),a}\right) \right] \\
      &=  \Pr_{(t,a,b) \sim \mc{Q}, i \in [C_{t,b}]}\left[\mc{E}_{(t,b,i),a} \bigwedge \left(\bigwedge_{j \in [J]}\mc{F}_{j, (t,b,i),a}\right) \land \left( T'_{1}[(t,b,i)]_a \in {\sf cand}(a)\right) \right] \\
     &+ \Pr_{(t,a,b) \sim \mc{Q}, i \in [C_{t,b}]}\left[\mc{E}_{(t,b,i),a} \bigwedge \left(\bigwedge_{j \in [J]}\mc{F}_{j, (t,b,i),a}\right) \land \left( T'_{1}[(t,b,i)]_a \notin {\sf cand}(a)\right) \right] \\
      &\leq  \Pr_{(t,a,b) \sim \mc{Q}, i \in [C_{t,b}]}\left[\mc{E}_{(t,b,i),a} \bigwedge \left(\bigwedge_{j \in [J]}\mc{F}_{j, (t,b,i),a}\right) \land \left( T'_{1}[(t,b,i)]_a \in {\sf cand}(a)\right) \right]  \\
     &+ \Pr_{(t,a,b) \sim \mc{Q}, i \in [C_{t,b}]}\left[\mc{E}_{(t,b,i),a} \land \left( T'_{1}[(t,b,i)]_a \notin {\sf cand}(a)\right) \right]\\
     &\leq \sum_{j \in [J]} \Pr_{(t,a,b) \sim \mc{Q}, i \in [C_{t,b}]}[\mc{E}_{(t,b,i),a} \land\mc{F}_{j, (t,b,i),a} \land \left(T'_{1}[(t,b,i)]_a = T_{1, j}[a]\right) ] \\
     &+ \Pr_{(t,a,b) \sim \mc{Q}, i \in [C_{t,b}]}\left[\mc{E}_{(t,b,i),a} \land \left( T'_{1}[(t,b,i)]_a \notin {\sf cand}(a)\right) \right]
    \end{split}
\end{equation}
In the last transition we used the union bound, and we next bound the terms on the right hand side. We upper bound the second term on the right hand side of~\eqref{eq: decoding degree reduction list dec error} as
\begin{equation} \label{eq: single prob}
\begin{split}
&\Pr_{(t,a,b) \sim \mc{Q}, i \in [C_{t,b}]}\left[\mc{E}_{(t,b,i),a} \land \left( T'_{1}[(t,b,i)]_a \notin {\sf cand}(a)\right) \right] 
\\
&\leq\E_{a}\left[
1_{T_2'[a]\not\in {\sf cand}(a)}
\Pr_{(t,b)\sim\mc{Q}(\cdot,a,\cdot), i \in [C_{t,b}]}[T_1'[(t,b,i)]_a=T_2'[a]] 
\right]\\
&\leq \eps^{1/5},
\end{split}
\end{equation}
where in the last inequality we used the definition of ${\sf cand}(a)$.

We now go on to bound the first term on the right hand side of~\eqref{eq: decoding degree reduction list dec error}. We bound the summand corresponding to each $j \in [J]$ separately, and we fix a $j$ henceforth. Define 
\begin{itemize}
    \item $\mc{E}'_{j, (t,b,i),a}$ be the event $\mc{E}_{(t,b,i),a} \land\mc{F}_{j, (t,b,i),a} \land \left( T_{1, j}[a] = T'_1[(t,b,i)]_a\right)$,
\end{itemize}
so our task is now to upper bound $\Pr_{(t,a,b) \sim \mc{Q}, i \in [C_{t,b}]}[\mc{E}'_{j, (t,b,i),a}]$. For any constraint of $\mc{D}$, say $(a,b) \in E$, and label $\sigma$ for $a$, we denote by $\sigma_{\to b}$ the unique label for $b$ such that $\Phi_{(a,b)}(\sigma, \sigma_{\to b}) = 1$. Here, the uniqueness follows since $\mc{D}$ is a projection dPCP. For each $(t,b,i) \in B'$, let $T'_{1}[(t,b,i)]_{\to b}$ be the label deduced for $b$ in $T'_{1}[(t,b,i)]$. In particular, by the hardcoded alphabet constraints we have that $T'_{1}[(t,b,i)]_{\to b} = \left(T'_{1}[(t,b,i)]_a\right)_{\to b}$ for any $a \in \Gamma(b)$.

We say a symbol $\sigma \in \Sigma_B$ is $b$-heavy if 
\begin{equation} \label{eq: heavy def} 
\Pr_{(t,a) \sim \mc{Q}(\cdot, \cdot, b)}[T_{1,j}[a]_{\to b} = \sigma] \geq \sqrt{\eps}.
\end{equation}

It is clear that for every $b \in B$, there are at most $J' = 1/\sqrt{\eps}$ many $b$-heavy symbols, hence we can find tables $H_1, \ldots, H_{J'}: B \to \Sigma_B$ such that for every $b$-heavy symbol $\sigma$, there is some $j' \in [J']$ such that $H_{j'}[b] = \sigma$. Thus, we may write
\begin{equation} \label{eq: heavy split new}
    \begin{split}
\Pr_{(t,a,b) \sim \mc{Q}}[\mc{E}'_{j, (t,b,i),a}] &= \Pr_{(t,a,b) \sim \mc{Q}}[\mc{E}'_{j, (t,b,i),a} \land (T'_{1}[(t,b,i)]_{\to b} \text{ is } b\text{-heavy})] \\
&+ \Pr_{(t,a,b) \sim \mc{Q}}[\mc{E}'_{j, (t,b,i),a} \land (T'_{1}[(t,b,i)]_{\to b}  \text{ is not } b\text{-heavy})].
\end{split}
\end{equation}
To bound the first probability, notice that if $\mc{E}'_{j, (t,b,i),a} \land (T'_{1}[(t,b,i)]_{\to b} \text{ is } b\text{-heavy})$ occurs, then for some $j' \in [J']$, both events below occur:
\begin{itemize}
\item the assignments to $a$ and $b$ given by $T_{1,j}[a]$ and $H_{j'}[b]$ respectively satisfy the constraint $\Phi_{(a,b)}$,
\item $D_t(a, T_{1,j}[a]) \notin \{(w_{j,1})_t, \ldots, (w_{j,L})_t \}$.
\end{itemize}
The first item requires unpacking the event $\mc{E}'_{j, (t,b,i),a}$. Indeed, if this event holds then $T'_{1}[(t,b,i)]_a = T_{1, j}[a]$, which implies $T'_{1}[(t,b,i)]_{\to b} = T_{1, j}[a]|_{\to b}$. Hence, $ T_{1, j}[a]|_{\to b} \in \{H_1[b], \ldots ,H_{J'}[b] \}$ because $T'_{1}[(t,b,i)]_{\to b}$ is heavy. Now, to bound the first probability in \eqref{eq: heavy split new}, notice that for a fixed $j' \in [J']$, the above two items occur simultaneously with probability at most $\eps$. Indeed, this follows by \eqref{eq: flip sides outer list dec soundness} applied to $T_{1,j}$ and $H_{j'}$. Union bounding over all $j' \in [J']$ gives that 
\[
\Pr_{(t,a,b) \sim \mc{Q}}[\mc{E}'_{j, (t,b,i),a} \land (T'_{1}[(t,b,i)]_{\to b} \text{ is } b\text{-heavy})]  \leq J' \cdot \eps.
\]
We move on to bounding the second probability in \eqref{eq: heavy split new}, and here we will use the agnostic property of $\mc{Q}$. We have
\begin{equation}
    \begin{split}
    &\Pr_{(t,a,b) \sim \mc{Q}, i \in [C_{t,b}]}[\mc{E}'_{j, (t,b,i),a} \land (T'_{1}[(t,b,i)]_{\to b} \text{ is not } b\text{-heavy})] \\
    &\leq \Pr_{(t,a,b) \sim \mc{Q}, i \in [C_{t,b}]}[T_{1,j}[a]_{\to b} = T'_{1}[(t,b,i)]_{\to b} \land T'_{1}[(t,b,i)]_{\to b} \text{ is not } b\text{-heavy}] \\
    &= \E_{(t,b) \sim \mc{Q}, i \in [C_{t,b}]}\left[1_{T'_{1}[(t,b,i)]_{\to b} \text{ is not } b\text{-heavy}}\Pr_{a \sim \mc{Q}(t, \cdot, b)}[T_{1,j}[a]_{\to b} = T'_{1}[(t,b,i)]_{\to b}] \right] \\
    &= \E_{(t,b) \sim \mc{Q}, i \in [C_{t,b}]}\left[1_{T'_{1}[(t,b,i)]_{\to b} \text{ is not } b\text{-heavy}}\Pr_{a \sim \mc{Q}(\cdot, \cdot, b)}[T_{1,j}[a]_{\to b} = T'_{1}[(t,b,i)]_{\to b} ] \right]\\
    &\leq \sqrt{\eps}.
    \end{split}
\end{equation}
The first transition is by definition of $\mc{E}'_{j, (t,b,i),a}$, the second transition is clear, the third transition uses the agnostic property of $\mc{Q}$, and the last transition uses 
the fact that the distribution of $a$ in the probability of the fourth line is the same as the distribution of $a$ in \eqref{eq: heavy def} and hence the probability in the expectation of the fourth line is always at most $\sqrt{\eps}$ by definition of not being $b$-heavy.

Altogether, this yields
\[
\Pr_{(t,a,b) \sim \mc{Q}}[\mc{E}'_{j, (t,b,i),a}]  \leq  J' \cdot \eps + \sqrt{\eps}.
\]
Hence, going back to the list-decoding error expression in \eqref{eq: decoding degree reduction list dec error}, we get that each term in the summation of the last expression is at most $ J' \cdot \eps + \sqrt{\eps}$, and combined with \eqref{eq: single prob} we get that the list-decoding error is at most 
\[
J \cdot (J' \cdot \eps + \sqrt{\eps}) + \eps^{1/5} \leq 2\eps^{1/5}.
\]
Hence $\mc{D}'$ has $(J \cdot L, 2\eps^{1/5})$-list-decoding soundness, where $J \cdot L = L \cdot \poly(1/\eps)$. 
\end{proof}

\subsection{An Inner dPCP}

Now we describe how to find a version of our dPCP that will be used as the inner dPCP for composition. Below, we highlight a few technical properties of this inner PCP construction which will be used later on in the composition step. 

\begin{itemize}
    \item The decoding degree is $1$.
    \item There is some integer $J$ such that for every pair of vertices $(a,b)$, the weight given by the complete decoding distribution can be written as 
     $\mc{Q}(\circ, a, b) = \frac{w(a,b)}{J}$ for some integer $w(a,b)$. Furthermore, for each $a$, we have that the sum $\sum_{b \in B} w(a,b)$ is the same. This latter condition also implies that the marginal of the complete decoding distribution over the left vertices is uniform.
\end{itemize}

We will obtain these properties by using \cref{lm: decoding degree reduction}, which gives a dPCP with decoding degree $1$, thereby establishing the first item above. To achieve the second item, we carefully choose the parameters $C_{t,b}$
so as to adjust the complete decoding distribution, following an approach similar to that of \cref{lm: discretize}.

\begin{lemma} \label{lm: left uniformity}
Suppose the language $\mc{L} \subseteq \Sigma_0^n$ has a projection dPCP
\[
\mc{D} = (A \cup B, E, \Sigma_A, \Sigma_B, \{\Phi_{e} \}_{e \in E}, \{D_t \}_{t \in [n]}, \{ \mc{P}_{t}\}_{t\in [n]} )
\]
 with perfect completeness, $(L, \eps)$-list-decoding soundness, agnostic complete decoding distribution,  projection decision complexity $\mathsf{ProjComp}$, decoding complexity $\mathsf{DecodeComp}$, and left degree, decoding degree, right degree that are at most $\mathsf{ld}, \mathsf{ldd}, d$ respectively.
 Then for every integer $M \geq |B|$, $\mc{L}$ has a projection dPCP
\[
\mc{D}' = (Z \cup A, E', \Sigma^{d}_A \times \Sigma_0, \Sigma_A, \{\Phi'_{e} \}_{e \in E'}, \{D'_t \}_{t \in [n]}, \{ \mc{P}'_t \}_{t\in [n]} )
\]
which satisfies the following properties:
\begin{itemize}
    \item \textbf{Length.} The left side has the form  $Z = [n] \times [M]$, so the length is $n\cdot M + |A|$. 
    \item \textbf{Degrees.} The decoding degree is $1$ and the left degrees are at most $d$. Furthermore if $\mc{D}$ is right-regular, then $\mc{D}'$ is left regular.
     \item \textbf{Projection Decision Complexity.} $O(d \cdot(\mathsf{ProjComp} + \mathsf{DecodeComp}))$.
        \item \textbf{Decoding Complexity.} $O(d\log |\Sigma_A| + \log |\Sigma_0|)$. 
         \item \textbf{Decoding Distribution.} For each $t \in [n]$, the marginal of $\mc{P}'_t$ over the left vertices, $Z$, is uniform in $\{t\} \times [M]$. 
    \item \textbf{Completeness.} The dPCP has perfect completeness.
    \item \textbf{List-Decoding Soundness.} $\left(L\cdot \poly(1/\eps), 2\eps^{1/5} + \frac{|B|}{M}\right)$-list-decoding soundness.
\end{itemize}
\end{lemma}
\begin{proof}
Fix $M \geq |B|$ and take $\mc{D}$ to be the dPCP from the lemma statement. We start by using \cref{lm: discretize} to obtain a discretized version of the marginal of $\mc{P}_{t}$ over $B$ for each $t \in [n]$. We apply \cref{lm: discretize} with denominator there set to $M$ for each $t \in [n]$, which yields the following. For each $t \in [n]$, there are nonnegative integers $z'_{t,b}$ for each $b \in B$ such that $M := \sum_{b \in B} z'_{t,b}$, and
\begin{equation} \label{eq: discretize error}
    \left| z'_{t,b} - M \cdot  \mc{P}_{t}(\circ, b) \right| \leq 1,
\end{equation}
for every $b \in B$. We let $\mc{R}_t$ be the distribution over $B$ given by 
\[
\mc{R}_t(b) = \frac{z'_{t,b}}{M}.
\]
Note that $\mc{R}_t$ is the discretized version of the marginal of $\mc{P}_t$ over $B$ given by \cref{lm: discretize}.

Apply \cref{lm: decoding degree reduction} to $\mc{D}$ with the nonnegative integers $C_{t,b}$ therein set to $C_{t,b} = z'_{t,b}$ to get a dPCP of the form
\[
\mc{D}_2 = \left(\left(\cup_{t \in [n]} S_t \right) \cup A, E_2, \Sigma^{d}_A \times \Sigma_0, \Sigma_A, \{\Phi_{e} \}_{e \in E_2}, \{D_{2,t} \}_{t \in [n]}, \{ \mc{P}_{2,t} \}_{t\in [n]} \right),
\]
where $S_t = \bigcup_{b \in B} \{(t,b)\} \times [z'_{t,b}]$ for each $t\in[n]$. Note that the sets $S_t$ are disjoint  and that when defining $S_t$, if $z'_{t,b} = 0$ we think of $\{(t,b)\} \times [z'_{t,b}]$ as the empty set. Furthermore, for each $t$, 
\begin{equation}
    |S_t| = \sum_{b \in B} z'_{t,b} = M.
\end{equation}
By \cref{lm: decoding degree reduction}, $\mc{D}_2$ has perfect completeness, $(L \cdot \poly(1/\eps), 2\eps^{1/5})$-list-decoding soundness, and decoding degree $1$. Likewise, we get that if $\mc{D}$ is right-regular, then $\mc{D}_2$ is left-regular.

Next, we will modify the decoding distributions of $\mc{D}_2$ to obtain our final dPCP, $\mc{D}'_2$. Our goal is to obtain a dPCP whose complete decoding distribution has uniform marginal over the left vertices (and we will later show that this implies the decoding distribution guarantee of the theorem statement). Let $\mc{Q}_2$ be the complete decoding distribution of $\mc{D}_2$. By the decoding distribution guarantee of \cref{lm: decoding degree reduction}, we have that 
\[
\mc{Q}_2(t, (t,b,i), \circ) = \frac{1}{n} \cdot \frac{\mc{P}_t(\circ, b)}{z'_{t,b}} = \frac{1}{n} \cdot \frac{p_{t,b}}{z'_{t,b}}
\]
The above shows that the marginal of $\mc{Q}_2$ over the left vertices is almost uniform, and hence we will only need to make a small modification.

We note that for each $t, b$ and $i$, the distribution of $a \sim \mc{Q}_2(t, (t,b,i), \cdot) $ is equal to $\mc{Q}(t, \cdot , b)$. Indeed, for any $a \in A$, using the decoding distribution guarantee of \cref{lm: decoding degree reduction} we have
\[
\frac{Q_2(t, (t,b,i), a)}{\mc{Q}_2(t, (t,b,i), \circ)} = \frac{\mc{Q}(t, a, b) / C_{t,b}}{\mc{Q}(t, \circ, b)/C_{t,b}} = \frac{Q(t,a,b)}{Q(t, \circ, b)},
\]
where in the second transition we used the fact that (using the decoding distribution guarantee of \cref{lm: decoding degree reduction} again)
\[
\mc{Q}_2(t, (t,b,i), \circ) = \sum_{a \in A} \mc{Q}_2(t, (t,b,i), a) =  \sum_{a \in A} \frac{\mc{Q}(t,a,b)}{C_{t,b}} = \frac{\mc{Q}(t,\circ,b)}{C_{t,b}}.
\]

On the other hand, consider the distribution $\mc{Q}'_{2}$, generated by choosing $t \in [n]$ uniformly, $b \sim \mc{R}_t$, $a \sim \mc{Q}(t, \cdot, b)$, $i \in [z'_{t,b}]$ uniformly and outputting $(t, (t,b,i), a)$. Then for any $t$ and $(t,b,i)$, we have
\[
\mc{Q}'_{2}(t, (t,b,i), \circ) = 
\frac{1}{n}\cdot\frac{z_{t,b}'}{M}\cdot\frac{1}{z_{t,b}'}
=
\frac{1}{n} \cdot \frac{1}{M}.
\]
Now we show that $\mc{Q}'_{2}$ is close to $\mc{Q}_2$ in TV-distance. First note that for every $t,b$ and $i$, by definition the distributions $\mc{Q}_2(t, \cdot, b) = \mc{Q}'_2(t, (t,b,i), \cdot)$, which in turn implies that, $\mc{Q}_2(t, (t,b,i), \cdot)$ and $\mc{Q}'_2(t, (t,b,i), \cdot)$ are identical. Therefore, the TV-distance between $\mc{Q}_2$ and $\mc{Q}'_{2}$ is equal to 
the TV-distance between their marginal on $(t,b,i)$. 
It follows that
\begin{align*}
   \tv(\mc{Q}_{2}, \mc{Q}'_2)
   \leq \sum_{t \in [n], b \in B, i \in [z'_{t,b}]} \left| \mc{Q}_{2}(t, (t,b,i), \circ) - \mc{Q}'_2(t, (t,b,i), \circ)\right|  
   \leq  \sum_{t \in [n], b \in B, i \in [z'_{t,b}]} \left|\frac{1}{n}\left( \frac{p_{t,b}}{z'_{t,b}} - \frac{1}{M}  \right)\right|.  
\end{align*}
Since the summation is only taken over $t,b$ such that $z'_{t,b} \geq 1$, so we can bound each summand there as
\[
\frac{1}{n}\left|\frac{p_{t,b}}{z'_{t,b}} - \frac{1}{M} \right|=
\frac{1}{n}\left| \frac{Mp_{t,b}-z'_{t,b}}{Mz'_{t,b}} \right| \leq \frac{1}{nMz_{t,b}'},
\]
where we used \eqref{eq: discretize error}. Hence, putting everything together, we get 
\begin{equation} \label{eq: tv distance in inner}
    \tv(\mc{Q}_{2}, \mc{Q}_2')\leq \frac{|B|}{M}.
\end{equation}

Our final dPCP, which we call $\mc{D}'_2$, is obtained by taking $\mc{D}_2$ and changing its complete decoding distribution to $\mc{Q}'_{2}$. In particular all of the parts of $\mc{D}'_2$ are the same as in $\mc{D}_2$ except for the decoding distributions. We have,
\[
\mc{D}'_2 = \left(\left(\cup_{t \in [n]} S_t \right) \cup A, E_2, \Sigma^{d}_A \times \Sigma_0, \Sigma_A, \{\Phi'_{e} \}_{e \in E_2}, \{D_{2,t} \}_{t \in [n]}, \{\mc{P}'_{2,t}\}_{t \in [n]}\right),
\]
where for each $t \in [n]$, the new decoding distribution $\mc{P}'_{2,t} $ is the distribution over $E_2$ given by $\mc{Q}'_{2}(t, \cdot, \cdot)$. 

It is clear that $\mc{D}'_2$ still has perfect completeness and the degree related parameters, projection decision complexity, and decoding complexity are the same as in $\mc{D}_2$. By \eqref{eq: tv distance in inner}, $\mc{D}'_2$ has $(L \cdot \poly(1/\eps), 2\eps^{1/5} + |B|/M)$-list decoding soundness.
Finally, for each $t \in [n]$, it can be checked that the decoding distribution, $\mc{P}'_{2,t}$, chooses a vertex in $S_t$ uniformly at random, and furthermore $|S_t| = M$. Hence, we can view the left vertices as $[n] \times [M]$ and the decoding distribution $\mc{P}'_{2,t}$ is indeed as described in the lemma. 
\end{proof}

We now describe how to obtain the inner dPCP that we will use for composition.

\begin{theorem}[Inner dPCP] \label{thm: inner dpcp}
    For every $\eps > 0$ there exist $L = {\sf poly}(1/\eps)$ and functions $M, K: \mathbb{N} \times \mathbb{N} \to \mathbb{N}$ satisfying
    \[
    M(T, N) = O_T(N \poly_{\eps}(\log N)) \quad \text{and} \quad K(T, N) = O_T( \poly_{\eps}(\log N))
    \]
    such that the following holds. For every alphabet $\Sigma_0$ of size at most $T$ and circuit $\varphi: \Sigma_0^n \to \{0,1\}$ of size at most $N$, the language $\sat(\varphi)$ has a projection dPCP of the form
    \[
    \mc{D} = \left(\left([n] \times [M]\right)\cup B, E, \Sigma_A, \Sigma_B, \{\Phi_{e}\}_{e\in E}, \{D_t\}_{t \in [n]}, \{\mc{P}_t \}_{t\in[n]} \right)
    \]
    which satisfies the following with $M:= M(T, N)$ and $K := K(T,N)$ 
     \begin{itemize}
        \item \textbf{Length.} $O_{T}(n M + N\poly_{\eps}(\log N))$.
        \item \textbf{Alphabet Sizes.} The left and right alphabet sizes are $2^{\left(O_{T}(\poly_{\eps}(\log N))\right)}$ and $O_{\eps}(1)$ respectively.
        \item \textbf{Degrees.} $\mc{D}$ has decoding degree $1$.
        \item \textbf{Projection Decision Complexity.} $ O_{|T|}(\poly_{\eps}(\log N))$.
        \item \textbf{Decoding Complexity.} $O_{|T|}(\poly_{\eps}(\log N))$.
        \item \textbf{Decoding Distribution.} Let $\mc{Q}$ be the complete decoding distribution. Then, $\mc{Q}$ satisfies the following:
        \begin{itemize}
            \item For each $(t,r) \in [n] \times [M]$ and $b \in B$ we can write $\mc{Q}(\circ, (t,r), b) = \frac{w((t,r),b)}{nMK}$ for some  $w((t,r),b) \in \mathbb{N}$. 
            \item For each $t \in [n]$, the distribution $\mc{Q}(t, \cdot, \circ)$ is uniform over $\{t\} \times [M]$ and hence for every $(t,r) \in [n] \times [M]$, 
            \[
            \sum_{b \in B} w((t,r),b) = K
            \]
        \end{itemize}
    
        \item \textbf{Completeness.} The dPCP has perfect completeness.
        \item \textbf{List-Decoding Soundness.} The dPCP has $(L, \eps)$-list-decoding soundness.
    \end{itemize}
\end{theorem}
\begin{proof}
     Fix $\eps > 0$, fix sizes $T$ and $N$ and let $\varphi: \Sigma_0^n \to \{0,1\}$ be a circuit of size at most $N$ with $|\Sigma_0| \leq T$. We take $\mc{D}_{\dip}$ to be the dPCP for $\sat(\varphi)$ in \cref{thm: dp dpcp} with soundness parameter set to $\eps$:
     \[
     \mc{D}_{\dip} = (X(k), X(\sqrt{k}), E_{\dip}, \Sigma^k, \Sigma^{\sqrt{k}}, \{\Phi_{e} \}_{e \in E_{\dip}}, \{D_{\dip,t} \}_{t \in [n]}, \{ \mc{P}_{\dip, t}\}_{t\in [n]} ).
     \]
    Now set $M = O_T(N \poly_{\eps}(\log N))$ sufficiently large relative to $|X\left(\sqrt{k}\right)| = O_T(N \poly_{\eps}(\log N))$. We also have that there is some $d_{r} = \poly_{\eps}(\log N)$ such that regardless of the initial circuit $\varphi$, the right degree of every vertex is at most $d_{r}$. We will transform $\mc{D}_{\dip}$ into the desired inner dPCP.

      \paragraph{Decoding Degree Reduction + Right Alphabet reduction.} We first apply \cref{lm: left uniformity} with $M$ set as above, and then apply right alphabet reduction from \cref{lm: right alphabet reduction} with sufficiently small $\eta = \poly(\eps)$ on the resulting dPCP. Call the dPCP after both transformations $\mc{D}_2$ and let its constraint graph be $([n] \times [M] \cup B, E)$. Note that we can assume such structure on the left vertices by \cref{lm: left uniformity}. Call the resulting complete decoding distribution $\mc{Q}_2$. Let us list a few properties we have that are straightforward to verify from \cref{lm: left uniformity} and \cref{lm: right alphabet reduction}.
      \begin{itemize}
            \item \textbf{Length.} $|B| = |X(k)| \cdot O_{\eta}(\sqrt{k}\log |\Sigma|) = O_{T}(N \poly_{\eps}(\log N))$.
            This follows from the length guarantees of \cref{lm: left uniformity} and the structure of the resulting constraint graph from \cref{lm: right alphabet reduction}. The length is therefore $nM + |B| = O_T(nM + N\poly_{\eps}(\log N))$.
            \item \textbf{Alphabets.} The left alphabet has size at most $|\Sigma|^{k \cdot d_r} \cdot T = 2^{O_T(\poly_{\eps}(\log N))}$ after applying \cref{lm: left uniformity}, and this is unaffected by \cref{lm: right alphabet reduction}. The right alphabet has size at most $O_{\eps}(1)$ by \cref{lm: right alphabet reduction}.
            \item \textbf{Degrees.} The decoding degree is $1$ after applying \cref{lm: left uniformity} and this is unaffected by \cref{lm: right alphabet reduction}.  The degree of each left vertex is at most $O_{T}(\poly_{\eps}(\log N))$. This follows from the left degree guarantee of \cref{lm: left uniformity} (which turns right degree of $\mc{D}_{\dip}$ into left degree) and the fact that \cref{lm: right alphabet reduction} multiplies the left degree by $O_{\eta}(\sqrt{k}\log |\Sigma|) = O_{T}(\poly_{\eps}(\log N))$. We refer to \cref{thm: dp dpcp} to bound $\sqrt{k} \log |\Sigma| \leq O_{T}(\poly_{\eps}(\log N))$.
            \item \textbf{Projection Decision Complexity.} $O_T(\poly_{\eps}(\log N))$. This is straightforward to check by using \cref{thm: dp dpcp}, then \cref{lm: left uniformity}, and finally \cref{lm: right alphabet reduction}.
            \item \textbf{Decoding Complexity.} $O_T(\poly_{\eps}(\log N))$. This is straightforward to check by using \cref{thm: dp dpcp}, then \cref{lm: left uniformity}, and finally \cref{lm: right alphabet reduction}.
          \item \textbf{Decoding Distribution.} The marginal of $\mc{Q}_2$ over $[n] \times [M]$ is uniform and each left vertex $(t, r) \in [n] \times [M]$ has exactly one decoding neighbor, $t$. This follows from the decoding distribution and decoding degree guarantees of \cref{lm: left uniformity},
          \item \textbf{Completeness.} $\mc{D}_2$ has perfect completeness because $\mc{D}_{\dip}$ does and both \cref{lm: left uniformity} and \cref{lm: right alphabet reduction} preserve it.
          \item \textbf{Soundness.} $\mc{D}_2$ has $(L \cdot \poly(1/\eps), O(\eps^{1/5}))$ soundness by \cref{thm: dp dpcp}, \cref{lm: left uniformity}, and \cref{lm: right alphabet reduction}. Here we use the fact that $M$ is sufficiently large relative to $|X(\sqrt{k})|$.
      \end{itemize}
    At this point, we have all of the desired properties, except for the discreteness of the marginal of the complete decoding distribution over the edges (i.e., the first guarantee of the decoding distribution above). We resolve this next, and some care is needed to accomplish this while also preserving a uniform marginal over the left vertices.
    \vspace{-2ex}
    \paragraph{Discretizing the Complete Decoding Distribution in a Left Uniform Manner.}
    We construct a discretized version of the distribution $\mc{Q}_2(t, (t,r), \cdot)$ for each $(t,r) \in [n] \times [M]$. Specifically, note that $|\supp(\mc{Q}_2(t, (t,r), \cdot))|$ is upper bounded by the degree of $(t,r)$, which is $O_{T}(\poly_{\eps}(\log N))$. We apply \cref{lm: discretize} with denominator set to a sufficiently large $M' = O_{T}(\poly_{\eps}(\log N))$ and let  $\Delta_{t,r, \disc}$ be the distribution obtained. Here we again note that we can choose the same $M'$ regardless of the specific circuit $\varphi$ fixed at the start. By \cref{lm: discretize} and the fact that $M'$ is sufficiently large, we have 
    \begin{equation} \label{eq: inner tv bound}
        \tv(\Delta_{t,r,\disc}, \mc{Q}_2(t, (t,r), \cdot)) \leq \eps 
    \end{equation}
    and $\supp(\Delta_{t,r,\disc}) \subseteq \supp( \mc{Q}_2(t, (t,r), \cdot))$. Now consider the following complete decoding distribution $\mc{Q}'_2$:
    \begin{itemize}
        \item Choose $(t, r) \in [n] \times [M]$ uniformly. 
        \item Choose $b \sim \Delta_{t,r,\disc}$.
        \item Output $(t, (t,r), b)$, 
    \end{itemize}
    and let $\mc{D}'_2$ be the dPCP obtained by changing the complete decoding distribution of $\mc{D}_2$ from $\mc{Q}_2$ to $\mc{Q}'_2$.
    
    Since  $\mc{Q}'_2$ and $\mc{Q}_2$ have the same marginal distribution over $[n] \times ([n] \times [M])$, we have 
    \begin{equation}  \label{eq: inner tv bound final}
    \begin{split}    
    \tv(\mc{Q}'_2, \mc{Q}_2) &=  \sum_{(t,r) \in [n] \times [M]}\mc{Q}_2(t, (t,r), \circ)\cdot  \tv(\mc{Q}'_2(t, (t,r), \cdot), \mc{Q}(t, (t,r), \cdot) ) \\
    &= \sum_{(t,r) \in [n] \times [M]}\mc{Q}_2(t, (t,r), \circ)\cdot  \tv(\Delta_{t,r, \disc}, \mc{Q}(t, (t,r), \cdot )) \\
    &\leq \eps,
    \end{split}
    \end{equation}
where the last transition uses \eqref{eq: inner tv bound}. We now verify that $\mc{D}'_2$ satisfies all of the desired properties. The length, degree, projection decision complexity, and decoding complexity are all unchanged from $\mc{D}_2$ and hence follow by the items above. Perfect completeness is preserved in $\mc{D}'_2$ because $\supp(\Delta_{t,r,\disc}) \subseteq \supp( \mc{Q}_2(t, (t,r), \cdot))$ as noted above, and $\mc{D}'_2$ still has $(L\cdot \poly(1/\eps), O(\eps^{1/5}))$-soundness because of \eqref{eq: inner tv bound final}. This shows the desired soundness condition because $\eps$ can be taken arbitrarily small and we can reparameterize the list size. Finally, we check the two properties of the decoding distribution. First, it is clear from the above that $\mc{Q}'_2$ satisfies that for each $t \in [n]$, the distribution $\mc{Q}'_2(t, \cdot, \circ)$ over $\{t\} \times [M]$ is uniform, so this takes care of the second condition there. For the first condition, we note that for any $(t,r) \in [n] \times [M], b \in B$ we can write
\[
\mc{Q}'_2(t, (t,r), b) = \frac{1}{n} \cdot \frac{1}{M} \cdot \Delta_{t,r, \disc}(b) = \frac{w((t,r), b)}{n M M'},
\]
for some $w((t,r), b) \in \mathbb{N}$. Hence, the first condition follows for $K(T, N) = M' = O_T(\poly_{\eps}(\log N))$.
\end{proof}

\subsection{An Outer dPCP}
The following result will be used as our outer dPCP. 
Compared to~\cref{lm: left uniformity}, the construction below has quasi-linear length and a complete decoding distribution whose marginal over the right vertices is uniform.

\begin{theorem} \label{thm: final dpcp before comp}
     For every $\eps > 0$ there exist $d', L = {\sf poly}(1/\eps)$ such that the following holds. For any alphabet $\Sigma_0$ and circuit $\varphi: \Sigma_0^n \to \{0,1\}$ of size at most $N$, the language $\sat(\varphi)$ has a projection dPCP of the form
    \[
    \mc{D} = (A\cup B, E, \Sigma_B^{\poly_{\eps}(\log N)} \times \Sigma_0, \Sigma_B, \{\Phi_{e}\}_{e\in E}, \{D_t\}_{t \in [n]}, \{\mc{P}_t \}_{t\in[n]})
    \]
    which satisfies the following
    \begin{itemize}
        \item \textbf{Length.} $O_{|\Sigma_0|}(N \cdot\poly_{\eps}(\log N))$.
        \item \textbf{Alphabet Sizes.} The left and right alphabet sizes are $2^{O_{|\Sigma_0|}(\poly_{\eps}(\log N))}$ and $O_{\eps}(1)$ respectively.
        \item \textbf{Degrees.} The dPCP has decoding degree $1$.
        \item \textbf{Projection Decision Complexity.} $O_{|\Sigma_0|}(\poly_{\eps}(\log N))$.
        \item \textbf{Decoding Complexity.} $O_{|\Sigma_0|}\left(\poly_{\eps}(\log N)\right)$.
        \item \textbf{Decoding Distribution.} The complete decoding distribution, $\mc{Q}$, satisfies
        \begin{itemize}
             \item The marginal of the complete decoding distribution over $A \times B$ is $K$-discrete for some 
             \[
             K = O_{|\Sigma_0|}(N\poly_{\eps}(\log N)).
             \]
            \item  Writing $\mc{Q}(\circ, a, b) = \frac{w(a,b)}{K}$ for integers $w(a,b)$, we have $\sum_{a \in A} w(a,b) = d'$ for all $b \in B$ and hence the marginal over the right vertices is uniform.
            \item  For each left vertex $a$, $\sum_{b \in B} w(a,b) \leq  O_{|\Sigma_0|}(\poly_{\eps}(\log N))$.
        \end{itemize}
        \item \textbf{Completeness.} The dPCP has perfect completeness.
        \item \textbf{List-Decoding Soundness.} The dPCP has $(L, \eps)$-list-decoding soundness.
    \end{itemize}
\end{theorem}
\begin{proof}
Fix $\eps > 0$, and let $\mc{D}_1 = (X(k), X(\sqrt{k}), E_1, \Sigma^k, \Sigma^{\sqrt{k}}, \{\Phi_{e} \}_{e \in E_{1}}, \{D_{1,t} \}_{t \in [n]}, \{ \mc{P}_{1, t}\}_{t\in [n]} )$ be the dPCP from \cref{thm: dp dpcp} with $(L, \eps)$-list-decoding soundness where $L = \poly(1/\eps)$. This dPCP has decoding degree, left degree, and right degree that are $\poly_{\eps}(\log N), O_{\eps}(1), \poly_{\eps}(\log N)$ respectively, and its complete decoding distribution, which we call $\mc{Q}_1$, is agnostic and has marginal over $X(k) \times X(\sqrt{k})$ which is $K$-discrete for some $K = N \poly_{\eps}(\log N)$. By the decoding distribution guarantee of \cref{thm: dp dpcp}, for each $\OutRightVertex \in X(\sqrt{k})$,  we have
\begin{equation} \label{eq: hdx disc} 
\mc{Q}_1(\circ,\circ, \OutRightVertex)  \leq \frac{\poly_{\eps}(\log N)}{N}.
\end{equation}
Starting from $\mc{D}_1$, we will first apply decoding degree reduction (\cref{lm: decoding degree reduction}) where all of the $C_{t,b}$ therein are set to $1$, and then apply right alphabet reduction (\cref{lm: right alphabet reduction}). We call the resulting dPCP $\mc{D}_2$, and obtain our final dPCP, which we call $\mc{D}_3$, by applying right degree reduction on $\mc{D}_2$.

\paragraph{Decoding Degree Reduction + Right Alphabet Reduction.} First, we note that $\mc{D}_1$ has agnostic complete decoding distribution by \cref{thm: dp dpcp}, so we may apply \cref{lm: decoding degree reduction} with the $C_{t,b}$ therein all set to $1$. We then apply \cref{lm: right alphabet reduction} with $\eta = O(\eps)$. The resulting dPCP,
\[
\mc{D}_2 = \left(A_2 \cup B_2, E_2, \Sigma_A, \Sigma_B, \{\Phi_{e} \}_{e \in E_2},\{D_{2,t} \}_{t\in [n]}, \{\mc{P}_{2,t} \}_{t \in [n]}\right),
\]
has the following properties, which are straightforward to verify.
\begin{itemize}
    \item \textbf{Completeness.} Perfect completeness.
    \item \textbf{Length.} We have $|A_2| \leq |X(\sqrt{k})| \cdot \poly_{\eps}(\log N)$ and $|B_2| \leq 
    O_{\eps}(|X(k)| \cdot \log(|\Sigma|^k))$, so the overall length is $O_{|\Sigma_0|}(N\cdot \poly_{\eps}(\log N))$.
    \item \textbf{Degrees.} The left degree is $O_{|\Sigma_0|}(\poly_{\eps}(\log N))$, the decoding degree is $1$, and the right degree is $\poly_{\eps}(\log N)$.
    \item \textbf{Projection Decision Complexity.} $O_{|\Sigma_0|}(\poly_{\eps}(\log N))$.
    \item \textbf{Decoding Complexity.} $O_{|\Sigma_0|}(\poly_{\eps}(\log N))$.
    \item \textbf{List-Decoding Soundness.} $(\poly(1/\eps), O(\eps^{1/5}))$-list-decoding soundness.
\end{itemize}

In order to apply right degree reduction next, we need to analyze the complete decoding distribution of $\mc{D}_2$, which we call $\mc{Q}_2$. Fix an arbitrary $a_2 \in A_2$ and let us bound $\mc{Q}_2(\circ, a_2, \circ)$. First, since the decoding degree is $1$, there is some $t$ such that 
\[
\mc{Q}_2(\circ, a_2, \circ) = \mc{Q}_2(t, a_2, \circ) = \mc{P}_{2,t}(a_2, \circ) 
\]
Then, by the decoding distribution guarantee of \cref{lm: decoding degree reduction}, there is some $\OutRightVertex \in X(\sqrt{k})$, such that 
\[
\mc{Q}_2(\circ, a_2, \circ) =  \mc{P}_{2,t}(a_2, \circ)  =  \mc{P}_{1,t}(\circ, \OutRightVertex) = \mc{Q}_1(t, \circ, \OutRightVertex).
\] 
Now applying \eqref{eq: hdx disc}, we have
\begin{equation} \label{eq: a not skewed in d2} 
\mc{Q}_2(\circ, a_2, \circ) = Q_1(t, \circ, \OutRightVertex) \leq Q_1(\circ, \circ, \OutRightVertex) \leq \frac{\poly_{\eps}(\log N)}{N}.
\end{equation}

\paragraph{Right Degree reduction.} Having analyzed the complete decoding distribution of $\mc{D}_2$, we are now ready to apply \cref{lm: gen right degree reduction} on it. Note that we need to apply the version of right degree reduction in \cref{lm: gen right degree reduction}, which changes the complete decoding distribution, because we did not analyze the discreteness of the marginal of $\mc{Q}_2$ over the edges.

To this end, set $$M = \Theta_{|\Sigma_0|}\left(\frac{N \poly_{\eps}(\log N)}{n}\right)$$ sufficiently large and degree parameter $d'$ sufficiently large relative to $\eps$ and apply \cref{lm: gen right degree reduction} with these parameters.
To see that we have set $M$ large enough to apply \cref{lm: gen right degree reduction}, we use the fact that $\mc{D}_2$ has left degree at most $O_{|\Sigma_0|}(\poly_{\eps}(\log N))$, and $|A_2| \leq N \poly_{\eps}(\log N)$.
Call the resulting dPCP
\[
\mc{D}_3 = \left(A_2 \cup B_3, E_3, \Sigma_A, \Sigma_B, \{\Phi_{e} \}_{e \in E_3},\{D_{2,t} \}_{t\in [n]}, \{\mc{P}_{3,t} \}_{t \in [n]}\right),
\]
and let $\mc{Q}_{3}$ be the complete decoding distribution of $\mc{D}_3$. 
By \cref{lm: gen right degree reduction} there are positive integers $w(a,b)$ for each $(a,b) \in A_2 \times B_3$ such that we can write  $\mc{Q}_{3}(\circ, a, b) = \frac{w(a,b)}{d'^2 n M}$. 
Fix a left vertex $a$ of $\mc{D}_3$. By the decoding distribution guarantee of \cref{lm: gen right degree reduction}, we have that either, 
\begin{equation} \label{eq: a not skewed} 
\mc{Q}_{3}(\circ, a, \circ)  \leq \mc{Q}_{2}(\circ, a, \circ) + \frac{|\Gamma_2(a)|}{nM} \leq \frac{O_{|\Sigma_0|}(\poly_{\eps}(\log N))}{N}
\end{equation}
or 
\begin{equation} \label{eq: a not skewed 2} 
\mc{Q}_{3}(\circ, a, \circ)  \leq \frac{1}{n M} + \frac{\eps}{10 |A_2|} \leq \frac{1}{N}.
\end{equation}
In the above inequalities, $\Gamma_2(a)$, denotes $a$'s neighborhood in $\mc{D}_2$, and the first inequality uses \cref{lm: gen right degree reduction} while the second inequality uses the fact that $\mc{D}_2$ has left degree $O_{|\Sigma_0|}(\poly_{\eps}(\log N))$ and $|A_2| \geq N$. We can now verify the listed properties for $\mc{D}_3$

\vspace{-2ex}
\paragraph{Length.} The right side has size $d' n M = O_{|\Sigma_0|}(\poly_{\eps}(N \log N))$ and the left side has size $|A_2|$, so the final length is $O_{|\Sigma_0|}(N\poly_{\eps}(\log N))$.
\vspace{-1ex}
\paragraph{Alphabet Sizes.} The left alphabet size is $|\Sigma_A| = 2^{O_{|\Sigma_0|}(\poly_{\eps} (\log N))}$ and the right alphabet size is $O_{\eps}(1)$ because we applied \cref{lm: right alphabet reduction}.

\vspace{-2ex}
\paragraph{Degrees.} The decoding degree is $1$ in $\mc{D}_2$ and this is preserved by \cref{lm: gen right degree reduction} when going to $\mc{D}_3$.

\vspace{-2ex}
\paragraph{Projection Decision Complexity.} The decision complexity is $O_{|\Sigma_0|}(\poly_{\eps}(\log N))$ in $\mc{D}_2$  and this is preserved by \cref{lm: gen right degree reduction}.

\vspace{-2ex}
\paragraph{Decoding Complexity.} The decoding complexity is $O_{|\Sigma_0|}(\poly_{\eps}(\log N))$ in $\mc{D}_2$ and this is preserved by \cref{lm: gen right degree reduction}.

\vspace{-2ex}
\paragraph{Decoding Distribution.} The complete decoding distribution is $\mc{Q}_{3}$ which is $K$-discrete for 
\[
K := d'^2 n M = O_{|\Sigma_0|}(N\cdot \poly_{\eps}(\log N)).
\]
The second item in the decoding distribution condition, follows because the complete decoding distribution guarantee of \cref{lm: gen right degree reduction} results in $\mc{Q}_3$ having uniform marginal over the right vertices. The third item follows from \eqref{eq: a not skewed} and \eqref{eq: a not skewed 2} as for any $a$ we have
\[
\sum_{b \in B} w(a,b) = K\cdot \mc{Q}_{3}(\circ, a, \circ) \leq O_{|\Sigma_0|}(\poly_{\eps}(\log N)).
\]

\vspace{-2ex}
\paragraph{Completeness.} We have perfect completeness in $\mc{D}_2$ and this is preserved by \cref{lm: gen right degree reduction}.

\vspace{-2ex}
\paragraph{Soundness.} It is straightforward to check that the $\mc{D}_3$ has $(\poly(1/\eps), O(\eps^{1/5}))$-list-decoding soundness by the list-decoding soundness of $\mc{D}_2$ and \cref{lm: gen right degree reduction}.
\end{proof}

%% file: composition.tex
\newcommand{\InEventConstraint}{\Event_{\In,1}}
\newcommand{\InEventDecoding}{\Event_{\In,2}}
\newcommand{\OutEventConstraint}{\Event_{\Out,1}}
\newcommand{\OutEventDecoding}{\Event_{\Out,2}}
\newcommand{\bb}{\textbf{b}}
\newcommand{\aaa}{\textbf{a}}
\newcommand{\co}{\mathsf{copy}}

\newcommand{\Comp}{\mathsf{comp}}
\newcommand{\CompLeftSide}{\LeftSide_{\Comp}}
\newcommand{\CompRightSide}{\RightSide_{\Comp}}
\newcommand{\CompConstraint}[1]{\Constraint{\Comp,#1}}
\newcommand{\CompEdgeDistributionName}{\EdgeDistributionName^{\Comp}}
\newcommand{\CompEdgeDistribution}[1]{\CompEdgeDistributionName_{#1}}
\newcommand{\CompLanguage}{\Language}
\newcommand{\CompDecoder}[2]{\Decoder{\Comp}{#1}}
\newcommand{\CompLeftVertex}{(\OutIndexInDecodingSet,\OutRightVertex,\InSpecificRandomness)}
\newcommand{\CompLeftVertexOther}{(\OutRightVertex, \InSpecificRandomness)}
\newcommand{\CompRightVertex}{(\OutLeftVertex; b)}
\newcommand{\CompLeftProof}{\TableProof{\CompLeftSide}}
\newcommand{\CompRightProof}{\TableProof{\CompRightSide}}
\newcommand{\CompDecisionComplexity}{\DecisionComplexity_{\Comp}}
\newcommand{\CompIndexInDecodingSet}{\OutIndexInDecodingSet}
\newcommand{\CompProofOutLeftProof}[1]{\TableProof{\OutLeftSide}^{(#1)}}
\newcommand{\CompProofOutRightProof}[1]{\TableProof{\OutRightSide}^{(#1)}}
\newcommand{\CompProofListIndex}{j}
\newcommand{\PCPList}{\mathsf{List}}
\newcommand{\InList}{\PCPList_{\In,\OutLeftVertex}}
\newcommand{\OutList}[1]{\PCPList_{\Out}^{(#1)}}
\newcommand{\CompList}{\PCPList_{\Comp}}
\newcommand{\InListWord}{y^{\OutLeftVertex}}
\newcommand{\OutListWord}[1]{w^{(#1)}}
\newcommand{\Event}{\mathcal{E}}
\newcommand{\DC}{\mathsf{DecisionComp}}

\section{A Composition Theorem for Decodable PCPs}\label{sec:comp_thms}
This section proves a composition theorem for dPCPs, which facilitates left alphabet reduction. In~\cref{sec: final dpcp} we show how to apply it on the constructions from~\cref{sec:dec_deg_red} (with a final composition step with a standard Hadamard-code based dPCP) to reduce the alphabet size from $2^{\polylogn}$ all the way down to $O(1)$.

Our composition result is as follows. Throughout the proof, we will use bold letters (e.g., $\OutLeftVertex,\OutRightVertex$) for outer dPCP vertices, and non-boldface symbols for inner dPCP vertices (e.g. $a,b$).

\begin{theorem} \label{thm: alph reduction comp}
Suppose we have the following two dPCPs, which we call outer dPCP and inner dPCP.

\begin{itemize}
    \item \textbf{Outer dPCP.} There exists $C_1 \in \mathbb{N}$ such that for every $\eps > 0$ there exist $d_1 \in \mathbb{N}$ and functions $F_1, \mathsf{DecisionComp}_1,$  $\mathsf{DecodeComp}_1, K_1, k_1: \mathbb{N} \times \mathbb{N} \to \mathbb{N}$ such that every circuit $\varphi: \Sigma_0^n \to \Sigma$ of size $N$ over some alphabet $\Sigma_0$ has a dPCP as follows

    \begin{itemize}
        \item \textbf{Parameters.} Length $S_1(N, |\Sigma_0|)$, left and right alphabets of sizes $F_1(N, |\Sigma_0|)$ and $O_{\eps}(1)$ respectively, decoding degree $O_{\eps}(1)$, projection decision complexity at most $\mathsf{DecisionComp}_1(N, |\Sigma_0|)$, and decoding complexity at most $\mathsf{DecodeComp}_1(N, |\Sigma_0|)$.
        \item \textbf{Complete Decoding Distribution.} Call the complete decoding distribution $\mc{Q}_{\out}$ and the sides of the constraint graph $\OutLeftSide, \OutRightSide$. Then, the following holds 
        \begin{itemize}
            \item (Discreteness): for every $\aaa \in \OutLeftSide,\bb \in \OutRightSide$, $\mc{Q}_{\out}(\circ, \aaa, \bb) = \frac{\wt(\aaa,\bb)}{K_1(N, |\Sigma_0|)}$ for $\wt(\aaa,\bb) \in \mathbb{N}$.
                \item (Right Uniform): for every right vertex $\bb$, $\sum_{\aaa \in \OutLeftSide} \wt(\aaa, \bb) = d_1$, and hence the marginal of $\mc{Q}_{\out}$ is uniform over $\OutRightSide$.
                \item (Left Bounded): for every left vertex $\aaa$, $\wt(\aaa) := \sum_{\bb \in \OutRightSide} w(\aaa, \bb) \leq k_1(N, |\Sigma_0|)$.
            \end{itemize}
         \item \textbf{Completeness.} Perfect completeness.
            \item \textbf{Soundness.} $(\OutListSize(\eps), \eps)$-list-decoding soundness where $\OutListSize = \eps^{-C_1}$.
    \end{itemize}
    \item \textbf{Inner dPCP.}  There exists $C_2 \in \mathbb{N}$ such that for every $\eps > 0$ there are bivariate functions $F_2, \mathsf{DecisionComp}_2,$  $\mathsf{DecodeComp}_2, K_2: \mathbb{N} \times \mathbb{N} \to \mathbb{N}$ and $M: \mathbb{N}\to \mathbb{N}$  such that every circuit $\varphi: \left(\Sigma'_0\right)^{n'} \to \Sigma$ of size at most $N$ over some alphabet $\Sigma'_0$ of size at most $T'$ has a dPCP as follows.
    \begin{itemize}
        \item \textbf{Parameters.} Length $S_2(N', T')$, left and right alphabets of sizes $F_2(N', T')$ and $O_{\eps}(1)$ respectively, decoding degree $1$, projection decision complexity at most $\mathsf{DecisionComp}_2(N', T')$, and decoding complexity at most $\mathsf{DecodeComp}_2(N', T')$.
         \item \textbf{Complete Decoding Distribution.} Call the complete decoding distribution $\mc{Q}_{\inner}$ and the sides of the constraint graph $A_{\inner}, B_{\inner}$. Then the following holds:
            \begin{itemize}
                \item (Discreteness): for every $a \in A_{\inner}, b \in B_{\inner}$, we have  $\mc{Q}_{\inner} (\circ,a,b)= \frac{ \wt'(a,b)}{n' \cdot M(N', T') \cdot K_2(N', T')}$ for $\wt'(a,b) \in \mathbb{N}$.
                
                \item (Left Uniform): the left vertices can be labeled as $[n'] \times [M(N', T')]$, so that for every $t \in [n']$, the distribution $\mc{Q}_{\inner}(t, \cdot, \circ)$ 
                is uniform over 
                $\{t\} \times [M(N', T')]$. It follows that the marginal distribution of $\mc{Q}_{\inner}$ over the left vertices is uniform and  for every left vertex $a$, $\wt'(a) := \sum_{b \in \InRightSide} \wt'(a,b) = K_2(N', T')$.
            \end{itemize}
             \item \textbf{Completeness.} perfect completeness.
            \item \textbf{Soundness.}$(\InListSize(\eps), \eps)$-list-decoding soundness where $\InListSize = \eps^{-C_2}$.
    \end{itemize}
\end{itemize}
  Then for every $\eps > 0$, and every circuit $\varphi: \Sigma_0^n \to \{0,1\}$ of size $N$, the language $\sat(\varphi)$ has a projection dPCP with the following properties, where we set
  \[
  N' = O(k_1(N, |\Sigma_0|) + \DC_1(N, |\Sigma_0|)) \quad\qquad \text{and} \quad\qquad T' = O_{\eps, |\Sigma_0|}(1).
  \]

    \begin{itemize}
        \item \textbf{Length.} $S_1(N, |\Sigma_0|) \cdot \left(M(N', T') + S_2(N', T') \right)$.
        \item \textbf{Left alphabet size.} $F_2(N', T')^{d_1}$.
        \item \textbf{Right alphabet size.} $O_{\eps}(1)$. 
        \item \textbf{Decoding Degree.} $O_\eps(d_1)$.
        \item \textbf{Projection Decision Complexity.} \[
        O(\log(F_1(N, |\Sigma_0|))) + O_{\eps, |\Sigma_0|}( \DC_2(N', T') + \mathsf{DecodeComp}_2(N', T')).
        \]
        \item \textbf{Decoding Complexity.} $\log(F_1(N, |\Sigma_0|)) + \mathsf{DecodeComp}_2(N', T')$. 
        \item \textbf{Complete Decoding Distribution.} Call the complete decoding distribution $\mc{Q}_{\comp}$ and the sides of the constraint graph $\CompLeftSide, \CompRightSide$. Then the following holds:
        \begin{itemize}
            \item (Discreteness): $\mc{Q}_{\comp}(\circ, a, b) = \frac{\wt''(a,b)}{J}$ for $\wt''(a,b)\in\mathbb{N}$ and $J = K_1(N, |\Sigma_0|) \cdot M(N', T') \cdot K_2(N', T')$.
            \item (Left Uniform): The marginal of $\mc{Q}_{\comp}$ over $\CompLeftSide$ is uniform, and for every $a \in A_{\comp}$ we have $\sum_{b \in B} \wt''(a,b) = d_1 K_2(N', T')$.
        \end{itemize}
        \item \textbf{Completeness.} The dPCP has perfect completeness.
        \item \textbf{Soundness.} The dPCP has $(\poly(1/\eps), \eps)$-list-decoding soundness.
    \end{itemize}
\end{theorem}

\paragraph{Proof overview.}
Our proof strategy follows the lines of the proof of \cite{MoshkovitzRaz} and \cite{dh}. Specifically, we construct a version of the inner dPCP for each left vertex $\aaa$ in the outer dPCP. The size of the circuit relevant for a given left vertex is determined by (1) the projection decision complexity, (2) the parameter $k_1$ which should be thought of as a weighted version of left degree (with weights coming from the complete decoding distribution), and (3) the decoding degree. In our application all of these are $O_{|\Sigma_0|}(\polylogn)$ or $O_{\eps, |\Sigma_0|}(1)$, so the inner dPCP is run on a much smaller scale. 

In order to compose all of these smaller inner dPCPs into one overall dPCP, we attempt to check several constraints ``in parallel''. Namely, as in~\cite{MoshkovitzRaz,dh} we think of choosing a right vertex $\bb$ of the outer PCP and checking all of the inner PCPs of $\aaa$ that are neighbours of $\bb$. This requires some alignment 
between the inner PCPs of various $\aaa$, which is handled by the extra flexibility we have arranged in~\cref{sec:dec_deg_red} in the size of the PCPs. 
Additional care is needed to make sure that the decoding is still maintained throughout, and this aspect is where we differ from the prior compositions mentioned. Indeed, this is why we went through extensive efforts in \cref{lm: decoding degree reduction} to reduce the decoding degree of the dPCP from \cref{thm: dp dpcp} to $1$. Roughly speaking, having small decoding degree allows us to pass decoding to the inner dPCPs while not increasing their alphabet size by too much.

\subsection{The Construction}\label{subsec:composition-construction}
Fix $\eps > 0$ and a size $N$ circuit $\varphi: \Sigma_0^n \to \{0,1\}$. Our goal is to construct a dPCP for the language $\Language := {\sf SAT}(\varphi) \subseteq \Alphabet{0}^{n}$, and we have the outer dPCP $\OutPCPDecoder$ and the inner dPCP $\InPCPDecoderSymbol$ at our disposal. We will apply the outer dPCP to $\varphi$ and the inner dPCP for circuits $\varphi_{\OutLeftVertex}$ of size roughly $O(k_1(N, |\Sigma_0|) + O(\mathsf{DecisionComp}_1(N, |\Sigma_0|))$, which should be thought of as much smaller than $N$.
\vspace{-1ex}
\paragraph{The Outer dPCP.} Take $\OutPCPDecoder$ to be the outer dPCP for $\mc{L}$ as in the premise with the soundness error set to $\OutPCPError := \eps^{C_2 + 10}$. We denote its parts by
\[
\OutPCPDecoder = \left(\OutLeftSide \cup \OutRightSide, \OutEdgeSet, \OutLeftAlphabet, \OutRightAlphabet ,\{\OutConstraint{e} \}_{e \in \OutEdgeSet}, \{\OutDecoder{t} \}_{t \in [n]}, \{\mc{P}^{\out}_{t}\}_{t \in [n]}\right).
\]
Let $q_1$ denote the decoding degree of $\mc{D}_{\out}$ so that $q_1 =   O_{\eps}(1)$ by assumption.

Henceforth, it will be convenient to let $\Gamma_{\disc, \out}(\aaa)$ denote the set of $\aaa$'s neighbors listed with multiplicity $\wt(\aaa, \bb)$. In order to avoid confusion when referring to multiple copies of $\bb$ (which appears in $\Gamma_{\disc, \out}(\aaa)$ more than once if $\wt(\aaa, \bb) > 1$), we refer to the members of $\Gamma_{\disc, \out}(\aaa)$ as $(\bb, i)$ where for each $\bb \in \Gamma_{\out}(\aaa)$, $i \in [\wt(\aaa, \bb)]$. Hence, 
\[
\Gamma_{\disc, \out}(\aaa) = \bigcup_{b \in \Gamma_{\out}(\aaa)} \{\bb\} \times [\wt(\aaa, \bb)].
\]
Notice that choosing $(\bb,i) \in \Gamma_{\disc, \out}(\aaa)$ uniformly, the marginal distribution of $\bb$ is $\mc{Q}_{\out}(\circ, \aaa, \cdot)$.

The following notation will be useful.
\begin{definition} \label{def: simplifying projection outer}
    For a symbol $\sigma$ for $\aaa \in \OutLeftSide$ and $\bb \in \Gamma_{\out}(\aaa)$, let $\sigma_{\to \bb}$ denote the unique $\Sigma_B$ symbol that satisfies the projection constraint on $(\aaa, \bb)$. For each $t \in \Gamma_{\dec, \out}(\aaa)$, we write $\sigma_{\rightsquigarrow t} = \OutDecoder{t}{}(\aaa, \sigma)$. Here, the vertex $\aaa$ will always be clear from context.
\end{definition} 

\vspace{-1ex}
\paragraph{The Inner dPCPs.} We will also construct an inner dPCP corresponding to each outer vertex $\OutLeftVertex$. For each $\aaa \in \OutLeftSide$ we start by defining a language 
\[
\mc{L}_{\OutLeftVertex} \subseteq \left(\OutRightAlphabet \times \Sigma_0^{q_1}\right)^{\wt(\aaa)},
\]
which is supposed to correspond to the set of valid alphabet symbols for $\aaa$. Set, 
\begin{equation} \label{eq: inner language alphabet size} 
T' = |\OutRightAlphabet| \times |\Sigma_0|^{q_1} = O_{\eps, |\Sigma_0|}(1),
\end{equation}
to be the size of the alphabet which each of the languages $\mc{L}_{\OutLeftVertex}$ is over. In the definition of $\mc{L}_{\aaa}$, we identify $[\wt(\aaa)]$ with the multiset $\Gamma_{\disc, \out}(\aaa)$. Let us now fix an $\aaa$ and define the language $\mc{L}_{\OutLeftVertex}$.

For each $w \in \mc{L}_{\OutLeftVertex}$, think of its $\wt(\aaa)$ indices as being indexed by $\Gamma_{\disc, \out}(\OutLeftVertex)$ and for each $(\bb,i) \in \Gamma_{\disc, \out}(\aaa)$ we write $w_{(\bb,i)}$ as the symbol in $w$ at index $(\bb, i)$. The symbol $w_{(\bb,i)}$ is thought of as consisting of one $\OutRightAlphabet$ symbol corresponding to an implied label for $\bb$, and $q_1$ symbols from $\Sigma_0$, corresponding to one implied label for each $t \in \Gamma_{\out, \dec}(\OutLeftVertex)$. With this setup in mind, the members of $\mc{L}_{\OutLeftVertex}$ are exactly the $w \in \left(\OutRightAlphabet \times \Sigma_0^{q_1}\right)^{\wt(\aaa)}$ such that there exists  $\sigma \in \Sigma_{\OutLeftVertex}$ satisfying the following two conditions:
\begin{itemize}
    \item for every $(\bb,i) \in \Gamma_{\disc, \out}(\OutLeftVertex)$, the $\OutRightAlphabet$ symbol in $w_{(\OutRightVertex, i)}$, which we denote by $\left(w_{(\OutRightVertex, i)}\right)_{\bb}$, is equal to $\sigma_{\to \bb}$, 
    \item for every $t \in \Gamma_{\out, \dec}(\OutLeftVertex)$ and $(\bb, i) \in \Gamma_{\disc, \out}(\OutLeftVertex)$, the symbol corresponding to $t$ in $w_{(\OutRightVertex, i)}$, which we denote by $\left(w_{(\OutRightVertex, i)}\right)_{t}$, is equal to $\sigma_{\rightsquigarrow t}$.
\end{itemize}
In words, the above two conditions mean that all of the symbols in $w$ are consistent with a single $\sigma \in \OutLeftAlphabet$. 

By the projection decision complexity of $\OutPCPDecoder$ and the bound on $\wt(\aaa)$, the language $\mc{L}_{\OutLeftVertex}$ can be decided by a circuit of size $O(\wt(\aaa) + \DC_1(N)) = O(k_1(N,|\Sigma_0|) + \DC_1(N,|\Sigma_0|))$, and hence there exists a uniform bound
\begin{equation}\label{eq:N'}
    N' =  O(k_1(N, |\Sigma_0|) + \DC_1(N, |\Sigma_0|))
\end{equation} 
such that for \emph{every} $\OutLeftVertex \in \OutLeftSide$, membership in the language $\mc{L}_{\aaa}$ can be decided by a circuit of size \textit{at most} $N'$. Henceforth, we fix 
\begin{equation} \label{eq: M and K2} 
M = M(N', T') \quad \text{and} \quad K_2 = K_2(N', T'),
\end{equation}
where $M(\cdot, \cdot)$ is the function related to the left vertices of the inner dPCP given by the theorem statement, and $K_2(\cdot, \cdot)$ is the function related to the complete decoding distribution of the inner dPCP given by the theorem statement.

Now, for each $\OutLeftVertex \in \OutLeftSide$ we take $\InPCPDecoder$ to be the inner dPCP guaranteed by the theorem statement for the language $\mc{L}_{\OutLeftVertex}$ with soundness error set to 
\[
\InPCPError := \eps.
\]
 Note that we can use the same value $M$ and $K_2$ here for every inner dPCP, $\InPCPDecoder$. This is because the inner dPCP assumption in the theorem statement, along with the fact that for all $\aaa\in \OutLeftSide$, the language $\mc{L}_{\aaa}$ is over strings of length at most  $k_1(N, |\Sigma_0|)$ with an alphabet of size at most $T'$, and membership in $\mc{L}_{\aaa}$ can be decided by a circuit of size at most $N'$.
 
 For each $\aaa \in \OutLeftSide$, we denote the parts of $\InPCPDecoder$ as 
\[
\InPCPDecoder = (\InLeftSide \cup \InRightSide, \InEdgeSet, \InLeftAlphabet, \InRightAlphabet, \{\InConstraint{e} \}_{e \in \InEdgeSet}, \{\InDecoder{(\bb,i)}{\OutRightVertex} \}_{(\OutRightVertex,i) \in \Gamma_{\out}(\OutLeftVertex)}, \{ \mc{P}^{\inner, \OutLeftVertex}_{(\OutRightVertex,i)}\}_{(\OutRightVertex,i) \in \Gamma_{\out}(\OutLeftVertex)}).
\]
We have that the alphabet sizes are
\begin{equation} \label{eq: inner alphabet sizes} 
|\InLeftAlphabet| \leq F_2(N', T') \quad \text{and} \quad |\InRightAlphabet| = O_{\eps}(1).
\end{equation} 
For, convenience, we will express $\InLeftSide$ as $\Gamma_{\disc, \out}(\OutLeftVertex) \times [M]$, rather than $[\wt(\aaa)] \times [M]$, where recall that $|\Gamma_{\disc, \out}(\OutLeftVertex)| = \wt(\aaa)$. Also, for simplicity, we refer to members of $\Gamma_{\disc, \out}(\OutLeftVertex) \times [M]$ as $(\bb, i, r)$ (rather than $((\bb, i), r)$.

 We will also use the following simplifying notation, which is similar to that in \cref{def: simplifying projection outer}, but now specified for the inner dPCP.

\begin{definition}
    Given a symbol $\sigma \in \InLeftAlphabet$ for $(\bb,i, r) \in \InLeftSide$, recall that $\InDecoder{(\bb,i)}{}((\bb,i, r), \sigma)$ outputs a $\OutRightAlphabet$ symbol for $\bb$ and one $\Sigma_0$ symbol for each $t \in \Gamma_{\out, \dec}(\aaa)$. We let $\left(\sigma_{\rightsquigarrow (\bb,i)}\right)_{\bb}$ denote the $\OutRightAlphabet$ symbol part of the output from $\InDecoder{(\bb,i)}{}((\bb,i,r), \sigma)$. Likewise, for each $t \in \Gamma_{\dec, \out}(\aaa)$, we write $\left(\sigma_{\rightsquigarrow (\bb, i) }\right)_{t}$ to denote the $\Sigma_0$ symbol for $t$ in $\InDecoder{(\bb,i)}{}((\bb,i, r), \sigma)$. When we use this notation, the identity of $r \in [M]$ will be clear from context.
\end{definition}

By design of the language, the output of the decoder, $\InDecoder{(\bb,i)}{}$ consists of one $\OutRightAlphabet$ symbol for $\OutRightVertex$ and one $\Sigma_0$ symbol for each $t \in \Gamma_{\dec}(\OutLeftVertex)$. The list-decoding soundness of $\InDecoder{(\bb,i)}{}$ translates to the following claim, which morally states that one can view $\InPCPDecoder$ as a decodable PCP for the language of valid $\aaa$-alphabet symbols, with $(\InListSize, \delta_{\inner})$-list decoding soundness.

\begin{claim} \label{cl: translate inner soundness}
For every $\aaa \in \OutLeftSide$ the following holds. For any left assignment $\InLeftProof: \Gamma_{\disc, \out}(\OutLeftVertex) \times [M] \to \InLeftAlphabet$ to $\InPCPDecoder$, there is a list  $\{\sigma^{\OutLeftVertex}_1, \ldots, \sigma^{\OutLeftVertex}_{\InListSize} \} \subseteq \Sigma_{\OutLeftVertex}$ such that for any right assignment $\InRightProof: \InRightSide \to \InRightAlphabet$
    \[
    \Pr_{ \OutRightVertex \sim \mc{Q}_{\out}(\circ, \OutLeftVertex, \cdot), i \in [\wt(\aaa,\bb)], r \in [M], b \sim \mc{P}^{\inner, \OutLeftVertex}_{(\OutRightVertex,i)}(( \OutRightVertex,i, r),\cdot)}[\Event_1 \land \Event_2] \leq \InPCPError.
    \]
Here $\Event_1$ and $\Event_2$ are the following events over the sampled $(\bb,i) \in \Gamma_{\out, \dec}(\aaa), r \in [M], b \in \InRightSide$.
    \begin{itemize}
        \item $\Event_1.$ the constraint on $(\OutRightVertex, i, r)$ and $b$ is satisfied by $\InLeftProof[(\bb, i, r)]$ and $\InRightProof[b]$.
        \item $\Event_2.$ the symbols for $(\OutRightVertex, i)$ 
        and $t \in \Gamma_{\out, \dec}(\OutLeftVertex)$ output by the decoder are not (simultaneously) equal to $\left(\sigma^{\OutLeftVertex}_j \right)_{\to \bb}$ and $\left(\sigma^{\OutLeftVertex}_j \right)_{\rightsquigarrow t}$ for any $j \in [\InListSize]$. That is, for each $j \in [\InListSize]$, either
        \[
        \left(\InDecoder{(\bb, i)}{}((\bb, i, r), \InLeftProof[(\bb,i,r)])\right)_{\bb} \neq \left(\sigma^{\OutLeftVertex}_j \right)_{\to \bb} 
        \]
        or there is $t \in \Gamma_{\out, \dec}(\aaa)$ such that 
        \[
         \left(\InDecoder{(\bb, i)}{}((\bb, i, r), \InLeftProof[(\bb,i,r)])\right)_{t} \neq \left(\sigma^{\OutLeftVertex}_j \right)_{\rightsquigarrow t}.
        \]
    \end{itemize}
\end{claim}
\begin{proof}
Fix $\InLeftProof$. Applying the list-decoding soundness of $\InPCPDecoder$, we get a list of words $\{w_{1},\ldots, w_{\InListSize}\} \subseteq \mc{L}_{\aaa}$ such that for any right assignment $\InRightProof$, the list-decoding error of $\InLeftProof$ and $\InRightProof$ relative to this list is at most $\InPCPError$. We can translate this list, $\{w_{1},\ldots, w_{\InListSize}\}$, into a corresponding list of $\InListSize$ symbols $\sigma^{\aaa}_1, \ldots, \sigma^{\aaa}_{\InListSize} \in \Sigma_{\aaa}$. Specifically, for each $w_\ell \in \{w_{1},\ldots, w_{\InListSize}\}$, by definition of being in  $\mc{L}_{\OutLeftVertex}$, we have that there exists $\sigma \in \Sigma_{\OutLeftVertex}$ such that for any $(\bb,i) \in \Gamma_{\disc, \out}(\aaa)$ and $t \in \Gamma_{\out, \dec}(\aaa)$ 
\[
(w_{\bb, i})_{\bb} = \sigma_{\to \bb} \quad \text{and} \quad (w_{\bb, i})_{t} = \sigma_{\rightsquigarrow t}.
\]

We set $\sigma^{\aaa}_\ell$ to be this symbol and if for some $w \in \{w_{1},\ldots, w_{\InListSize}\}$ there exist multiple such $\sigma \in \Sigma_{\OutLeftVertex}$, we choose one of them arbitrarily to be $\sigma^{\aaa}_\ell$. Now it is straightforward to see that the list-decoding soundness of $\InPCPDecoder$ for the language $\mc{L}_{\aaa}$ relative to the list $\{w_{1},\ldots, w_{\InListSize}\}$ translates exactly into the statement of the claim with the list $\{\sigma^{\aaa}_1, \ldots, \sigma^{\aaa}_{\InListSize}\} \subseteq \Sigma_{\aaa}$. Here, we use the fact that sampling $(\bb,i)$ in the multiset, $\Gamma_{\disc, \out}(\aaa)$, uniformly, we have $\bb \sim \mc{Q}_{\out}(\circ, \aaa, \cdot)$ and $i$ uniform in $[\wt(\aaa,\bb)]$.
\end{proof}
Finally, let us write the complete decoding distribution of each inner dPCP, $\InPCPDecoder$, as $\mc{Q}_{\inner, \aaa}$. Using the discreteness property of the inner dPCP's complete decoding distribution, we have that for every $((\bb,i, r), b) \in \InLeftSide \times \InRightSide$, there is $\wt'_{\aaa}((\bb,i, r), b)$ such that
\begin{equation} \label{eq: inner discrete express}
\mc{Q}_{\inner, \aaa}(\circ, (\bb,i, r), b) = \mc{Q}_{\inner, \aaa}((\bb,i), (\bb, i, r), b) = \frac{\wt'_{\aaa}((\bb,i, r), b)}{\wt(\aaa) \cdot M \cdot K_2}.
\end{equation}

By the third property of the complete decoding distribution of the inner dPCP, for all $\aaa \in \OutLeftSide$ and $(\bb, i, r) \in \InLeftSide$ we have 
\begin{equation} \label{eq: inner left uniform}
\sum_{b \in \Gamma_{\aaa, \inner}((\bb, i,r))} \wt'_{\aaa}((\bb, i, r), b) = K_2,
\end{equation}
where in the above $\Gamma_{\aaa, \inner}$ denotes the neighborhood of a vertex in $\InPCPDecoder$. We emphasize that, by assumption, $K_2$ above is the same for all of the inner dPCPs, $\InPCPDecoder$. We are ready to describe our composed dPCP.

\vspace{-1ex}
\paragraph{The composed dPCP $\OutPCPDecoder \circ \InPCPDecoderSymbol$.}
\begin{itemize}

    \item \textbf{Constraint Graph:} The constraint graph is bipartite with the following left and right sides respectively: 
     \begin{align*}
         &\CompLeftSide\coloneqq\{\CompLeftVertexOther \ :\ \InSpecificRandomness \in [M]\text{ and } \ \OutRightVertex \in \OutRightSide\}, \\
         &\CompRightSide\coloneqq
         \{\CompRightVertex \ : \  \OutLeftVertex\in\OutLeftSide,  b\in\InRightSide\}=\bigcup_{\OutLeftVertex\in\OutLeftSide}\{\aaa\}\times\InRightSide.
     \end{align*} 
      Let us explain how the vertices in the composed PCP should be thought of. First, take the inner constraint graphs, $(\InLeftSide \cup \InRightSide, \InEdgeSet)$, for each $\OutLeftVertex \in \OutLeftSide$. Viewing each left side as $\Gamma_{\disc, \out}(\OutLeftVertex) \times [M]$, we see that for each $\bb \in \OutRightSide$ and $r \in [M]$, there are $d_1$-many vertices of the form $(\bb, i, r)$ which appear in $\CompLeftSide$. This is because of the assumption that
      \begin{equation}   \label{eq: out left regularity}
      \sum_{\aaa \in \OutLeftSide} \wt(\aaa, \bb) = d_1,
    \end{equation}
from the outer dPCP's complete decoding distribution. Henceforth, when we are in the context of the composed dPCP, we will denote the copy of  $(\OutRightVertex, i, r)$ in $\InLeftSide$ by $(\aaa; \bb, i, r)$.
     
     The left side of the composed constraint graph is obtained by identifying, for each $\bb \in \OutRightSide$ and $r \in [M]$, the vertices of $(\aaa; \OutRightVertex, i, r)$ over all $\aaa$ adjacent to $\bb$ and $i \in [\wt(\aaa,\bb)]$. One can think of the vertex $(\bb, r) \in \CompLeftSide$ as a cloud consisting of all $(\aaa; \OutRightVertex, i, r)$ appearing over all copies of the inner dPCPs $\InPCPDecoder$. 

     The right side of the composed constraint graph is the union of the right sides of the inner constraint graph, and the edges are the same after identification. 
     \item \textbf{Left Alphabet:} Fix a vertex $\CompLeftVertexOther\in\CompLeftSide$. For each $\aaa \in \Gamma_{\out}(\bb)$ and $i \in [\wt(\aaa, \bb)]$, there is a vertex $(\aaa; \bb, i, r)$ coming from the inner dPCP of $\mc{D}_{\inner, \aaa}$. Then, the composed vertex $(\bb, r)$ holds a symbol from $\Sigma_{\InLeftSide, \aaa}$ for each $(\aaa; \bb, i, r)$. Overall, using \eqref{eq: out left regularity} see that these are $d_1$ inner left alphabet symbols.
     
     If $\sigma$ is an alphabet symbol for $(\bb, r)$, we use $\sigma \{ \aaa, i \}$ to denote the $\InLeftAlphabet$ symbol corresponding to a label for $(\aaa; \bb, i, r)$. We also constrain the alphabet of $(\bb,r)$ to only contain $\sigma$ as follows.
     \begin{itemize}
        \item \textbf{Hardcoding Outer Constraints.} 
        We require that $\left(\sigma \{ \aaa, i \} \right)_{\rightsquigarrow \bb}$ is the same for all $\aaa\in \Gamma_{\disc}(\bb)$ and $i \in [\wt(\aaa, \bb)]$. Henceforth, we denote this value as $\sigma_{\rightsquigarrow \bb}$ so that 
        \begin{equation}\label{eq:def_squiq_eq}
        \sigma_{\rightsquigarrow \bb}  = \left(D^{\inner, \aaa}_{(\bb,i)}((\aaa; \bb, i, r), \sigma\{\aaa, i\})\right)_{\bb}.
        \end{equation}
        On the right side, we are first running the inner decoder from $\InPCPDecoder$, where the index to decode is $(\bb, i)$, the left vertex is $(\aaa; \bb, i, r) \in \InLeftSide$, and the symbol is $\sigma\{\aaa,i\}$. We then take the symbol for $\bb$ given by the output symbol. We note that whenever we use the notation of \eqref{eq:def_squiq_eq}, the symbol $\sigma$ will satisfy the hardcoding outer constraints, so all the possible $\aaa$ lead to the same value. 
        \item \textbf{Hardcoding Decoding Consistency.} For any $t \in [n]$ and all $\aaa, i$ such that $(\aaa, \bb) \in \supp(\mc{P}_t)$ and $i \in [\wt(\aaa, \bb)]$, we require that the decoded values $\left(D^{\inner, \aaa}_\bb((\aaa; \bb, i, r), \sigma\{\aaa, i\})\right)_{t}$ are the same, and we denote this value by $\sigma_{\rightsquigarrow t}$.
     \end{itemize}
     We denote the left alphabet by $\Alphabet{\CompLeftSide}$.
    \item \textbf{Right Alphabet:} The right alphabet is $\InRightAlphabet$.
    
    \item \textbf{Constraints:} By our construction, $(\CompLeftVertexOther,\CompRightVertex)\in\EdgeSet_{\Comp}$ if and only if there is an $i$ such that $(\bb,i) \in \Gamma_{\disc, \out}(\aaa)$ and, in $\InPCPDecoder$, the vertex $(\aaa; \bb, i, r)$ is adjacent to $(\aaa; b)$. Let us describe the constraint on this edge now.
    
    Given valid labels $\sigma \in \Alphabet{\CompLeftSide}$ and $\sigma' \in \InRightAlphabet$ for $(\bb,r)$ and $(\aaa; b)$ respectively in the composed dPCP, the constraint on $(\CompLeftVertexOther,\CompRightVertex)$  checks that, for all such $\aaa, i$ as above, where $(\bb, i, r)$ is adjacent to $b$ in $\InPCPDecoder$, the inner constraint of $\InPCPDecoder$, between $(\aaa; \bb, i, r)$ and $(\aaa; b)$ is satisfied by the labels $\sigma\{\aaa, i\} \in \InLeftAlphabet$ and $\sigma' \in \InRightAlphabet$.
    That is, 
    \[
\Phi_{(\CompLeftVertexOther,\CompRightVertex)}(\sigma, \sigma') = 1 
    \]
    if and only if 
    \[
    \Phi_{\inner,  ((\aaa;\bb,i, r),(\aaa; b))}(\sigma ,\sigma') = 1
    \]
    over all $\aaa, i$ such that $(\aaa; \bb, i, r)$ is adjacent to $(\aaa; b)$ in $\InPCPDecoder$. In the second equality, the constraint is coming from $\InPCPDecoder$.
    
    \item \textbf{Decoding Distributions:} Let $\mc{Q}_{\out}$ be the complete decoding distribution of the outer dPCP. Then the complete decoding distribution of the composed dPCP, $\mc{Q}_{\comp}$, is generated as follows.
        \begin{itemize}
        \item Choose $(t, \OutLeftVertex,\OutRightVertex)\sim \mc{Q}_{\out}$, and $i \in [\wt(\aaa, \bb)]$ uniformly at random.
        \item Sample $r \in [M]$ uniformly.
         \item Choose $(\aaa; b) \sim \mc{Q}_{\In,\OutLeftVertex}((\OutRightVertex,i),(\aaa; \OutRightVertex,i, r), \cdot )$, 
        \item Output $(t, \CompLeftVertexOther,\CompRightVertex)$.
    \end{itemize}
    It is clear that the marginal of $\mc{Q}_{\comp}$ over $[n]$ is uniform, and for each $t \in [n]$ we take the decoding distribution to be $\CompEdgeDistribution{t} = \mc{Q}_{\comp}(t, \cdot, \cdot)$.

     \item \textbf{Decoder.} For 
    $\CompLeftVertexOther \in\CompLeftSide$ and label $\sigma$, the decoder outputs $
   D^{\Comp}_{t}(\CompLeftVertexOther, \sigma) = \sigma_{\rightsquigarrow t}$.
\end{itemize}
This completes the description of the composed dPCP,  $\OutPCPDecoder \circ \InPCPDecoderSymbol$, and we will now show that it satisfies the requirements of \cref{thm: alph reduction comp}.

\subsection{Basic Properties of the Composed PCP}
In this section we establish all of the properties of $\OutPCPDecoder \circ \InPCPDecoderSymbol$ as promised in the statement of \cref{thm: alph reduction comp}, except for the list-decoding soundness, which is deferred to the next section. In the below, recall that $T'$ is given by \eqref{eq: inner language alphabet size}, $N'$ is given by \eqref{eq:N'}, $K_2$ and $M$ are given by \eqref{eq: M and K2}.
\vspace{-2ex}
\paragraph{Length.} By construction, left side has size at most $|\OutRightSide|\cdot M$ vertices and the right side has size $\sum_{\OutLeftVertex\in\OutLeftSide} |\InRightSide|$, so overall the length is 
\[
S_1(N, |\Sigma_0|) \cdot M(N', T') + S_1(N, |\Sigma_0|) \cdot S_2(N', T').
\]

\vspace{-2ex}
\paragraph{Alphabet Sizes.} The left alphabet size is at most $|\InLeftAlphabet|^{d_1}$.
Indeed, fix a left vertex $(\bb, r)$. Then for every $\OutLeftVertex \in \OutLeftSide$ and $i \in [\wt(\aaa,\bb)]$, $(\bb, r)$ holds one symbol from $\InLeftAlphabet$, corresponding to a label to $(\aaa; \bb,i,r)$ in $\InPCPDecoder$, 
which by~\eqref{eq: out left regularity} is a total of $d_1$ 
symbols from $\InLeftAlphabet$. Hence the alphabet size is $F_2(N', T')^{d_1}$, where we use \eqref{eq: inner alphabet sizes} to bound the size of $|\InLeftAlphabet|$. The right alphabet is $\InRightAlphabet$ and its size is $O_{\eps}(1)$ by \eqref{eq: inner alphabet sizes}.
\vspace{-2ex}
\paragraph{Decoding Degree.} 
The decoding degree is $O_{\eps}(d_1)$. Each $(\OutRightVertex, r)$ is responsible for decoding $t$ if there is $\aaa \in \Gamma_{\out}(\bb)$ such that $t \in \Gamma_{\out, \dec}(\aaa)$. The decoding degree bound follows from the fact that $|\Gamma_{\out, \dec}(\aaa)| \leq O_{\eps}(1)$ and that 
$|\Gamma_{\out}(\bb)| \leq d_1$ by~\eqref{eq: out left regularity}.
\vspace{-2ex}
\paragraph{Projection Constraints.} It is straightforward to check that the composed dPCP has projection constraints because the inner dPCPs are all projection dPCPs.
\vspace{-2ex}
\paragraph{Projection Decision Complexity.} Fix a constraint $((\bb, r), \InRightVertex)$ and let $\sigma$ be the left alphabet symbol. Recall that $\sigma$ contains, as its restrictions, one $\InLeftAlphabet$ symbol for each of $d_1$ vertices corresponding to all $\aaa, i$ such that $\aaa \in \Gamma_{\disc}(\bb), i \in [\wt(\aaa, \bb)]$.

The projection circuit needs to do the following: 
(1) check that the left alphabet symbol satisfies hardcoding of outer constraints, (2) check that the left alphabet symbol satisfies hardcoding decoding consistency, and (3) compute the unique right alphabet symbol that satisfies the constraint.

Each of (1) and (2) can be done by computing a restriction of the given left alphabet symbol and then running the inner decoder $d_1$ times and then checking an equality between $d_1$ symbols of either $\OutRightAlphabet$, in the case of (1), or $\Sigma_0$ in the case of (2). Altogether, the circuit complexity of (1) and (2) is a constant multiple of
\begin{align*}
    &\log(F_1(N, |\Sigma_0|)) + d_1 \cdot (\mathsf{DecodeComp}_2(N', T') + \log|\OutRightAlphabet| + \log |\Sigma_0|) \\
    &= \log(F_1(N, |\Sigma_0|)) + O_{\eps, |\Sigma_0|}(\mathsf{DecodeComp}_2(N', T')).
\end{align*}

To accomplish (3), the projection circuit runs the projection circuit of the inner dPCP, $\InPCPDecoder$, on a restriction of the given left alphabet symbol which requires circuit complexity
\[
\log(F_1(N, |\Sigma_0|)) + \DC_2(N', T').
\]
Overall this gives projection decision complexity that is, 
\[
O(\log(F_1(N, |\Sigma_0|))) + O_{\eps, |\Sigma_0|}( \DC_2(N', T') + \mathsf{DecodeComp}_2(N', T')).
\]

\vspace{-2ex}
\paragraph{Decoding Complexity.} The decoder restricts its input to obtain a symbol for the inner dPCP and runs the decoder of the inner dPCP. Hence, the decoding complexity is $\log(F_1(N, |\Sigma_0|)) + \mathsf{DecodeComp}_2(N', T')$.

\vspace{-2ex}
\paragraph{Complete Decoding Distribution.} Fix $(\bb, r) \in \CompLeftSide$ and $\InRightVertex \in \CompRightSide$. Then
\begin{equation} \label{eq: composed dPCP cdd}
    \begin{split}
    \mc{Q}_{\comp}(\circ, (\bb, r), \InRightVertex) &= \mc{Q}_{\out}(\circ, \aaa, \bb) \cdot \frac{1}{M} \cdot\E_{i\in[\wt(\aaa,\bb)]}\left[\frac{\mc{Q}_{\inner,\aaa}((\bb,i), (\aaa;\bb,i,r),(\aaa; b))}{\mc{Q}_{\inner,\aaa}((\bb,i), (\aaa;\bb,i,r),\circ)}\right]\\
 &= \mc{Q}_{\out}(\circ, \aaa, \bb) \cdot \frac{1}{M} \cdot \E_{i \in [\wt(\aaa, \bb)]} \left[ \frac{\wt'_{\aaa}((\aaa; \bb, i, r), (\aaa;b))}{K_2} \right]\\
    &= \frac{\wt(\aaa, \bb)}{K_1 M K_2} \E_{i \in [\wt(\aaa,\bb)]} \left[ \wt'_{\aaa}((\aaa; \bb,i, r), (\aaa;b))\right] \\
    &= \frac{\sum_{i \in [\wt(\aaa, \bb)]}\wt'_{\aaa}((\aaa;\bb,i,r), (\aaa;b))}{K_1 M K_2} 
\end{split}
\end{equation}

In the first transition, we are using the definition of $\mc{Q}_{\comp}$. In the second transition, we are using \eqref{eq: inner discrete express} to plug in the term inside the expectation, and then~\eqref{eq: inner left uniform}. In the third transition, we are writing out $\mc{Q}_{\out}(\circ, \aaa, \bb)$. In the final transition, we are expanding out the expectation from the second line.

By the above, we get that the first property of the complete decoding distribution holds with discreteness parameter 
\begin{equation} \label{eq: comp discrete}  
J = K_1 \cdot M \cdot K_2 \quad \text{and} \quad \wt''((\bb, r), \InRightVertex) =  \sum_{i \in [\wt(\aaa, \bb)]}\wt'_{\aaa}((\aaa; \bb,i,r), (\aaa;b)).
\end{equation}

We move onto the second property which requires showing that the marginal over $A_{\comp}$ is uniform. From the description, it is clear that
\[
\mc{Q}_{\comp}(\circ, (\bb, r), \circ)  = \mc{Q}_{\out}(\circ, \circ, \bb) \cdot \frac{1}{M} = \frac{d_1}{ K_1(N, |\Sigma_0|) M}. 
\]
In the second equality, we use the fact that  $\mc{Q}_{\out}(\circ, \circ, \bb) = \frac{d_1}{ K_1(N, |\Sigma_0|)}$, by the right uniformity property of the outer dPCP's complete decoding distribution. It follows that for any $(\bb, r) \in \CompLeftSide$
\[
\sum_{(\aaa; b) \in \CompRightSide}\wt''((\bb, r), (\aaa; b)) = \frac{J d_1}{ K_1(N, |\Sigma_0|) M} = d_1K_2,
\]
as desired.

\vspace{-1ex}
\paragraph{Completeness.} Since the outer dPCP $\OutPCPDecoder$ has perfect completeness, for any $w\in \mathsf{SAT}(\varphi)$ there are assignments $\OutLeftProof:\OutLeftSide\to\OutLeftAlphabet$ and $\OutRightProof:\OutRightSide\to\OutRightSide$ such that all of the constraints in $\OutPCPDecoder$ are satisfied and the decoder outputs $w_t$ with probability $1$ for all $t \in [n]$. Now, for each $\OutLeftVertex\in \OutLeftSide$, using the value  $\OutLeftProof[\OutLeftVertex]$, one can construct $w_{\In,\OutLeftVertex}\in \left(\OutRightAlphabet \times \Sigma_0^{q_1}\right)^{\wt(\aaa)}$ such that for every $(\bb,i) \in \Gamma_{\disc, \out}(\OutLeftVertex)$ and $t \in \Gamma_{\out, \dec}(\OutLeftVertex)$, we have that the $\OutRightAlphabet$ part of $(w_{\In,\OutLeftVertex})_{(\OutRightVertex, i)}$ equals $(\OutLeftProof[\OutLeftVertex])_{\to \OutRightVertex}$, and the symbols corresponding to $t \in \Gamma_{\out, \dec}(\OutLeftVertex)$ in $(w_{\In,\OutLeftVertex})_{\bb}$ are consistent $\left(\OutLeftProof[\aaa]\right)_{\rightsquigarrow t}$. For each such  $w_{\In,\OutLeftVertex}$, which is in $\mc{L}_{\aaa}$, one can then use the perfect completeness of the inner dPCP's to construct $\InLeftProof:\InLeftSide\to\InLeftAlphabet$ and $\InRightProof:\InRightSide\to\InRightSide$  which satisfy all of the constraints on the inner dPCP $\InPCPDecoder$ and always output a decoding that is consistent with $(w_{\inner, \aaa})_{\to \bb}$ for all $(\bb, i) \in \Gamma_{\disc, \out}(\aaa)$ and $(w_{\inner, \aaa})_{\rightsquigarrow t}$ for all $t \in \Gamma_{\out, \dec}(\aaa)$.

From these assignments, we construct the composed assignments $\CompLeftProof:\CompLeftSide\to\Alphabet{\CompLeftSide}$ and $\CompRightProof:\CompRightSide\to\Alphabet{\CompRightSide}$ as follows:
\begin{itemize}
    \item For every $(\aaa; b) \in \CompLeftSide$, note that $b$ is a right vertex in $\InPCPDecoder$, and hence we set $\CompRightProof[\CompRightVertex]=\InRightProof[b]$.
\item  For each $(\bb, r) \in \CompLeftSide$, $\CompLeftProof[(\bb,r)]$ consists of the label $\InLeftProof[(\bb, i,r)]$ for each $\aaa \in \Gamma_{\out}(\bb)$ and $i \in [\wt(\aaa,\bb)]$.
\end{itemize}
One can check that this assignment indeed satisfies all of the constraints of the composed dPCP and the decoder outputs $w_t$ with probability $1$ for all $t \in [n]$.

\subsection{Proof of List-Decoding Soundness}
This section is devoted to the proof of the following lemma, establishing the list-decoding soundness of~\cref{thm: alph reduction comp} and thereby completing its proof. 
\begin{lemma}\label{lemma:left-alphabet-reduction-soundness}
The composed dPCP constructed above has decoding error $\delta_{\inner} + \InListSize \cdot \delta_{\out} = \eps + \eps^{10}$ relative to a list of size $\InListSize \cdot \OutListSize = \eps^{C_1 + C_2 + 10}$, and hence satisfies $(\poly(1/\eps), O(\eps))$-list-decoding soundness.
\end{lemma}
\begin{proof}
Fix an assignment. $\CompLeftProof:\CompLeftSide\to\Alphabet{\CompLeftSide}$
We will first define a list of assignments for $\OutPCPDecoder$. Specifically, we will have a left outer assignment, $\CompProofOutLeftProof{\CompProofListIndex
}: \OutLeftSide \to \OutLeftAlphabet$, for each $\CompProofListIndex \in [\InListSize]$.

\paragraph{Obtaining assignments to the left side of the Outer PCP.}
For each $\OutLeftVertex \in \OutLeftSide$, we will attempt to obtain candidate assignments to $\OutLeftVertex$ in the outer dPCP by using the list-decoded values of the inner dPCP, $\InPCPDecoder$. To do so, we first use $\CompLeftProof$ to obtain an assignment $\InLeftProof$ for $\InPCPDecoder$, and subsequently apply the list-decoding condition of $\InPCPDecoder$ to these assignments. Set
\[
\InLeftProof[(\aaa; \OutRightVertex, i, r)] := {\CompLeftProof[\CompLeftVertexOther]}\{\OutLeftVertex, i\}, \quad 
\forall (\aaa; \bb, i, r) \in\InLeftSide.
\]
Now that $\InLeftProof$ is an assignment to $\InPCPDecoder$, we can apply the list-decoding guarantee from \cref{cl: translate inner soundness}. We get that there is a list $\InList = \{y^\aaa_1, \ldots, y^\aaa_{\InListSize} \}\subseteq \Sigma_{\OutLeftVertex}$ such that the following holds for any right assignment $T': \InRightSide \to \InRightAlphabet$, 
\begin{equation}
\label{eq: inner list decoding condition other}
\Pr_{(\OutRightVertex,i) \in \Gamma_{\disc, \out}(\OutLeftVertex), r \in [M], (\aaa;b) \sim \mc{P}^{\inner, \OutLeftVertex}_{(\OutRightVertex, i)}((\aaa; \OutRightVertex,i,r), \cdot)} \left[\InEventConstraint  \land \InEventDecoding\right] \leq \InPCPError.
\end{equation}
Here, $\InEventConstraint$ and $\InEventDecoding$ are the following events:
\begin{itemize}
    \item $\InEventConstraint$ (the inner constraint is satisfied): this is the event that $\InLeftProof[(\aaa; \bb, i,r)]$ and $T'[(\aaa;b)]$ satisfy the constraint on $((\aaa; \bb,i,r) , (\aaa,b))$ from $\InPCPDecoder$. 
    \item $\InEventDecoding$ (the inner decoding fails): this is the event that the inner decoder does not output something consistent with the list. That is, the symbols for $\bb$ and $t \in \Gamma_{\out, \dec}(\OutLeftVertex)$ output by the decoder are not (simultaneously) equal to $\left(y^{\OutLeftVertex}_j \right)_{\to \bb}$ and $\left(y^{\OutLeftVertex}_j \right)_{\rightsquigarrow t}$ for any $j \in [\InListSize]$.
\end{itemize}
Using $\InList$ we can now define $\InListSize$ assignments to $\OutLeftSide$. More specifically, for each $\CompProofListIndex \in [\InListSize]$, define the assignment $T^{(j)}_{\OutLeftSide}: \OutLeftSide \to \OutLeftAlphabet$ by
\[
\CompProofOutLeftProof{\CompProofListIndex}[\OutLeftVertex] = {\InListWord_{\CompProofListIndex}}, \quad \forall \OutLeftVertex \in \OutLeftSide.
\]

\paragraph{Outer List decodings for each $\CompProofOutLeftProof{\CompProofListIndex}$.}

Fix a $j \in [\InListSize]$ and consider one of the left outer assignments $\CompProofOutLeftProof{\CompProofListIndex}$. By the list-decoding condition of $\OutPCPDecoder$, there exists a list $\OutList{\CompProofListIndex} = \{\OutListWord{\CompProofListIndex}_1, \ldots, \OutListWord{\CompProofListIndex}_{\OutListSize}\} \subseteq  \CompLanguage$ such that for any assignment to the right side of the outer dPCP, $\OutRightProof:\OutRightSide\to\OutRightAlphabet$, the list-decoding error relative to $\OutList{j}$ is bounded by $\OutPCPError$. Specifically, we have that for each $\CompProofListIndex \in [\InListSize]$ and any right assignment $\OutRightProof$,
\begin{equation} \label{eq: outer list decoding condition other}   
\Pr_{\substack{\OutIndexInDecodingSet\in[n]\\ (\OutLeftVertex, \OutRightVertex)\sim \OutEdgeDistribution{\OutIndexInDecodingSet}}} \left[\OutEventConstraint  \land \OutEventDecoding\right] \leq \OutPCPError,
\end{equation}
where the events above are as follows:
\begin{itemize}
    \item $\OutEventConstraint$ (satisfied outer constraint): this is the event that $\OutConstraint{(\OutLeftVertex,\OutRightVertex)}\left(\CompProofOutLeftProof{\CompProofListIndex}[\OutLeftVertex], \OutRightProof[\OutRightVertex]\right) = 1$, meaning that the outer constraint is satisfied on $(\OutLeftVertex, \OutRightVertex)$.
   
    \item $\OutEventDecoding$ (the outer decoding fails): this is the event that $\OutDecoder{\OutIndexInDecodingSet}(\OutLeftVertex,\CompProofOutLeftProof{\CompProofListIndex}[\OutLeftVertex]) \notin \left\{\Restricted{\left(\OutListWord{\CompProofListIndex}_\ell\right)}{t}\right\}_{\ell\in[\OutListSize]}$, meaning that the decoding does not agree with any member of $\OutList{j}$.
\end{itemize}

\paragraph{Showing List-Decoding Soundness.}

We are now ready to show that the composed dPCP for $\Language$ satisfies $(\OutListSize \cdot \InListSize, \InPCPError + \InListSize\cdot \OutPCPError)$-list-decoding soundness. The list-decoding soundness will be with respect to the following list from $\Language$:
\[
\CompList := \bigcup_{\CompProofListIndex \in [\InListSize]} \OutList{\CompProofListIndex} = \{\OutListWord{\CompProofListIndex}_\ell\}_{\CompProofListIndex \in [\InListSize], \ell \in [\OutListSize]}. 
\]

Fix any right assignment $\CompRightProof: \CompRightSide \to \InRightAlphabet$. For each $\aaa\in\OutLeftSide$ define the following inner right assignment:
\[
\InRightProof[(\aaa; b)] :=  {\CompRightProof[\CompRightVertex]},
\quad \ \forall (\aaa;b) \in \InRightSide.
\]
Let us also define a candidate right assignments for the outer dPCP. For each $\CompLeftVertexOther \in \CompLeftSide$, recall that $\Restricted{\CompLeftProof[\CompLeftVertexOther]}{\OutRightVertex}$ gives an assignment to $\OutRightVertex\in\OutRightSide$. For each $\InSpecificRandomness \in [M]$, recalling the notation from~\eqref{eq:def_squiq_eq}, we define a table $\CompProofOutRightProof{\InSpecificRandomness}$ as
\[
\CompProofOutRightProof{\InSpecificRandomness}[\OutRightVertex] = \Restricted{\CompLeftProof[\CompLeftVertexOther]}{\rightsquigarrow \OutRightVertex}.
\]
Now, we define the following events:
\newcommand{\CompEventConstraint}{\Event_{\Comp,1}}
\newcommand{\CompEventDecoding}{\Event_{\Comp,2}}
\begin{itemize}
    \item $\CompEventConstraint$: this is the event that the constraint on $\CompLeftVertexOther$ and $\CompRightVertex$ is satisfied in the composed dPCP:
\[\CompConstraint{(\CompLeftVertexOther,\CompRightVertex)}\left(\CompLeftProof[\CompLeftVertexOther], \CompRightProof[\CompRightVertex]\right) = 1.
\]
Unpacking the definition of the constraint, this is also the event that the constraint on the edge $(\CompLeftVertexOther,\CompRightVertex)$ of the \emph{inner} dPCP is satisfied by the induced assignments, $\InLeftProof[(\aaa; \bb, i, r)]$ and $\InRightProof[(\aaa; b)]$.

\item $\CompEventDecoding$ (the composed decoding is outside the list): is the event that 
\[
\CompDecoder{t}{\CompLeftVertexOther}(\CompLeftVertexOther,\CompLeftProof[\CompLeftVertexOther]) \notin \left\{\Restricted{\left(\OutListWord{\CompProofListIndex'}_{\ell}\right)}{t}\right\}_{{\CompProofListIndex'\in[\InListSize],\ell \in [\OutListSize]}}
\]
meaning the output of the decoder for the decoding set $t$ is \emph{not} consistent with $w^{j}_\ell$ for any $j \in [\InListSize]$ or $\ell \in [\OutListSize]$.   
\end{itemize}
In the notation of the events above, showing list-decoding soundness of the composed dPCP amounts to showing the following inequality, where the term on the left hand side is precisely the list-decoding error.
\newcommand{\HeavyFraction}{\rho}

\begin{equation} \label{eq: composed list decoding eqn 1 other}
\Pr_{(t, (\bb, r), (\aaa; b)) \sim \mc{Q}_{\comp}} \left[\CompEventConstraint  \land\CompEventDecoding\right] \leq \InPCPError + \InListSize \cdot \OutPCPError.
\end{equation}

It will be convenient now to define a few events. Call the event on the left hand side of \eqref{eq: composed list decoding eqn 1 other} $\Event=\CompEventConstraint  \land\CompEventDecoding$ and define the following events over $t, (\bb, r), (\aaa; b)$ output by the distribution from \eqref{eq: composed list decoding eqn 1 other}:
\begin{itemize}
    \item $\Event_1$ is the event that there does not exist $j \in [\InListSize]$ such that $\CompLeftProof[(\bb, r)]_{\rightsquigarrow \bb} = \left(y^\aaa_j \right)_{\to \bb}$ and $\CompLeftProof[(\bb, r)]_{\rightsquigarrow t} = \left(y^\aaa_j \right)_{\rightsquigarrow t}$ for all $t$ such that $(\aaa, \bb) \in \supp(\mc{P}^{\out}_t)$. 
    
\item $\Event_2$ is the complementary event of the above. That is, that there does exist $j \in [\InListSize]$ such that $\CompLeftProof[(\bb, r)]_{\rightsquigarrow \bb} = \left(y^\aaa_j \right)_{\to \bb}$ and $\CompLeftProof[(\bb, r)]_{\rightsquigarrow t} = \left(y^\aaa_j \right)_{\rightsquigarrow t}$ for all $t$ such that $(\aaa, \bb) \in \supp(\mc{P}^{\out}_t)$.
\end{itemize}
It is clear that $\Event = \left(\Event \land \Event_1\right) \lor \left( \Event \land \Event_2\right)$ so the probability of interest is
\begin{equation}\label{eq: composed list decoding 1 other}
\begin{split}
\Pr_{(t, (\bb, r), (\aaa; b)) \sim \mc{Q}_{\comp}} [\Event] &= \Pr_{(t, (\bb, r), (\aaa; b)) \sim \mc{Q}_{\comp}} [\Event \land \Event_1] + \Pr_{(t, (\bb, r), (\aaa; b)) \sim \mc{Q}_{\comp}}[\Event \; \land \; \Event_2] \\
&\leq\Pr_{(t, (\bb, r), (\aaa; b)) \sim \mc{Q}_{\comp}} [\Event \land \Event_1] + \Pr_{(t, (\bb, r), (\aaa; b)) \sim \mc{Q}_{\comp}}[\Event_{\comp,2} \; \land \; \Event_2].
\end{split}
\end{equation}
We start by bounding $\Pr[\Event \land \Event_1]$. We think of $(t, (\bb, r), (\aaa; b)) \sim \mc{Q}_{\comp}$ as being sampled in the following manner. First, $\aaa \sim \mc{Q}_{\out}(\circ,\cdot,\circ)$ is sampled, then $(\bb,i) \in \Gamma_{\disc, \out}(\aaa), r\in[M]$ are chosen uniformly at random. Finally, $t$ is chosen according to the marginal $\mc{Q}_{\out}(\cdot, \aaa, \bb)$, and  $(t, (\bb, r), (\aaa; b))$ is output. Rewriting the distribution in this manner, one can check that we have
\begin{equation}\label{eq:composed_list_dec_1}
 \Pr_{(t, (\bb, r), (\aaa; b)) \sim \mc{Q}_{\comp}} [\Event \land \Event_1] \leq \E_{\aaa \sim \mc{Q}_{\out}(\circ,\cdot,\circ)}\left[\Pr_{(\bb,i) \in \Gamma_{\disc, \out}(\aaa), r \in [M]}[\mc{F}_{\aaa, \bb,i, r} \land \mc{E}_1] \right],
\end{equation}
where $\mc{F}_{\aaa, \bb,i, r}$ is the event that the assignments, $\InLeftProof[(\aaa; \bb, i, r)] := \CompLeftProof[(\bb,r)]\{\aaa, i\}$ and $ \InRightProof[(\aaa;b)] := \CompRightProof[(\aaa; b)]$, satisfy the constraint on $((\aaa; \bb, i, r), (\aaa;b))$ in $\InPCPDecoder$. We note that the event $\mc{E}_1$ contains the event that the inner decoding on $\InPCPDecoder$ for $(\bb,i)$ is not consistent with any member of $\InList$. That is, that for every $j \in [\InListSize]$, either,
\[
\left(\InDecoder{(\bb, i)}{}((\aaa;\bb, i, r), \InLeftProof[(\aaa; \bb, i, r)])\right)_{\bb} = \CompLeftProof[(\bb, r)]_{\rightsquigarrow \bb} \neq \left(y^\aaa_j\right)_{\to \bb}
\]
or for some $t \in \Gamma_{\dec, \out}(\aaa)$,
\[
\left(\InDecoder{(\bb, i)}{}((\aaa;\bb, i, r), \InLeftProof[(\aaa;\bb, i, r)])\right)_{t} = \CompLeftProof[(\bb, r)]_{\rightsquigarrow t} \neq \left(y^\aaa_j\right)_{\rightsquigarrow t}.
\]
In the above two equations we are using the hardcoding of the alphabet of $(\bb, r)$. Specifically, the equality in the first equation is by the ``hardcoding outer constraints'' and  the equality in the second equation is by the ``hardcoding decoding consistency''.

Hence for each $\aaa$, the inner probability on the right hand side of~\eqref{eq:composed_list_dec_1} is at most the list-decoding error of $\InPCPDecoder$ applied to left and right assignments $\InRightProof, \InLeftProof$. This probability is at most $\InPCPError$ by \cref{cl: translate inner soundness}, so 
\[
 \Pr_{(t, (\bb, r), (\aaa; b)) \sim \mc{Q}_{\comp}} [\Event \land \Event_1] \leq \InPCPError,
 \]
and going back to \eqref{eq: composed list decoding 1 other}, we get
\begin{equation} \label{eq: comp list decoding inter} 
\Pr_{(t, (\bb, r), (\aaa; b)) \sim \mc{Q}_{\comp}} [\Event]  \leq \InPCPError + \Pr_{(t, (\bb, r), (\aaa; b)) \sim \mc{Q}_{\comp}} [\Event_2 \; \land \; \CompEventDecoding].
\end{equation}

It remains to bound the probability on the right hand side of~\eqref{eq: comp list decoding inter}. To this end, for each $\CompProofListIndex \in [\InListSize], \InSpecificRandomness\in [M]$ let $\Event_{\CompProofListIndex,\InSpecificRandomness}$ be the event that $\CompProofOutLeftProof{\CompProofListIndex}[\OutLeftVertex]$ and $\CompProofOutRightProof{\InSpecificRandomness}[\OutRightVertex]$ satisfy the outer constraint, $\OutConstraint{(\OutLeftVertex,\OutRightVertex)}$, or equivalently, the event that $(\CompProofOutLeftProof{\CompProofListIndex}[\OutLeftVertex])_{\to \bb} = \Restricted{(\InListWord_{j})}{\to \bb}$. Also define the event $\mc{F}_{j}$ to be the event that $D^{\out}_t(\aaa, T^{(j)}_{\OutLeftSide}[\aaa]) \notin \{\left(y^{\aaa}_1\right)_{\rightsquigarrow t}, \ldots, \left(y^{\aaa}_{\InListSize}\right)_{\rightsquigarrow t} \}.$
The probability on the right hand side of ~\eqref{eq: comp list decoding inter} can be bounded as
\begin{equation} \label{eq: composed list decoding 2 other}
\begin{split}
\Pr_{(t, (\bb, r), (\aaa; b)) \sim \mc{Q}_{\comp}} [\Event_2 \; \land \; \CompEventDecoding]&
\leq\E_{\InSpecificRandomness \in [M]} \left[\Pr_{(t, (\bb, r), (\aaa; b)) \sim \mc{Q}_{\comp}}\left[\bigvee_{\CompProofListIndex \in [\InListSize]} \Event_{\CompProofListIndex,\InSpecificRandomness} \land \mc{F}_j ~\Bigg|~ r\right]\right]
\\ &\leq \sum_{\CompProofListIndex \in [\InListSize]}\E_{\substack{\InSpecificRandomness \in [M]}}\left[ \Pr_{\substack{ t \in [n]\\(\OutLeftVertex, \OutRightVertex)\sim \OutEdgeDistribution{\CompIndexInDecodingSet}}}\left[\Event_{\CompProofListIndex,\InSpecificRandomness} \land \mc{F}_j \right] \right].
\end{split}
\end{equation}
In the first transition, we use the observation that $r \in [M]$ is chosen uniformly and independently of $t, \bb, \aaa, b$ in the distribution of the probability on the left. In addition, we use the fact that if $\mc{E}_2$ occurs, it must be the case that for some $j \in [\InListSize]$, $T^{(j)}_{\OutLeftSide}[\aaa]$ is consistent with $\CompLeftProof[(\bb, r)]$ on $\bb$ and all $t$ such that $(\aaa, \bb) \in \supp(\mc{P}^{\out}_t)$. Hence, if $\mc{E}_2 \land \mc{E}_{\comp,2}$ occurs, then there must exist $j \in [\InListSize]$ such that  
\[
(\CompProofOutLeftProof{\CompProofListIndex}[\OutLeftVertex])_{\to \bb} = \Restricted{(\InListWord_{j})}{\to \OutRightVertex} \quad \text{and} \quad D^{\comp}_{t}((\bb, r), \CompLeftProof[(\bb,r)]) = D^{\out}_{t}(\aaa, T^{(j)}_{\out}[\aaa]) \notin \left\{\Restricted{\left(\OutListWord{\CompProofListIndex'}_\ell \right)}{t}\right\}_{{\CompProofListIndex'\in[\InListSize],\ell\in [\OutListSize]}}.
\]
This takes care of the first transition. The second transition is by a union bound. Now fix a $j \in [\InListSize]$ and notice that the probability, in the summand of the last line is precisely the list-decoding error of $T^{(j)}_{\OutLeftSide}, T^{(r)}_{\OutRightSide}$, and hence for any $j \in [\InListSize]$ and $r \in [M]$,
\[
\Pr_{\substack{ t \in [n]\\(\OutLeftVertex, \OutRightVertex)\sim \OutEdgeDistribution{\CompIndexInDecodingSet}}}\left[\Event_{\CompProofListIndex,\InSpecificRandomness} \land \mc{F}_j \right] \leq \delta_{\out}.
\]
Plugging this into \eqref{eq: composed list decoding 2 other} and going back to \eqref{eq: comp list decoding inter} we get that the list-decoding error is 
\[
\Pr_{(t, (\bb, r), (\aaa; b)) \sim \mc{Q}_{\comp}} [\Event]\le \InPCPError+\InListSize\cdot \OutPCPError.
\qedhere
\]
\end{proof}

\section{The Final dPCP} \label{sec: final dpcp}
We are now ready to construct our final dPCP, using the constructions and tools developed up to this point. For the remainder of the section, fix a size $N$ circuit $\varphi\colon\Sigma_0^n\to\{0,1\}$ and let $\mc{L} := \sat(\varphi)$.

\subsection{An Intermediate dPCP with Nearly Constant Alphabet Size}\label{subsec:comp-to-polynomial-alph}
Our starting point is the dPCP for $\mc{L}$ given by \cref{thm: final dpcp before comp} which we denote by $\PCPDecoder_1$. By composing it with the inner dPCP from \cref{thm: inner dpcp}, we get the following result.

\begin{theorem}[Intermediate dPCP] \label{thm: inter-pcp}
    For every $\eps > 0$ there exist $L = {\sf poly}(1/\eps)$ and
    $M = O_{|\Sigma_0|}(\poly_{\eps}(\log N))$, $d = O_{\eps}(1)$, $k = O_{|\Sigma_0|}(\poly_{\eps}(\log\log\log N))$ such that the following holds. The language $\mc{L}$ has a projection dPCP of the form
    \[
    \mc{D} = \left(A\cup B, E, \Sigma_A, \Sigma_B, \{\Phi_{e}\}_{e\in E}, \{D_t\}_{t \in [n]}, \{\mc{P}_t \}_{t\in[n]} \right)
    \]
    which satisfies the following
     \begin{itemize}
        \item \textbf{Length.} $O_{|\Sigma_0|}\left( N \cdot \poly_{\eps}(\log N) \right)$.
        \item \textbf{Alphabet Sizes.} The left and right alphabet sizes are $2^{O_{|\Sigma_0|}\left(\poly_{\eps}(\log\log\log N)\right)}$ and $O_{\eps}(1)$ respectively.
        \item \textbf{Decoding Degree.} $O_{\eps}(1)$.
        \item \textbf{Projection Decision Complexity.} $O_{|\Sigma_0|}\left(\poly_{\eps}(\log \log \log N) \right)$.
        \item \textbf{Decoding Complexity.} $O_{\eps, |\Sigma_0|}(1)$.
        \item \textbf{Complete Decoding Distribution.} Call the complete decoding distribution $\mc{Q}$. Then the following holds:
        \begin{itemize}
            \item For all $a \in A, b \in B$, we can write $\mc{Q}(\circ, a, b) = \frac{\wt(a,b)}{M}$ for some integer $\wt(a,b)$.
            \item For every $b \in B$, $\sum_{a \in A} \wt(a,b) = d$.
            \item For every $a \in A$, $\sum_{b \in B} \wt(a,b) = k$.
        \end{itemize}
        \item \textbf{Completeness.} The dPCP has perfect completeness.
        \item \textbf{Soundness.} The dPCP has $(\poly(1/\eps), \eps)$-list-decoding soundness.
    \end{itemize}
\end{theorem}
\begin{proof}
 Let us start by recalling the parameters of our Outer dPCP and inner dPCP in \cref{tab:outer-dpcp} and \cref{tab:inner-dpcp}. Ultimately we will compose two times, where both times the inner dPCP is taken from \cref{tab:inner-dpcp} (although instantiated on different languages and hence different language-related parameters $N', n', T'$). Throughout all tables in this section, we abbreviate complete decoding distribution as CDD. Also, once we specify that a complete decoding distribution is $M$-discrete and can be written as $\mc{Q}(\circ, a, b) = \wt(a,b) / M$, we refer to the left weighted degrees as the sum $\sum_{b \in B} \wt(a,b)$ for each left vertex $a$, and the right weighted degree as the sum $\sum_{a \in A} \wt(a,b)$ for each right vertex $b$. Viewing the bipartite constraint graph as having $\wt(a,b)$ edges between $a$ and $b$ for each pair of vertices $a$ and $b$, these quantities indeed correspond to degrees in this graph. Finally, all of the dPCPs in this section have perfect completeness, and all of the transformations we use preserve perfect completeness, so we do not include perfect completeness in any of the tables.

 Henceforth, fix the circuit $\varphi: \Sigma_0^n  \to \{0,1\}$ from the theorem statement. We will construct a dPCP for $\sat(\varphi)$.
\renewcommand{\arraystretch}{1.4}
\setlength{\tabcolsep}{8pt}

\begin{table}[H]
\centering
\begin{tabular}{@{} l l @{}}
\toprule
\textbf{Parameter} & \textbf{Outer dPCP: $\mc{D}_1$} \\
\midrule
Length & $O_{|\Sigma_0|}( N \cdot \poly_\eps(\log N))$ \\
Left alphabet size & $2^{O_{|\Sigma_0|}(\poly_\eps(\log N))}$ \\
Right alphabet size & $O_\eps(1)$ \\
Decoding Degree & $O_\eps(1)$ \\
Projection Decision Complexity & $O_{|\Sigma_0|}(\poly_{\eps}(\log N))$ \\
Decoding Complexity & $O_{|\Sigma_0|}(\poly_{\eps} (\log N))$ \\
List-Decoding Soundness & $(\poly(1/\eps), \eps)$ \\
CDD Marginal over Edges & 
  $O_{|\Sigma_0|}(N\poly_{\eps} (\log N))$-discrete \\
Left Weighted Degree & all at most $ O_{|\Sigma_0|}(\poly_{\eps} (\log N))$ \\
Right Weighted Degree & all equal and $O_{\eps}(1)$ \\
\bottomrule
\end{tabular}
\caption{The parameters of the dPCP from \cref{thm: final dpcp before comp} applied to the language $\sat(\varphi)$ 
with list-decoding soundness error $\eps$.}
\label{tab:outer-dpcp}
\end{table}

\begin{table}[H]
\centering
\begin{tabular}{@{} l l @{}}
\toprule
\textbf{Parameter} & \textbf{Inner dPCP: $\PCPDecoder_2$} \\
\midrule
Length & $O_{T'}( N'^2 \cdot \poly_\eps(\log N'))$ \\
Left Vertices & $[n'] \cdot [M]$ for $M = O_{T'}(N'\poly_{\eps}(\log N'))$ \\
Left alphabet size & $2^{O_{T'}(\poly_\eps(\log N'))}$ \\
Right alphabet size & $O_\eps(1)$ \\
Decoding Degree & $O_\eps(1)$ \\
Projection Decision Complexity & $O_{T'}(\poly_{\eps}(\log N'))$ \\
Decoding Complexity & $O_{T'}(\poly_{\eps}(\log N'))$ \\
List-Decoding Soundness & $(\poly(1/\eps), \eps)$ \\
CDD Marginal over Edges & $n' \cdot M \cdot O_{T'}(\poly_{\eps} (\log N'))$-discrete\\
&(so that $K(T', N')$ from~\cref{thm: inner dpcp} is $O_{T'}(\poly_{\eps}(\log N'))$) \\
Left Weighted Degree & all equal and $O_{T'}(\poly_{\eps}(\log N'))$\\
\bottomrule
\end{tabular}
\caption{The parameters of the dPCP from \cref{thm: inner dpcp} applied to the language $\sat(\varphi')$ for some circuit $\varphi: (\Sigma'_0)^{n'} \to \{0,1\}$ of size at most $N'$ over an alphabet of size $|\Sigma'_0| \le T'$ with list-decoding soundness error $\eps$.}
\label{tab:inner-dpcp}
\end{table}

\paragraph{First Composition.} We apply the composition of \cref{thm: alph reduction comp} with outer dPCP $\mc{D}_1$ and with inner dPCPs given \cref{thm: inner dpcp}. Referring to the parameters from \cref{tab:outer-dpcp} and \cref{tab:inner-dpcp}, one can check that we get a dPCP, called $\mc{D}_3$, with parameters as described in \cref{tab:composed-dpcp-parameters-1}. In particular, the parameters $N', T'$ arising in the composed dPCP guarantee of \cref{thm: alph reduction comp} are $O_{|\Sigma_0|}(\poly_{\eps}(\log N))$ and $O_{\eps, |\Sigma_0|}(1)$ respectively.

\begin{table}[h!]
\centering
\renewcommand{\arraystretch}{1.4}
\setlength{\tabcolsep}{8pt}
\begin{tabular}{@{} l l @{}}
\toprule
\textbf{Parameter} & \textbf{Composed dPCP: $\PCPDecoder_3$} \\
\midrule
Length &
$O_{|\Sigma_0|}\left( N \cdot \poly_{\eps}(\log N) \right)$ \\
Left alphabet size &
$2^{O_{|\Sigma_0|}\left(\poly_{\eps}(\log\log N)\right)}$ \\

Right alphabet size &
$O_{\eps}(1)$ \\

Decoding Degree& $O_{\eps}(1)$\\

Projection Decision Complexity &
 $O_{|\Sigma_0|}\left(\poly_{\eps}(\log\log N)\right)$ \\

Decoding Complexity & $O_{|\Sigma_0|}\left(\poly_{\eps}(\log\log N)\right)$\\ 
List-Decoding Soundness & $(\poly(1/\eps), \eps)$\\
CDD Marginal over Edges & $O_{|\Sigma_0|}(N \cdot \poly_{\eps}(\log N))$-discrete\\
Left Weighted Degree & all equal and $O_{|\Sigma_0|}(\poly_{\eps}(\log \log N))$\\
\bottomrule
\end{tabular}
\caption{Composing $\PCPDecoder_1$ and $\PCPDecoder_2$. In $\mc{D}_2$, the parameters $N'$ and $T'$ are  $O_{|\Sigma_0|}(\poly_{\eps}(\log N))$ and $O_{\eps, |\Sigma_0|}(1)$ respectively.}
\label{tab:composed-dpcp-parameters-1}
\end{table}
Let $A_3 \cup B_3$ be the sides of the constraint graph, and call $\mc{Q}_3$ the complete decoding  distribution of $\mc{D}_3$. By the discreteness of the marginal of $\mc{Q}_3$ over the edges, we can thus write, for every $(a,b) \in A_3 \times B_3$, $\mc{Q}_3(\circ, a, b) = \frac{\wt(a,b)}{J}$ for some  $\wt(a,b) \in \mathbb{N}$, where $J = O_{|\Sigma_0|}(N \cdot \poly_{\eps}(\log N))$.

The alphabet size is now significantly smaller but is still super-constant. We would like to compose again to further reduce the alphabet size, but we note that now $\mc{D}_3$ does not satisfy all of the requirements of the outer dPCP in \cref{thm: alph reduction comp}. In particular, \cref{thm: alph reduction comp} requires the outer dPCP's complete decoding distribution to have uniform marginal over the right vertices --- a property that we do not have in \cref{tab:composed-dpcp-parameters-1}.

To fix this, we set $d' = O(1/\eps^{100})$, multiply all of the weights by $d'$, and then apply \cref{lm: right degree reduction} with degree $d'$. We let the resulting dPCP be $\mc{D}'_3$, and let its constraint graph have sides $A_3 \cup B'_3$ with complete decoding distribution $\mc{Q}'_{3}$. It is straightforward to check that $\mc{D}'_3$ has all of the same parameters as $\mc{D}_3$ in \cref{tab:composed-dpcp-parameters-1}, except for the following changes
\begin{itemize}
    \item The length becomes $d' J$, where $J = O_{|\Sigma_0|}(N\poly_{\eps}(\log  N))$ is the discreteness parameter of the complete decoding distribution's marginal over the edges.
    \item The discreteness parameter of the CDD Marginal over edges is multiplied by $d'$ and is still $O_{|\Sigma_0|}(N \cdot \poly_{\eps}(\log N))$.
    \item The right weighted degrees are all $d'= O_{\eps}(1)$.
    \item The left weighted degrees are all multiplied by $d'$ and are still all equal to some $O_{|\Sigma_0|}( \poly_{\eps}(\log \log N))$.
\end{itemize}

Now we can compose $\mc{D}'_3$ with the inner dPCP from \cref{thm: inner dpcp} one more time using \cref{thm: alph reduction comp} and again apply  \cref{lm: right degree reduction} with degree $d'$ on the resulting dPCP. This time, the inner dPCPs are applied with parameters $N'$ and $T'$ that are $O_{|\Sigma_0|}(\poly_{\eps}(\log \log N))$ and $O_{\eps, |\Sigma_0|}(1)$. Altogether, we obtain the dPCP, $\mc{D}_4$, whose parameters are listed in \cref{tab:composed-dpcp-parameters-2}. We note that, similarly to in $\mc{D}'_3$, all of the parameters except for weighted right degrees come from \cref{thm: alph reduction comp}. The weighted right degrees from \cref{tab:composed-dpcp-parameters-2} are obtained by applying \cref{lm: right degree reduction} which changes the length to $O_{\eps}(1)$ times the discreteness parameter of the complete decoding distribution's marginal over the edges, and increases all of the other parameters by an $O_{\eps}(1)$ factor which is absorbed. 
\end{proof}

\begin{table}[h!]
\centering
\renewcommand{\arraystretch}{1.4}
\setlength{\tabcolsep}{8pt}
\begin{tabular}{@{} l l @{}}
\toprule
\textbf{Parameter} & \textbf{Composed dPCP: $\PCPDecoder_4$} \\
\midrule

Length &
$O_{|\Sigma_0|}\left( N \cdot \poly_{\eps}(\log N) \right)$ \\
Left alphabet size &
$2^{O_{|\Sigma_0|}\left(\poly_{\eps}(\log \log\log N)\right)}$ \\

Right alphabet size &
$O_{\eps}(1)$ \\

Decoding Degree& $O_{\eps}(1)$ \\

Projection Decision complexity &
 $O_{|\Sigma_0|}(\poly_{\eps}(\log\log \log N))$ \\
 
Decoding Complexity &   $O_{|\Sigma_0|}(\poly_{\eps}(\log\log \log N))$ \\ 

List-Decoding Soundness & $(\poly(1/\eps), \eps)$\\
CDD Marginal over Edges & $O_{|\Sigma_0|}(N \cdot \poly_{\eps}(\log N))$-discrete\\
Left Weighted Degree & all equal and $O_{|\Sigma_0|}(\poly_{\eps}(\log \log \log N))$\\
Right Weighted Degree & all equal and $O_{\eps}(1)$\\
\bottomrule
\end{tabular}
\caption{Composing $\PCPDecoder'_3$ and $\mc{D}_2$. In $\mc{D}_2$, the parameters $N'$ and $T'$ are $O_{|\Sigma_0|}(\poly_{\eps}(\log \log N))$ and $O_{\eps, |\Sigma_0|}(1)$ respectively.}
\label{tab:composed-dpcp-parameters-2}
\end{table}
\subsection{The Final dPCP with Constant Alphabet Size}

Now we are finally ready to prove our main dPCP theorem, \cref{thm:dPCP main}, which we restate below for convenience.

\begin{theorem}\label{thm: main dpcp restate}
    For all $\eps>0$ there exists $C \in \mathbb{N}$ and $L = \poly(1/\eps)$ such that the following holds for all size $N$ circuits $\varphi\colon \Sigma_0^n\to\{0,1\}$ over some alphabet $\Sigma_0$. The language ${\sf SAT}(\varphi)$ has a $2$-query, poly-time constructible dPCP with the following properties
    \begin{itemize}
        \item \textbf{Length.} $O_{|\Sigma_0|}\left(N\log^C N\right)$,
        \item \textbf{Alphabet Size.} $O_{\eps, |\Sigma_0|}(1)$,
        \item \textbf{Complete Decoding Distribution.} The marginal of the complete decoding distribution over the right vertices is uniform and the marginal of the complete decoding distribution over the left vertices is uniform.
        \item \textbf{Completeness.} The dPCP has perfect completeness,
        \item \textbf{Soundness.}  $(L,2\eps^{1/5})$-list-decoding soundness.
    \end{itemize}
\end{theorem}
\cref{thm: main dpcp restate} is obtained from~\cref{thm: inter-pcp} by performing yet 
another composition step. This time, though, 
we will compose with a Hadamard-code based 
dPCP construction, which has a constant size 
alphabet (but an exponential length). Starting from $\mc{D}_4$, however, this exponential size dPCP is applied to a circuit of size $\poly(\log \log \log N)$, hence the overall length increase is only a factor of $\polylogn$, which is affordable to us.
Below we state the properties of the Hadamard dPCPs, which can be found in slightly different forms in \cite{dh, bmv}.

\begin{theorem}[Hadamard dPCP] \label{thm: hadamard-pcp}
    For every $\eps > 0$ there exist $L = {\sf poly}(1/\eps)$ and functions $M, K: \mathbb{N} \times \mathbb{N} \to \mathbb{N}$ satisfying
    \[
    M(T, N) = 2^{O_{\eps,T}(N^2)} \quad \text{and} \quad K(T, N) = O_{\eps}(1)
    \]
    such that the following holds. For every alphabet $\Sigma_0$ of size at most $T$ and circuit $\varphi: \Sigma_0^n \to \{0,1\}$ of size at most $N$, the language $\sat(\varphi)$ has a projection dPCP of the form
    \[
    \mc{D} = \left(\left([n] \times [M]\right)\cup B, E, \Sigma_A, \Sigma_B, \{\Phi_{e}\}_{e\in E}, \{D_t\}_{t \in [n]}, \{\mc{P}_t \}_{t\in[n]} \right)
    \]
    which satisfies the following with $M:= M(T, N)$ and $K := K(T,N)$.
     \begin{itemize}
        \item \textbf{Length.} $O(nM)$.
        \item \textbf{Alphabet Sizes.} The left  alphabet size is $O_{\eps, T}(1)$ and the right alphabet size is $O_{\eps}(1)$. 
        \item \textbf{Degrees.} $\mc{D}$ has decoding degree $1$.
        \item \textbf{Projection Decision Complexity.} $O_{\eps, T}(1)$.
        \item \textbf{Decoding Complexity.} $O_{\eps, T}(1)$.
        \item \textbf{Decoding Distribution.} Let $\mc{Q}$ be the complete decoding distribution. Then, $\mc{Q}$ satisfies the following:
        \begin{itemize}
            \item For each $(t,r) \in [n] \times [M]$ and $b \in B$ we can write $\mc{Q}(\circ, (t,r), b) = \frac{\wt((t,r),b)}{n M K}$ for some  $\wt((t,r),b) \in \mathbb{N}$. 
            \item For each $t \in [n]$, the distribution $\mc{Q}(t, \cdot, \circ)$ is uniform over $\{t\} \times [M]$ and hence for every $(t,r) \in [n] \times [M]$, 
            \[
            \sum_{b \in B} w((t,r),b) = K.
            \]
        \end{itemize}
    
        \item \textbf{Completeness.} The dPCP has perfect completeness.
        \item \textbf{Soundness.} The dPCP has $(L, 2\eps^{1/5})$-list-decoding soundness.
    \end{itemize}
\end{theorem}

The proof appears in \cref{apx:Hadamard}. On a high-level, the construction is identical to the one in~\cite[Theorem 7.7]{bmv}, with minor additional considerations in the analysis.

We can now complete the proof of \cref{thm: main dpcp restate} by composing $\mc{D}_4$ with the Hadamard dPCP from \cref{thm: hadamard-pcp} which we denote by $\mc{D}_{\mathsf{Had}}$.

\vspace{0.2cm}

\begin{proof}[Proof of \cref{thm: main dpcp restate}]
    We use \cref{thm: alph reduction comp} with $\PCPDecoder_4$ as the outer dPCP, and the Hadamard based dPCP $\PCPDecoder_{\sf Had}$ from~\cref{thm: hadamard-pcp} as the inner dPCP. When applying $\PCPDecoder_{\sf Had}$, the parameters $N'$ and $T'$ are $\poly_{\eps}(\log \log \log N)$ and $O_{\eps, |\Sigma_0|}(1)$ respectively. We also have to apply \cref{lm: right degree reduction} again on the resulting dPCP to obtain regularity for the right weighted degrees (and hence uniformity of the marginal of the complete decoding distribution over the right vertices).  We call the resulting dPCP $\mc{D}_5$ and list its parameters in \cref{tab:composed-dpcp-parameters3}. Here, we remark that it is important that the dPCP given by \cref{thm: alph reduction comp} has complete decoding distribution whose marginal over the edges is $O_{|\Sigma_0|}\left(N\cdot \poly_{\eps} (\log N)\right)$-discrete, so the length blowup from applying \cref{lm: right degree reduction} is controlled. Similar to in the previous two compositions, all of the parameters in \cref{tab:composed-dpcp-parameters3} are given by \cref{thm: alph reduction comp}, except for the weighted right degrees. These are obtained by applying \cref{lm: right degree reduction} which changes the length to $O_{\eps}(1)$ times the discreteness parameter of the complete decoding distribution's marginal over the edges, and increases all of the other parameters by an $O_{\eps}(1)$ factor which is absorbed. 
\end{proof}

\begin{table}[h!]
\centering
\renewcommand{\arraystretch}{1.4}
\setlength{\tabcolsep}{8pt}
\begin{tabular}{@{} l l @{}}
\toprule
\textbf{Parameter} & \textbf{Composed dPCP: $\PCPDecoder_5$} \\
\midrule

Length &
$
O_{|\Sigma_0|}\left(N\cdot \poly_{\eps} (\log N)\right) $ \\

Left alphabet size &
$O_{\eps, |\Sigma_0|}(1)$ \\

Right alphabet size &
$O_{\eps}(1)$ \\

List-decoding soundness &
$(\poly(1/\eps), \eps)$ \\

CDD Marginal over Edges & $O_{|\Sigma_0|}\left(N\cdot \poly_{\eps} (\log N)\right)$-discrete \\

Left Weighted Degree & all equal \\
Right Weighted Degree & all equal \\
\bottomrule
\end{tabular}
\caption{Composing $\PCPDecoder_4$ and $\mc{D}_{\mathsf{Had}}$. In $\mc{D}_{\mathsf{Had}}$, the parameters $N'$ and $T'$ are $O_{|\Sigma_0|}(\poly_{\eps}(\log \log \log N))$ and $O_{\eps, |\Sigma_0|}(1)$ respectively.}
\label{tab:composed-dpcp-parameters3}
\end{table}

%% file: applications.tex
\section{Applications}
In this section we describe applications of our main dPCP result (\cref{thm:dPCP main}) to PCPPs, RLDCs, and CSP-reconfiguration.

\subsection{Proof of~\cref{thm:PCPP_intro}}
Our 3-query PCPP construction is based on the dPCP from~\cref{thm:dPCP main}, along with an additional consistency check. Due to the latter, we have to take into account the list size of the dPCP, which makes our PCPP effective only for words which are very far from the language. 
\begin{theorem}\label{thm:near-linear-pcpp}
For any $\eps>0$ the following holds for any alphabet $\Sigma_0$. For all circuits $\varphi:\Sigma_0^n\to \{0,1\}$ of size $N$, there exists a 3-query PCPP $\Psi=(X\cup A\cup B, E,\mc{P}, \Sigma_0\cup \Sigma_A\cup \Sigma_B, \{\Phi_{e}\}_{e\in E})$ for the language $\mathsf{SAT}(\varphi)$ with the following parameters:
\begin{itemize}
    \item \textbf{Length.} $|A|,|B|=O_{|\Sigma_0|}( N \cdot \poly_{\eps} \log N)$, $|X|=n$ and the edge set $E\subseteq X\times A\times B$ consists of $O_{|\Sigma_0|}( N \cdot \poly_{\eps} \log N)$ hyperedges, each has size $3$ and includes one vertex from each one of $A,B$ and $X$.
    \item \textbf{Distribution.} The marginal distribution of $\mc{P}$ on $B$ is uniform.  
    \item \textbf{Alphabet Size. } Vertices in $A,B$ have alphabets $\Sigma_A, \Sigma_B$ of size $O_{|\Sigma_0|,\eps}(1)$ each, the alphabet of $X$ is $\Sigma_0$.
    \item \textbf{Soundness.} There exists a universal constant $C>0$ such that for any $\eta\ge 1-\eps^{C+1}$, if $w:X\to \Sigma_0$ is $\eta$-far from $\mathsf{SAT}(\varphi)$ then for any $T_A:A\to \Sigma_A, T_B:B\to \Sigma_B$, 
    \[
    \Pr_{(x,a,b)\sim \mc{P}}[\Phi_{(x,a,b)}(w_t,T_A[a],T_B[b])=1]\le 2\eps.
    \]
\end{itemize}
\end{theorem}

\begin{proof}Fix $\eps>0$. We use \cref{thm:dPCP main}
to construct a dPCP for $\mathsf{SAT}(\varphi)$: $$    \PCPDecoder = \left(\LeftSide\cup \RightSide,\EdgeSet,\LeftAlphabet,\RightAlphabet,\{\Constraint{e} \}_{e\in \EdgeSet},\{\EdgeDistribution{\IndexInDecodingSet}\}_{\IndexInDecodingSet \in \SideDecodings}, \{\DecoderSymbol_{\IndexInDecodingSet} \}_{\IndexInDecodingSet \in \SideDecodings}\right)$$ 
with $(\mathsf{L}, \eps)$ list-decoding. Then, we construct a PCPP $\Psi = (V', E',\mc{P}', \Sigma', \{\Phi'_e\}_{e \in E'})$ with $X\subseteq V'$ as follows: $V'=\LeftSide\cup \RightSide\cup X$ where $|X|=n$, $E'=\{(t,\LeftVertex,\RightVertex)\ : \ t\in X, (a,b)\in\supp(\EdgeDistribution{t})\}$, $\Sigma'=\LeftAlphabet\cup\RightAlphabet\cup \Sigma_0$. For an edge $(t,a,b)\in E'$, the satisfying assignments of the constraint $\Phi_{(t,a,b)}$ are tuples $(\sigma,\alpha,\beta)$ such that $(\alpha,\beta)$ satisfies the constraint $\Phi_{(a,b)}$ of $\mc{D}$, and $D_t(a,\alpha)) = \sigma$.  Finally, the distribution ${\mc P'}$ is the complete distribution of $\mc{D}$, which we denote by $\mc{Q}$.

The completeness, query complexity, number of vertices, edges and alphabet size of $\Psi$ follow directly from the construction.  The marginal distribution over $B$ is uniform as $\mc{Q}(\circ,\circ,\cdot)=\mc{P}'(\circ,\circ,\cdot)$ is uniform by~\cref{thm:dPCP main}. 

We prove the soundness condition. Suppose $w\in \Alphabet{0}^n$ is $\eta$-far from $\mathsf{SAT}(\varphi)$, and fix assignments $\LeftProof:\LeftSide\to \LeftAlphabet$ and $\RightProof:\RightSide\to \RightAlphabet$ to the dPCP $\PCPDecoder$. By the list-decoding property of the dPCP, there exists a list $\{u_1, \ldots, u_\mathsf{L} \} \subseteq \mathsf{SAT}(\varphi)$ such that 
\begin{equation}\label{eq:pcpp_deri_list}
    \Pr_{(t,a,b)\sim \mc{Q}}[\Phi_{(\LeftVertex,\RightVertex)}(\Assignment_{\LeftSide}[\LeftVertex], \Assignment_{\RightSide}[\RightVertex])=1 \land \DecoderSymbol_{\IndexInDecodingSet}(\LeftVertex,\Assignment_{\LeftSide}[\LeftVertex]) \notin \{\Restricted{(u_i)}{\IndexInDecodingSet} \}_{i \in [\ListSize]}] \leq \eps.
\end{equation}
We conclude that
\begin{align*}
&\Pr_{(t,a,b)\sim \mc{P}'}[\Phi_{(t,a,b)}(w_t,\LeftProof[a],\RightProof[b]) = 1]\\ 
&= \Pr_{(t,a,b)\sim \AllDistribution}[\Phi_{(a,b)}(\LeftProof[a],\RightProof[b])=1\land (w_t=\DecoderSymbol_{t}(a,\LeftProof[a]))]\\
&= \Pr_{(t,a,b)\sim \AllDistribution}[\Phi_{(a,b)}(\LeftProof[a],\RightProof[b])=1\land (w_t=\DecoderSymbol_{t}(a,\LeftProof[a]))\land\DecoderSymbol_{t}(\LeftVertex,\Assignment_{\LeftSide}[\LeftVertex]) \notin \{\Restricted{(u_i)}{t} \}_{i \in [\ListSize]}]\\&+ \Pr_{(t,a,b)\sim \AllDistribution}[\Phi_{(a,b)}(\LeftProof[a],\RightProof[b])=1\land (w_t=\DecoderSymbol_{t}(a,\LeftProof[a]))\land\DecoderSymbol_{t}(\LeftVertex,\Assignment_{\LeftSide}[\LeftVertex]) \in \{\Restricted{(u_i)}{t} \}_{i \in [\ListSize]}]\\
&\le \eps + \Pr_{(t,a,b)\sim \AllDistribution}[ w_t \in \{\Restricted{(u_i)}{t} \}_{i \in [\ListSize]}]\\
&\le \eps + \sum_{i\in[\ListSize]}\Pr_{(t,a,b)\sim \AllDistribution}[ w_t =\Restricted{(u_i)}{t} ]\\
&\le \eps+\ListSize\cdot (1-\eta),
\end{align*}
where in the third transition we us~\eqref{eq:pcpp_deri_list}, in the fourth transition we use the union bound, and in the last transition we use the fact that $u_i\in\sat(\varphi)$, that $w$ is $\eta$-far from $\mathsf{SAT}(\varphi)$ and that the marginal distribution of $\mc{Q}$ over $t$ is uniform. Since $L=1/\eps^C$ for some absolute constant $C>0$, given that $\eta\ge 1-\eps^{C+1}$ we conclude the the soundness error is at most $2\eps$. 
\end{proof}

\subsection{Proof of~\cref{thm:4PCPP_intro}}
Our 4-query PCPP with optimal proximity/soundness guarantees is constructed via a similar transformation as in~\cref{thm:near-linear-pcpp} with an additional encoding idea.
   
\begin{theorem}\label{thm:polynomial-size-pcpp}
For any $\eps>0$ the following holds for any alphabet $\Sigma_0$. For all circuits $\varphi:\Sigma_0^n\to \{0,1\}$ of size $N$ over $\Sigma_0$, there exists a 4-query PCPP $\Psi=(X\cup Y\cup A\cup B, E, \mc{P},\Sigma_0\cup\Sigma\cup \Sigma_A\cup \Sigma_B, \{\Phi_{e}\}_{e\in E})$ for the language $\mathsf{SAT}(\varphi)$ with the following parameters:
\begin{itemize}
    \item \textbf{Length.} $|A|,|B|=O_{|\Sigma_0|}(N\cdot \poly_{\eps}\log(N))$, $|X|=n, |Y|=O_{\eps}(n)$ and the edge set 
    $E\subseteq X\times Y\times  A\times B$ 
    consists of $O_{|\Sigma_0|}(N\cdot \poly_{\eps}\log(N))$ hyperedges, each has size $4$ and includes one vertex from each one of $A,B,X,Y$.
    \item \textbf{Alphabet Size. } Vertices in $A,B$ have alphabets $\Sigma_A, \Sigma_B$ of size $O_{|\Sigma_0|,\eps}(1)$ each, the alphabet of $X$ is $\Sigma_0$ and the alphabet of $Y$, denoted by $\Sigma_Y$, has size $O_{\eps}(1)$.
    \item \textbf{Soundness.} For any $\eta>0$, if $w:X\to \Sigma_0$ is $\eta$-far from $\mathsf{SAT}(\varphi)$ then for any $T_A:A\to \Sigma_A, T_B:B\to \Sigma_B, T_Y:Y\to \Sigma_Y$, 
    \[
    \Pr_{(x,y,a,b)\sim \mc{P}}[\Phi_{(x,y,a,b)}(w_x,T_Y[y],T_A[a],T_B[b])=1]\le 1-\eta+\eps.
    \]
\end{itemize}
\end{theorem}

\begin{proof}
Let $\mc{C}:\Sigma_0^{n}\to\Sigma^{n/\rho}$ be a 
linear time encodable code with distances $1-\mu$ as in~\cref{lm: code GI}; we will choose the parameter $\mu$ later and then $\rho = \rho(\mu)>0$ is picked. Define the circuit $\varphi':\Sigma_0^n\times \Sigma^{n/\rho}\to \{0,1\}$ which on input $(m,m')\in \Sigma_0^n\times \Sigma^{n/\rho}$ 
outputs $1$ only if $\mc{C}(m)=m'$ and $\varphi(m)=1$.  Note that since $\varphi$ is of size $N$ and the encoding circuit is linear in $n$, the size of $\varphi'$ is $\widetilde{O}(N+n/\rho)$. Also, we have 
\[
\sat(\varphi') = 
\{(m,m')\in \Sigma_0^n\times \Sigma^{n/\rho}~|~m\in \sat(\varphi), m' = \mc{C}(m)\}.
\]

Let $d\in\mathbb{N}$ be a parameter to be chosen later, and let $\mc{S}\subseteq[n]\times [n/\rho]$ 
be the edges of a bipartite expander with sides $[n]\times [n/\rho]$ and right degree $d$, left degree $d/\rho$  and second singular value at most $O(1/\sqrt{d})$. We note that $\mc{S}$ can be constructed in polynomial time by~\cref{lm: poly time bip expander}.

Consider the circuit $\varphi'':(\Sigma\times \Sigma_0)^{|\mc{S}|}\to \{0,1\}$ defined as follows. Given an input $w'':\mc{S}\to \Sigma\times \Sigma_0$, define $w:[n]\to \Sigma_0$ by picking for each $i\in [n]$ an arbitrary $j$ that $(i,j)\in\mc{S}$ and
then defining $w(i) = w''(i,j)|_{i}$, and similarly
define $u:[n/\rho]\to \Sigma$. 
The output $\varphi''(w'')$ is $1$ if:
\begin{enumerate}
    \item For all $(i,j)\in\mc{S}$, $w''(i,j) = (w(i),u(j))$.
    \item $\phi'(w,u) = 1$.
\end{enumerate} 
We note that the size of $\varphi''$ is
$\widetilde{O}(|\mc{S}|)$ plus the size of
$\varphi'$, which is $\widetilde{O}_{d,\rho}(N)$.
Using~\cref{thm:dPCP main} we get a dPCP for the language $\mathsf{SAT}(\varphi'')$ with $(\mathsf{L}, \eps)$-list decoding, which we denote as
$$\PCPDecoder = \left(\LeftSide\cup \RightSide,\EdgeSet,\LeftAlphabet,\RightAlphabet,\{\Constraint{(a,b)} \}_{(a,b)\in \EdgeSet},\{\EdgeDistribution{(x,y)}\}_{(x,y)\in\mc{S}}, \{\DecoderSymbol_{(x,y)} \}_{(x,y)\in\mc{S}}\right).
$$  
We denote the complete decoding distribution of $\mc{D}$ by $\mc{Q}$. We think of the decoder on $(x,y)\in \mc{S}$ as outputting two elements, which we call $D_{(x,y)}(\cdot,\cdot)_x$ and $D_{(x,y)}(\cdot,\cdot)_y$.

We construct the PCPP for $\mathsf{SAT}(\varphi)$ as $\Psi=(X\cup Y\cup A\cup B, E', \mc{P}',\Sigma_0\cup \Sigma_A\cup \Sigma_B,\{\Phi_{e'}\}_{e\in E'})$ where $|X|=n,|Y|=n/\rho$, the edge set is
\[
E' = \{((x,y),a,b)~|~(x,y)\in\mc{S}, (a,b)\in\supp(\mc{P}_{(x,y)})\}
\]
and the distribution $\mc{P}'$ is taken to be $\mc{Q}$.
The alphabet of $A$ is $\Sigma_A$, the alphabet of $B$ is $\Sigma_B$, the alphabet of $X$ is $\Sigma_0$ and the alphabet of $Y$ is $\Sigma$. For $((x,y),a,b)\in E'$, a tuple $((\sigma,\tau),\alpha,\beta)$ satisfies the constraint $\Phi_{((x,y),a,b)}$ if 
$(\alpha,\beta)$ satisfies $\Phi_{(a,b)}$, $D_{(x,y)}(a,\alpha)_x = \sigma$, and $D_{(x,y)}(a,\alpha)_y=\tau$. 

The size and the other parameters of the PCPP follow from the parameters of the dPCP from \cref{thm:dPCP main} where the circuit size is $|\varphi''|$, and in the rest of the argument we argue the soundness of $\Psi$. Let $w:X\to\Sigma_0$ be $\eta$-far from $\sat(\varphi)$, and fix $T_A:A\to\Sigma_A$, $T_B:B\to\Sigma_B$ and  $T_Y:Y\to\Sigma$. By the list-decoding soundness property of the $\mc{D}$ there exists a list $\{w''_i\}_{i\in \mathsf{L}}\subseteq\sat(\varphi'')$ such that
\begin{equation}\label{eq:4_q_pcpp_list}
\Pr_{((x,y),a,b)\sim \AllDistribution}[\Phi_{(a,b)}(T[a],T[b]) = 1\; \land \; \DecoderSymbol_{(x,y)}(a,T[a])\notin \{w_i''(x,y)\}_{i\in [\mathsf{L}] }]\le \eps.
\end{equation}
As $w''_i\in \sat(\varphi'')$ we conclude that there exists $u_i\in\sat(\varphi)$, such that we have $w''(x,y) = u_i(x),\mc{C}(u_i)_y$ for all $(x,y)\in \mc{S}$, and we fix such $u_i$ henceforth.
Let $\mc{E}$ be the event that $\Phi_{(a,b)}(T_A[a],T_B[b]) = 1\; \land \; \DecoderSymbol_{(x,y)}(a,T_A[a])_x=w_x\; \land \; \DecoderSymbol_{(x,y)}(a,T_A[a])_y=T_Y[y]$. 
The soundness of the PCPP is thus $\Pr_{((x,y),a,b)\sim \mc{P}'}[\mc{E}]$, 
which can be written as
\begin{align*}
\Pr_{((x,y),a,b)\sim \AllDistribution}[\mc{E}]&= \Pr_{((x,y),a,b)\sim \AllDistribution}[\mc{E}\; \land \; \DecoderSymbol_{(x,y)}(a,T_A[a])\notin \{w_i''(x,y)\}_{i\in [\mathsf{L}] }]\\&+\Pr_{((x,y),a,b)\sim \AllDistribution}[\mc{E} \land \; \DecoderSymbol_{(x,y)}(a,T_A[a])\in \{((u_i)_{x},(\mc{C}(u_i))_{y})\}_{i\in [\mathsf{L}] }].
\end{align*}
The first term is at most $\eps$ by~\eqref{eq:4_q_pcpp_list}, and we now focus on the second term. Let $\mc{E}_y$ be the event that there are $i\neq i''$ such that $\mc{C}(u_i)_y = \mc{C}(u_{i'})_y$. We get that
\begin{align*}
&\Pr_{((x,y),a,b)\sim \AllDistribution}[\mc{E}\; \land \; \DecoderSymbol_{(x,y)}(a,T_A[a])\in \{((u_i)_{x},(\mc{C}(u_i))_{y})\}_{i\in [\mathsf{L}] }]\\
&\le \Pr_{((x,y),a,b)\sim \AllDistribution}[\mc{E}\; \land \; \DecoderSymbol_{(x,y)}(a,T_A[a])\in \{((u_i)_{x},(\mc{C}(u_i))_{y})\}_{i\in [\mathsf{L}] }\; \land \; \bar{\mc{E}}_y]+\Pr_{((x,y),a,b)\sim \AllDistribution}[ \mc{E}_y].
\end{align*}
Note that for any $i\neq i'$, the words $\mc{C}(u_i)$, $\mc{C}(u_{i'})$ agree on at most $\mu$ fraction of coordinates, so by the union bound 
$\Pr_{((x,y),a,b)\sim \AllDistribution}[ \mc{E}_y]\leq \binom{L}{2}\mu$. As for the other term on the right hand side, note that as $\mc{E}$ requires that $\DecoderSymbol_{(x,y)}(a,T_A[a])=(w_x,T_Y[y])$ we get
\begin{align*}
&\Pr_{((x,y),a,b)\sim \AllDistribution}[\mc{E}\; \land \; \DecoderSymbol_{(x,y)}(a,T_A[a])\in \{((u_i)_{x},(\mc{C}(u_i))_{y})\}_{i\in [\mathsf{L}] }\; \land \; \bar{\mc{E}}_y]\\&\le \Pr_{(x,y)\in \mc{S}}[(w_x,T_Y[y])\in\{((u_i)_{x},(\mc{C}(u_i))_{y})\}_{i\in [\mathsf{L}] }\; \land \; \bar{\mc{E}}_y].  
\end{align*}
If the event above holds, then the values $\{(\mc{C}(u_i))_y\}$ are all different, so at most one of them is equal to $T_Y[y]$.   
Defining $A_I=\{y\ : \ \bar{\mc{E}_y}\; \land \; (\mc{C}(u_I))_y=T_Y[y]\}$ and $B_I=\{x\ : \ w_x= (u_I)_x\}$ for $I\in [\mathsf{L}]$, we get that
\begin{align*}
\Pr_{(x,y)\in \mc{S}}[(w_x,T_Y[y])\in\{((u_i)_{x},(\mc{C}(u_i))_{y})\}_{i\in [\mathsf{L}] }\; \land \; \bar{\mc{E}_y}]&= \sum_{I\in [\mathsf{L}]}\mathbb{E}_{(x,y)\in \mathcal{S}}[1_{y\in A_I}\cdot 1_{x\in B_I}]\\
&\le \sum_{I\in [\mathsf{L}]}\mathbb{E}_{y\in [n/\rho]}[1_{y\in A_I}]\cdot \mathbb{E}_{x\in[n]}[1_{x\in B_I}]+\frac{1}{\sqrt{d}}\cdot \sqrt{\frac{|A_I|}{n/\rho}\cdot \frac{|B_I|}{n}}\\
&\le 1-\eta+\frac{\mathsf{L}}{\sqrt{d}}.
\end{align*}
In the second transition we used~\cref{lm: expander mixing}, and in the last transition we used the fact that $\mathbb{E}_{x\in[n]}[1_{x\in B_I}]\leq 1-\eta$ for all $I$ (as $w$ is $\eta$-far from $\sat(\varphi)$)  
and that $\sum_{I\in [\mathsf{L}]}\mathbb{E}_{y\in [n/\rho]}[1_{y\in A_I}]=\Pr_{((x,y),a,b)\sim \AllDistribution}[ \bar{\mc{E}}_y]\leq 1$.

Combining all of the above, we get that
\begin{align*}
\Pr_{((x,y),a,b)\sim \AllDistribution}[\Phi_{(a,b)}(T_A[a],T_B[b]) = 1\; \land \; \DecoderSymbol_{(x,y)}(a,T_A[a])_x=w_x]\le 1-\eta+\binom{\mathsf{L}}{2}\cdot \mu
+\frac{\mathsf{L}}{\sqrt{d}}.
\end{align*}
As $\mathsf{L}=1/\eps^C$ for some absolute constant $C$ by \cref{thm:dPCP main}, choosing $\mu=\eps^{2C}$ and $d=4\eps^{-2C}$, 
we get that the error is at most $1-\eta+\eps$.  
\end{proof}

\subsection{A Query-Preserving PCPP to RLDC Transformation}
To prove~\cref{thm:RLDC_intro} we describe a generic transformation that converts a sufficiently good $3$-query PCPP into a $3$-query RLDC. 

\begin{theorem}\label{thm:transformation_rldc}
Suppose for some absolute constant $C>0$ we have the following two ingredients:
\begin{itemize}
    \item For any alphabet $\Sigma_0$, constant $\delta>0$, and circuit $\varphi\colon \Sigma_0^n\to\{0,1\}$ of size $N$ there is a 3-query PCPP $\Psi = (X\cup A\cup B, E,\mc{P}, \Sigma_0\cup \Sigma, \{\Phi_e\}_{e \in E})$ satisfying:
\begin{itemize}
    \item \textbf{Length.} $O_{|\Sigma_0|}( N \cdot \poly_{\eps}(\log N))$.
    \item \textbf{Queries.} Each hyperedge $e = (e_1, e_2, e_3)\in E$ satisfies $e_1 \in X, e_2 \in A, e_3 \in B$, and choosing $e$ uniformly, $e_1$ is uniformly random in $X$, $e_2$ is uniformly random in $A$, and $e_3$ is uniformly in $B$.
    \item \textbf{Soundness.} For all $\eta \geq 1 - \delta^{C}$, if $w:X\to \Sigma_0$ is $\eta$-far from $\mathsf{SAT}(\varphi)$, then for any $T_A:A\to \Sigma_A, T_B:B\to \Sigma_B$ 
    \[
    \Pr_{(t,a,b)\sim \mc{P}}[\Phi_{(t,a,b)}(w_t,T_A[a],T_B[b])=1]\le 2\delta.
    \]
\end{itemize}
\item For any $\delta\in (0,1)$ there is $\rho>0$, an alphabet $\Sigma_0$ of size $O_{\delta}(1)$, and a code $\mathcal{C}_0:\{0,1\}^k\to\Sigma_0^{k/\rho}$ with relative distance $\delta$, that is both quasi-linear time encodable and decodable.
\end{itemize}
Then for all $\mu \in (0,1)$ sufficiently small, and sufficiently large $k$, there is a $3$-query RLDC with message length $k$, blocklength $k^2\cdot \poly_{\mu}(\log k)$, alphabet size $O_{\mu}(1)$, distance $\frac{1-\mu}{2}$, and decoding radius $\mu$. Furthermore, for all $\delta>0$, the RLDC has success-rate $1-16\delta-o(1)$ on inputs that are $\delta$-close to a valid codeword.
\end{theorem}

\begin{proof}
Fix $\mu > 0$ sufficiently small and take $\rho>0$ and a code $\mc{C}_0$ as per the second item with distance $1-\mu$, rate $\rho$, and appropriate alphabet $\Sigma_0$. For an index $i \in [k]$ and a bit $b \in \{0,1\}$, let $\varphi_{i,b}:\Sigma_0^{k/\rho_0}\to\{0,1\}$ be a circuit that given a word $w$ of length $k/\rho$, outputs $1$ if $w$ is a member of $\mc{C}_0$, and furthermore the $i$'th bit of the encoded message is equal to $b$. More precisely on an input $w\in \Sigma_0^{k/\rho}$, run the decoding map to get $m=\mathsf{Dec_{\mc{C}_0}}(w)$, then output $1$ if and only if  $\mathsf{\mc{C}_0}(m)=w$ and $m_i = b$. Since the code $\mc{C}_0$ is quasi-linear time encodable and decodable, the circuit size of $\varphi_{i,b}$ is $\widetilde{O}_{\mu}(k)$. Define  $\mc{L}_{i,b}\coloneqq\mathsf{SAT}(\varphi_{i,b})$.  

\paragraph{Unifying the proof length for PCPP invocations.} For $i \in [k], b \in \{0,1\}$, take 
\[
\Psi_{i,b} = ([k/\rho]\cup A_{i,b}\cup B_{i,b}, E_{i,b}, \Sigma, \{\Phi_e\}_{e \in E},\mc{P}_{i,b})
\]
to be the PCPP from the premise of the statement for the language $\mc{L}_{i,b}$. Without loss of generality, we assume that all the PCPPs have the same alphabet which has size $O_{\eps}(1)$. 
The code we construct will contains multiple encodings of the message and relevant proofs, but we first address a small technical issue. The PCPPs are guaranteed to have length that is bounded by $k\cdot \log^C k$ (as the circuit size is $\widetilde{O}_{\mu}(k)$), but $A_{i,b}$ (as well as $B_{i,b}$) may differ in size across $i$ and across $b$. Since each proof should be contained in a predefined part of the codeword we construct below, we must ensure that the right and left sizes of the PCPP have the same size across all invocations. Towards that, we set $C'=C+10$ and let $K=k\cdot \log ^{C'}k$. Given assignments to the PCPP  $T_{A_{i,b}}:A_{i,b}\to\Sigma$ and $T_{B_{i,b}}:B_{i,b}\to\Sigma$ construct $\pi_{i,b}=(\pi_{i,b,A},\pi_{i,b,B})$ as follows:
\begin{itemize}
    \item $\pi_{i,b,A}=\left(T_{A_{i,b}}^{\lfloor K/|A_{i,b}|\rfloor},\nu^{K-\lfloor K/|A_{i,b}|\rfloor \cdot |A_{i,b}|}\right)$, that is, repeat $T_{A_{i,b}}$ in full $\lfloor K/|A_{i,b}|\rfloor$-times, and put $\nu$ in the rest to make the overall size exactly $K$. Here, $\nu$ is an arbitrary symbol outside of $\Sigma$.
    \item $\pi_{i,b,B}=\left(T_{B_{i,b}}^{\lfloor K/|B_{i,b}|\rfloor}, \nu^{K-\lfloor K/|B_{i,b}|\rfloor\cdot |B_{i,b}|}\right)$, that is, repeat $T_{B_{i,b}}$ in full $\lfloor K/|B_{i,b}|\rfloor$-times, and put $\nu$ in the rest to make the overall size exactly $K$.
\end{itemize}
Note that the uniform distribution over $[K]$ and the distribution over $[K]$ which is uniform over the subset $[\lfloor K/|A_{i,b}|\rfloor\cdot |A_{i,b}|]$, are $1/\log^{10}k$-close in statistical distance; similarly for $B_{i,b}$.
We think of the PCPP as the string $\pi_{i,b}=(\pi_{i,b,A},\pi_{i,b,B})$ of length $2K$ over alphabet symbols from $\Sigma$. We denote by $\pi_{i,b}(w)$ an arbitrary assignment to the PCPP in the completeness case for $w$. By the perfect completeness of the PCPP, it is guaranteed that such an assignment exists which satisfies all of the constraints with $w$.
\vspace{-1ex}
\paragraph{The construction.} 
Let $\Sigma' = \Sigma \times \{0,1\}$. We define $\mc{C}\colon \{0, 1\}^k \to \Sigma_0^{r \cdot k/\rho}\times (\Sigma')^{2K \cdot k}$ as follows:
\[
\mc{C}(m) = \mc{C}_0(m)^{r} \pi'_{1, m_1}(\mc{C}_0(m)) \cdots \pi'_{k, m_k}(\mc{C}_0(m))
\]
for $r = 2\rho K$, where the proof $\pi_{i,m_i}'$ are constructed as follows. We first construct $\pi_{i,m_i}(\mc{C}_0(m))$, and to get $\pi'_{i, m_i}(\mc{C}_0(m))$ we append the bit $m_i$ to each entry of $\pi_{i,m_i}(\mc{C}_0(m))$. This results in doubling the alphabet size of the proof. Overall, we get that a codeword consists of two equal length blocks:
\begin{itemize}
    \item The first block consists of $r$ repetitions of $\mc{C}_0(m)$. 
    \item The second block consists of the concatenation of $\pi'_{i, m_i}(\mc{C}_0(m))$ for $i = 1, \ldots, k$, where each entry of the proof has the bit $m_i$ ``appended''.
\end{itemize}
Moreover, the second block can be further partitioned into two equal sized blocks consisting of the ``$A$''-vertices and ``$B$''-vertices of each PCPP.

\vspace{-1ex}
\paragraph{The decoder.} 
Given an index $i \in [k]$ and $w \in \Sigma_0^{r\cdot k/\rho}\times \Sigma^{k\cdot r}$, the decoder $\Dec(w, i)$ proceeds in the following way: 
\begin{enumerate}
    \item Choose a random repetition of out of the $r$ repetitions in the first block of $w$. Call this $w_0$.
    \item Sample a uniform index $b'$ in $[K]$, and read $\pi_{i,m_i,B}'$ at index $b'$. 
    If that symbol is $\nu$ output $\bot$, and else proceed.

    Recall that $\pi_{i,m_i,B}'$ contains multiple copies of assignments to the right-side $B_{i,m_i}$. The index $b'$ corresponds to a vertex in $b\in B_{i,m_i}$, and the distribution we chose is $(1/\log^{10}k)$-close to being uniform over $B_{i,m_i}$. 
    We note that as the supposed value of $m_i$ is
    not known to the decoder yet, it does not know the correct distribution over the locations that need to be read in the PCPP (as they may be different for each one of the languages $\mathcal{L}_{i,m_i}$). However, since the distribution over the query on the $B$-side is uniform for each of $\mathcal{L}_{i,m_i}$, the query we made in this step is sampled according to a distribution which is statistically close to the correct distribution.
        
    \item Let $x$ be the bit appended, so that supposedly $x = m_i$. 

    \item The decoder now knows which language $\mc{L}_{i,x}$ was supposedly used for the PCPP encoding, and in particular the decoder can calculate the length of $A_{i,x}$. Thus, the decoder also knows which locations in the $A$-part of $\pi_{i,x}(\mc{C}_0(m))$ are the non-$\nu$ ones. The decoder samples a uniformly random index $j \in [\lfloor K / |A_{i,x}\rfloor]$ and looks at the $j$th copy of $\pi_{i,x}(\mc{C}_0(m))$ from the $A$-part.
    The verifier samples $(t,a,{\bf b})\sim \mc{P}_{i,x}$ conditioned on ${\bf b}=b$, and then queries the locations corresponding to $a$ in the $j$th copy of the $A$-part of $\pi_{i,x}(\mc{C}_0(m))$ and $t$ in $w_0$. 
    
    \item If the read assignment to $a$ does not have the bit $x$ appended, output $\perp$. 
    Else, remove the appended bit from each of the read symbols and proceed.
    \item If the read assignment satisfies the constraint on $(t,a,b)$ of $\Psi_{i,x}$, output $x$. Otherwise, output $\perp$.
\end{enumerate}

\paragraph{Analysis of the construction.}We now analyze the construction. 
It is easy to see that the code has distance at least $\frac{1-\mu}{2}$ due to the repetitions in the first block. The length and completeness for relaxed local decoding are clear from the corresponding properties of the PCPP.

For the success rate, suppose the codeword, $w$, is $\delta$-close to some valid codeword $w'$. Let $I_1, I_2,$ and $I_3$ be the $\mc{C}_0(m)^r$ part, the $A$-parts of the proofs, and the $B$-parts of the proofs respectively. The queries of the decoder are distributed as follows. Over a uniformly random index to decode, the marginal distribution of the queries of the decoder is: (1) a uniformly random query in $I_1$, 
(2) a query according to a distribution which is $o(1)$-close to uniform over $I_2$, and (3) a uniformly random query in $I_3$.
Let $\delta_s$ be the fraction of locations in $I_s$ on which $w$ and $w'$ differ, so that $\frac{1}{2}|I_1|\delta_1+\frac{1}{4}|I_2|\delta_2+\frac{1}{4}|I_3|\delta_3$ is the Hamming-distance between $w$ and $w'$ and hence at most $\delta \cdot 4Kk$. As $|I_1| = 2Kk$, $|I_2|=Kk$ and $|I_3| = Kk$ we get that 
$\delta_1 + \frac{1}{4}\delta_2+\frac{1}{4}\delta_3\leq 4\delta$.
By the union bound, the probability that any of the queries of the decoder is a location in which $w$ and $w'$ differ is at most
\[
\delta_1 + \delta_2 + \delta_3 + o(1)
\leq 4(\delta_1 + \frac{1}{4}\delta_2+\frac{1}{4}\delta_3) + o(1)
\leq 16\delta + o(1).
\]
In this case, the decoder acts on $w$ as it would act on $w'$, hence it outputs the correct bit by the completeness. We conclude that the success rate is at least $1-16\delta - o(1)$.

We now analyze the soundness of the decoder. Suppose that there is a message $m \in \{0,1\}^k$ such that $$\mathsf{dist}(w, \mc{C}(m)) \leq \mu,$$
and fix an $i\in [k]$ to decode. Choosing the copy $w_0$ of the supposed encoding of $m$ under $\mc{C}_0$ from the first block uniformly, we get that, $$\E_{w_0}[\mathsf{dist}(w_0, \mc{C}_0(m))] \leq 2\mu.$$ Thus, by Markov's inequality with probability at least $9/10$ we have that the $w_0$ chosen satisfies 
\begin{equation} \label{eq: close to code} 
\mathsf{dist}(w_0, \mc{C}_0(m)) \leq 20\mu.
\end{equation}
Now the decoder, as described above, looks at the proof $\pi'$ from the second block which is supposedly the proof $\pi'_{i, m_i}(w_0)$. If one of the proof queries has the bit $m_i$ appended, then the decoder will output either $m_i$ or $\perp$, and we are done. Otherwise, they both have the bit $m_i\oplus 1$ appended. In that case, the verifier of the PCPP of $\mc{L}_{i, m_i\oplus 1}$ is run, and we lower bound its rejection probability.
By \eqref{eq: close to code} and the fact $\mc{C}_0$ has distance $1-\mu$, we get that $\mathsf{dist}(w_i, \mc{L}_{i, m_i\oplus 1}) \geq (1-\mu)-20\mu =  1-21\mu$. In this case, the decoder will output $\perp$ with probability at least $1-2 \cdot (21\mu)^{1/C}-\frac{1}{\log^{10} k}$ for some fixed constant $C > 0$, by the soundness of the PCPP.\footnote{The $\frac{1}{\log^{10} k}$ term is because the query we made to the $B$-side is only $\frac{1}{\log^{10} k}$-close to uniform in total variation distance.} In particular, for sufficiently large $k$ and sufficiently small $\mu$, the decoder outputs either the correct bit or $\perp$ with probability at least
\[
\frac{9}{10} \cdot  \left(1-2 \cdot (21\mu)^{1/C}-\frac{1}{\log^{10} k}\right) \geq \frac{2}{3},
\]
where the inequality follows by setting $\mu$ to be sufficiently small.
\end{proof}

\begin{proof}[Proof of~\cref{thm:RLDC_intro}]
   Invoke~\cref{thm:transformation_rldc} with the PCPP from~\cref{thm:near-linear-pcpp} and the codes from~\cref{lm: code GI}.
\end{proof}

\subsection{Improved PSPACE Inapproximability of CSP-Reconfiguration}
The reconfiguration inapproximability problem~\cite{HiraharaOhsaka2024,KarthikManurangsi2024,guruswami_reconfiguration} studies the complexity of moving between two given solutions to a given CSP instance, by modifying the value of only a single variable at each step, and only going through assignments that satisfy a large fraction of the constraints. More precisely:

\begin{definition}[Gap$_{1,\varepsilon}$ $q$-CSP$_\Sigma$ Reconfiguration]
Let $q \ge 1$ be an integer, $\varepsilon \in (0,1]$, and $\Sigma$ be a finite
alphabet. An instance of $\mathrm{Gap}_{1,\varepsilon}\, q$-\textrm{CSP}$_\Sigma$
Reconfiguration is specified by a $q$-CSP$_\Sigma$ instance $\Phi$ and two
solutions $\sigma,\sigma' \in \Sigma^n$. The task is to distinguish the
following two cases:
\begin{itemize}
  \item \textbf{Yes Case.} There exist assignments
  $\sigma^{(0)},\ldots,\sigma^{(t)} \in \Sigma^n$ such that
  $\sigma^{(0)}=\sigma$, $\sigma^{(t)}=\sigma'$, each
  $\sigma^{(i)},\sigma^{(i-1)}$ differ on at most one coordinate, and each
  $\sigma^{(i)}$ satisfies all constraints in $\Phi$.
  \item \textbf{No Case.} There exist no assignments
  $\sigma^{(0)},\ldots,\sigma^{(t)} \in \Sigma^n$ such that
  $\sigma^{(0)}=\sigma$, $\sigma^{(t)}=\sigma'$, each
  $\sigma^{(i)},\sigma^{(i-1)}$ differ on at most one coordinate, and each
  $\sigma^{(i)}$ satisfies at least $\varepsilon$ fraction of the constraints
  in $\Phi$.
\end{itemize}
\end{definition}
The recent work of \cite{guruswami_reconfiguration} establishes a formal connection between PCPPs and hardness of CSP reconfiguration problems, where the better the parameters of the PCPPs are, the better of a gap one obtains for the reconfiguration problem.

\begin{theorem}[Theorem 2,  \cite{guruswami_reconfiguration}]\label{thm:reconf_ven}
Suppose for some $\eta \in (0,1)$, there is an alphabet $\Sigma_0$ such that we have the following two ingredients:
\begin{itemize}
    \item Any error-correcting code $\mc{C}: \{0,1\}^{n} \to \Sigma_0^{k}$ \footnote{In \cite{guruswami_reconfiguration}, a binary code (rather than a code over a larger alphabet) is used. However, an inspection of their argument shows that they only use the fact the \emph{messages} are binary. Thus, the alphabet size of the encoding could be taken as a large constant.} with relative distance $\eta\in (0,1)$ that is encodable in space $\poly(n)$.
    \item For any constant $\delta < \eta/2$ and every circuit $\varphi: \Sigma_0^n \to \{0,1\}$ of size $N$ there is a $q$-query PCPP which is constructible in space polynomial in $|\Sigma_0|, N$, $(\delta, \varepsilon(\delta))$-soundness, randomness $O(\log N)$, alphabet size $O_{\delta, |\Sigma_0|}(1)$.
\end{itemize}
Then for all $\delta < \eta/2$, there is an alphabet $\Sigma$ of size $O_{\delta, |\Sigma_0|}(1)$, $\mathrm{Gap}_{1,1-\varepsilon(\delta)}\,(q+1)$-\textrm{CSP}$_\Sigma$ Reconfiguration is $\mathsf{PSPACE}$-hard.
\end{theorem}

Since our PCPP from \cref{thm:polynomial-size-pcpp} achieves strong relation between the proximity and rejection probability we can use it to instantiate the theorems above with it.

We can now prove ~\cref{cor:recon1,cor:recon2}. 
\begin{proof}[Proof of~\cref{cor:recon1}]
Take the code from \cref{lm: code GI} with distance greater than $1-\delta$ for some arbitrarily small $\delta$ and take the $4$-query PCPP of ~\cref{thm:polynomial-size-pcpp}, which rejects $\frac{1-\delta}{2}$-far inputs with probability $\frac{1-\delta}{2} + \eps$ for some $\eps = \eps(\delta)$ arbitrarily small relative to $\delta$. Then, the result follows by plugging these two pieces into \cref{thm:reconf_ven}.
\end{proof}

\begin{proof}[Proof of \cref{cor:recon2}]
The result follows by a simple reduction from $\mathrm{Gap}_{1,1-\eps}\,(q+1)$-\textrm{CSP}$_\Sigma$ to $\mathrm{Gap}_{1,1-\frac{\varepsilon}{q+1}}\,2$-\textrm{CSP}$_\Sigma$,
\end{proof}

\section*{Acknowledgments}
T.G.~thanks Oded Goldreich for numerous discussions about the LDC vs RLDC problem over the years. G.W.~would like to thank Alessandro Chiesa for his support. 

%% file: hadamard.tex
\section{The Final Inner dPCP (Hadamard dPCP)}\label{apx:Hadamard}

In this subsection we prove \cref{thm: hadamard-pcp}. We do so in three steps:
\begin{itemize}
    \item We present the notion of right-assignment decoding soundness, which most dPCP constructions in the literature satisfy. We show that it implies our notion of list-decoding soundness.  
    \item We recall the definition of robust dPCPs from \cite{dh}. A construction of such dPCPs yields the desired structure of the left-vertices in the inner dPCPs for composition (see \cref{thm: inner dpcp}).   
    \item Then we show that the folklore Hadamard dPCP (taken from \cite[Lemma 7.7]{bmv}) satisfies all of our requirement for the inner dPCP. 
\end{itemize}

\subsection{Right-Assignment Decoding to Left-Assignment Decoding}\label{sec:switch_side}
Below is the more commonly used list-decoding soundness for dPCPs. 
In it the list is constructed from a fixed assignment to the right side of the dPCP, as opposed to ~\cref{def:list-decoding-soundness} wherein it is constructed from a fixed assignment to the left side of the dPCP. 

\begin{definition}\label{def:right-decodinglist-decoding-soundness}
Let $\PCPDecoder$ be a decodable PCP for $\Language$. Given $\eps \in [0,1]$, $L \in \mathbb{N}$ we say that $\mc{D}$ satisfies \emph{right-assignment $(L, \eps)$-list-decoding soundness} if the following holds. For any right assignment $\Assignment_{\RightSide}: \RightSide \to \Alphabet{\RightSide}$ there exists a list $\{\Word_1, \ldots, \Word_L \} \subseteq \Language$ such that for any $\Assignment_{\LeftSide}: \LeftSide \to \Alphabet{\LeftSide}$, 
\[
\Pr_{\substack{\IndexInDecodingSet \in [n]\\ (\LeftVertex,\RightVertex)\sim \EdgeDistribution{\IndexInDecodingSet}}}[\Phi_{(\LeftVertex,\RightVertex)}(\Assignment_{\LeftSide}[\LeftVertex], \Assignment_{\RightSide}[\RightVertex])=1 \land \DecoderSymbol_{\IndexInDecodingSet}(\LeftVertex,\Assignment_{\LeftSide}[\LeftVertex]) \notin \{\Restricted{(\Word_i)}{\IndexInDecodingSet} \}_{i \in [L]}] \leq \eps.
\]
\end{definition}
In the definition above, the right-assignment $T_{B}$ is fixed and induces a list of candidate words in the language, whereas in our definition of list-decoding soundness it is the left-assignment $T_A$ that is fixed. 
In the following we show a simple transformation between the two notations: 

\begin{lemma} \label{lm: right to left}
Suppose a language $\mc{L}$ has a projection dPCP $\mc{D} = \left(A \cup B, E, \Sigma_A, \Sigma_B , \{\mc{P}_t \}_{t \in [n]},\{D_t \}_{t \in [n]} \right)$
with right-assignment-decoding $(L,\eps)$-list-decoding soundness. Then $\mc{D}$ also satisfies $(L\cdot \eps^{-1/5},2\eps^{1/5})$-list-decoding soundness for the language $\mc{L}$.
\end{lemma}

\begin{proof}
First, note that since $\mc{D}$ is a projection dPCP, for each $a\in A$ and an assignment to it $\sigma$, each neighbour $b\in\Neighbourhood{a}$ has a unique assignment that would satisfy the constraint $\Phi_{(a,b)}$, and we denote it by $(\sigma)_b$ (the vertex $a$ will always be clear from context). Let $\mc{Q}$ be the complete decoding distribution of $\mc{D}$.

Fix an assignment $T_A:A\to \Sigma_A$, and consider the following candidate list: for each $b \in B$, let ${\sf cand}(b)$ be the set of symbols $\sigma \in \mathbb{F}_q$ such that $\Pr_{a\sim \mc{Q}(\circ,\cdot,b)}[T_A[a]_b = \sigma]  \geq \eps^{1/5}$. One can find $J=1/\eps^{1/5}$ assignments  $T_1',...T_J':B\to \Sigma_B$, such that for every $b\in B$ and every $\sigma\in  \mathsf{cand}(b)$ there is some $j \in [J]$ such that $T_j' [b] = \sigma$. We use the assignments $T_1',...,T_J'$ to invoke the right-assignment list-decoding property of the dPCP with respect to fixed right assignments (and with $T_A$ being the left assignment). Thus, for each $j$ we get a list $L_j'=\{\Word_1^{j}, \ldots, \Word_L^j \} \subseteq \Language$ such that
\[
\Pr_{\substack{\IndexInDecodingSet \sim [n],\\ (\LeftVertex,\RightVertex)\sim \EdgeDistribution{\IndexInDecodingSet}}}[\Phi_{(\LeftVertex,\RightVertex)}(\Assignment_{\LeftSide}[\LeftVertex], \Assignment_{j}'[\RightVertex])=1 \land \DecoderSymbol_{\IndexInDecodingSet}(\LeftVertex,\Assignment_{\LeftSide}[\LeftVertex]) \notin \{\Restricted{(\Word_i^j)}{\IndexInDecodingSet} \}_{i \in [L]}] \leq \eps.
\]
Define the list list $L'=\cup_{j\in [J]}L_j'$. To complete the proof, we will show that for any $T_B:B\to \Sigma_B$ we have 
\[
\Pr_{\substack{\IndexInDecodingSet \sim [n],\\ (\LeftVertex,\RightVertex)\sim \EdgeDistribution{\IndexInDecodingSet}}}[\Phi_{(\LeftVertex,\RightVertex)}(\Assignment_{\LeftSide}[\LeftVertex], \Assignment_{B}[\RightVertex])=1 \land \DecoderSymbol_{\IndexInDecodingSet}(\LeftVertex,\Assignment_{\LeftSide}[\LeftVertex]) \notin \{\Word_i^j[\IndexInDecodingSet] \}_{i \in [L],j\in [J]}] \leq J\cdot \eps+\eps^{1/5}\le 2\eps^{1/5}.
\]
First, note that the probability on the left hand side can be written as
\begin{align*}
&\Pr_{\substack{\IndexInDecodingSet \sim [n],\\ (\LeftVertex,\RightVertex)\sim \EdgeDistribution{\IndexInDecodingSet}}}[\Phi_{(\LeftVertex,\RightVertex)}(\Assignment_{\LeftSide}[\LeftVertex], \Assignment_{B}[\RightVertex])=1 \land \DecoderSymbol_{\IndexInDecodingSet}(\LeftVertex,\Assignment_{\LeftSide}[\LeftVertex]) \notin \{\Word_i^j[\IndexInDecodingSet] \}_{i \in [L],j\in [J]}]\\
&=\Pr_{\substack{\IndexInDecodingSet \sim [n],\\ (\LeftVertex,\RightVertex)\sim \EdgeDistribution{\IndexInDecodingSet}}}[\Phi_{(\LeftVertex,\RightVertex)}(\Assignment_{\LeftSide}[\LeftVertex], \Assignment_{B}[\RightVertex])=1 \land \DecoderSymbol_{\IndexInDecodingSet}(\LeftVertex,\Assignment_{\LeftSide}[\LeftVertex]) \notin \{\Word_i^j[\IndexInDecodingSet] \}_{i \in [L],j\in [J]}\land T_A[a]_b\in \mathsf{cand}(b)]\\
&+\Pr_{\substack{\IndexInDecodingSet \sim [n],\\ (\LeftVertex,\RightVertex)\sim \EdgeDistribution{\IndexInDecodingSet}}}[\Phi_{(\LeftVertex,\RightVertex)}(\Assignment_{\LeftSide}[\LeftVertex], \Assignment_{B}[\RightVertex])=1 \land \DecoderSymbol_{\IndexInDecodingSet}(\LeftVertex,\Assignment_{\LeftSide}[\LeftVertex]) \notin \{\Word_i^j[\IndexInDecodingSet] \}_{i \in [L],j\in [J]}\land T_A[a]\notin \mathsf{cand}(b)].
\end{align*}
The second term is bounded by $\eps^{1/5}$ by the definition of the $\mathsf{cand}(b)$ list. For the first term, we use union bound:
\begin{align*}
&\Pr_{\substack{\IndexInDecodingSet \sim [n],\\ (\LeftVertex,\RightVertex)\sim \EdgeDistribution{\IndexInDecodingSet}}}[\Phi_{(\LeftVertex,\RightVertex)}(\Assignment_{\LeftSide}[\LeftVertex], \Assignment_{B}[\RightVertex])=1 \land \DecoderSymbol_{\IndexInDecodingSet}(\LeftVertex,\Assignment_{\LeftSide}[\LeftVertex]) \notin \{\Word_i^j[\IndexInDecodingSet] \}_{i \in [L],j\in [J]}\land T_A[a]_b\in \mathsf{cand}(b)]\\
&=\Pr_{\substack{\IndexInDecodingSet \sim [n],\\ (\LeftVertex,\RightVertex)\sim \EdgeDistribution{\IndexInDecodingSet}}}[\Phi_{(\LeftVertex,\RightVertex)}(\Assignment_{\LeftSide}[\LeftVertex], \Assignment_{B}[\RightVertex])=1 \land \DecoderSymbol_{\IndexInDecodingSet}(\LeftVertex,\Assignment_{\LeftSide}[\LeftVertex]) \notin \{\Word_i^j[\IndexInDecodingSet] \}_{i \in [L],j\in [J]}\land (\cup_{j\in [J]}T_A[a]_b=T_j'[b])]\\
&\le \sum_{j\in [J]}\Pr_{\substack{\IndexInDecodingSet \sim [n],\\ (\LeftVertex,\RightVertex)\sim \EdgeDistribution{\IndexInDecodingSet}}}[\Phi_{(\LeftVertex,\RightVertex)}(\Assignment_{\LeftSide}[\LeftVertex], \Assignment_{B}[\RightVertex])=1 \land \DecoderSymbol_{\IndexInDecodingSet}(\LeftVertex,\Assignment_{\LeftSide}[\LeftVertex]) \notin \{\Word_i^j[\IndexInDecodingSet] \}_{i \in [L],j\in [J]} \land T_A[a]_b=T_j'[b])]\\
&= \sum_{j\in [J]}\Pr_{\substack{\IndexInDecodingSet \sim [n],\\ (\LeftVertex,\RightVertex)\sim \EdgeDistribution{\IndexInDecodingSet}}}[\Phi_{(\LeftVertex,\RightVertex)}(\Assignment_{\LeftSide}[\LeftVertex], T_j'[b])=1 \land \DecoderSymbol_{\IndexInDecodingSet}(\LeftVertex,\Assignment_{\LeftSide}[\LeftVertex]) \notin \{\Word_i^j[\IndexInDecodingSet] \}_{i \in [L],j\in [J]})]\\
&\le J\cdot \eps.
\end{align*}
In this third transition we used the fact that if $\Phi_{(\LeftVertex,\RightVertex)}(\Assignment_{\LeftSide}[\LeftVertex], \Assignment_{B}[\RightVertex])=1$ and $ T_A[a]_b=T_j'[b]$, then we must have $ \Assignment_{B}[\RightVertex]=T_j'[b]$. In the last transition we used 
the right-assignment $(L,\eps)$-list-decoding soundness of $\mc{D}$.
\end{proof}

\subsection{Going from Robust dPCPs to dPCPs}
In order to present the Hadamard dPCP from prior work, we need to first present an alternative definition of dPCPs, called robust dPCPs, introduced by \cite{dh}. The main difference between robust dPCPs and our notion of dPCP is that robust dPCPs have a different notion of list-decoding soundness. 

\begin{definition}[Robust Decodable PCPs]\label{def:previous-dpcp-def}
Let $\varphi\colon \Sigma_0^n\to\{0,1\}$ be a circuit of size $N$. A PCP decoder for the language ${\sf SAT}(\varphi)$ over a proof alphabet $\Sigma$
is a probabilistic polynomial-time algorithm $\mathcal{D}$ that on an input circuit
$\varphi : \Sigma_0^n \to \{0,1\}$ of size $N$ and an index $j \in [n]$,
tosses $r = r(N)$ random coins to sample from a distribution $\mathcal{P}_j$ supported over tuples containing $(I,f)$ which are: 
\begin{itemize}
    \item[(1)] A sequence of $q = q(N)$ locations $I = (i_1,\ldots,i_q)$ in a proof of length $m(N)$ over the alphabet $\Sigma$; the parameter $q$ is referred to as the query complexity of $\mc{D}$.
    \item[(2)] A (local decoding) function $f : \Sigma^q \to \Sigma_0 \cup \{\bot\}$
    whose corresponding circuit has size at most $s(N)$, referred to henceforth as the \emph{decision complexity} of the decoder.
\end{itemize}
 For functions $\delta : \mathbb{Z}^+ \to [0,1]$ and $L : \mathbb{Z}^+ \to \mathbb{Z}^+$, we say that a PCP
decoder $\mathcal{D}$ is a dPCP for
$\mathsf{SAT}(\varphi)$ with soundness error $\delta$ and list size $L$ if the following
completeness and soundness properties hold for every circuit
$\varphi : \Sigma_0^n \to \{0,1\}$ of size $N$:
\begin{itemize}
    \item \textbf{Completeness.}
    For any $y \in \Sigma_0^n$ such that $\varphi(y) = 1$, there exists a proof $\pi \in \Sigma^m$, also called
    a decodable PCP, such that
    \[
        \Pr_{j,I,f}[f(\pi_I) = y_j] = 1,
    \]
    where $j \in [n]$ is chosen uniformly at random and $I,f$ are distributed according to $\mathcal{P}_j$ and
    the verifier's random coins.

    \item \textbf{Soundness.}
    For any $\pi \in \Sigma^m$, there is a list of $0 \le \ell \le L$ strings
    $y^1,\ldots,y^\ell$ satisfying $\varphi(y^i)=1$ for all $i$, and furthermore
    \[
        \Pr_{j\in[n],(I,f)\sim\mc{P}_j}\big[f(\pi_I) \notin \{\bot, y^1_j,\ldots,y^\ell_j\} \,\big] \le \delta.
    \]

    \item \textbf{Robust Soundness.}
    We say that $\mc{D}$ is a robust dPCP system for $\mathsf{SAT}(\varphi)$ with robust soundness
    error $\delta$ if the soundness criterion can be strengthened to the following robust condition:
    \[
        \mathbb{E}_{j\in[n],(I,f)\sim\mc{P}_j}\left[ \mathrm{agr}(\pi_I, \mathrm{BAD}(f)) \right] \le \delta,
    \]
    where $\mathrm{BAD}(f) := \{w \in \Sigma^q \mid f(w) \notin \{\bot, y^1_j,\ldots,y^\ell_j\} \}.$
\end{itemize}
\end{definition}

We have the following correspondence 
between robust dPCPs as in \cref{def:previous-dpcp-def} and our more usual notion of dPCPs.

\begin{lemma}\label{lemma:old-new-dPCP-transformation}
Suppose a language $\mc{L} \subseteq \Sigma_0^n$ has a robust dPCP with the following parameters
\begin{itemize}
        \item \textbf{Length.} $m$.
        \item \textbf{Query complexity.} $d_1$.
        \item \textbf{Proof alphabet.} $\Sigma$.
        \item \textbf{Randomness Complexity.} $r$.
        \item \textbf{Completeness.} Perfect completeness.
        \item \textbf{Soundness.} Robust soundness error $\eps$ with list size $L$.
\end{itemize}
Then $\mc{L}$ has a dPCP with the following parameters
\begin{itemize}
    \item \textbf{Length.} The left vertex set has the form $[n] \times [2^r]$ and the length is $m + n \cdot 2^r$.
    \item \textbf{Alphabets.} The left alphabet is a subset of $\Sigma^{d_1}$ and the right alphabet is $\Sigma$.
    \item \textbf{Degree.} Decoding degree is $1$.
    \item \textbf{Complete Decoding distribution.} Let $\mc{Q}$ be the complete decoding distribution and let the left and right sides be $A$ and $B$ respectively. Then the complete decoding distribution satisfies the following.
    \begin{itemize}
        \item The marginal over $A$ is uniform.
        \item For each $a \in A$, the distribution $\mc{Q}(\circ, a, \cdot)$ is uniform over $d_1$ vertices in $B$.
    \end{itemize}
    \item \textbf{Completeness.} Perfect completeness.
    \item \textbf{Soundness.} $(L \cdot \eps^{-1/5}, 2\eps^{1/5})$-list-decoding soundness.
\end{itemize}
\end{lemma}
\begin{proof} 
    Fix a robust dPCP for $\mc{L}$ as in the premise.
    We construct a dPCP $\mc{D}$ for $\mc{L}$, whose parts are 
    \[
    \mc{D} = (A \cup B, E, \Sigma_A, \Sigma_B, \{\Phi_{e}\}_{e\in E}, \{D_t \}_{t\in [n]} ,\{\mc{P}_t \}_{t \in [n]} ),
    \]
    as follows.
    The right vertices correspond to the proof locations of the robust dPCP. The left vertices are labeled as $[n] \times [2^r]$ and they correspond to ever possible pair of index $t \in [n]$ and randomness $\ell \in [2^r]$ when choosing the queries given index $t$. By the randomness complexity of the robust dPCP, each query set corresponds to one of $2^r$ random strings. The edges $e$ correspond to all pair $(t, \ell) \in A$ and $b \in B$ such that the decoder of the robust dPCP given index $t$ and randomness $\ell$ queries the location corresponding to $b$ in the proof. The left alphabet is $\Sigma^{d_1}$. For the vertex $(t, \ell) \in A$, we think of an alphabet symbol as corresponding to answer to the $d_1$ queries made by the robust dPCP given index $t$ and randomness $\ell$. It also corresponds to a labeling of its $d_1$ neighbors on the right side.
    We constrain the alphabet to only consist of values which would have made the robust dPCP verifier accept. The right alphabet is $\Sigma$, and each constraint between $(t, \ell) \in A$ and $b \in B$ checks that both assignments agree on the value assigned to $b$. The decoder is the same as in the robust dPCP. That is, it thinks of the symbol assigned to $(t, \ell) \in A$ as the $d_1$ answers to its queries in the robust dPCP and decodes accordingly. Finally, for each $t \in [n]$, the decoding distribution $\mc{P}_t$ chooses $\ell \in [n]$ randomly, chooses one of $d_1$ many $b \in B$ adjacent to $(t, \ell)$ uniformly at random, and outputs $(t, (t, \ell), b)$.

    It is straightforward to check that $\mc{D}$ has the desired length, alphabets, decoding degree, and properties for the complete decoding distribution. For perfect completeness, for any $w \in \mc{L}$, using the perfect completeness of the robust dPCP we get that there is an assignment to $B$ such that the robust dPCP always accepts and decodes to $w_t$. Then, we assign the left side according to the assignment to $B$. That is, for each $(t, \ell) \in A$, the assignment to $(t, \ell)$ consists of the $d_1$ symbols from $\Sigma$ assigned to its $d_1$ neighbors. It is straightforward to check that all of the constraints of $\mc{D}$ are satisfied and the decoder for $\mc{D}$ decodes $w_t$ with probability $1$ for all $t \in [n]$. Finally, it is straightforward to see that the definition of robust soundness corresponds directly to right-assignment $(L, \eps)$-list-decoding soundness. Then, the desired list decoding soundness follows by applying \cref{lm: right to left}.
\end{proof}

\subsection{The Hadamard dPCP}

We are now ready to present the Hadamard dPCP and prove \cref{thm: hadamard-pcp}. The result follows by plugging the robust (Hadamard) dPCP presented in \cite{bmv} into the transformation of \cref{lemma:old-new-dPCP-transformation}.

\begin{lemma}[Lemma 7.7 from \cite{bmv}]\label{lemma:bmv hadamard} For every $T,N \in \mathbb{N}$ and $\eps > 0$ the following holds. Let $\varphi\colon \Sigma_0^n\to\{0,1\}$ be a circuit of size at most $N$ over an alphabet of size at most $T$. There exists $q = 1/\eps^{O(1)}$ such that the language $\mathsf{SAT}(\varphi)$ has a robust dPCP (according to \cref{def:previous-dpcp-def}) with the following parameters:

\begin{enumerate}
    \item \textbf{Length.} $m(N) = q^{O(N^2)}$.
    \item \textbf{Alphabet Size.} $q$.
    \item \textbf{Randomness complexity.} $r(N) = O(N^2 \log(q))$.
    \item \textbf{Query Complexity.} $q^{O(\log |\Sigma_0|)}$.
    \item \textbf{Soundness.} Robust soundness $\eps$ with list size $1/\eps^{O(1)}$.
\end{enumerate}
\end{lemma}

\begin{proof}[Proof of \cref{thm: hadamard-pcp}]
The result follows by plugging  the robust dPCP of \cref{lemma:bmv hadamard} into the transformation of \cref{lemma:old-new-dPCP-transformation}. The projection decision complexity and decoding complexity follow from the fact that each can be implemented as a circuit which outputs $0$ or $1$ and has an input consisting of at most two symbols from alphabets of size at most $O_{\eps, T}(1)$. 
Finally, the complete decoding distribution condition follows from the complete decoding distribution guarantee of \cref{lemma:old-new-dPCP-transformation}.
\end{proof}